\documentclass[aps,prb,twocolumn,superscriptaddress,floatfix,showpacs,10pt,letterpaper]{revtex4-1}

\usepackage{graphicx}
\usepackage{graphics} 
\usepackage{ifpdf}
\ifpdf
\usepackage{epstopdf}
\usepackage{hyperref}
\else
\usepackage[hypertex]{hyperref}
\fi
\usepackage{xypic}

\input xypic
\usepackage[all,tips]{xy}

\allowdisplaybreaks[1]

\newlength{\fighskip} \fighskip=2pt
\newlength{\figvskip} \figvskip=3pt

\newcommand*{\figbox}[2]{{
  \def\figscale{#1}
  \def\arraystretch{0.8}
  \arraycolsep=0pt
  \begin{array}{c}
    \vbox{\vskip\figscale\figvskip
      \hbox{\hskip\figscale\fighskip
        \includegraphics[scale=\figscale]{#2}}}
  \end{array}}}

\newcommand\void[1]       {}

\pdfoutput=1

\newcommand{\A}[2]{A_{\ v_{#2}}^{v_{#1}}}
\newcommand{\Amm}[2]{A^{\ v_{#2}}_{v_{#1}}}

\newcommand{\As}[3]{{A_{#3}}_{v_{#1} v_{#2}}}

\newcommand{\Aw}[2]{A_{v_{#1} v_{#2}}^{e_{#1#2}}} 

\newcommand{\B}[3]{B^{e_{#1#2} e_{#2#3}}_{v_{#1} v_{#2} v_{#3}; e_{#1#3} }}
\newcommand{\Bmm}[3]{B_{e_{#1#2} e_{#2#3}}^{v_{#1} v_{#2} v_{#3}; e_{#1#3} }}

\newcommand{\Bs}[4]{{B_{#4}}^{e_{#1#2} e_{#2#3} e_{#1#3}}_{v_{#1} v_{#2} v_{#3}}}

\newcommand{\C}[4]{C^{e_{#1#2} e_{#1#3} e_{#1#4} e_{#2#3} e_{#2#4} e_{#3#4};\phi_{#1#2#3}\phi_{#1#3#4}}_{v_{#1} v_{#2} v_{#3} v_{#4}; \phi_{#1#2#4}\phi_{#2#3#4}}}

\newcommand{\Cmm}[4]{C_{e_{#1#2} e_{#1#3} e_{#1#4} e_{#2#3} e_{#2#4} e_{#3#4};\phi_{#1#2#3}\phi_{#1#3#4}}^{v_{#1} v_{#2} v_{#3} v_{#4}; \phi_{#1#2#4}\phi_{#2#3#4}}}

\newcommand{\Cs}[5]{{C_{#5}}^{e_{#1#2} e_{#1#3} e_{#1#4} e_{#2#3} e_{#2#4} e_{#3#4}}_{v_{#1} v_{#2} v_{#3} v_{#4}; \phi_{#2#3#4}\phi_{#1#2#4}\phi_{#1#3#4}\phi_{#1#2#3}}}

\newcommand\hBF[1]           {BF$^\text{H}_{#1}$}     
\newcommand\lBF[1]           {BF$^\text{L}_{#1}$}     

\newcommand\BF           {\mathrm{BF}}     
\newcommand\MBF        {\mathrm{MBF}}
\newcommand\TO           {\mathrm{TO}}     
\newcommand\aTO           {\mathrm{aTO}}     

\newcommand\id            {\mathrm{id}}
\newcommand\op          {\mathrm{op}}
\newcommand\ob          {\mathrm{Ob}}
\newcommand\ev          {\mathrm{ev}}
\newcommand\coev      {\mathrm{coev}}
\newcommand\End    {\mathrm{End}}
\newcommand\aut      {\mathcal{A}\mathrm{ut}}

\newcommand\hilb   {\EuScript{H}\mathrm{ilb}}
\newcommand\vect    {\EuScript{V}\mathrm{ect}}
\newcommand\rep     {\EuScript{R}\mathrm{ep}}
\newcommand\fun     {\EuScript{F}\mathrm{un}}

\newcommand\Set    {\EuScript{S}\mathrm{et}}
\newcommand\TC    {\mathrm{TC}}
\newcommand\fact   {\mathrm{Fact}}
\newcommand\lw     {\mathrm{LW}}
\newcommand\prebf  {\mathrm{BF}^{pre}}
\newcommand\prembf {\mathrm{MBF}^{pre}}
\newcommand\core   {\mathrm{core}}

\newcommand{\bulk}   {\underline{\sl bulk} }

\usepackage{amsthm}
\usepackage{euscript}
\usepackage{epstopdf} % for pdflatex

\newtheoremstyle{wenthm}% name of the style to be used
  {3pt}% measure of space to leave above the theorem. E.g.: 3pt
  {3pt}% measure of space to leave below the theorem. E.g.: 3pt
  {\slshape}% name of font to use in the body of the theorem
  {}% measure of space to indent
  {\bfseries}% name of head font
  {:}% punctuation between head and body
  {.5em}% space after theorem head; " " = normal interword space
  {}% Manually specify head

\theoremstyle{wenthm}

\newtheorem{thm}{Theorem}
\newtheorem{prop}{Proposition}
\newtheorem{cor}{Corollary}
\newtheorem{lemma}{Lemma}

\newtheorem{conj}{Conjecture}

\theoremstyle{definition}

\newtheorem{defn}{Definition}
\newtheorem{rema}{Remark}
\newtheorem{expl}{Example}

\newtheorem{app}{Application}

\newcommand\nn             {\nonumber \\}

\newcommand\bea           {\begin{eqnarray}}
\newcommand\eea         {\end{eqnarray}}
\newcommand\bnu          {\begin{enumerate}}
\newcommand\enu          {\end{enumerate}}

\newcommand\bearll        {\begin{array}{ll}\displaystyle}

\newcommand\eear          {\end{array}}

\newcommand{\pf}{\begin{proof}}
\newcommand{\epf}{\end{proof}}

\newcommand\Zb            {\mathbb{Z}}

\newcommand\EA           {\EuScript{A}}
\newcommand\EB           {\EuScript{B}}
\newcommand\EC           {\EuScript{C}}
\newcommand\ED           {\EuScript{D}}
\newcommand\EE          {\EuScript{E}}
\newcommand\EF          {\EuScript{F}}

\newcommand\EL          {\EuScript{L}}
\newcommand\EM          {\EuScript{M}}
\newcommand\EN         {\EuScript{N}}

\newcommand\ER         {\EuScript{R}}

\newcommand\cM          {\mathcal{M}}

\newcommand\rH            {\mathrm{H}}

\newcommand\rw            {\mathrm{w}}
\newcommand\ksim        {\overset{k}{\sim}}
\newcommand\rwsim      {\overset{\mathrm{w}}{\sim}}
\newcommand\gwsim     {\overset{\mathrm{gw}}{\sim}}

\begin{document}

\begin{titlepage}

\title{Braided fusion categories, gravitational anomalies, and\\
the mathematical framework for topological orders in any dimensions}

\author{Liang Kong} 
\affiliation{Department of Mathematics \& Statistics, University of New Hampshire, Durham, NH, 03824, USA}
\affiliation{Institute for Advanced Study, Tsinghua University, Beijing, 100084, P. R. China}
\author{Xiao-Gang Wen} 
\affiliation{Perimeter Institute for Theoretical Physics, Waterloo, Ontario, N2L 2Y5 Canada} 
\affiliation{Department of Physics, Massachusetts Institute of Technology, Cambridge, Massachusetts 02139, USA}

\begin{abstract} 
Topological order describes a new kind of order in gapped quantum liquid states
of matter that correspond to patterns of long-range entanglement, while
gravitational anomaly describes the obstruction that a seemingly consistent low
energy effective theory cannot be realized by any well defined quantum model in
the same dimension.  Amazingly, topological order and gravitational anomaly
have a very direct relation: gravitational anomalies can be realized on the
boundary of topologically ordered states in one higher dimension and are
described by topological orders in one higher dimension.  
%Measuring the gravitational anomaly on the boundary of a topologically ordered
%state provide a universal way to probe topological order.  
In this paper, we try to develop a general theory for topological order and
gravitational anomaly in any dimensions.  (1) We introduce the notion of BF
category to describe the braiding and fusion properties of topological
excitations that can be point-like, string-like, etc. A subset of BF categories
-- closed BF categories -- classify topological orders in any dimensions, while
generic BF categories classify (potentially) anomalous topological orders that
can appear at a boundary of a gapped quantum liquid in one higher dimension.
(2) We introduce topological path integral based on tensor network to realize
those topological orders.  (3) Bosonic topological orders have an important
topological invariant: the vector bundles of the degenerate ground states over
the moduli spaces of closed spaces with different metrics. They may fully
characterize topological orders.  (4) We conjecture that a topological order
has a gappable boundary iff the above mentioned vector bundles are flat.  (5)
We find a holographic phenomenon that every topological order with a gappable boundary can be uniquely determined by
the knowledge of the boundary. As a consequence, 
BF categories in different dimensions form a (monoid) cochain
complex, that reveals the structure and relation of topological orders and
gravitational anomalies in different dimensions. We also studied the
simplest kind of bosonic topological orders that have no non-trivial topological
excitations.  We find that this kind of
topological orders form a $\Zb$ class in 2+1D (with gapless edge), a $\Zb_2$ class in 4+1D (with gappable boundary),
and a $\Zb\oplus \Zb$ class in 6+1D  (with gapless boundary).

\end{abstract}


\maketitle

\end{titlepage}

{\small \setcounter{tocdepth}{1} \tableofcontents }

\section{Introduction of topological order}
\label{intro}

In 1989, through a theoretical study of chiral spin
liquid,\cite{KL8795,WWZ8913} we realized that there exists a new kind of order
-- topological order\cite{Wtop,WNtop,Wrig} -- beyond Landau symmetry breaking
theory.  Topological order cannot be characterized by the local order
parameters associated with the symmetry breaking.  However, topological order
can be characterized/defined by (a) the topology-dependent ground state
degeneracy\cite{Wtop,WNtop} and (b) the non-Abelian geometric phases of the
degenerate ground states\cite{Wrig,KW9327}, where both of them are \emph{robust
against any local perturbations} that can break any symmetries.\cite{WNtop}
This is just like superfluid order is characterized/defined by zero-viscosity
and quantized vorticity that are robust against any local perturbations that
preserve the $U(1)$ symmetry.  

We know that, microscopically, superfluid order is originated from boson or
fermion-pair condensation.  Then, what is the  microscopic origin of
topological order?  What is  the  microscopic origin of robustness against
\emph{any} local perturbations?  Recently, it was found that, microscopically,
topological order is related to long-range entanglement.\cite{LW0605,KP0604} In
fact, we can regard topological order as pattern of long-range entanglement in
many-body ground states,\cite{CGW1038} which is defined as the equivalent
classes of stable gapped  quantum liquid\cite{ZW14} states under local unitary
transformations.\cite{LWstrnet,VCL0501,V0705}
% or the equivalent classes of gapped quantum liquid states under local
% invertible transformations.\cite{ZW14} 
The notion of topological orders and many-body
quantum entanglement leads to a new point of view of quantum phases and quantum
phase transitions:\cite{CGW1038} for bosonic Hamiltonian quantum systems
without any symmetry, their gapped quantum liquid phases\cite{ZW14} can be
divided into two classes: short-range entangled (SRE) states and long-range
entangled (LRE) states.

SRE states are states that can be transformed into tensor-product states via
local unitary transformations. All SRE states belong to the same phase.  LRE
states are states that cannot be transformed into  tensor-product states via
local unitary transformations. LRE states can belong to different quantum
phases, which are nothing but the topologically ordered phases.  Chiral spin
liquids\cite{KL8795,WWZ8913}, integral/fractional quantum Hall
states\cite{KDP8094,TSG8259,L8395}, $\Zb_2$ spin
liquids\cite{RS9173,W9164,MS0181}, non-Abelian fractional quantum Hall
states\cite{MR9162,W9102,WES8776,RMM0899}, etc., are examples of topologically
ordered phases.

Topological order and long-range entanglement are truly new phenomena. They
require new mathematical language to describe them.  It appears that tensor
category theory\cite{FNS0428,LWstrnet,CGW1038,GWW1017,KK1251,GWW1332} and
simple current algebra\cite{MR9162,BW9215,WW9455,LWW1024} (or pattern of zeros
\cite{WW0808,WW0809,BW0932,SL0604,BKW0608,SY0802,BH0802,BH0802a,BH0882}) may be
part of the new  mathematical language.  Using tensor category theory, we have
developed a systematic and quantitative theory for topological orders with
gapped edge for 2+1D interacting boson and fermion
systems.\cite{LWstrnet,CGW1038,GWW1017,GWW1332} For 2+1D topological orders
(with gapped or gapless edge) that have only Abelian statistics, we find that
we can use integer $K$-matrices to classify them and use the following $U(1)$
Chern-Simons theory to describe
them\cite{BW9045,R9002,FK9169,WZ9290,BM0535,KS1193}
\begin{align}
\label{csK}
 {\cal L}= \frac{K_{IJ}}{4\pi} a_{I\mu} \prt_\nu a_{J\la}\eps^{\mu\nu\la} .
\end{align}

\section{A summary of main results/conjectures}

In this paper, we try to develop a general theory for topological order and
gravitational anomaly for local bosonic quantum systems in any
dimensions.\cite{W1313}  We would like to consider the following basic issues:
\begin{enumerate}
\item
How to classify topological orders in any dimensions.  (Previous works have
classified topological orders in 1+1 space-time dimensions: there is no
nontrivial topological order in 1+1D.\cite{VCL0501,CGW1107}.  The 2+1D
topological orders with gappable boundary are classified by unitary fusion
category.\cite{LWstrnet,CGW1038,KK1251, kong-anyon,LW1384}) 
%Can we obtain a list of topological orders in various dimensions?  
\item
How to classify anomalous topological orders that can only appear on a boundary
of a gapped system, but cannot be realized by any well-defined system in the
same dimension?  (In this paper, we follow the tradition to use the term
``topological order'' to mean anomaly-free topological order.)
\item
Given a low energy effective theory
 (for example given the data that describes the fusion
and braiding of topological excitations), how to determine if it is anomalous or
anomaly-free? Can we realize a given set of fusion and braiding properties 
by a well defined model in the same dimension?
\item
Given an anomaly-free topological order, how to determine if its boundary can
be gapped or not? (See \Ref{KS1193,WW1263,L1355,K1354} for discussions about
the gappable boundary of 2+1D Abelian topological orders and \Ref{kong-anyon}
for general 2+1D topological orders.)
\item
Given two bosonic Hamiltonians, how to determine if their ground states have
the same topological order or
not?\cite{Wrig,KW9327,W1221,ZGT1251,TZQ1251,ZMP1233,CV1308,HMW1457}
\end{enumerate}
In this paper, we try to address the above issues.  Let us first summarize the
main results of this paper. They include
\begin{enumerate}
\item We define BF category, closed BF category, and exact BF category in
various dimensions based on higher category theory. We will explain why the
structures of an higher category automatically encode the information of the
fusion and braiding of topological excitations which can be point-like,
string-like, etc (see Sections \ref{topexc} and \ref{sec:math-def}). These
definitions are based on many intuitive physical consideration, and are
conjectural and incomplete.  The precise definition is not important to us at
the current stage.  What is important to us is the general framework we
provide, and how conjectures and physically important questions can be
formulated in this framework.  These conjectures and questions will serve as a
blueprint for future studies. 

We argue that the closed BF categories (defined in higher category theory)
classify anomaly-free topological orders (defined in many-body wave functions). 
%(up to certain unquantized gravitational Chern-Simons terms which appear in
%$4k+3$ space-time dimensions).  
The exact BF categories classify topological orders with gappable boundary, and
the BF categories classify all potentially anomalous topological orders
(that can be realized on the boundary of well defined quantum models in one
higher dimension), except a small class of anomalous topological orders
described by unquantized gravitational Chern-Simons terms in $4k+3$ space-time
dimensions.  See Section \ref{ceBF}.
\item 
We show that the above three kinds of BF categories in different dimensions
form a (commutative-monoid-valued) cochain complex (just like closed and exact
differential forms form a cochain complex). 
%The cohomology classes of the cochain complex describe the types of gapless
%boundaries.
As a result, all the topological orders with gappable boundary are fully
characterized (in a many-to-one fashion) and classified by anomalous
topological orders in one lower dimension.  
\item The perturbative and global gravitational anomalies (except those
described by unquantized gravitational Chern-Simons terms in $4k+3$ space-time 
dimensions) 
%in space-time dimensions lower than 7 
are classified by closed BF categories in one higher dimension (see Section
\ref{ganm}).
\item We develop a tensor network approach that produce a large class of exact
topological orders \emph{in any
dimensions}.  We also use tensor networks in one higher dimension to produce (a
large class of) topological orders and anomalous topological orders in any
dimensions (see Section \ref{TNappr}). 
%\item We conjecture that anomaly-free topological orders (\ie the closed BF
%categories) in bosonic Hamiltonian quantum theory can be fully
%characterized/probed by the projective representations of the homeomorphism
%groups of the space with various topologies
%(see Section \ref{PMtopo}). 
\item
As an application of the developed theory, we studied the simplest bosonic
topological orders that have no non-trivial topological excitations.  We find
that this kind of topological orders form a $\Zb$ class in 2+1D, a $\Zb_2$
class in 4+1D,\cite{K1467,K1459} and a $\Zb\oplus \Zb$ class in 6+1D (see
Section \ref{invTop}).\cite{FT1292,F14,freed2014} The boundary of $\Zb$-class
topological orders must be gapless (with perturbative gravitational anomalies),
while the boundary of $\Zb_2$-class topological orders can be gapped. But such
a gapped boundary must be topological, which contains non-trivial topological
excitations and has global gravitational anomalies.
\item
As another application, we show that, for a 2+1D bosonic topological order, the
chiral central charge $c$ of the edge state must satisfy $c D_g/2 \in \Z$ for
$g>2$, where $D_g$ is the ground state degeneracy on genus $g$ surface.  
%For  a 2+1D fermionic topological order, we have $2c D_g \in \Z$ for $g>2$.
\end{enumerate}

The above main results are built upon many new concepts and results.
In the following, we will summarize them in detail.

\subsection{Braided fusion category}

In this paper, we will only consider local (short-range interacting) bosonic
quantum systems with a finite gap.  
To develop a theory of topological order in
$n+1$-dimensional space-time, we assume that such \emph{a topologically ordered
phase (a gapped phase) is characterized by the gravitational responses, as well
as the topological properties of its topological excitations of spatial
dimension $p$ for $0\leq p \leq n-1$ (such as particle-like, string-like, and
membrane-like excitations).} The gravitational responses includes the thermal
Hall effect in 2+1D (which is related to the chiral central charge of the edge
states).  The topological properties of the topological excitations include
their fusion and braiding properties.  The collection of all those topological
properties defines a categorical notion, which generalizes the usual
mathematical notion of braided tensor category and will be called a
\textbf{BF$_{n+1}$ category} (see Section \ref{topexc}), where ``B" stands for
``braiding", ``F" stands for fusion and the subscript always means the
space-time dimension.  In physics, the term ``BF category'' is synonymous to
``gapped effective theory''.  
%In this paper, we will mainly use the term ``BF category'' to stress that we
%want to use the data for the fusion and the braiding to characterize the
%gapped effective theories.  
In other word, a ``gapped effective theory'' is really a collection of data
that describes the fusion and the braiding of topological excitations.

The main result of this paper is to develop an mathematical definition of
BF$_{n+1}$ category, which will allow us to develop a general theory for
topological order and gravitational anomaly (perturbative and global) in any
dimensions. We will first try to define BF$_{n+1}$ category physically, trying
to bring in relevant concepts for the definition  (see Section
\ref{sec:univ-prop-hbfcat}).  Then we will define BF$_{n+1}$ category
mathematically using the $n$-category theory\cite{bd,baez,K1007,KT1321} (see
Section \ref{sec:math-def}). 

To have a simple understanding of the mathematical definition of BF$_{n+1}$
category, we can start with a class of 0-categories -- Hilbert spaces. A
1-category is a category enriched by 0-categories. Namely, it has a set of objects and 
a hom space $\hom(a,b)$ (or a space of arrows $\{ a \to b \}$) for each ordered pair of objects $(a, b)$, 
and each hom space is a 0-category, i.e. a Hilbert space. A 1-category with only one
object $\ast$ can describe a 0+1D quantum system, and the space of morphisms
$\hom(\ast,\ast)$ is the local operator (observables) algebra of the quantum
system (see Section \ref{0dTO}). An morphism in $\hom(a,b)$ (or an arrow $a\to b$) 
can also be viewed as a defect in the time direction, i.e. an instanton. 
A 2-category is a category enriched by
1-categories. More precisely, a 2-category consists of a set of objects (or 0-morphisms), a set of 1-morphisms $\{ a \to b \}$ from object $a$ to object $b$ and a set of 2-morphisms $\{ x\Rightarrow y\}$ for 1-morphisms $x,y:a \to b$. The full hom space $\hom(a,b)$ between $a$ and $b$, consisting all 1-morphisms from $a$ to $b$ and all 2-morphisms between these 1-morphisms, form a 1-category. 
A 2-category with one object $\ast$ and additional assumptions on unitarity
describes a 1+1D gapped quantum systems.  It contains the information of the
fusion of point-like excitations in the 1+1D systems.  The point-like
excitations are described by the 1-morphisms from $\ast \to \ast$ and the fusion of point-like excitations are described by composition of arrows (see Section\ref{1dTO}).  
2-morphisms are fusion/splitting channels of the point-like excitations and can also be viewed as instantons as they are defects in the time direction. 
A 3-category is a category enriched by 2-categories. A 3-category
with one object and additional assumptions describes a 2+1D gapped quantum
system.  It contains the information of the fusion and the braiding of string-like excitations (1-morphisms), point-like excitations (2-morphisms) and instantons (3-morphisms) in the 2+1D systems (see Section \ref{preM}).  
More generally, the notion of an $(n+1)$-category automatically
includes the fusion and braiding structures of excitations of all codimensions in an $(n+1)$ space-time dimensional system (see Section\,\ref{sec:fb-in-ncat}). 
One does not need to mention fusion and braiding at all. 

In addition to BF$_{n+1}$ category, we will also introduce the notion of a
$\prebf_{n+1}$-category, which can describe (not in a minimal way) many
interesting constructions of topological orders from concrete models, and that
of a $\prembf_{n+1}$-category, which can describe multiple phases connected by
gapped domain walls, including the gapped boundary cases. 

%Physically speaking, a  mathematical definition of BF$_{n+1}$ category defines
%a topological order by using only ``measurable'' quantities (also known as
%topological invariants). Here ``measurable'' means measurable at least in
%numerical calculations. On the other hand, a physical ``definition'' is more
%constructive.  The definition is formulated via various constructions.

\medskip \noindent {\bf Terminology of dimensions}: Both space-time dimensions
and spatial dimension will be used in this work. To avoid confusion, we will
always try to make it clear which one we mean. In general, by a $p$-dimensional
topological excitation or defect, we always mean the spatial dimension; by an
$n$-dimensional topological order, we always mean the space-time dimension.  To
avoid confusion, we will also use the term: an $l$-codimensional excitation.
The subscript $n$ in $\BF_n$-category always means the space-time dimension.
Sometimes we will use $n+1$ for space-time dimensions instead of $n$ for the
obvious reason. 

\subsection{Simple, composite, and elementary topological excitations}
\label{simpcomp}

The excitations above a topologically ordered ground state play a key role in
developing a definition  $\BF_n$-category (and topological order).  To use
those excitations to define a $\BF_n$-category, we introduced the notion of
\emph{topological} excitations which can be point-like, string-like,
membrane-like, etc.  We discussed the notions of simple and composite
topological excitations (see Section \ref{topexc})  We also introduced the
notion of elementary topological excitations (see Section \ref{unbele}.) 

We conjecture that topological orders are fully determined via the fusion and
braiding properties of the elementary topological excitations alone (plus the
gravitational responses). (See Section \ref{sec:math-def}.) This allows us to
identify a special type of $n$-categories -- $\BF_n$-categories -- that
describe/define the topological orders.

\subsection{Stacking operation and tensor product $\boxtimes$} 
\label{tprod}

Let us use $\TO_n$ to denote a (anomaly-free) topologically ordered phase and
$\aTO_n$ to denote a potentially anomalous topologically ordered phase, in
$n$-dimensional space-time.  Clearly, the set of (anomaly-free) topologically
ordered phases $\{\TO_n\}$ is a subset of potential anomalous topologically
ordered phases $\{\aTO_n\}$.  Both sets $\{\TO_n\}$ and $\{\aTO_n\}$ admit a
multiplication operation: we can stack two physical systems that realize two
topological orders ($\EC_n^1$ and $\EC_n^2$) to obtain a double layer system
that realizes another topological order.  Such a stacking operation is a
symmetric tensor product $\boxtimes$: $\EC_n^1\boxtimes \EC_n^2=\ED_n$.  In
general, a topological order may not have an inverse.  So the two sets
$\{\TO_n\}$ and $\{\aTO_n\}$, with the stacking $\boxtimes$, form commutative
monoids. (A monoid is like a group except that some elements may not have
inverse.  See Section \ref{mnd}).

We like to point out that some topological orders do have an inverse under the
stacking operation $\boxtimes$, which are called invertible.\cite{FT1292,F14}
The collection of all invertible topological orders form a group under
the stacking $\boxtimes$.

%The invertible anomaly-free
%topological orders are simplest
%kind of topological orders that have no non-trivial bulk topological
%excitations and no topological ground state degeneracy on any closed spaces.
%But those states are still topological since their boundary must be
%gapless\cite{Wedge,Wtoprev} or topological.
% and the
%thermal Hall conductivity is quantized and non zero\cite{KF9732} (see Section
%\ref{invTop}). 

An example of 
%such invertible anomaly-free topological states is the $E_8$ quantum Hall
%state $\EC_3^{E_8}$ in 2+1D described in Example \ref{E8} which generate an
%Abelian group $\Zb$.  We also have an  example of 
invertible \emph{anomalous} topological orders, is described by an effective
theory given by a gravitational Chern-Simons 3-form with a \emph{unquantized}
coefficient $\ka_{gCS}$.  $\ka_{gCS}$ generates a unquantized thermal Hall
conductivity.\cite{KF9732,HLP1242} We denote such anomalous topological orders
as $\aTO_3^{\ka_{gCS}}$ and call them gCS anomalous topological orders.  Under
the stacking, we have $\aTO_3^{\ka_{gCS}}\boxtimes \aTO_3^{\t
\ka_{gCS}}=\aTO_3^{\ka_{gCS}+\t \ka_{gCS}}$.  So such invertible anomalous
topological orders form an Abelian group isomorphic to the real numbers.  \

The gCS anomalous topological orders only appear in $4k+3$ space-time
dimensions.  They are all described by  gravitational Chern-Simons forms which
exist only in $4k+3$ space-time dimensions.  It is not entirely clear to us how
to include such gCS anomalous topological orders in our BF category approach.
%(for a possible approach see Remark\,\ref{rema:infty-cat}).  
So in this paper,
we will take a quotient. More precisely, we will use the term ``anomalous
topological orders'' to refer to the quotient $\{\aTO_n\}/\{\text{gCS anomalous
topological orders}\}$, which is also a monoid. The set of BF categories,
defined as higher categories, form a monoid as well. We conjecture that the two
monoids are isomorphic:
\begin{align}
\label{topoBF}
&\ \ \ \ 
 \{\text{BF categories}\} 
\nonumber\\
&\simeq  \frac{\{\text{potentially anomalous topological orders}\}}{\{\text{gCS anomalous topological orders}\}}.
\end{align}
This is a key expression of this paper. It relates a mathematical construction
(BF category) to a physical phenomenon (topological order on a boundary). 

%Since the quotient is also a monoid, 
%the set of BF categories in the same
%dimension also form a commutative monoid under $\boxtimes$, with the unit given
%by the trivial phase, which is denoted as $\one$. 

Group structure can also be recovered on the certain quotient of the monoid of
the $\BF_n$ categories.  We introduce two equivalence relations $\sim$ and
$\rwsim$ between two $\BF_n$ categories.  Two $\BF_n$ categories $\EC_n$ and
$\ED_n$ are called quasi-equivalent if $\EC_n \sim \ED_n$ and Witt equivalent
if $\EC_n \rwsim \ED_n$. Witt equivalence $\EC_n \rwsim \ED_n$ means that two
corresponding phases $\EC_n$ and $\ED_n$ can be connected by a gapped domain
wall.  If the domain wall is not only gapped, but its topological excitations
also all come from the dimension reduction of the topological excitations in an
$n$-dimensional BF categories, 
%$\EC_n$ and $\ED_n$, 
then we say  $\EC_n \rwsim \ED_n$. 
%whose equivalence classes describe types of non-trivial boundaries of the
%topological states.
Under $\boxtimes$ operation, the equivalence classes of BF$_{n}$ categories
under  $\sim$ or $\rwsim$ form Abelian groups (see Section \ref{wgrp}).

\subsection{Two versions of quantum theories}

We point out that local bosonic quantum theory has two versions: Hamiltonian
version and the Lagrangian version. The two versions are really different
theories.  The  Hamiltonian version will be called {\it local bosonic
Hamiltonian quantum} (lbH) {\it theory} which are described by lattice bosonic
Hamiltonian with short range interactions.  The Lagrangian version will be
referred to as {\it local bosonic Lagrangian quantum} (lbL) {\it theory}, which
are described by local bosonic path integral with short range interactions (see
Appendix \ref{path}).  As a result, there are two version of BF categories,
which will be referred to as H-type \hBF{n} category (for lbH theory) and
L-type \lBF{n} category (for lbL theory).

In this paper, \emph{topological order} and \emph{long-range entanglement}
belong to the Hamiltonian version of quantum theory and are described by
\hBF{n} category.  While topological quantum field theories (TQFT) studied in
high energy physics and mathematics mostly belong to the Lagrangian version of
quantum theory and are associated with \lBF{n}  category.  
%In this paper, we will mainly discuss \hBF{n} categories. 

%These TQFTs (based on Lagrangian) cannot precisely describe topological order
%and long-range entanglement (see Sections \ref{ceBF}, \ref{ganm}, and
%\ref{WWmdl}).  We need a modified TQFT based on the lbH theory to describe
%topological order and long-range entanglement. In this paper, we developed
%such a theory (see Sections \ref{TNeBF} and \ref{TNcBF}).  

%We note that a closed L-type \lBF{} category (described by a TQFT based on
%Lagrangian) is always a closed H-type \hBF{} category (described by lbH
%theory).  So the L-type topological orders (closed \lBF{n}  categories) form a
%subset of H-type topological orders (closed \hBF{n} categories). 

%The BF categories mention before can be either \hBF{n} categories and \lBF{}
%categories. The corresponding results remain to be valid for both cases.

%Similarly, the equivalence classes of BF categories in different dimensions
%form a cochain complex, with $\cZ_n(\cdot)$ acting like the ``differential''
%operator.

\subsection{The boundary-bulk relation}

Not all the locally consistent sets of topological properties (\ie not
all \hBF{n} categories or not all gapped effective theories) can be realized by
lattice qubit models (\ie the lbH systems) in the same dimension (see Section
\ref{ganm}).  Those \hBF{n} categories, which are not realizable by lattice
models in the same dimension, are called \textit{anomalous} \hBF{n} categories.
We argue that (1) a generic \hBF{n} category $\EC_n$ (or a potentially
anomalous gapped effective theory) in $n$-dimensional space-time can always be
realized by a boundary of a lattice qubit model in $(n+1)$-dimensional
space-time whose bulk realizes another \hBF{n+1} category $\EC_{n+1}$.  (2)
$\EC_{n+1}$ is uniquely determined by $\EC_n$. Therefore, we introduce the
notion of the \bulk of $\EC_n$ (see Definition\,\ref{def:bulk}), denoted by
$\cZ_n(\EC_n)$ and defined by $\cZ_n(\EC_n):=\EC_{n+1}$ (see
Lemma\,\ref{lemma:unique-bulk}). Clearly, under such a definition, the \bulk of
a \bulk is trivial: $\cZ_{n+1}(\cZ_n(\EC_n))=\one_{n+2}$ (see Section
\ref{cfun}). 

%XGW:
Physically, an topological order in $n+1$-dimensional space-time is defined as
an equivalent class of many-body wave functions (see Section \ref{topdef}).  If
a topological order can have a gapped boundary, then it can have many different
types of gapped boundary, described by different anomalous topological orders
in $n$-dimensional space-time.  However, for a given boundary anomalous
topological order, there can be only one unique bulk topological order.  This
has a flavour of holographic principle: the topological class of the surface
part of a many-body wave functions determines the topological class of the
whole bulk many-body wave functions.  Thus we have a mapping \bulk: boundary
topological orders $\to$ bulk topological orders. We see that the \bulk
operator has a geometric meaning of describing a boundary-bulk relation of a
many-body wave function. 

Since topological order can be described by an algebraic structure -- BF
categroy.  The geometrically or physically defined \bulk operator corresponds to 
an algebraic construction of {\it center} in
category theory.  In fact, we will show in \Ref{kong-wen-zheng} that the notion
of the \bulk is equivalent to a mathematical and a purely algebraic notion of
the center.  Such a connection between a geometric notion of \bulk and an
algebraic notion of center is quite amazing and deep, and was confirmed in
2+1D.\cite{KK1251,LW1384,fsv}
%:XGW

Similarly, we can also define a notion of the \bulk for \lBF{n}
categories, which also satisfy $\cZ_{n+1}(\cZ_n(\EC_n))=\one_{n+2}$.

\subsection{Closed and exact BF categories}

If a \hBF{n} category $\EC_n$ can be realized by a lattice model in the same
dimension (\ie the \bulk $\EC_{n+1}$ is trivial), such a \hBF{n} category is
said to be \emph{closed} (and the corresponding gapped effective theory is said
to be \emph{free of anomaly}).  In other words, $\EC_n$ is closed iff
$\cZ_n(\EC_n)=\one_{n+1}$.  If the qubit model that realizes the closed \hBF{n}
category $\EC_n$ also has a gapped boundary, which is described by a \hBF{n-1}
category $\EC_{n-1}$ in one lower dimension, then the \hBF{n} category $\EC_n$
is said to be exact.  In other words, $\EC_n$ is exact iff there exists a
$(n-1)$-dimensional \hBF{n-1} category $\EC_{n-1}$ such that
$\EC_n=\cZ_{n-1}(\EC_{n-1})$ (see Section \ref{ceBF}).  Similarly, we can also
define closed/exact \lBF{n}  categories.  

\begin{rema} \label{LHclosed}
A closed/exact \lBF{n} category is automatically a closed/exact \hBF{n} category.  More precisely, we have the monoid homomorphism
\begin{align} 
 \{\text{closed/exact }\BF^{\text{L}}_n \text{ Cat.}\} \to \{ \text{closed/exact }\BF^{\text{H}}_n  \text{ Cat.} \}   \nonumber
\end{align}
A non-trivial closed  \lBF{} category might correspond to a trivial  closed
\hBF{n} category.  
%At the moment, we do not have an example of a non-trivial
%\lBF{} category that correspond to a trivial \hBF{} category, and we do not
%have an example of a non-trivial  \hBF{} category that cannot be realized by
%\lBF{} category.  Since we fail to show they are the same, so in this paper, we
%will discuss \hBF{} category and  \lBF{} category separately,
%
%Here we would like to remark that the  \hBF{} category and \lBF{} category have
%different mathematical definitions: 
Mathematically, a closed \lBF{n} category may correspond to an
$n$-$(n-1)$-$\cdots$-$1$-$0$ fully extended TQFT, while a closed \hBF{d}
category may correspond to an ``x-$(n-1)$-$\cdots$-$1$-$0$ extended
TQFT'',\cite{W06,morrison-walker-1108,morrison-walker-1009,S1430} 
although the definitions of the those concepts are quite different.
An ``x-$(n-1)$-$\cdots$-$1$-$0$ extended
TQFT'' is defined as a theory where
we can assign a Hilbert space (the space of
degenerate ground states) to every closed orientable $d$-manifold, but we do
not require the path integral to be well defined for every closed orientable
$d$-dimensional space-time manifold.  We only require the path integral to be
well defined for every closed orientable $d$-dimensional mapping torus (a
mapping torus is a fiber bundle over $S^1$, where the fiber is the space and
$S^1$ is the time).  In \Ref{W06,morrison-walker-1108,morrison-walker-1009}, it was shown that any finite unitary
x-$(n-1)$-$\cdots$-$1$-$0$ extended TQFT extends to an
$n$-$(n-1)$-$\cdots$-$1$-$0$ fully extended TQFT.  So it is also possible that
a closed \lBF{d} category is the same as a closed \hBF{d} category.
\end{rema}

Although the BF categories sounds abstract, in low dimensions, they correspond
to some well known tensor categories.  Let us give some examples.  In 1+1D, the
closed and the exact \hBF{2} categories are always trivial (see Section
\ref{TN2D}).  The generic \hBF{1+1} categories are unitary fusion 1-categories
(UFC) which are always anomalous except the trivial UFC, i.e. the category
$\hilb$ of finite dimensional Hilbert spaces (see Section \ref{FZ2} and Example
\ref{C2FZ2}).  

In 2+1D, we believe that the closed \hBF{3} categories
are classified by the unitary modular tensor categories (UMTC), up to
some $E_8$ quantum Hall states $\EC_3^{E_8}$.\cite{WW1132}
% For the L-type topological orders, we
%believe that the closed \lBF{3} categories are classified by the UMTCs whose
%chiral central charges are zero mod 8 (up to some $E_8$ quantum Hall states).
Or more precisely, there is a many-to-one surjective map that maps the set of
closed \hBF{3} categories to the set of  UMTCs, and the kernel of
the map is the set of the $E_8$ quantum Hall states:
% and there is a many-to-one
%surjective map that maps the set of closed \lBF{3} categories to the set of
%UMTCs with zero chiral central charge.  The kernel of the above surjective map
%is the invertible topologically ordered states obtained by stacking the $E_8$
%quantum Hall states:
\begin{align}
& \one_3 \to \{ (\EC_3^{E_8})^{\boxtimes n}\} 
 \to \{\text{closed }\BF^{\text{H}}_3\text{ categories}\} \nonumber \\
&\hspace{4.5cm}  \to \{\text{UMTCs}\} \to \one_3  \label{diag:E8-umtc}
\end{align}
where the arrows are monoid homomorphisms. 

The exact \hBF{3} and \lBF{3} categories are the monoidal center of UFC's.  The
monoidal center $\cZ$ actually maps a \hBF{2} category to an exact \hBF{3}
category.  The unitary braided fusion categories (which may not be modular) are
examples of generic (potentially anomalous) \hBF{3} and \lBF{3} categories (see
Section \ref{preM}).

\subsection{Gravitational anomaly and its classification}

In this paper, as in \Ref{W1313}, we define gravitational anomaly as the
obstruction that a set of seemly consistent low energy properties cannot be
realized by any well defined quantum models in the same dimension.  If the  low
energy properties cannot be realized by any well defined bosonic Hamiltonian
models, we say there is a H-type gravitational anomaly.  If the  low energy
properties cannot be realized by any well defined local bosonic path integrals,
we say there is a L-type gravitational anomaly (see Section \ref{ganm}).  

Because a potentially anomalous theory (gapped or gapless) can always be
realized as a boundary of a gapped state in one-higher dimension (see
Corollary
\ref{cor:dim-reduction}), and because
the theory in one-higher dimension are described by closed BF category,  we see
that gravitational anomalies and closed BF categories (\ie topological orders)
in one-higher dimension are closely related.\cite{W1313} More precisely,
anomaly-free topological orders (or closed BF categories), gravitational anomalies, and gCS anomalous  topological orders
form three monoids, and we have a short exact sequence of
monoid homomorphisms
%\begin{align}
%\one \to 
%& d+1\text{D gCS anomalous  topological orders} \to
%\nonumber\\
%& d+1\text{D gravitational anomalies} \to
%\nonumber\\
%%& (d+2)\text{D topological orders} \to \one
%\end{align}
%or equivalently
\begin{align}
\one \to 
& \{ d+1\text{D gCS anomalous  topological orders} \}
\nonumber\\
& \to \{ d+1\text{D gravitational anomalies} \}
\nonumber\\
& \to \{ \text{closed BF}_{d+2} \text{ categories} \}\to \one
\end{align}
Note that $d+1$D gCS anomalous topological orders only appear in $d+1=4k+3$.
In other dimensions, $d+1$D gravitational anomalies are fully classified by
closed BF$_{d+2}$ categories.  In 2+1D, the only gravitational anomaly that is
not classified by closed BF$_{4}$ categories is the one described by a
unquantized gravitational Chern-Simons term.

%To see the difference of H-type and L-type gravitational anomalies, let us
%consider a 2+1D gapped theory with only one type of topological excitations:
%point-like semion excitation.  Such a theory has an L-type  gravitational
%anomaly (see Section \ref{}), and corresponds to an anomalous L-type
%topological order described by a non-closed \lBF{2+1} category.  This means\\
%(1) The partition function of the theory on closed 2+1D space-time
%%is not diffeomorphism invariant.\\
%(2) There is no local bosonic path integral in 2+1D that can realize such a low
%energy property of having a single type of topological excitations with a
%semion statistics.\\
%On the other hand, there is a 3+1D local bosonic path integral, whose boundary
%realizes the above theory. Forthermore, the 3+1D local bosonic path integral
%describes a trivial \lBF{3+1} category (\ie a trivial L-type topological order
%in 3+1D).
%
%But the above 2+1D theory does not have any H-type gravitational anomaly, and
%corresponds to a H-type topological order described by a closed \hBF{2+1}
%category.  We can realize the low energy property of having a single type of
%topological excitations with a semion statistics through  a bosonic $\nu=1/2$
%fractional quantum Hall (FQH) state.
 
\subsection{A classification of topological order}

Since all possible topological orders in lattice qubit models are described by
the closed \hBF{n} categories, the closed \hBF{n} categories classify bosonic
topological orders (and anomaly-free gapped effective theories).  Restricting
\eqn{topoBF} to closed \hBF{} categories, we obtain a monoid isomorphism
\begin{align}
\{ \text{Topological orders}\} \simeq \{\text{closed }\BF^H_n \text{ categories} \}.
\end{align} 
Again, we have an expression that relates a mathematical construction (closed
\hBF{} category) to a physical phenomenon (topologically ordered phases or
long-range entanglement). 

Similarly, the exact \hBF{n} categories classify topological orders with
gappable boundary.  We have a monoid isomorphism
\begin{align}
&\ \ \ \
 \{ \text{Topological orders with gappable boundary} \}
\nonumber\\
&\hspace{2cm} \simeq  \{\text{exact }\BF^H_n \text{ categories} \}.
\end{align}

\subsection{Monoid-cochain complex and cochain complex}

The sequence $\cdots \overset{\cZ_{n-1}}{\rightarrow} \EC_n
\overset{\cZ_n}{\rightarrow} \EC_{n+1} \overset{\cZ_{n+1}}{\rightarrow} \cdots
$ and the fact that $\cZ_{n}\cZ_{n-1}=0$ imply that the BF categories (\ie
\hBF{} categories or \lBF{} categories) in different dimensions form a
monoid-cochain complex,\cite{F1291} where taking the \bulk (or the center)
$\cZ_n(\cdot)$ acts like the ``differential'' operator that maps a BF category
to another BF category in one-higher dimension.  We also show that the
equivalence classes of BF categories in different dimensions (under the
equivalence relation $\sim$ mentioned in Section \ref{tprod}) form a cochain
complex.  Such monoid-cochain complex and cochain complex reveal the structure
and connection between gravitational anomalies, BF categories, and topological
orders in different dimensions (see Section \ref{mcc}).  The cohomology classes
of the cochain complex, $\rH^n=\text{ker}(\cZ_n)/\text{img}(\cZ_{n-1})$,
describe types of boundaries (including gapless ones) 
of the topological states.

\subsection{Tensor network approach to topological order in any dimensions}

We also develop tensor network path integrals, hopefully to realize all
generic/closed/exact \lBF{n}  categories.  This in turn allows us to realize a
large subset of \hBF{n} categories.  Here we collect some important
conjectures:
\begin{enumerate}
\item 
All exact \lBF{n}  categories can be realized via \emph{topological path
integrals} on space-time complex, which can be expressed as tensor network (see
Section \ref{TNeBF}).  
%(which has a topology of a mapping torus).  A subset of exact \hBF{n}
%categories are described by topological path integrals where the space-time is
%a fiber bundle over $S^1$ -- a mapping torus 
%XGW:
(\emph{Topological path integrals} are defined as path integrals that produce
\emph{topological invariant partition functions} for arbitrary closed
space-time.  They are the fixed points of renormalization flow for gapped
quantum liquids.  \emph{Topological path integrals} can always be described by
tensor network.)  
%:XGW
%Since we are using path integrals to describe the Hamiltonian version of
%quantum theory, as a result, we require the space-time to have a topology of a
%mapping torus 
%\item 
%The partition function $Z(M^n)$ of a topological path integral is a function of
%space-time $M^n$ which depends only on the topology of $M^n$.  When
%$n<8$, there is an one-to-one correspondence between different exact
%\hBF{n} categories and different topological partition functions on
%space-time which are mapping tori (see Section \ref{TNeBF}).  
%The trivial \hBF{n} category is described by $Z(M^n)=1$.
\item 
Different exact \lBF{n} categories always have different topological partition
functions.  But the reverse is not true.  Partition functions differ by a
factor $W^{\ch(M^n)} \ee^{\ii \sum_{\{n_i\}} \phi_{n_1n_2\cdots} \int_{M^n}
P_{n_1n_2\cdots}} $ describe the same \lBF{} category.  Here $\ch(M^n)$ is the
Euler character and $\int_{M^n} P_{n_1n_2\cdots}$ are the Pontryagin numbers of
the space-time $M^n$.  The trivial \lBF{n} category is described by partition
functions of form $Z(M^n)=W^{\ch(M^n)} \ee^{\ii \sum_{\{n_i\}}
\phi_{n_1n_2\cdots} \int_{M^n} P_{n_1n_2\cdots}} $ (see \eqn{ZElBF}). 
\item 
Not all  topological path integrals describe exact \lBF{n}  categories.  Only
stable topological path integrals describe  exact \lBF{n}  categories.  A
topological path integral in $(n+1)$D space-time is stable iff $|Z(S^1\times
S^n)|=1$ (see Section \ref{TNeBF}).\cite{KW13}
\item 
$n$-dimensional (potentially anomalous) \lBF{n} categories can be
described by \emph{topological path integrals} in one higher dimensions (see
Section \ref{TNBF}).
%\item
%A topological path integral in $n$-dimensional space-time describes a trivial
%\hBF{n} category iff its topological partition function $Z(M^n)=1$ for any
%%orientable 
%space-time $M^n$ which is a mapping torus (see Section \ref{TNcBF}).
\item 
$n$-dimensional closed \lBF{n} categories can be described by
``trivial'' \emph{topological path integrals} 
in one higher dimensions.  There can be
many different ``trivial'' topological path integrals that describe the same
trivial  \lBF{n} category. Those different ``trivial'' topological path
integrals describe different closed \lBF{n} categories in one lower dimension
(see Section \ref{TNcBF}).\cite{WW1132,KBS1307,BCFV1372,CFV1350}
\end{enumerate}

We note that the above results give us a concrete, practical, and constructive
definition of exact \lBF{} category via \emph{topological path integorals} (or
tensor network) in any dimensions.  Then the generic (closed) \lBF{} categories
can be defined as the boundary of exact \lBF{} category (trivial  \lBF{}
category) in one higher dimension.

\subsection{Probing and measuring topological orders}

We propose some ways to probe and measure  topological orders (see Section
\ref{PMtopo}):
\begin{enumerate}
\item 
For a L-type quantum system defined by a path integral, we can compute its
imaginary time partition function on closed space-time $M^{d+1}$.  The system
describes an exact \lBF{{d+1}} category (\ie the topological order has a
gappable boundary) iff the corresponding volume independent part of the
partition function 
$Z_0(M^{d+1})$ is a topological invariant of the space-time $M^{d+1}$.  
\item
For a H-type quantum system described by a Hamiltonian on a closed space $\Si^d$,
the degenerate ground states form a vector space $V$. As we change the metrics
on $\Si^d$, we obtain the moduli space $\cM_{\Si^d}$ of $\Si^d$.  Together with the vector
space $V$ for each point in $\cM_{\Si^d}$, we obtain a vector bundle 
 on the moduli space $\cM_{\Si^d}$.  For different space topologies, we
will get different vector bundles.  The collection of those  vector bundles
should fully characterize the closed \hBF{d+1} category (\ie the topological
order).
\item
A \hBF{d+1}  category is exact iff the above mentioned vector bundle
is flat (see Conjecture \ref{flatVB}).
%\item Note that a space-time mapping torus are characterized by an
%homeomorphism of the space.  The homeomorphism form a group
%$G_\text{homeo}(M_\text{space})$.  
%%For Hamlitonian quantum systems that must have gapless boundaries (closed
%%\hBF{n} categories), 
%The volume independent part of the partition function is given by the trace of
%the above mentioned projective representation of the homeomorphism group
%$G_\text{homeo}(M_\text{space})$.  For Hamiltonian quantum systems that can
%have gapped boundaries (exact \hBF{n}  categories), the volume independent
%part of
%%the partition function is given by the trace of the representation of the
%discrete homeomorphism group $G^\text{dis}_\text{homeo}(M_\text{space})\equiv
%\pi_0[G_\text{homeo}(M_\text{space})]$.  Therefore, topological orders are
%fully characterized by projective representations of the homeomorphism group
%$G_\text{homeo}(M_\text{space})$ (for closed \hBF{n}  categories), or by
%%representations of the discrete homeomorphism group
%$G^\text{dis}_\text{homeo}(M_\text{space})$ (for exact \hBF{n}  categories).
\item
A \hBF{n}  category is closed iff every nontrivial topological excitation in it
has a nontrivial mutual braiding property (or a nontrivial mutual statistics)
with at least one topological excitation.  This is the condition for a gapped
effective theory to be free of H-type gravitational anomaly (see Conjecture
\ref{msta}). This principle was discussed in detail in \Ref{L1309}.
\end{enumerate}

%The monoid of the \hBF{} categories contain a subset which form a group.  Those
%\hBF{} categories are called the invertible \hBF{} categories, which contain no
%non-trivial bulk topological excitations and no ground state degeneracies.
%Similarly, we can define invertible \lBF{} categories, 

As an application of the above conjectures, let us consider the simplest
topological orders that have no non-trivial topological excitations and no
degenerate ground states.  We find that this kind of topological orders have
two defining properties: i) their partition functions on closed space-time can
always be chosen to be a pure $U(1)$ phase; ii) they are invertible under the
stacking $\boxtimes$ operation. Using those properties, one can try to classify
those invertible topological orders.\cite{FT1292,F14,freed2014}  
We find that there is no non-trivial invertible
\lBF{} categories in 3+1D, and 5+1D.  The invertible \lBF{}
categories in 2+1D form an Abelian group $\Zb$ generated by the $E_8$ bosonic
fractional quantum Hall state.\cite{PMN1372} The invertible \lBF{}
categories in 6+1D form an Abelian group $\Zb\times \Zb$.  The boundary of
those $\Zb$-class topological orders must be gapless.  The 
invertible \lBF{} categories in 4+1D form an Abelian group $\Zb_2$.  (This
result has been obtained in \Ref{K1459}).  The boundary of the non-trivial
$\Zb_2$-class topological order can be gapped but must carry a anomalous
topological order with non-trivial topological excitations on the boundary.
The invertible \lBF{} categories are also invertible \hBF{} categories (see
Remark\,\ref{LHclosed}), and we show that the non-trivial  invertible \lBF{} categories
are also non-trivial when viewed as \hBF{} categories.  Thus the invertible
\hBF{} categories contains a $\Zb$ class in 2+1D, a $\Zb_2$ class in 4+1D, and
a $\Zb\oplus \Zb$ class in 6+1D.

\section{Physical definition of topologically ordered phase}
\label{topdef}

The topologically ordered states that we will discuss in this paper are ground
states of bosonic local \emph{gapped} Hamiltonians.  However, not all gapped
ground states are topologically ordered states.  Only the special gapped ground
states, called gapped quantum liquids,\cite{ZW14} 
are  topologically ordered states. 
%and topologically ordered states are gapped quantum liquids.  
For example, 2+1D FQH states and 3+1D $\Zb_2$ gauge theory are gapped quantum
liquids (\ie topologically ordered states), while the 3+1D gapped state formed
by layers of 2+1D FQH states is not a gapped quantum liquids.  The notion of
gapped quantum liquids (\ie topologically ordered states) are discussed in
detail in \Ref{ZW14}.  We will not repeat them here.  

Now we are ready to define topological order (or topologically ordered phase).
We will give two definitions:\\
(1)
Let us call the Hamiltonian that realizes a gapped quantum liquid \emph{an
l-gapped Hamiltonian}.  If the ground state degeneracy of a  l-gapped
Hamiltonian is robust against any perturbations, then the l-gapped Hamiltonian
is said to be stable.  Let $M_{slgH}$ be the space of stable l-gapped
Hamiltonians.  Then the elements of $\pi_0(M_{slgH})$ define the topologically
ordered phases.  \\
(2)
We can also use local unitary transformations to define topologically ordered
phases: two \emph{stable} gapped quantum liquids belong to the same
topologically ordered phase iff they are connected by a  local unitary
transformation.\cite{CGW1038} 
%(3)
%There is a third way to define topologically
%ordered phases: two gapped quantum liquids belong to the same topologically
%ordered phase iff they are connected by a local invertible
%etransformation.\cite{ZW14} 

Local unitary (LU) transformation\cite{LWstrnet,VCL0501,V0705,CGW1038} is an
important concept which is directly related to the definition of quantum phases
.\cite{CGW1038}  To explain LU transformation, let us first introduce 
\emph{local unitary evolution}.  A LU evolution is defined
as the following unitary operator that act on the degrees of freedom in a
lbH system:
\begin{align}
\label{LUdef}
  \cT[e^{-i\int_0^1 dg\, \t H(g)}]
\end{align}
where $\cT$ is the path-ordering operator and $\t H(g)=\sum_{\v i} O_{\v i}(g)$
is a sum of local Hermitian operators.  Two \emph{gapped} quantum states belong
to the same phase if and only if they are related by a LU
evolution.\cite{HW0541,BHM1044,CGW1038} Note that, in this paper, the term ``a
gapped quantum state'' really means ``the subspace of the degenerate ground
states''.

\begin{figure}[tb]
\begin{center}
\includegraphics[scale=0.5]{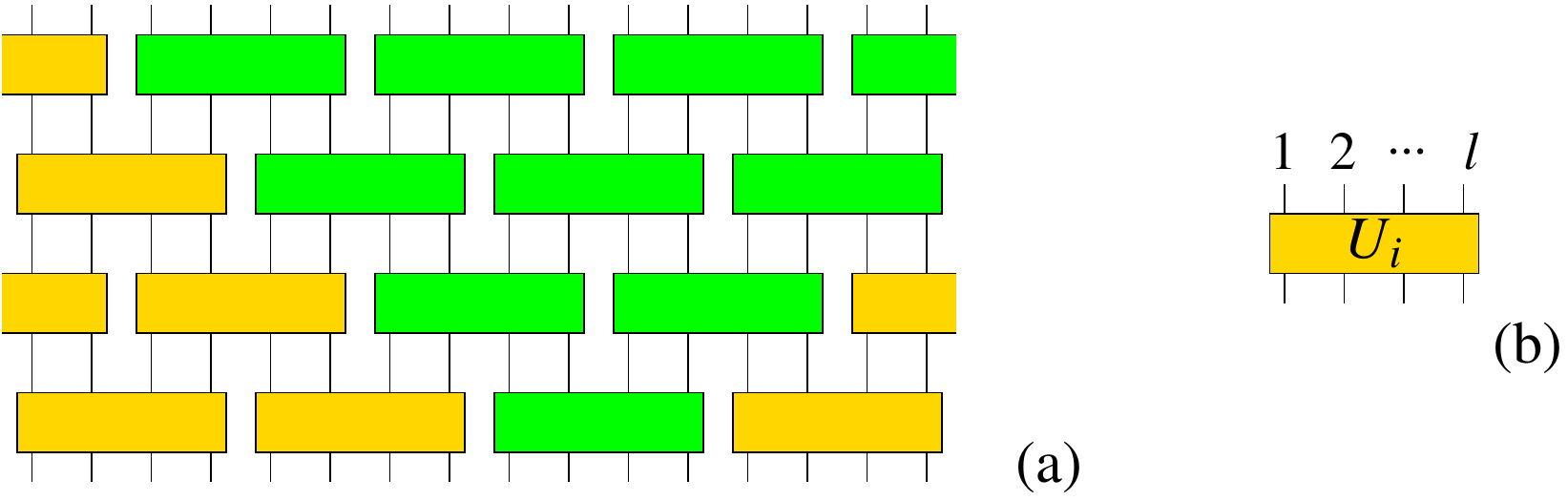}
%Fig. 3
\end{center}
\caption{
(Color online)
(a) A graphic representation of a quantum circuit, which is form by (b) unitary
operations on blocks of finite size $l$. The green shading represents a causal
structure.
}
\label{qc}
\end{figure}

The LU evolutions are closely related to \emph{quantum circuits with finite
depth}.  To define quantum circuits, let us introduce  piecewise local unitary
operators.  A piecewise local unitary operator has a form
\begin{equation*}
 U_{pwl}= \prod_{i} U^i
\end{equation*}
where $\{ U^i \}$ is a set of unitary operators that act on non overlapping
regions. The size of each region is less than a finite number $l$. 
%The
%unitary operator $U_{pwl}$ defined in this way is called a piecewise local
%unitary operator with range $l$.  
A quantum circuit with depth $M$ is given by
the product of $M$ piecewise local unitary operators:
\begin{equation*}
 U^M_{circ}= U_{pwl}^{(1)} U_{pwl}^{(2)} \cdots U_{pwl}^{(M)}
\end{equation*}
We will call $U^M_{circ}$ a LU transformation.  In quantum information theory,
it is known that a finite time unitary evolution with a local Hamiltonian (a LU
evolution defined above) can be simulated with a constant depth quantum circuit
(\ie a  LU transformation) and vice-verse:
\begin{align}
  \cT[e^{-i\int_0^1 dg\, \t H(g)}] =U^M_{circ}.
\end{align}
So two gapped quantum states belong to the same phase if and only if they are
related by a LU transformation.

Using the LU transofrmations, we can define the concept of short-range
and long-range entanglement.\cite{CGW1038}
\begin{defn} \textbf{Short-range entanglement}\\
A state is short-range entangled (SRE) if it can be transformed into product
state by a LU transformation of a fixed depth regardless how large the system
is.
\end{defn} \noindent
We can show that all short-range entangled states belong to the same phase:
\begin{cor} 
All short-range entangled states can be transformed into each other via
LU transformations.
\end{cor} \noindent
We can also show that
\begin{cor} 
For any short-range entangled state $|\Psi\>$, 
there exists a gapped local Hamiltonian $H$
such that  $|\Psi\>$ is the only ground state of $H$.
\end{cor} \noindent

\begin{defn} \textbf{Long-range entanglement}\\
A stable gapped state is long-range entangled if it is not short-range entangled.
\end{defn} \noindent
Here ``stable'' means that the ground state degeneracy is 
robust against any small perturbations.
\begin{defn} \textbf{Topologically ordered states}\\
Topologically ordered states are LRE gapped liquid states.  
In other words, a gapped liquid state has a
nontrivial topological order iff it cannot be transformed to a product state by
any LU transformations of finite depth.
\end{defn} \noindent
Not all long-range entangled (LRE) gapped liquid states can be transformed into each other
via LU transformations.  Thus LRE states can belong to different phases: \ie
the LRE states that are not connected by LU transformations belong to different
phases.  Those different phases are nothing but the topologically ordered
phases:\cite{Wtop,WNtop,Wrig,KW9327}
\begin{defn} \textbf{Topologically ordered phases}\\
LU transformations are equivalence relations.  Topologically ordered phases are 
equivalence classes of topologically ordered states under the LU 
transformations.
\end{defn} \noindent
In this paper, we plan to study topologically ordered phases (\ie topological
orders) in any dimensions. We would like to find a way to classify all
topologically ordered phases. As we have stressed above, our study is limited
to gapped quantum liquids, and does not apply to other more complicated gapped
quantum states.

\section{Excitations in topologically ordered states} \label{topexc}

Topological orders (or patterns of long-range entanglement) can be
characterized by the appearance of the ``topological excitations''.  In this
section, we will discuss/define the notion of topological excitations.  

In higher dimensions, the allowed topological excitations can be quite
complicated.  They can be ``particle-like'', ``string-like'',
``membrane-like'', etc.  Many concepts need to be introduced to describe and
understand these excitations.

\begin{figure}[tb]
  \centering
  \includegraphics[scale=0.5]{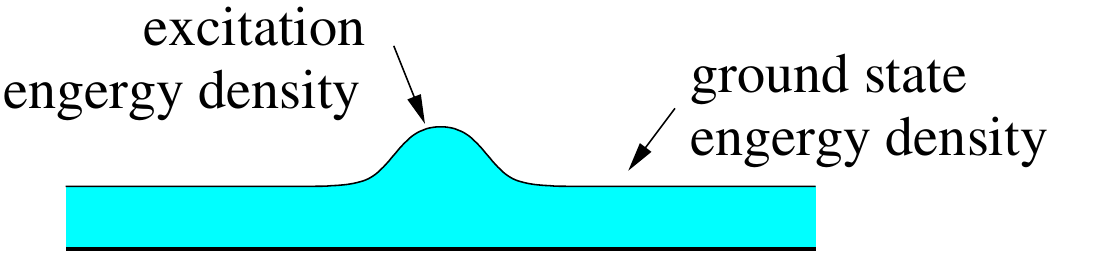}
  \caption{The energy density distribution of a particle-like excitation.}
  \label{exceng}
\end{figure}

\subsection{Particle-like excitations}

%In this subsection we will discuss particle-like excitations in detail and
%introduce many important concepts. In the next subsection, we will discuss
%string-like, membrane-like excitations.
%
%\subsubsection{Local quasiparticle excitations and topological quasiparticle
%excitations}

First we define the notion of ``particle-like'' excitations.  Consider a gapped
system with translation symmetry.  The ground state has a uniform energy
density.  If we have a state with an excitation, we can measure the energy
distribution of the state over the space.  If for some local area, the energy
density is higher than ground state, while for the rest area the energy density
is the same as ground state, one may say there is a ``particle-like''
excitation, or a quasiparticle, in this area (see Figure \ref{exceng}).
Quasiparticles defined like this can be further divided into two types.  The
first type can be created or annihilated by local operators, such as a spin
flip.  So the first type of the particle-like excitations is called local
quasiparticle excitations.  The second type cannot be  created or annihilated
by any finite number of local operators (in the infinite system size limit).
In other words, the higher local energy density cannot be created or removed by
\emph{any} local operators in that area.  The second type of the particle-like
excitations is called topological quasiparticle excitations. They are characterized by the modules over the local operator algebras\cite{KK1251, kong-icmp12, LW1384}

%\begin{figure}[tb]
%  \centering
%  \includegraphics[scale=0.5]{exceng}
%  \caption{The energy density distribution of a quasiparticle.}
%  \label{exceng}
%\end{figure}

From the notions of local quasiparticles and topological quasiparticles, we can
also introduce a notion of topological quasiparticle type, or simply,
quasiparticle type.  We say that local quasiparticles are of the trivial type,
while topological quasiparticles are of nontrivial types.  Also two
topological quasiparticles are of the same type if and only if they differ by
local quasiparticles.  In other words, we can turn one topological
quasiparticle into the other one of the same type by applying some local
operators.

\subsection{$p$-dimensional topological excitations}

In the above, we only discussed the notion of ``particle-like'' topological
excitations. Similarly, we can also introduce  the notion of ``string-like''
topological excitations, or even more general $p$-dimensional topological
excitations, where $p$ is the spatial dimension. 
To define a $p$-dimensional topological excitations, let us first
define 
\begin{defn} \label{def:p-excitation}\textbf{$p$-dimensional excitations:}\\
Consider a gapped lbH system defined by a local bosonic Hamiltonian $H_0$ in
$n$ spatial dimensions.  For $p<n$, a $p$-dimensional excitation is the
\emph{gapped} ground state of $H_0+\Del H$ where $\Del H$ is a local hermitian
operator which is non-zero only on a $p$-dimensional subspace $M^p$ and is
almost uniform on $M^p$ in the large $M^p$ limit.
%We also require that the  ground state corresponds to a stable quantum
%phase on $M^D$ in the large $M^D$ limit. \\
%(2) For $D=d-1$, a $p$-dimensional excitation is a boundary of the system where
%the boundary is  almost uniform on the boundary in large boundary limit.  We
%also require that the boundary state corresponds to a stable quantum phase in
%large boundary limit.  Note that there are many different boundaries which
%correspond to many different $p$-dimensional space-time defects.
\end{defn} \noindent
\begin{rema} 
(1) A $p$-dimensional excitation is defined only for large $M^p$ if $p>0$.\\
(2) If the ground state of  $H_0+\Del H$ has gapless modes for large $M^p$,
then, by definition, $\Del H$ does not create a $p$-dimensional excitation.
\end{rema} \noindent

We note that we can view a $p$-dimensional excitation as a gapped system with
$p$ spatial dimensions, which has a thermal dynamical limit when $M^p$ is
large.  This allows us to define the equivalence relation between
$p$-dimensional excitations:
\begin{defn} \label{def:3-eq-rel} %\textbf{Three equivalence relations}  \\
Two  $p$-dimensional excitations on $M^p$ are equivalent if they can be\\
(1) transformed into each other via a local unitary transformation of finite
depth (see Section \ref{topdef})\cite{CGW1038} on a neighborhood of $M^p$ or\\
(2) transformed into each other
via a tensor product of an unentangled state on $M_p$.\\
The equivalence class is called the type of topological excitations.
\end{defn} 

When we refer to a topological excitation, we usually mean its equivalence
class (or type). We have the following conjecture. 
\begin{conj}
Two $p$-dimensional excitations localized on the two submanifolds $M^p$ and on $N^{p}$ are equivalent if we can deform them into each other smoothly without phase transition in the large $M_p$ and $N_p$ limit.
\end{conj}

%The $p$-dimensional excitations that we have just defined are very simple type
%of excitations. In general, one can have $q$-dimensional excitations nested in
%a $p$-dimensional excitations for $q<p$.  Moreover, a $p$-dimensional
%excitation defined in Definition\,\ref{def:p-excitation} can also be viewed as
%a sub-excitation nested in the trivial $r$-dimensional excitation for $r>p$.
%For convenience, we will call those defined in
%Definition\,\ref{def:p-excitation} or those nested in only trivial excitations
%as pure $p$-dimensional excitations. 

\subsection{Wall and pure excitations}

\begin{figure}[tb]
  \centering
  \includegraphics[scale=0.6]{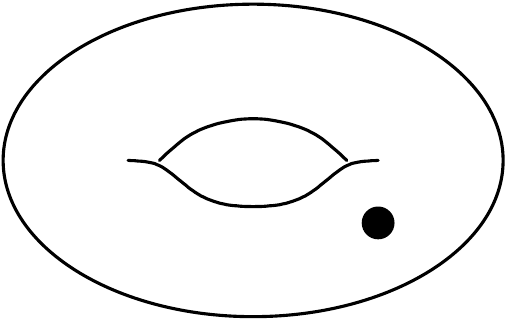}
  \caption{
%A mixed 2-dimensional excitation in $d$-dimensional space.
The 2-dimensional excitation has a torus topology
and carries a 0-dimensional (point-like) wall excitation (or a sub-defect).
}
  \label{torus}
\end{figure}

We also like to point out that the $p$-dimensional topological excitations
discussed above are not all the topological excitations that can appear in a
topologically ordered state.  We can have more general
topological excitations, such as a $p$-dimensional topological excitation nested
in a $p'$-dimensional topological excitations with $p'>p$ (see Fig.
\ref{torus}).  We will 
%call such an excitation as a mixed $p$-dimensional topological excitation, and 
call the $p$-dimensional topological excitation as a wall excitation (or a
sub-defect) of the $p$-dimensional topological excitation.  We will call
$p$-dimensional topological excitations that  only nested in a trivial higher
dimensional excitation as pure excitations.  In this paper, we usually use the
term ``topological excitation'' to refer to pure excitation.

\section{Universal low energy properties and a physical definition of \hBF{n}  category}
\label{sec:univ-prop-hbfcat}

\begin{figure}[tb]
  \centering
  \includegraphics[scale=0.6]{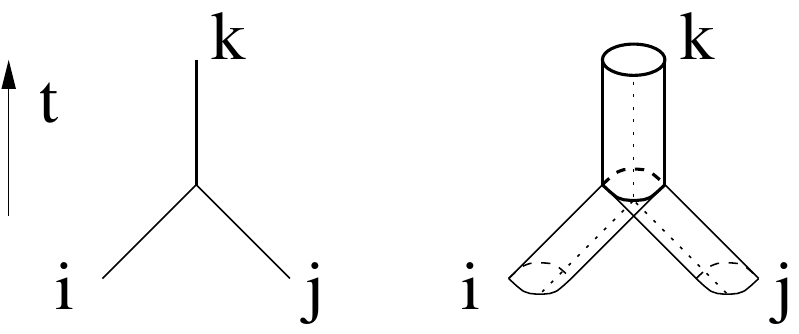}
  \caption{The fusion of topological particles (the 0-branes)
and  topological loops (the 1-branes). Here $i$- and $j$-types
of topological excitations are fused into a $k$-type  topological excitation.
The time direction points upwards.}
  \label{fuse}
\end{figure}

\begin{figure}[tb]
  \centering
  \includegraphics[scale=0.6]{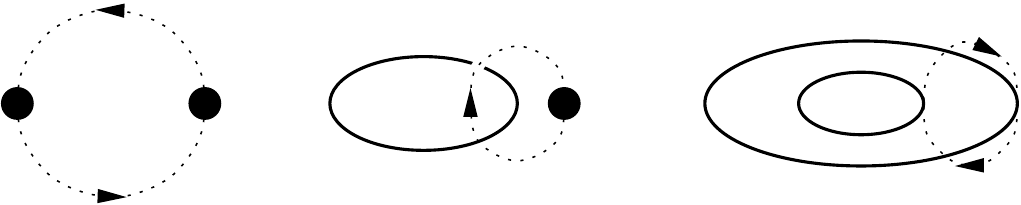}
  \caption{The braiding between topological particles
and topological loops in 3 dimensional space.}
  \label{braid}
\end{figure}

\subsection{A physical definition of \hBF{}  category}

The $p$-dimensional topological excitations can have
some universal properties (or topological properties), which, by definition,
are robust against any local perturbations and can be used to physically
characterize topological phases.
\begin{defn} \textbf{\hBF{}  category}\\
The collection of all topological (or universal) properties of all the
topological excitations in $n$ space dimensions, as well as the perturbative
gravitational responses defines a $(n+1)$-dimensional braided fusion (\hBF{}
or \hBF{n+1}) category.  (Here $n+1$ is the space-time dimension.)
\end{defn} \noindent
So, in physics, the term ``\hBF{}  category'' and the term ``the set of topological
properties'' can be used interchangeably.  In physics, a set of topological properties also defines a gapped low energy effective theory, so the term ``
\hBF{} category'' and the term ``gapped low energy effective theory'' can also
be used interchangeably.  In this paper, we will mainly use the term  ``\hBF{}
category''.  %But such a term can be replaced by ``the set of topological properties'', or ``gapped  low energy effective theory''.

%At the moment, it is not clear what are \emph{all} the universal properties.
In Section \ref{sec:math-def}, we will give a more detailed definition of
\hBF{}  category, trying to describe a subset of universal properties that
completely specify a \hBF{}  category, \ie completely specify all other
universal (or topological) properties.  In the following, we will describe some
of the simple universal (or topological) properties for $p$-dimensional
topological excitations.  
%We note that, in general, a $D$-dimensional topological excitations may have
%gapless modes on the $D$-dimensional defects if $D\geq 1$.  We do not know how
%to describe the topological properties of those $D$-dimensional topological
%excitations that carry gapless modes.  In the following, we will only discuss
%the universal (or topological) properties for $D$-dimensional topological
%excitations that do not carry any gapless modes.

\subsection{Universal low energy properties}
\label{uniprop}

What is the subset of universal properties that completely specify a
topological order? Here we propose that
\begin{conj} 
The fusion and the braiding properties of topological excitations, plus the
universal correlations of energy-momentum tensor completely
characterize the topological order (\ie the \hBF{} category.)
\end{conj}\noindent
In this section, we will explain the above conjecture, in particular, what are
the fusion and the braiding properties.  Also, we only need to include the
fusion and the braiding of a subset of topological excitations to fully define
the topological order. We will discuss what this subset is.

\subsubsection{The fusion space}

The first and the most important universal property is 
\begin{defn} \textbf{the generalized Fusion space:}\\
If we put $p$-dimensional topological
excitations (labeled by $i,j,k,\cdots$) on an $n$-dimensional closed space
$M^n$, a generalized fusion space is the $D(M^n, \cT_{ijk\cdots}; i,j,k,\cdots)$ dimensional (nearly) degenerate space 
$\cV^F(M^n, \cT_{ijk\cdots}; i,j,k,\cdots)$ of the lowest energy eigenstates.
\end{defn}\noindent
A $p$-dimensional topological excitation can have a nontrivial topology
and linking with other topological excitations described by $\cT_{ijk\cdots}$.
Also $i,j,k,\cdots$ label different types of topological excitations which can
have different dimensions $p$.  We also fixed the locations of the topological
excitations, and assume that the topological excitations have a large size and
are well separated.  In this limit, the (near) degeneracy is well defined.

$ D(M^n, \cT_{ijk\cdots}; i,j,k,\cdots)$ reduces to the ground state degeneracy
$\text{GSD}(M^n)$ on a closed topological space $M^n$, when there is no
topological excitation:
\begin{align}
 \text{GSD}(M^n)= D(M^n).
\end{align} 
On a sphere $S^n$, 
$ D(M^n, \cT_{ijk\cdots}; i,j,k,\cdots)$
reduces to
\begin{align}
 D_\one (\cT_{ijk\cdots}; i,j,k,\cdots)
= D(S^n, \cT_{ijk\cdots}; i,j,k,\cdots)
\end{align} 
which is called the dimension of the fusion space of topological excitations
$i,j,k,\cdots$ into the trivial one $\one$ (see Fig. \ref{fuse}).  The
\emph{fusion space} is the degenerate space of the lowest energy eigenstates
with topological excitations $i,j,k,\cdots$ on a sphere $S^n$, with the fixed positions and shapes of the topological excitations.

\subsubsection{A ``local'' description of a topological excitation}
\label{local}

The key to understand topological ordered states is to understand the fusion
space $\cV^F(M^d,\cT_{i_1,i_2,\cdots};i_1,i_2,\cdots )$.  To understand the
structure of the  fusion space, it s very tempting to assign a vector space
$\cV^F_{i_k}$ for each excitation $i_k$, and a  vector space $\cV^F_{M^d}$ for
the closed space, and view the fusion space as a tensor product of those spaces
$ \cV^F(M^d,\cT_{i_1,i_2,\cdots};i_1,i_2,\cdots ) =\cV^F_{M^d} \otimes_k
\cV^F_{i_k}$.  If this is true, we may regard $\cV^F_{i_k}$ as the space for
the local degree of freedom carried by the excitation $i_k$.  However, in
general,  the fusion space $\cV^F(M^d,\cT_{i_1,i_2,\cdots};i_1,i_2,\cdots )$
does not have the above tensor product structure.  This makes it very difficult
to understand the structure of the fusion space from the local properties of
the excitation $i_k$.

However, we still insist on using local properties of the excitation $i_1$ to
understand and to construct the total fusion space
$\cV^F(M^d,\cT_{i_1,i_2,\cdots};i_1,i_2,\cdots )$.  To achieve this, we simply
define the ``local properties'' of the excitation $i_1$ as a map that defines the notion of a topological excitation as follows:
\begin{defn} 
A topological excitation $i$ is a map that maps a collection
of  topological excitations $i_2,i_3,\cdots$ to a fusion space:
$i: i_2,i_3,\cdots \to \cV^F(M^d,\cT_{i,i_2,\cdots};i,i_2,\cdots )$.
\end{defn} \noindent
Such a map $i$ represents a ``local'' description of topological excitation
$i$. In mathematical language, it is nothing but the Yoneda Lemma. 

\subsubsection{Simple type and composite type}

To understand the notion of simple type and composite type, remember that a
gapped lbH system is defined by a local Hamiltonian $H_0$ in
$d$ dimensional space $X^d$ without boundary.  A collection of excitations
labeled by $i_1$, $i_2$, \etc and located at $M_1$, $M_2$ \etc can be produced
as \emph{gapped} ground states of $H_0+\Delta H$ where $\Delta H$ is non-zero
only near $M_i$'s. Here $M_i$ is a sub manifold of $X^d$.  By choosing
different $\Delta H$ we can create all kinds of excitations.  

The gapped ground states of $H_0+\Delta H$ may have a degeneracy $ D(M^d,
\cT_{i_1,i_2,\cdots}; i_1,i_2,\cdots)$ as discussed above The degeneracy is not
exact, but becomes exact in the large $M_i$ space and large excitation
separation limit.  
%We will use $\cV^F(M^d,\cT_{i_1,i_2,\cdots};i_1,i_2,\cdots )$ to denote the
%space of the degenerate ground states.

If the Hamiltonian $H_0+\Delta H$ is not gapped, we will say $D( M^d,,
\cT_{i_1,i_2,\cdots};i_1,i_2,\cdots )=0$ ({\it i.e.} $\mathcal V(M^d,
\cT_{i_1,i_2,\cdots};i_1,i_2,\cdots )$ has zero dimension).  If $H_0+\Delta H$ is
gapped, but if $\Delta H$ also creates excitations away from $M_i$'s
(indicated by the bump in the energy density away from $M_i$'s), we will also
say $D( M^d, \cT_{i_1,i_2,\cdots};i_1,i_2,\cdots )=0$.  (In this case
excitations at $M_i$'s do not fuse to trivial excitations.) So if $D(
M^d,\cT_{i_1,i_2,\cdots};i_1,i_2,\cdots )>0$, $\Delta H$ only creates
excitations at $M_i$'s.

\begin{defn} \textbf{Simple and composite types:}\\
If the degeneracy $D(M^d,\cT_{i_1,i_2,\cdots};i_1,i_2,\cdots)$ cannot be lifted
by \emph{any} small local perturbation near $M_1$, then the topological type
$i_1$ at $M_1$ is said to be simple. Otherwise, the  topological type $i_1$ at
$M_1$ is said to be composite.  
\end{defn} \noindent
The degeneracy
$D(M^d,\cT_{i_1,i_2,\cdots};i_1,i_2,\cdots)$ for simple topological types $i_i$
is a universal property ({\it i.e.} a topological invariant) of the
topologically ordered state.

When $i_1$ is composite, then the space of the degenerate ground states
$\cV^F(M^d;i_1,i_2,i_3,\cdots)$ will have a direct sum
decomposition:
\begin{align}
\label{comp}
&\ \ \ \cV^F(M^d;\cT_{i_1,i_2,\cdots};i_1,i_2,i_3,\cdots)
\nonumber\\
&\hspace{1cm} =
 \cV^F(M^d;\cT_{j_1,i_2,\cdots};j_1,i_2,i_3,\cdots)
\\
& \hspace{1.5cm} \oplus
 \cV^F(M^d;\cT_{k_1,i_2,\cdots};k_1,i_2,i_3,\cdots)
\nonumber\\
& \hspace{1.5cm} \oplus
 \cV^F(M^d;\cT_{l_1,i_2,\cdots};l_1,i_2,i_3,\cdots) \oplus \cdots
\nonumber 
\end{align}
where $j_1$, $k_1$, $l_1$, {\it etc.} are simple types.  To see the above
result, we note that when $i_1$ is composite the ground state degeneracy can be
split by adding some small perturbations near $M_1$.  After splitting, the
original degenerate ground states become groups of degenerate  states, each
group of degenerate  states span the space $\mathcal
V(M^d;\cT_{j_1,i_2,\cdots};j_1,i_2,i_3,\cdots)$ or $\mathcal
V(M^d;\cT_{k_1,i_2,\cdots};k_1,i_2,i_3,\cdots)$ etc., where $j_1$ and $k_1$ correspond to simple quasiparticle types at $M_1$.  
%We denote the composite type $i_1$ as \begin{align}
% i_1=j_1\oplus k_1\oplus l_1\oplus \cdots.
%\end{align}
The decomposition \eqn{comp} is valid for any choices of $i_2,i_3,\cdots$.
From the discussion in Section \ref{local}, we see that a composite excitation
$i_1$ can be written as $i_1=j_1\oplus k_1 \oplus\cdots$, where $j_1\oplus k_1
\oplus\cdots$ is a map that maps  a collection of  topological excitations
$i_2,i_3,\cdots$ to a fusion space $
\cV^F(M^d,\cT_{j_1,i_2,\cdots};j_1,i_2,\cdots ) \oplus
\cV^F(M^d,\cT_{k_1,i_2,\cdots};k_1,i_2,\cdots ) \oplus \cdots $ .

\subsubsection{Quasiparticle fusion algebra}

%In the following, we will only consider particle-like excitations.  In this
%case the topology of the exactions, $\cT_{i_1,i_2,\cdots}$, is trivial and we
%will drop $\cT_{i_1,i_2,\cdots}$.

When we fuse two topological excitations, $i$ and $j$, of simple types together,
it may become a topological excitation of a composite type:
\begin{align}
 i\otimes j=q=k_1\oplus k_2 \oplus \cdots,
\end{align}
where $i,j,k_i$ are simple types and $q$ is a composite type.  Here the fusion
is denoted as $i\otimes j$ which represents a map that maps a
collection of  topological excitations $i_2,i_3,\cdots$ to a fusion space:
$\cV^F(M^d,\cT_{i,j,i_2,\cdots};i,j,i_2,i_3,\cdots )$.

We can also use an integer tenser $N_{ij}^k$ to describe the
quasiparticle fusion, where $i$,$j$,$k$ label simple types.  When $N_{ij}^k=0$,
the fusion of $i$ and $j$ does not contain $k$.
When $N_{ij}^k=1$, the fusion of $i$ and $j$ contains one
$k$: $i\otimes j=k \oplus k_1
\oplus k_2 \cdots$.  When $N_{ij}^k=2$, the fusion of $i$
and $j$ contains two $k$'s: $i\otimes j
=k \oplus k \oplus k_1  \oplus k_2
\cdots$.  This way, we can denote that fusion of simple types as 
\begin{align} 
i\otimes j=\oplus_k N_{ij}^k k .  
\end{align} 

In physics, the quasiparticle types always refer to simple types. The fusion
tensor $N_{ij}^k$ is another universal property of the topologically ordered
state.  The degeneracy $D(S^d;i_1,i_2,\cdots)$ is determined completely by the
fusion tensor $N_{ij}^k$ if we only have quasiparticles.

\subsubsection{Braiding properties}

If the spatial dimension is higher than 1, we can braid the 
topological excitations of codimension 2 or higher (see Fig. \ref{braid}), which will
induce an non-Abelian geometric phase described by
$N^{i,j,k,\cdots}_{\one}$-dimensional unitary matrix.  Since the overall phase
of the  unitary matrix is path dependent (which may depend on the size of the
excitations $i,j,k$, etc.,), the unitary matrices from different
braidings form a \emph{projective} unitary representation of the ``braid
group'' of  the $p$-dimensional excitations.  Such a   \emph{projective}
unitary representation is also an universal property which is independent of
(homologous) braiding paths and local perturbations to the Hamiltonian.  

For particle-like
topological excitations, even the overall phase of the unitary matrix is well
defined and path independent.  The projective unitary representation of the
braid group becomes a unitary representation, which describes the statistics of
the topological quasiparticles.

\subsubsection{Universal perturbative gravitational responses}  \label{sec:univ-grav-response}

The above fusion and braiding properties of topological excitations are not
enough to characterize topological orders.  The 2+1D $E_8$ bosonic quantum Hall
state (see Example \ref{E8}), containing no non-trivial 0-dimensional and
1-dimensional topological excitations, is a counter example.  However, if we
put the $E_8$ bosonic quantum Hall state on a curved space-time and integrate
out the bosons, we will obtain an effective theory that contains gravitational
Chern-Simons term, whose coefficient is proportional to the chiral central
charge $c_R-c_L$ of the edge state of the  $E_8$ bosonic quantum Hall state.
The  gravitational Chern-Simons term is an example of the universal
perturbative gravitational responses. We also need such gravitational responses
to characterize an topological order.

This consideration motivates us to introduce
\begin{defn} \textbf{universal perturbative gravitational responses}:\\
Putting a system on curved space-time and integrating out all matter fields
will produce an effective Lagrangian that depends on the vielbein 1-form
and the Lorentz connection 1-form.\cite{Z1253}
The universal perturbative gravitational responses correspond to terms
that only depend on the  Lorentz connection 1-form and independent of
the vielbein 1-form.
\end{defn} \noindent
The universal perturbative gravitational responses are given by the
Chern-Simons forms of the gravity.  They correspond to the volume independent
but shape dependent partition function discussed in Section \ref{PMtopo}.  They
describe the perturbative gravitational anomalies of the corresponding boundary
theory.  

It is known that gravitational Chern-Simons terms exist only in $4k+3$
space-time dimensions.  In 2+1D, there is only one kind of gravitational
Chern-Simons term, which correspond to the thermal Hall effect.  In  6+1D,
there are two kinds of gravitational Chern-Simons terms.

\subsubsection{Elementary and finite topological excitations}
\label{unbele}

We like to point out that, according to the Definition \ref{def:3-eq-rel} for
type of topological excitations, even trivial topological states (\ie the
product states) can have (infinitely many) non-trivial types of $p$-dimensional
topological excitations for $p>1$, since non-trivial topologically ordered
states can exist for spatial dimension $p>1$.\cite{CGW1038}  Adding (stacking)
a nontrivial topologically ordered state (such as a FQH state) defined on the
subspace $M^p$ to the $d$ dimensional ground state will create a nontrivial
type of $p$-dimensional topological excitation.  
%Due infinite many types of topological excitations, we need infinite data to
%describe their fusion and braiding.
This makes the description and definition of BF category very difficult.

To fix this problem, we note that most of the topological excitations are
descendant. They come from other lower dimensional topological excitations.  We
can exclude them without hurting our ability to characterize the topologically
ordered state.  So in the following, we will describe ways to exclude those
``descendant'' topological excitations.

There are three way to create ``descendant'' topological excitations:\\
(A) Adding a $p$-dimensional topological state of a qubit system to a
$p$-dimensional subspace $M^p$ creates a ``descendant'' $p$-dimensional
excitation.\\
(B) Proliferating pure topological excitations with dimensions less than $p$ on
subspace $M^p$ creates a ``descendant'' $p$-dimensional excitation.\\
(C) Assume we have $p$-dimensional topological excitation on $M^p$ which may
carry wall excitations with  dimensions less than $p$.  Proliferating  those
wall excitations on $M^p$  creates a ``descendant'' $p$-dimensional excitation.
\\
So the goal is to exclude the above three types of ``descendant'' topological
excitations.  In fact, if we allow to proliferate trivial wall excitations, the
case C include the case A and case B.  So in the following, we will only
discuss the case C.

Let consider a $p$-dimensional topological excitation labeled by $i$ on a
$p$-dimensional subspace $M^p$.  We can create a  ``descendant''
$p$-dimensional excitation $i'$ on $M^p$ by proliferating the wall excitations (or sub-defects) of $i$ on $M^p$. Since we can always choose to  proliferate the wall
excitations (or sub-defects) of $i$ in a subregion of  $M^p$, the excitations $i$ and $i'$ must have a property that they can be connected by a $(p-1)$-dimensional domain wall. Such excitations $i$ and $i'$ are also called ``Witt equivalent". It indeed defines an equivalence relation and the associated equivalence class will be called the Witt class. So we can exclude the ``descendant'' excitations by not allowing the
excitations that can be connected by a domain wall.  This will solve our
problem.

However, since the domain wall between $i$ and $i'$ can be gapless, it is hard
to describe/define condition within tensor category theory, which only deal
with gapped states. So we have to use a weaker condition:\\
(1) we do not allow excitations that can be joined alone a \emph{gapped}
domain wall.

But the condition (1) is too weak to exclude all ``descendant'' excitations.
So we will add another condition:\\
(2) we only allow finite excitations.  \\
Here the notions of \emph{finite} is defined below:
\begin{defn} 
A topological excitation $x$ is finite if 
%the finite fusion product of $x$ contains
%the trivial excitation $1$: $x^{\otimes n}  =1\oplus \cdots$. 
the set $\{ x^{\otimes n} \}_{n=1}^\infty$ contains only finite number of simple excitations as direct summands. 
\end{defn} \noindent

The finiteness is a powerful condition, as implied by the following
corollary:
\begin{cor} 
\label{topinf}
Stacking a non-invertible topologically ordered state repeatedly always generate
infinitely many different non-trivial topological orders. 
\end{cor} \noindent
So most topological excitations generated by proliferating local excitations
(\ie trivial excitations) are not finite, and can be excluded by the finiteness condition.
  
More generally, we like to conjecture that
\begin{conj} \label{elem}
If topological excitations $i$ and $i'$ described above are both finite, then
the domain wall between  $i$ and $i'$ must be gappable.  
\end{conj} \noindent 
To understand the above conjecture, let us assume that the wall excitations
that create $i'$ form a ``short-range entangled''\cite{CGW1038,ZW14} state on
$M_p$.  In this case, the domain wall between $i$ and $i'$ can be gapped.  If
the domain wall must be gapless, then the wall excitations on $M^p$ must form a
``long-range entangled''\cite{CGW1038,ZW14} state, and such a state must be
invertible in order for $i'$ to be finite.  The invertible topological orders
that belong to $Z_2$-class are discussed in Section \ref{invTop}, which first
appear in 4+1D and always have a gappable boundary.  So the finiteness of $i'$
implies that the domain wall can be gapped.

%In fact, we can make a stronger conjecture
%\begin{conj} \label{elemS}
%If topological excitations $i$ and $i'$ described above are both finite, then
%$i$ and $i'$ must be of the same type.
%\end{conj} \noindent
%This is because the finiteness implies that the  wall excitations  on $M^p$
%must form a ``short-range entangled'' state, which implies that $i$ and $i'$
%must be of the same type.

The above discussion allow us to introduce
\begin{defn}  \label{def:elementary-excitation}
\textbf{Elementary topological excitation:}\\
The set of all topological excitations is closed under the fusion operations.  Let us consider the maximal subset of topological excitations that is  closed under the fusion operations, and also satisfies
the following conditions:\\
(1) Any two different simple topological excitations with the same dimension
in the subset cannot be joined by a gapped  domain wall.\\
(2) All the  topological excitations in the subset are finite.\\
The simple topological excitations in  the subset are called elementary
topological excitations.
\end{defn} \noindent
We believe that all the topological excitations can be obtained by fusing the
elementary topological excitations with the ``descendant'' topological
excitations.  
%We also believe that there is one and only one elementary topological
%excitation in each equivalence class of basic type.  
This motivates the following conjecture
\begin{conj}  \label{conj:ele-topological-order}
The fusion and the braiding properties of the elementary topological
excitations (plus the universal perturbative gravitational responses) fully
characterize the topological order (or the corresponding BF category).
\end{conj} \noindent
Therefore, we can limit ourselves to consider only  elementary topological
excitations, and use their  fusion and braiding properties to define BF
category. We believe that all the topological orders contain only a finite
number of  elementary topological excitations.  This makes the task of defining the BF category
a finite problem. In other words, we can use a finite amount of data to
define a BF category.

\subsubsection{Examples}

We see that topological excitations in high dimensions can be very complicated.
In this section, let us give some simple examples of topological excitations.

\begin{expl} 
\label{Z2exc3D}
Consider a $\Zb_2$ topologically ordered state\cite{RS9173,W9164,MS0181,K032} in 2+1 dimensions whose effective theory is a
$\Zb_2$ gauge theory. The $\Zb_2$ charge,
denoted as $e$, is a particle-like topological excitation.  The $\Zb_2$ vortex,
denoted as $v$, is another particle-like topological excitation.  The bound
state of $e$ and $v$, denoted as $\eps$, is the third  particle-like
topological excitation.  Since the number of $e$-excitations (the $\Zb_2$ gauge
charge) is conserved mod 2, which leads to an effective $\Zb_2$ symmetry, so if
the $e$ excitations form a 1D gas, such a 1D system may spontaneously break the
$\Zb_2$ effective symmetry.  In this case, the 1D gas of $e$ becomes an
string-like topological excitation of a nontrivial type.  Similarly, the 1D gas
of $v$ and 1D gas of $\eps$ can also form two other non-trivial string-like
excitations.  The above three nontrivial string-like excitations are all Witt equivalent to the trivial string-like excitation. Thus, the 2+1D $\Zb_2$ topological order has only three nontrivial elementary topological excitations, $e$, $v$, and $\eps$, which are all ``particle-like''. It does not have any nontrivial elementary string-like topological excitations. But it has (at least) three nontrivial non-elementary string-like topological
excitations.
\end{expl}

\begin{expl} 
\label{Z2exc4D}
Consider a $\Zb_2$ topologically ordered state in 3+1 dimensions. The $\Zb_2$
charge, denoted as $e$, is a particle-like topological excitation.  The $\Zb_2$
vortex-line, denoted as $s$, is a string-like topological excitation.  There is
also a trivial  string-like excitation, denoted as $\one_1$.  We note that the
non-trivial string $s$ and the trivial string $\one_1$ cannot be connected by a gapped domain wall. In other words, the non-trivial string $s$ cannot have an end. Thus $s$ is an elementary excitation.  In fact, $e$ and $s$ are the only two pure elementary
topological excitations which are non-trivial.  On the other hand, the 3+1D
$\Zb_2$ topologically ordered state can have many nontrivial non-elementary
string-like and membrane-like topological excitations, generated by the
condensation of $e$ or $s$ on the string or membrane.
\end{expl}

In the rest of this paper (except Section\,\ref{sec:math-def}), when we consider topological excitations, we will
only consider  elementary topological excitations.  When we say there are $N$
topological excitations in a topologically ordered state, we mean there are $N$
elementary topological excitations.

%\begin{expl} 
%\label{FCnexc}
%Consider an $n$-dimensional \hBF{n}  category $C_n$ on $M^n$. Let $M^D$ is
%$D$-dimensional subspace $M^D \subset M^n$ (see Fig. \ref{FCn}).  (Note that $n$ and $D$ are the
%space-time dimensions).  We may simply view $M^D$ as a $D-1$ dimensional
%excitation, which is a trivial $D-1$ dimensional excitation.
%\end{expl}

\subsection{A complete characterization of topological order}

In the above, we discussed several universal properties of pure $p$-dimensional
topological excitations:\\
(1)  the number of elementary topological types for each $p$,\\
(2)  the fusion spaces,\\
(3)  the  projective unitary representations of the ``braid
group'' acting on the fusion spaces,
as well as\\
(4) the gravitational Chern-Simons terms.\\
Those topological data are needed to define a \hBF{}  category (or a gapped low
energy effective theory).  We hope that the above definition is complete.
Certainly, our description is not rigorous. It just illustrates the physical
ideas to develop a rigorous definition.  In Section \ref{sec:math-def}, we will
give a more rigorous definition of \hBF{}  category. In the Section
\ref{BFexample}, we will discuss some simple examples.

%, which we will compute in this paper. It turns out that this  topological
%property is directly related to another topological property for 2+1D
%topological states: The number of the  topological quasiparticle types equal
%to the ground state degeneracy on torus. This is one of many amazing and deep
%relations in topological order.

\subsection{Examples of \hBF{}  categories}
\label{BFexample}

Now let us list some examples of \hBF{}  categories, to gain a more intuitive
understanding of \hBF{}  category (or set of topological properties).  In those
examples (and in the rest of this paper), we will only consider pure elementary
topological excitations.  So when we say ``topological excitations'', we mean
``pure elementary topological excitations''.  Those examples are both \hBF{}
and \lBF{} categories.

\subsubsection{Examples of \hBF{2}  categories in 1+1D}

First let us consider 1+1D gapped systems. In this case, we can only have
particle-like topological excitations.
\begin{expl} 
\label{C2FZ2}
A 1+1D system with only one type of particle-like topological excitation
labeled by $e$ (not including the trivial type).  The fusion of two $e$'s gives
rise to a trivial excitations $e\otimes e = \one$.  In 1-dimensional space,
particle-like excitations cannot braid and there is no braiding property.  We
will denote such a \hBF{1+1} category as $\EC_2^{F\Zb_2}$.  
\end{expl}

\subsubsection{Examples of \hBF{3}  categories in 2+1D}

In 2+1D, we can have both particle-like and string-like topological
excitations.

\begin{expl} 
\label{E8}
A 2+1D system that contains no non-trivial particle-like or string-like
topological excitations.
Such a system correspond to a few copies of $E_8$ bosonic quantum Hall states.
The $E_8$ bosonic quantum Hall state is described by the following wave function with 8 kinds of bosons:
\begin{align}
\prod_{I;i<j} (z_i^I-z_j^I)^{K_{II}}
\prod_{I<J;i,j} (z_i^I-z_j^J)^{K_{IJ}}
\ee^{-\frac14 \sum_{i,I} |z_i^I|^2},
\end{align}
whose low energy effective theory is given by
\eqn{csK} with
\begin{align}
 K=\begin{pmatrix}
2&1&0&0&0&0&0&0\\
1&2&1&0&0&0&1&0\\
0&1&2&1&0&0&0&0\\
0&0&1&2&1&0&0&0\\
0&0&0&1&2&1&0&0\\
0&0&0&0&1&2&0&0\\
0&1&0&0&0&0&2&1\\
0&0&0&0&0&0&1&2\\
\end{pmatrix}.
\end{align}
Despite there is no non-trivial topological excitations (since det$(K)=1$), the
system has a non-trivial thermal Hall effect\cite{KF9732} and chiral edge
states\cite{Wedge,Wtoprev} with chiral central charge $c_R-c_L=8\times$
integer.  In other word, the system has a non-trivial perturbative gravitational
response (\ie the non-trivial thermal Hall effect).  (In some
papers,\cite{PMN1372} the $E_8$ bosonic quantum Hall state is called
short-range entangled state. However, according to our definition based on the
local unitary or local invertible transformations,\cite{CGW1038,ZW14} the $E_8$
bosonic quantum Hall state is long-range entangled.)
\end{expl}

\begin{expl} 
\label{C3FZ2b}
A 2+1D system whose only topological excitation is particle-like which is
labeled by $e$.  The fusion of two $e$'s gives rise to a trivial excitation
$e\otimes e = \one$.  The braiding property is given by the Bose statistics of
$e$.  There are no other string-like topological excitations, except the one
formed by 1D gas of $e$'s.  In fact, the $\Zb_2$ conserved boson $e$ may form a
$\Zb_2$ symmetry breaking state.  We will say such a string is formed by the
condensation of $e$ and denoted by $s$. It corresponds to the trivial Witt class 
of string-like topological excitation.  The fusion of $s$ is
given by $s\otimes s=\one$.  The only nontrivial type of topological excitation
is $e$.  We will denote such a \hBF{2+1} category as $\EC_3^{F\Zb_2b}$.
\end{expl}

\begin{expl} 
\label{C3FZ2f}
A 2+1D system whose only particle-like topological excitation is labeled by
$e$.  The fusion of two $e$'s gives rise to a trivial excitation $e\otimes e =
\one$.  The braiding property is given by the Fermi statistics of $e$.  There
are no other string-like topological excitations, except the one formed by 1D
gas of $e$'s.  In fact, the Majorana fermions $e$ may form a 1D BCS p-wave
condensed state, with Majorana zero-modes at the ends of the 1D
system.\cite{K0131} Such a string can be viewed as formed by 
the condensation of $e$\cite{B1003,YJW1306} and denoted by $s$.  
It is Witt equivalent to the trivial string-like topological excitation.  
The fusion of $s$ is given by
$s\otimes s=\one$.  Again, the only nontrivial type of topological excitation is
$e$.  We will denote such a \hBF{2+1} category as $\EC_3^{F\Zb_2f}$.
\end{expl}

\begin{expl} 
\label{C3FZ2s}
A 2+1D system whose only topological excitation is the particle-like
topological excitation labeled by $e$.  There are no other string-like
topological excitations.  The fusion of two $e$'s gives rise to a trivial
excitation $e\otimes e = \one$.  The braiding property is given by the semion
statistics of $e$.  We will denote such a \hBF{2+1} category as
$\EC_3^{F\Zb_2s}$.
\end{expl}

In the above three examples, the only particle-like topological excitation $e$
is assigned a Bose, a Fermi, or a semion statistics.  Such three choices are
consistent with the fusion rule $e\otimes e = \one$, since the bond state of
two bosons,  two fermions, or two semions is a boson.  We also discussed the
string-like topological excitations, which are all Witt equivalent to the trivial excitation.  
In the following, we will only discuss excitations in nontrivial Witt classes, which are sufficient to characterize \hBF{3}  categories (up to $E_8$
bosonic quantum Hall states in Example \ref{E8}).

\begin{expl} 
\label{C3Z2}
A 2+1D system whose only topological excitations
are three types of particle-like topological
excitations labeled by $e$, $v$, and $\eps$.  The fusion rules of the
particle-like excitations are given by $e\otimes e = v\otimes v = \eps \otimes
\eps =\one$, $e\otimes v = \eps$, $e\otimes \eps = v$, and $v\otimes \eps = e$.
The braiding properties are described by (1) $e$ and $v$ are bosons and $\eps$
is a fermion; (2) moving $\eps$ around $e$ or $v$ will induce a phase factor
$-1$.  
%As discussed above, the Majorana fermions $\eps$ may form a 1D BCS
%p-wave condensed state, with Majorana zero-modes at the ends of the 1D
%system.\cite{K0131} This corresponds to a mixed  string-like topological
%excitation (denoted as $s$).  $s$ is the only nontrivial string-like
%topological excitation.  There is no other  nontrivial string-like topological
%excitations.  The fusion of $s$ is given by $s\otimes s=\one$, and the braiding
%property give by: moving a particle-like excitation around one end of string
%$s$ will induce an automorphism $(\one,e,v,\eps)\to (\one,v,e,\eps)$.  
We will
denote such a \hBF{2+1} category as $\EC_3^{\Zb_2}$.
\end{expl}

\begin{expl} 
\label{C3Z2ds}
A 2+1D system whose only topological excitations are three types of
particle-like topological excitations labeled by $e$, $v$, and $\eps$.  The
fusion of those excitations is given by $e\otimes e = v\otimes v = \eps \otimes
\eps =\one$, $e\otimes v = \eps$, $e\otimes \eps = v$, and $v\otimes \eps = e$.
The braiding properties are described by (1) $e$ and $v$ are semions with
statistics $\pm \pi/2$ respectively and $\eps$ is a boson; (2) moving $\eps$
around $e$ or $v$ will induce a phase factor $-1$.  We will denote such a
\hBF{2+1} category as $\EC_3^{\Zb_2ds}$.
\end{expl}

\begin{expl} 
\label{C3Z2f3}
A 2+1D system whose only topological excitations are three types of
particle-like topological excitations labeled by $e$, $v$, and $\eps$.  The
fusion of those excitations is given by $e\otimes e = v\otimes v = \eps \otimes
\eps =\one$, $e\otimes v = \eps$, $e\otimes \eps = v$, and $v\otimes \eps = e$.
The braiding properties are described by (1) $e$, $v$, and $\eps$ are all
fermions; (2) they all have a mutual $\pi$ statistics.\cite{VS1306}  We will
denote such a \hBF{2+1} category as $\EC_3^{\Zb_2f^3}$.
\end{expl}

\begin{expl} 
\label{C3sFZ2}
A 2+1D system whose only topological excitations are one type of
string-like topological excitations denoted as $s$, and there is no
particle-like topological excitations. The fusion of the two string excitations
give rise to a trivial string $s\otimes s =\one$.  There is no nontrivial
braiding property between the strings.  We will denote such a \hBF{2+1}
category as $\EC_3^{sF\Zb_2}$.
\end{expl}

\subsubsection{Examples of \hBF{4}  categories in 3+1D}

In 3+1D, we can have particle-like, string-like and membrane-like (2-brane-like)
topological excitations.

\begin{expl} 
\label{C4Z2}
A 3+1D system whose only topological excitations are one type of
particle-like topological excitations denoted as $e$ and one type of
string-like topological excitations denoted as $s$.  The fusion of those
excitations is given by $e\otimes e = s\otimes s =\one$.  The only nontrivial
braiding property is the phase factor $-1$ as we move the particle $e$ around
the string $s$.  We will denote such a \hBF{3+1} category as $\EC_4^{\Zb_2}$.
\end{expl}

\begin{expl} 
\label{C4sFZ2}
A 3+1D system whose only topological excitations are one type of
string-like topological excitations denoted as $s$. The fusion of the two
string excitations give rise to a trivial string $s\otimes s =\one$.  There is
no nontrivial braiding property between the strings.  We will denote such a
\hBF{3+1} category as $\EC_4^{sF\Zb_2}$.
\end{expl}

\begin{expl} 
\label{C4mFZ2}
A 3+1D system whose only topological excitations are one type of
membrane-like topological excitations denoted as $m$.  The fusion of the two
membrane-like excitations give rise to a trivial membrane $m\otimes m =\one$.
There is no nontrivial braiding property between the membranes.  We will denote
such a \hBF{3+1} category as $\EC_4^{mF\Zb_2}$.
\end{expl}

\subsubsection{Summary}

In the above simple examples in various dimensions, we have described some
topological properties of particle-like, string-like, and membrane-like
excitations. Those properties are needed to define a \hBF{}  category.  We may
also need some additional topological properties, such as the non-Abelian
geometric phases of the degenerate ground
states\cite{Wrig,KW9327,BM0535,KS1193} and the linear relation between the
fusion spaces (see Section \ref{GBFexample}), to completely define a
\hBF{} category.

A natural question is ``can those topological properties be realized by well
defined lbH model in the same space-time dimensions?''.  In fact, some
topological properties can be realized in the same dimension and the
corresponding \hBF{} category (or topological phase) is anomaly-free, while
some other topological properties cannot be realized in the same dimension and
the correspond \hBF{} category (or topological phase) is anomalous.  This is a
issue of gravitational anomaly, which will be discussed in the next section.

\section{A general discussion of gravitational anomaly}

\subsection{Closed and exact \hBF{}  categories}
\label{ceBF}

%With the above introduction of topological order (BF category) and their topological excitations, 
We are now ready to have a general discussion of
gravitational anomaly.  Usually,  gravitational anomaly is defined through
the variance of the path integral under the homeomorphic transformations
of the space-time.  Here, we will introduce a more general
definition as in \Ref{W1313}, which are closely related to anomaly inflow (the
first examples were discovered in \Ref{L8132,CH8527}).

Let us consider an $(n+1)$-dimensional topologically ordered state $x$ and a
spatial $p$-dimensional ($p<n$) defect $y$ in $x$ (e.g. a gapped boundary of
$x$). Let us assume that the excitations on the defect are also gapped.  Such a
$p$-dimensional defect $y$ can have lower dimensional wall excitations (or
sub-defects).  These sub-defects in $y$ have the universal properties, which
again are described by a gapped low energy effective theory, i.e a \hBF{p+1}
category.  We now ask, can we realize such a $(p+1)$-dimensional effective
theory on the defect $y$ by a well-defined local $p$-dimensional lattice model
without a higher dimensional bulk?  The answer can be yes or no. This leads to
two kinds of \hBF{p+1} categories (or two kinds of gapped low energy effective
theories).  This line of thinking leads to the notion of a \emph{closed}
\hBF{n+1}  category. 
\begin{defn} \textbf{Closed \hBF{n+1}  category}\\ 
If the topological properties
of a gapped state in $n$ space-time dimensions, described by a \hBF{n+1}
category $\EC_{n+1}$, can be realized by a well-defined lbH system in the same
dimension, then $\EC_{n+1}$ is said to be closed.  
\end{defn}\noindent 

%If we assume that the \hBF{n+1}  category contains all the topological properties of a gapped phases, then from the above definition, we find
\begin{conj} 
The closed \hBF{n+1}  categories classify the topological orders (\ie the
patterns of long-range entanglement). In other words, the topological
excitations and gravitational responses in two gapped states are described by
the same closed \hBF{n+1}  category iff the two gapped states are in the same
phase. 
\end{conj}\noindent
Since a closed \hBF{n+1}  category $\EC_{n+1}$ can be realized by a lbH system
in the same dimension, this allows us to consider the boundary of the lbH
system.  If the boundary of such a system can be gapped, the topological
properties of the boundary will define a \hBF{n}  category $\EC_{n}$ in
one-lower dimension. This leads to the concept of 
\begin{defn} \textbf{Exact \hBF{n+1}  category}\\
If a $(n+1)$ space-time dimensional gapped lbH system, which realizes a closed \hBF{n+1} category $\EC_{n+1}$, can have a gapped boundary, then the \hBF{n+1}  category $\EC_{n+1}$ is said to be exact.
\end{defn}\noindent

%In the above, we have used lbH system to define \hBF{} categories and closed/exact \hBF{}  categories.  
Similarly, we can also use \emph{lbL system} to define the notions of a \lBF{} category and a closed/exact \lBF{}  category.  
%(To stress the difference, we may call the BF categories and the closed/exact
%BF categories defined above as \hBF{}  categories and closed/exact \hBF{}
%categories.)
\begin{defn} \textbf{\lBF{}  category}\\
The collection of all topological (or universal) properties of 
the instantons, the world line of particle-like topological excitations,
the world sheet of string-like topological excitations, etc., in an $(n+1)$-dimensional
space-time defines a $(n+1)$-dimensional \lBF{}  or \lBF{n+1}
category.  
\end{defn} \noindent
\begin{defn} \textbf{Closed/exact \lBF{n+1}  category}\\
If the topological properties in $n+1$ space-time dimensions, described by a
\lBF{n+1}  category $\EC^L_{n+1}$, can be realized by a well-defined \emph{lbL system} in the same dimension, then $\EC^L_{n+1}$ is said to
be closed.  If the lbL system also has a short-range
correlated boundary, then $\EC^L_{n+1}$ is said to be exact.
\end{defn}\noindent
%We will not discuss the \lBF{}  categories and the closed/exact \lBF{}
%categories in details here. We just want to stress that \emph{lbH system} and
%\emph{lbL system} are not exactly the same.  Thus, the derived concepts \hBF{}
%category and \lBF{}  category are also not exactly the same.  Later, we will
%give an example of a path integral that describe both a lbH system and a  lbL
%system. When viewed as a lbH system, the path integral describes a trivial
%\hBF{}  category. But when viewed as a lbL system, the path integral describes
%a nontrivial \lBF{}  category.

%In this paper, we will concentrate on lbH (\ie local Hamiltonian qubit)
%systems and \hBF{}  categories. 
%All the BF categories discussed in this paper are \hBF{}  categories.

\subsection{A definition of gravitational anomaly}
\label{ganm}

In the above, we have discussed whether a gapped low energy effective theory
(\ie the fusion and braiding properties of gapped topological excitations) can
be realized by a lbH system in the same dimension or has to
appear as a boundary theory of a gapped lbH system in one-higher
dimension.  More generally, a ``low energy effective theory'' is a  collection
of all the low energy properties, which may or may not be gapped.  We want to
consider when a low energy effective theory can be realized by a local
Hamiltonian system in the same dimension or has to appear as a boundary theory
of a gapped lbH system in one-higher dimension. 
%The  ``gapped low energy effective theory'' is a special kind of ``low energy
%effective theory'' which is gapped.  
This leads to the following concept:
\begin{defn} \textbf{H-type gravitational anomaly}\\
If we can realize a low energy effective theory (gapped or gapless) by a
lbH system in the same space-time dimension, we say the
low energy effective theory is free of H-type gravitational anomaly.  
\end{defn}\noindent
Let us assume that 
\begin{conj} 
A potentially anomalous $n$-dimensional low energy effective theory (gapped or
gapless) can always be realized on an $n$-dimensional defect in a $\t
n$-dimensional lbH system with an energy gap, where $\t n$ is finite and $\t
n>n$.
\end{conj}

\begin{figure}[tb] 
\begin{center} 
\includegraphics[scale=0.4]{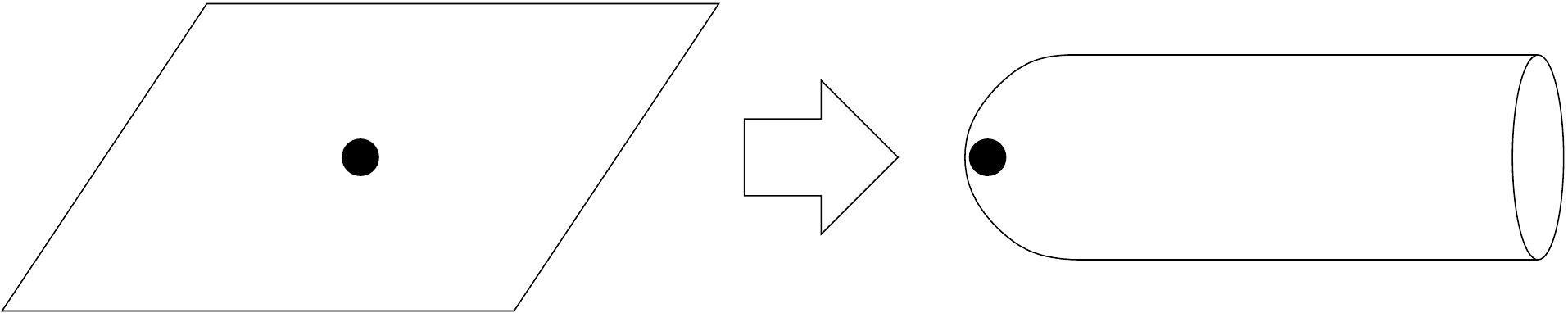} \end{center}
%Fig. 1
\caption{Dimensional reduction: 
A point defect in 2D looks like a boundary
of an effective 1D system, if we wrap the 2D
space into a cylinder. 
} 
\label{boundary} 
\end{figure}

Note that given an $n$-dimensional defect $M^n$ in a higher $\t n$-dimensional
space $\t M^{\t n}$, we can always deform the higher dimensional space $\t
M^{\t n}$ so that the defect looks like a boundary when viewed from far away
(see Fig.  \ref{boundary}).\cite{W1313} This process will be called {\it
dimensional reduction}. We have the following result of dimensional reduction. 
\begin{cor}  \label{cor:dim-reduction}
A potentially anomalous $n$-dimensional low energy effective theory (gapped or
gapless) can always be realized by a boundary of a $(n+1)$-dimensional local
Hamiltonian system with an energy gap.  
\end{cor}
\begin{rema}
We see that a $p$-dimensional excitation can be viewed as an anomalous
\hBF{p+1} category in $(p+1)$ dimensional space-time. A simple
$p$-dimensional excitation may, however, correspond to an composite \hBF{p+1} category.
\end{rema}

It is clear that
the low energy effective theory on the boundary of short-range entangled state
can be realized as a pure boundary theory without the bulk.  So the low energy
effective theory on the boundary of short-range entangled state is always free
of gravitational anomaly, while the low energy effective theory on the boundary
of long-range entangled state always have gravitational anomaly. This line of
thinking allows us to show that
\begin{cor} 
(1) The H-type gravitational anomalies in $n$ space-time dimensions are
classified by topological orders\cite{Wtop,Wrig} (\ie patterns of long-range
entanglement\cite{CGW1038}) in one-higher dimension. In other words, the H-type
gravitational anomalies in $n$ space-time dimensions are classified by closed
\hBF{n+1}  categories  $\EC^\text{closed}_{n+1}$ in
one-higher dimension.  \\ 
(2) The gapped H-type gravitational anomalies in $n$ space-time dimensions are
classified by exact \hBF{n+1}  categories
$\EC^\text{exact}_{n+1}$ in one-higher dimension.\\ 
(3) A gapped system described by a \hBF{n}  category $\EC_n$ has a H-type
gravitational anomaly if $\EC_n$ is not closed.  So a non-closed  \hBF{n}  category
$\EC_n$ describes a gravitationally anomalous theory of H-type.  We also call a
non-closed  \hBF{n}  category as an anomalous \hBF{n}  category.
\end{cor}\noindent

%In this paper, we will use tensor category theory and tensor network 
%approach to study topological order and the boundary of topologically ordered
%states, following \Ref{KK1251,LN0701,GW0931}.  As discussed above, such a
%theory is also a theory of gravitational anomaly.

%First, we need to find a quantitative way (a data set) to describe the
%topological properties of a gapped low energy effective theory which may or may
%not have gravitational anomaly.  In fact, in this paper, we will refer such a
%data set as a  \emph{gapped low energy effective theory}.  Second, we need to
%show that the gapped low energy effective theory can always realized as a
%boundary of a state in one-higher dimension.  If the state in one-higher
%dimension has a trivial topological order, then, the  gapped low energy
%effective theory is free of gravitational anomaly.  If the state in one-higher
%dimension has a nontrivial topological order, then, the  gapped low energy
%effective theory has a gravitational anomaly.

\medskip
Similarly, we can also define
\begin{defn} \textbf{L-type gravitational anomaly}\\
If we can realize a low energy effective theory (gapped or gapless) by a
lbL system in the same space-time dimension, we say the
low energy effective theory is free of L-type gravitational anomaly.  
\end{defn}\noindent
We also have 
\begin{conj} 
a potentially anomalous $n$-dimensional low energy effective theory (gapped or
gapless) can always be realized by a boundary of a $(n+1)$-dimensional lbL system with an energy gap.  
\end{conj}\noindent
Thus
\begin{cor} 
(1) The L-type gravitational anomalies in $n$ space-time dimensions are
classified by closed \lBF{n+1}  categories $\EC^\text{L,closed}_{n+1}$ in one-higher
dimension.  \\ 
(2) The short-range correlated L-type gravitational anomalies in $n$ space-time
dimensions are classified by exact \lBF{n+1}  categories $\EC^\text{L,exact}_{n+1}$
in one-higher dimension.\\ 
(3) A system described by a \lBF{n}  category $\EC^L_n$ has a gravitational
anomaly if $\EC^L_n$ is not closed.  So a non-closed  \lBF{n}  category $\EC_n$
describes a gravitationally anomalous theory of L-type.  
\end{cor}\noindent

%Later, we will show that fractional quantum Hall systems in 2+1D (described by
%closed \hBF{3}  categories) are examples of quantum systems that are free of
%H-type gravitational anomalies, but have  L-type gravitational anomalies (see
%Sections \ref{WWmdl} and \ref{WWmdl1}).

We have listed many examples of BF categories in Section \ref{BFexample}.  In
Appendix \ref{app:examples}, we will discuss those simple examples further to
illustrate the notions of exact, closed, and anomalous BF categories, and to
see how those simple examples fit into the above three classes of BF
categories. 

\section{Boundary-bulk relation for BF categories in different dimensions} 
\label{sec:boundary-bulk-relation}

The results in this section apply to both \hBF{}  and \lBF{}  categories.  We
will refer them as BF categories.

\subsection{The boundary of a given bulk}

We have introduced $\BF_n$ category to describe a set of topological
excitations in $n$-dimensional space-time.  
%Let us assume that the $p$-dimensional topological excitations all have a
%topology of D-sphere.  
Those topological excitations have a property that they are closed under the
local fusion and braiding operations (see Fig.  \ref{fuse} and  Fig.
\ref{braid}).  Their fusion rules braiding properties are consistent among
themselves.  We also introduced the notions of closed/exact $\BF_n$  category.

%also satisfy an addition condition: any two different topological
%excitations in the set can be distinguished by their different braiding
%properties with some other topological excitations in the set.  We will also
%refer a closed-set of topological excitations as a closed BF category.

The (generic) BF categories in $n$ space-time dimension are closely related to
the exact BF categories in $n+1$ space-time dimension. In this section, we will
explore this relation in details. 

Consider a well-defined gapped state in $n+1$ space-time dimension, whose
topological excitations are described by an exact BF category $\EC_{n+1}$.
Since  $\EC_{n+1}$ is exact, the gapped state in $n+1$ space-time dimension has
a gapped boundary of $n$ space-time dimension.  Some of the  topological
excitations on the  gapped boundary come from the topological excitations in
the bulk, while others are confined on the boundary and only appear on the
boundary.  Since the boundary topological excitations can still fuse and braid
within the $n$-dimensional boundary, they are described by a BF category
$\EC_n$. In general, such $\EC_n$ is not unique. Moreover, even if $\EC_n$ is
fixed, we still can not fix the boundary type. For example, in toric code
model, there are two types of gapped boundaries: a rough boundary and a smooth
boundary (see Fig.\,\ref{toric}). The boundary excitations in both cases are given by the unitary
fusion category $\rep_{\Zb_2}$, which is the category of representations of
$\Zb_2$ group. 

\subsection{The bulk-to-boundary map}

The additional data that is needed to determine the boundary is the so-called
the bulk-to-boundary map, which is a functor $f: \EC_{n+1} \to \EC_n$ that maps
the bulk topological excitations into a subset of boundary topological
excitations. Those bulk topological excitations, that are mapped into the
trivial excitations on the boundary, are said to be condensed on the boundary.
For example, in the toric code model, the smooth boundary in Fig.\,\ref{toric}
corresponds to the condensation of $e$-particles; the rough boundary in
Fig.\,\ref{toric} corresponds to the condensation of $m$-particles. In general,
there might be different sets of boundary excitations $\EC_n$ for a given
$\EC_{n+1}$ as in Levin-Wen types of lattice models\cite{KK1251}.  
%For example, in a Levin-Wen lattice model based on a unitary fusion category
%$\ED$, the bulk excitations are given by the monoidal center $\cZ(\ED)$ of
%$\ED$, and boundary excitations can be the fusion category $\EC_\EM^\ast=
%\fun_\ED(\EM, \EM)$ of $\ED$-module functors from a $\ED$-module to itself. 
So an $n$-dimensional boundary of a given $(n+1)$-dimensional $\EC_{n+1}$ is
determined by a pair $(\EC_n, \EC_{n+1} \xrightarrow{f} \EC_{n})$. The functor
$f$ can not be arbitrary. It must satisfy some consistency conditions. We will
return to this point later. 

\begin{rema}
As we will argue later that $\EC_n$ determines the \bulk $\EC_{n+1}$ uniquely up to isomorphisms. If $\EC_{n+1}$ has non-trivial automorphisms, then a bulk-to-boundary map can be twisted by these automorphisms to give different bulk-to-boundary maps. The non-trivial automorphism is physically detectable. So we will treat $\EC_n$ associated to a different bulk-to-boundary maps as different boundary types.   There is no contradiction between different boundary types associated to the same $\EC_n$ and the uniqueness of the \bulk up to isomorphisms in Lemma\,\ref{lemma:unique-bulk}.
\end{rema}

Another way to characterize a boundary is to specify the condensation (of the
bulk excitations) that can create a trivial condensed phase $\one_{n+1}$ and a
gapped boundary. For example, in the Levin-Wen models, the boundary types can
be classified by Lagrangian algebras in the tensor category of bulk
excitations\cite{kong-anyon}. Two types of boundary in toric model corresponds to two different condensations\cite{kong-anyon}. 
%Actually, the condensation of bulk excitations is a more general phenomena
%than the gapped boundary. In general, a gapped domain wall can be created
%between the initial phase and the condensed phase, which is not necessarily
%trivial. This general situation is not the interest of this paper. 
What is really important to us is that if an $n+1$ dimensional BF category
$\ED_{n+1}$ is exact, it is reasonable that one can always create the trivial phase and a gapped
boundary via a condensation of the bulk excitations in $\ED_{n+1}$.  We assume this for the rest of this section. We will use it, in particular, in the proof of Theorem\,\ref{thm:w=gw}. 
%Before we discuss the boundary-bulk relation in the reversed order.  We would
%also like to remark that the non-uniqueness of boundary types for a given bulk
%theory is a ubiquitous phenomenon which might have deep
%meanings\cite{Kong-new-geometry}.

\subsection{The bulk of a given boundary}
\label{cfun}

Now we consider the bulk-boundary relation in the reversed order.
Given an $n$ space-time dimensional boundary BF category $\EC_n$, it turns out
that it determines uniquely an $(n+1)$ space-time dimensional \bulk BF category
$\EC_{n+1}$. Here, we must make it very clear what we mean by ``a bulk". A
given $n$ space-time dimensional $\BF$ category can always be realized as an
$n$-dimensional defect in a higher dimensional (possibly trivial) topological
order. But such realization is almost never unique. However, by the dimensional
reduction given in Fig.\,\ref{boundary} and Corollary\,\ref{cor:dim-reduction},
we can always reduce such a realization down to an exact $(n+1)$ space-time
dimensional BF category $\EC_{n+1}$ with a gapped boundary given by $\EC_n$.
Such BF category $\EC_{n+1}$ is unique.  This will be our first important
result (Lemma\,\ref{lemma:unique-bulk}), which leads to many interesting
consequences. 

Before we state Lemma\,\ref{lemma:unique-bulk}, let us first state a generalization of the results in \Ref{CGW1038} to the anomalous topological phases.
\begin{lemma}  \label{lemma:H-H'}
If there are two $n+1$-dimensional lbH systems $H$ and $H'$ realizing the same
$n$-dimensional topological phase as their boundaries, then there is a
neighborhood $U$ of the boundary such that the restriction of $H$ in $U$,
denoted by $H|_U$, can be deformed smoothly to $H'|_U$ without closing the gap.
In other words, there is a smooth family $H_t$ for $t\in [0,1]$ without closing
the gap such that $H_0=H$, and $H_r$ and $H_s$ differ only in $U$ for $s,t\in
[0,1]$ and $H_1|_U=H'|_U$. 
\end{lemma}

\begin{lemma} \label{lemma:unique-bulk}
Two exact BF categories $\EC_{n+1}$ and $\EC_{n+1}'$ must be equivalent $\EC_{n+1}=\EC_{n+1}'$ if they can have gapped boundaries described by the same BF category $\EC_n$.
\end{lemma}
\begin{proof}
We need use Lemma\,\ref{lemma:H-H'}. Let $H$ be a local Hamiltonian qubit system that realizes the topological bulk phase $\EC_{n+1}$, 
and $H'$ be the one that realizes the topological bulk phase $\EC_{n+1}'$. 
By Lemma\,\ref{lemma:H-H'}, we are able to deform $H_0=H$ smoothly only in a neighborhood $U$ of the boundary such that $H_t$ does not close the gap for all $t\in [0,1]$ and $H_1|_U=H'|_U$.
Therefore, we can connect two local Hamiltonian qubit systems $H$ and $H'$ by adding a region which contains only a neighborhood $V$ of the boundary depicted in Fig.\,\ref{fig:unique-bulk} as the dotted box. 
In the region $V$, the lbH system is smoothly deformed from $H|_U$ to $H'|_U$. The remaining bulk are glued by brutal force, which creates a domain wall labeled by $X$ between two bulk phases as shown in Fig.\,\ref{fig:unique-bulk}. 

 If the domain wall is trivial, then we are done. Assume that the domain $X$ is non-trivial. Then the domain wall
must end near the boundary but outside the region $V$ and create a non-trivial defect junction (a
defect of codimension 2). Since, in the dotted neighborhood of the boundary (see Fig.\,\ref{fig:unique-bulk}), 
$H_t$ does not close the gap, all observables $\langle O\rangle(t)$, including the topological excitations, 
can cross from the left side of $V$ to the right side of $V$ smoothly (without crossing any singularities).
Consequentially, there is no macroscopic detectable defects between the boundary
and the defect junction. 

Since the bulk phase and the wall phase are
topological, we can move the $X$-wall up and create a larger neighborhood 
and continue this move until the domain wall $X$ is completely removed. As a consequence, two
exact BF categories $\EC_{n+1}$ and $\EC_{n+1}'$ must be equivalent.  
\end{proof}

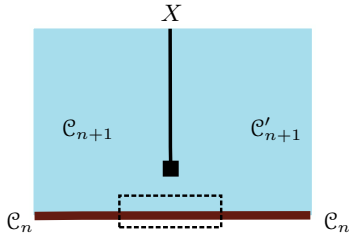
\begin{figure}[t] 
$$
 \begin{picture}(120, 70)
   \put(0, 0){\scalebox{1}{\includegraphics{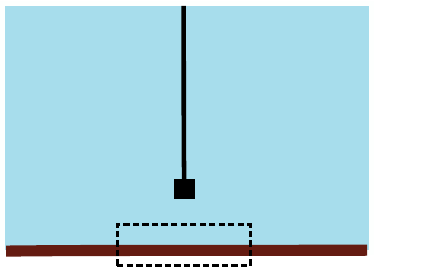}}}
   \put(0, 0){
     \setlength{\unitlength}{.75pt}\put(-18,-19){
     \put(125, 67)     { $\EC_{n+1}'$ } 
     \put(80, 125)   { $ X $}
     \put(30, 67)     { $\EC_{n+1}$ }
     \put(2, 20)   { $\EC_{n}$ }     
     \put(162, 20)  { $\EC_n$}
     }\setlength{\unitlength}{1pt}}
  \end{picture}
$$  
  \caption{{\small The bulk phase of a given boundary is unique. The dotted box $V$ on the boundary represents the region where the Hamiltonian is smoothly deformed from $H_0$ to $H_1$.}}
  \label{fig:unique-bulk}
\end{figure}

In other words, two gapped topological states belong to the same phase if they
can have gapped boundaries which are of the same type (\ie described by the
same BF category). As a consequence, we introduce the following definition. 

\begin{defn}  \label{def:bulk}
If an $n$-dimensional BF category $\EC_n$ describes the topological excitations
on the boundary of a gapped lbH system in  $(n+1)$-dimensional space-time, then we call the unique $\BF_{n+1}$ category determined by this $(n+1)$-dimensional space-time gapped lbQH system as the \bulk of $\EC_n$, denoted by $\cZ_n(\EC_n)$. 
\end{defn}

%We will give a very important remark which will be constantly referred to in this work. 
\begin{rema}  \label{rema:bulk=center}
We choose the letter ``$\cZ$" because it is also used in algebra for the
notion of center. In particular, $\cZ_n$ is somewhat similar to the so-called
$E_n$-center\cite{lurie2}. We will show in the next paper\cite{kong-wen-zheng}
that the \bulk is exactly equivalent to the mathematical notion of
center\cite{lurie2}. In this work, we don't need this result in such
generality. We only need it in a few lower dimensional cases. In the case of
$3$ space-time dimensional topological phases with a gapped boundary, this
result has already been rigorously proved. Indeed, consider a Levin-Wen type of
lattice model with a bulk lattice constructed from a unitary fusion category
$\EC$ and a boundary lattice from a $\EC$-module $\EM$. It was proved
rigorously in \Ref{KK1251} that the excitations on the boundary constructed on
an ${}_\EC \EM$-lattice,  are given by the unitary fusion category
$\EC_\EM^\vee:=\fun_\EC(\EM, \EM)$ of $\EC$-module functors. The bulk
excitations is given by $\cZ(\EC)$ which is the monoidal center of
$\EC$.\cite{LWstrnet,KK1251}. And we have $\cZ(\EC) \simeq \cZ(\EC_\EM^\vee)$
as unitary modular tensor categories\cite{ostrik}. A model independent proof of
this boundary-bulk relation for $3$ space-time dimensional theories was also
given in \Ref{fsv}. For this reason, it is harmless for readers to take the
\bulk $\cZ$ simply as a synonym of the center.  
\end{rema}

It is clear that $\cZ_n(\one_n) \simeq \one_{n+1}$, where $\one_{n/n+1}$ is the
$n/n+1$-dimensional trivial phase. Using this notation, a closed $\BF_n$
category $\EC$ means $\cZ_n(\EC_n)\simeq \one_{n+1}$, and an exact $\BF_{n+1}$
category is equivalent to $\cZ_n(\EC_n)$ for some $\BF_n$ category $\EC_n$. 

%Moreover, the notation $\cZ_{n-1}(\EC_n)$ also makes sense when $\EC_n$ is
%viewed as a $\BF_{n-1}$ category. Therefore, $\cZ_{n-1}(\EC_n)$ is a $\BF_n$
%category. We will give an explicit formula of this category later when $\EC_n$
%is closed. In general, the notation $\cZ_n(\EC_m)$ for $m>n$, where $\EC_m$ is
%viewed as a $\BF_n$ category,  defines a $(n+1)$ dimensional BF category. 

The following result follows from Lemma\,\ref{lemma:unique-bulk} immediately. 
\begin{cor} \label{cor:Z2=0}
The \bulk of the \bulk of a $\BF_n$ category $\EC_n$ is trivial, i.e. $\cZ_{n+1}(\cZ_n(\EC_{n}))= \one_{n+2}$.
\end{cor}

\begin{rema} \label{rema:center2=0}
As we remarked in Remark\,\ref{rema:bulk=center} that the \bulk is equivalent to the center, Corollary\,\ref{cor:Z2=0} also means that the center of a center is trivial. This is a very interesting and non-trivial result, which is the dual of the statement that the boundary of a boundary is empty. As the physical or geometric intuition is so obvious, its mathematical meaning is very non-trivial\cite{kong-wen-zheng}. We believe that the triviality of the center of a center is also a robust phenomena which can be proved in many different contexts using different notions of center in mathematics. For example, it seems also plausible to establish this result in the framework of factorization algebras\cite{lurie2}. 
\end{rema}

\void{
\begin{figure}[t] 
$$
 \begin{picture}(120, 70)
   \put(-25, 0){\scalebox{1}{\includegraphics{pic-ZZ-zero}}}
   \put(-25, 0){
     \setlength{\unitlength}{.75pt}\put(-18,-19){
     \put(165, 67)     { $\one_{n+2}$ } 
     \put(122, 122)   { $ \ED_{n+1}' $}
     \put(80, 67)     { $\EC_{n+2}$ }
     \put(200, 20)  { $\one_n$}
     \put(28, 20)   { $\ED_{n+1}$ }
     \put(122, 8) { $\EE_{n}$}
     }\setlength{\unitlength}{1pt}}
  \end{picture}
$$  
  \caption{{\small $\cZ_{n+1}(\cZ_n(-))=\one_{n+2}$.}}
  \label{fig:ZZ-zero}
\end{figure}
}

%We also say that $\EC_{n+1}$ is the center of $\EC_n$:
%$\EC_{n+1}=\cZ_n(\EC_n)$.  Those different $\EC_n$'s (for different
%boundaries) all have the same center $\cZ_n(\EC_n)$ (\ie the same bulk state).
%This leads to the following definition and conjecture: ??  \noindent

\section{Monoidal and group structure of $\BF_n$ categories in the same
dimension}

\subsection{A tensor product of $\BF_n$ categories}
\label{mnd}

Let a closed BF category $\EC^1_n$ be realized by a gapped lbH system
$\La_1$, and a closed BF category $\EC^2_n$ by another gapped lbH system
$\La_2$.  If we stack the two lbH systems together to form an lbH system
$\La_{12}$, then the gapped model $\La_{12}$ will give rise to a new closed BF
category, denoted by $\EC_n^{1}\boxtimes \EC_n^{2}$.

\void{
\begin{defn} 
The closed BF category $\EC_n^{12}$ constructed above is the tensor product of
$\EC^1_n$ and $\EC^2_n$, and we denote such a tensor product as
$\EC_n^{12}:=\EC_n^{1}\boxtimes \EC_n^{2}$.
\end{defn}
}

This defines a tensor product among all closed BF categories. It turns out that such tensor product can be generalized to generic BF categories.  Let an $n$-dimensional BF category $\EC^1_n$ be realized by the boundary of a gapped
lbH system $\La_1$ in $(n+1)$-dimensional space-time.  The
topological excitations in the model $\La_1$ is described by an exact
$(n+1)$-dimensional BF category $\EC^1_{n+1}$. Similarly, let another
$n$-dimensional BF category $\EC^2_n$ be realized by the boundary of a gapped
lbH system $\La_2$.  The topological excitations in $\La_2$
is described by an exact BF category $\EC^2_{n+1}$.  If we stack the two local
Hamiltonian systems together to form the third lbH system
$\La_{12}$, then the gapped boundary of $\La_{12}$ will give rise to a BF
category denoted by $\EC_n^{1}\boxtimes \EC_n^{2}$.

\void{
\begin{defn} 
The BF category $\EC_n^{12}$ constructed above is the tensor product of
$\EC^1_n$ and $\EC^2_n$, and we denote such a tensor product as
$\EC_n^{12}:=\EC_n^{1}\boxtimes \EC_n^{2}$.
\end{defn}
}

\medskip
Let $\cM^n$ be the set of 
%equivalence classes of 
$\BF_n$ categories. The tensor product $\boxtimes$ defines a multiplication on the set $\cM^n$. It is clear that $\one_n \boxtimes \EC_n \simeq \EC_n \simeq \EC_n \boxtimes \one_n$, and the multiplication is associative, i.e. $(\EC_n \boxtimes \ED_n) \boxtimes \EE_n \simeq 
\EC_n \boxtimes (\ED_n \boxtimes \EE_n)$, and commutative, i.e. $\EC_n \boxtimes \ED_n \simeq \ED_n \boxtimes \EC_n$. 
\begin{lemma}
The multiplication $\boxtimes$ and the unit $\one_n$ provide the set $\cM^n$ a structure of commutative monoid. 
\end{lemma}

Notice that the tensor product commutes with $\cZ$. More precisely, we have 
\begin{equation} \label{eq:Z-boxtimes}
\cZ_n(\EC_n \boxtimes \ED_n) \simeq \cZ_n(\EC_n) \boxtimes \cZ(\ED_n)
\end{equation}
as $\BF_{n+1}$ categories. In other words, $\cZ_n: \cM^n \to \cM^{n+1}$ is a homomorphism between monoids.

\begin{rema} \label{rema:compBF}
We also like to make a remark on the action of the center functor on composite BF categories $\EC_n\oplus\ED_n$.  In general, $\cZ_n(\EC_n\oplus\ED_n) \neq \cZ_n(\EC_n)\oplus \cZ_n(\ED_n)$.
When $\EC_n = \ED_n$, we have $\cZ_n(\EC_n\oplus\EC_n) = \cZ_n(\EC_n) \times M_{2\times 2}$,\cite{kong-wen-zheng}   where $M_{2\times 2}$ is the $2\times 2$ matrix algebra. The phase $\cZ_n(\EC_n) \times M_{2\times 2}$ is unstable and can flow to the stable one $\cZ_n(\EC_n)$. 
%This is because $\EC_n\oplus\ED_n$ can appear as a boundary of $\cZ_n(\EC_n)$ or $\cZ_n(\ED_n)$.  Such as result is related to how we define composition for anomalous BF categories $\EC_n$ and $\ED_n$. If $\cZ_n(\EC_n) \neq \cZ_n(\ED_n)$, the composition $\EC_n\oplus\ED_n$ also implies a composition of the \bulk $\cZ_n(\EC_n)\oplus \cZ_n(\ED_n)$. If $\cZ_n(\EC_n) = \cZ_n(\ED_n)$, the composition $\EC_n\oplus\ED_n$ only implies a composition at the boundary, without affecting the \bulk. 
The composition $\oplus$ and the tensor product $\boxtimes$, together with the tensor unit $\one_n$ and zero category $0_n$, give a commutative ring structure to all BF categories. We will not discuss it further in this work. More details will be given \Ref{kong-wen-zheng}. 
\end{rema}

\begin{rema}  \label{rema:general-boxtimes}
A general tensor product between two $n$-dimensional BF categories can be
defined. We can stack one topological order $\EC_n$ on the top of the other
$\ED_n$ and glue them by inserting between them a $(n+1)$-dimensional ``glue",
the topological type of which is given by a $(n+1)$-dimensional BF category
$\EE_{n+1}$. Such a physical gluing process creates a new (possibly anomalous)
$n$-dimensional topological phase, denoted by $\EC_n\boxtimes_{\EE_{n+1}}
\ED_n$, where $\boxtimes_{\EE_{n+1}}$ defines a new type of tensor product. It
actually contains the old tensor product $\boxtimes$ as a special case, i.e.
$\boxtimes = \boxtimes_{\one_{n+1}}$. 
%Notice that if $\EE_{n+1}$ is not trivial, then both $\EC_n$ and $\ED_n$ are not closed. If the new $n$-dimensional topological order $\EC_n\boxtimes_{\EE_{n+1}} \ED_n$ is closed, then we must have $\cZ_n(\EC_n) \simeq \EE_{n+1} \simeq \cZ_n(\overline{\ED}_n)$. 
\end{rema}

Let $\cM^n_{\text{closed}}$ and $\cM^n_{\text{exact}}$ be the subsets of $\cM^n$ consisting of the equivalence classes of closed and exact $\BF_n$ categories, respectively. Clearly, the multiplication $\boxtimes$ is closed on the subsets $\cM^n_{\text{closed}}$ and $\cM^n_{\text{exact}}$. Therefore, $\cM^n_{\text{closed}}$ and $\cM^n_{\text{exact}}$ are two sub-monoids of $\cM^n$. 

\smallskip
%The tensor product makes the BF categories in the same dimensions to form a monoid.  
A monoid is not a group since the inverse may not exist. In our case,
there is no group structure on $\cM^n$. Because the long range entanglement on
different layer can not cancel each other,  a double layer system $\EC_n
\boxtimes \ED_n$ has no long range entanglement if and only if each factor has
no long range entanglement. In other words (see also Conjecture \ref{topinf}), 
\begin{align}
\label{eq:C=0=D}
\EC_n \boxtimes \ED_n &\simeq \one_n \ \ \mbox{iff } \EC_n \text{ and } \ED_n \text{ have no non-trivial}
\nonumber\\
&
\text{elementary topological excitations}.
\end{align}
Equivalently, all elements in $\cM^n$, that have non-trivial elementary
topological excitations, are not invertible.  On the other hand, the  elements
in $\cM^n$, that have no non-trivial elementary topological excitations, are
invertible. 

In order to obtain a group structure, we have to consider certain quotient sets
of $\cM^n$, and to obtain the  quotient sets, we need to introduce a few new
concepts. 

%\subsection{Reduced tensor product of BF categories in the same dimension }
%
%Furthermore, if we view $\La$ and  $\bar\La$
%in Fig.  \ref{TConj}b as a boundary of a local Hamiltonian qubit system in $n+1$ dimensional
%space-time (the $(n+1)$-dimensional model fills between the two sheets $\La$
%and  $\bar\La$), we can even obtain
%\begin{cor}
%$\EC_{n}\boxtimes_R \overline{\EC}_{n}$ is exact for any BF category $\EC_n$,
%where $\one$ represents a trivial \hBF{}  category.
%\end{cor}

%Since the stacking of $\La$ and $\bar\La$
%can have a trivial boundary,
%the corresponding bulk BF category $\EC_{n+1}\boxtimes \overline{\EC}_{n+1} $ must be
%trivial as well. Thus, we obtain
%\begin{cor}
%$\EC_{n+1}\boxtimes \overline{\EC}_{n+1} = \one$, where $\one$ represents a
%trivial BF category.
%\end{cor}
%Since the ordering of the stacking is not important, we also have
%\begin{cor}
%BF categories in the same dimension plus the stacking operation $\boxtimes$ 
%form an Abelian group.
%\end{cor}

\subsection{Dimension reduction of a BF category}

\begin{figure}[tb]
  \centering
  \includegraphics[scale=0.6]{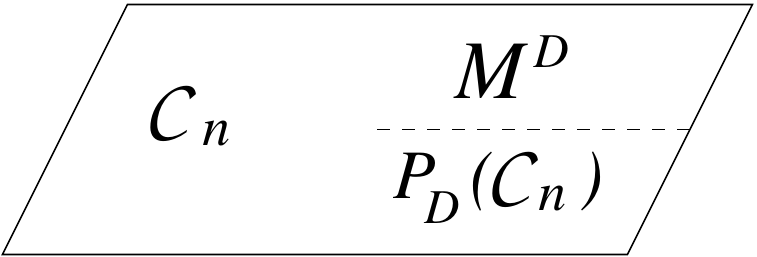}
  \caption{Consider a BF category $\EC_n$ realized by a boundary $M^n$ of a
system in $(n+1)$-dimensional space-time. Now consider a $D$-dimensional sub
space-time $M^D$ on the $n$-dimensional boundary. The topological excitations
in the $D$-dimensional sub space-time $M^D$ all come from the BF category
$\EC_n$.  All  topological excitations and their topological properties define
a $D$-dimensional category $P_D(\EC_n)$.
}
  \label{FCn}
\end{figure}

A $n$-dimensional BF category (\ie a \hBF{n} or \lBF{n} category)
can be viewed as a $\BF_p$ category for $p<n$:
\begin{defn} \textbf{The project functor $P_D$}\\
\label{FDCn}
Consider an $n$-dimensional BF category $\EC_n$ on $X^n$ which is a boundary of
$(n+1)$ dimensional space-time. Let $M^D$ be a $D$-dimensional sub space-time
$M^D \subset M^n$ (see Fig. \ref{FCn}).  If we view $M^D$ as a subsystem whose
topological excitations all come from the $n$-dimensional $\EC_n$,  then
topological excitations on $M^D$ define a $D$-dimensional BF category which is
denoted as $P_D(\EC_n)$.  $P_D$ is a functor that maps $\EC_n$ to
$\EC_D$. 
\end{defn}\noindent
We call $P_D(\EC_n)$ a projection of $\EC_n$ from $n$-dimensions to
$D$-dimensions.  We know that if $M^D$ contains no topological excitations,
then the BF category on $M^D$ is trivial.  The BF category $P_D(\EC_n)$ on
$M^D$ does contain topological excitations and thus nontrivial. But all the
topological excitations come trivially from its higher dimensional parent.  So,
in some sense, $P_D(\EC_n)$ is ``trivial''.  In the following, we are going to
introduce an equivalence relation $\sim$ between BF categories that makes
$P_D(\EC_n)$ equivalent to  trivial BF category when $\EC_n$ is closed.

%Using such a projection, we can have different understanding
%of the equivalence relation defined in Definition \ref{equvCC}:
%\begin{thm}
%Two $n$-dimensional BF categories, $\EC_n$ and $\EC'_n$, are equivalent $\EC_n
%\sim \EC'_n$ iff\\ 
%(1) $\cZ_n(\EC_n)=\cZ_n(\EC_n')$,\\
%(2) there exists a boundary between $\EC_n$ and $\EC_n'$ which is gapped and is
%described by a BF category of the form $P_{n-1}(\ED)$ for a certain
%\emph{closed} BF category $\ED$.
%\end{thm}\noindent

%We note that the boundary between  $\EC_n$ and $\EC_n$ is $P_{n-1}(\EC_n)$.
%Also the boundary between $\EC_n \boxtimes P_n(\ED)$ and $\EC_n$ is
%$P_{n-1}(\ED)\boxtimes P_{n-1}(\EC_n)=P_{n-1}(\ED \boxtimes \EC_n)$.  

\subsection{Dual BF category}

\begin{figure}[tb]
\begin{center}
\includegraphics[scale=0.35]{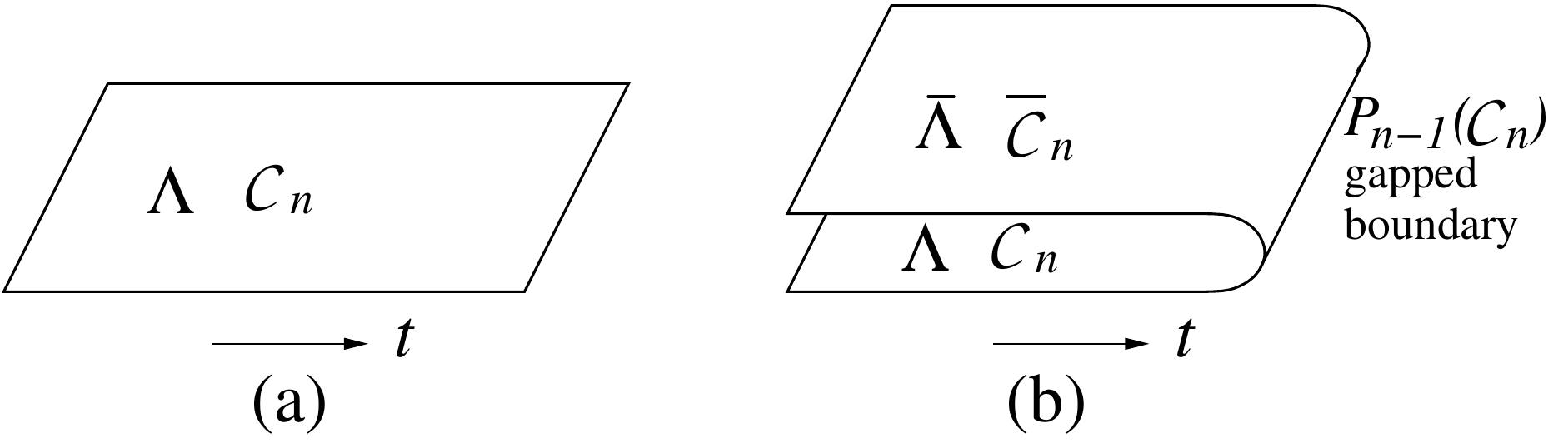}
%Fig. 1
\end{center}
\caption{
(a) A local Hamiltonian qubit system $\La$ defined by a path integral in $n$
space-time dimensions.  Its topological excitations are described by a
$n$-dimensional BF category $\EC_n$.
(b) If we fold the time direction, we obtain a local Hamiltonian qubit system
$\overline{\La}$ which is the time-reversal transformation of the local Hamiltonian
qubit system $\La$.  Its topological excitations are described by a
$n$-dimensional BF category $\bar C_n$.  The gapped boundary of
$\EC_n\boxtimes \overline{\EC}_n$ is described by  a
$(n-1)$-dimensional BF category $P_{n-1}(\EC_n)$:
$\cZ_n[P_{n-1}(\EC_n)]=\EC_n\boxtimes \overline{\EC}_n$.  
%Thus $\EC_n \boxtimes
%\overline{\EC}_n \sim \one_n$, since $\EC_n$ is closed.
}
\label{TConj}
\end{figure}

\begin{figure}[tb]
\begin{center}
\includegraphics[scale=0.35]{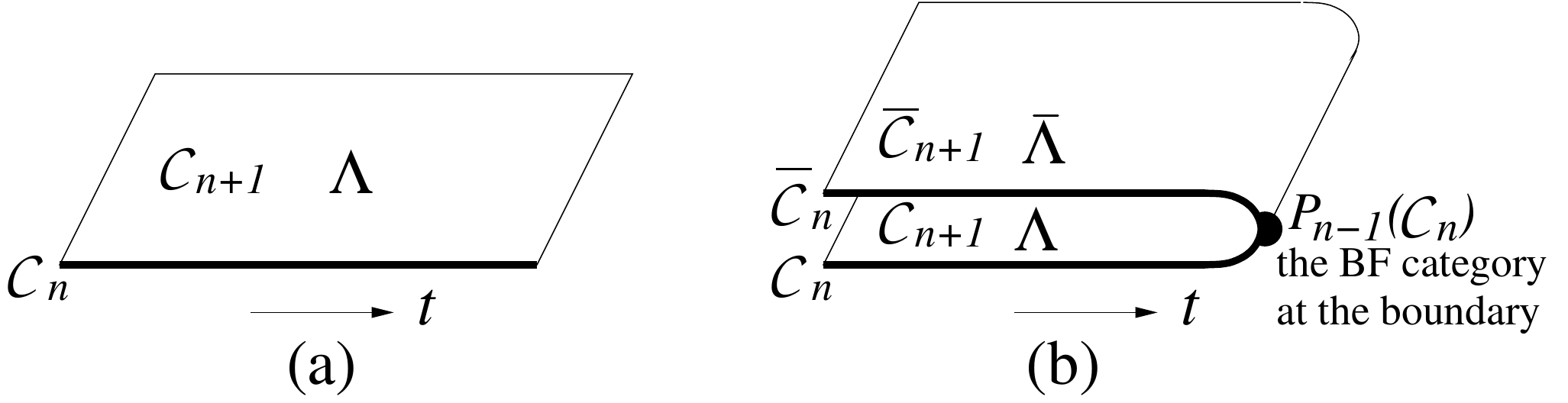}
%Fig. 1
\end{center}
\caption{(a) A local Hamiltonian qubit system $\La$ defined by a path integral
in $n+1$ space-time dimensions.  Its topological excitations are described by
an exact $n+1$-dimensional BF category $\EC_{n+1}$.  The topological
excitations on its gapped boundary (represented by the thick line) are
described by an $n$-dimensional BF category $\EC_n$.  
(b) If we fold the time direction, we obtain a local Hamiltonian qubit system
$\overline{\La}$ which is the time-reversal transformation of the local Hamiltonian
qubit system $\La$.  Its topological excitations are described by an exact
$n+1$-dimensional BF category $\overline{\EC}_{n+1}$.  The topological excitations on
its gapped boundary (represented by the thick line) are described by a
$n$-dimensional BF category $\overline{\EC}_n$.  If we stack the two boundaries
$\EC_n$ and $\overline{\EC}_n$ together, we will obtain a new boundary $\EC_n
\boxtimes \overline{\EC}_n$.  We see that the boundary between $\EC_n \boxtimes \overline{\EC}_n$ and $\one_n$ is described by $P_{n-1}(\EC_n)$. 
%But it is not obvious that
%$\EC_n \boxtimes \overline{\EC}_n \sim \one_n$, since $\EC_n$ is not closed.  In
%fact, we do have $\EC_n \boxtimes \overline{\EC}_n \sim \one_n$.
}
\label{TConjG}
\end{figure}

In order to obtain a group structure, 
we have to consider certain quotient sets of $\cM^n$. 
Before we do that, we need first introduce the dual of a BF category. 
\begin{defn} \textbf{Dual BF category}\\
(1) Let $\La$ be a lbH system in an $(n+1)$-dimensional
space-time.  The lbH system can always be described by a
path integral.  Then the dual lbH system $\overline{\La}$ is the
lbH system described by the time-reversal transformation of
the path integral that defines the first lbH system $\La$
(see Fig. \ref{TConj}). For more details, see Appendix \ref{path}.\\
(2) Let $\EC_n$ be the BF category realized by a boundary of a lbH system $\La$.  Then the dual $\overline{\EC}_n$ is the BF category
realized by the boundary of the dual lbH system $\overline{\La}$
(see Fig. \ref{TConjG}).
\end{defn}

We collect a few basic properties of the dual BF categories below. 
\begin{lemma} \label{lem:dual-Z}
Let $\EC_n$ and $\ED_n$ be two $\BF_n$ categories.  We have: 
\bnu
\item $\overline{\one}_n = \one_n$,
\item $\cZ_n(\overline{\EC}_n) = \overline{\cZ_n(\EC_n)}$, 
\item If $\ED_n$ is closed, then 
$\ED_n \boxtimes \overline{\ED}_n$ is exact and
$\cZ_{n-1}[P_{n-1}(\ED_n)] =  \ED_n \boxtimes \overline{\ED}_n$
\enu
\end{lemma}
%Recall that the BF category $\ED_n$ in $\cZ_{n-1}(\ED_n)$ is viewed as an $(n-1)$-dimensional BF category. 
The last result can be proved by folding an $n$-dimensional topological phase
defined by $\ED_n$ along a codimension 1 hyperplane (see Fig. \ref{TConj}).

From Lemma\,\ref{lem:dual-Z}, we can easily see that if $\EC_n$ is a closed BF
category realized by a lbH system $\La$, then the dual $\overline{\EC}_n$ is
also closed which can be realized by the time-reversal transformed lbH system $\bar
\La$. By folding a $\EC_{n+1}$-phase and by (\ref{eq:C=0=D}), we obtain an
interesting corollary of Lemma\,\ref{lem:dual-Z}. 
\begin{cor}
For a $\BF_{n+1}$ category $\EC_{n+1}$, $\cZ(\EC_{n+1})$ is invertible if $\cZ_n[P_n(\EC_{n+1})] \simeq \EC_{n+1} \boxtimes \overline{\EC}_{n+1}$. 
\end{cor}

%(???) How to show ``if'', take the bulk on both sides of the condition, then the right hand side is trivial if and only if each factor is trivial. 

\begin{rema}
If the only invertible $\BF_{n+2}$ category is the trivial one for certain $n$
(possibly for $n<6$ and $n\neq 1,3$ see Section\,\ref{invTop}), $\EC_{n+1}$ is
closed if and only if $\cZ_n[P_n(\EC_{n+1})] \simeq \EC_{n+1} \boxtimes
\overline{\EC}_{n+1}$. When $n=2$, this result reproduces a well-known
mathematical result, which says that a premodular category $\EC$ is modular if
and only if $Z(\EC) \simeq \EC\boxtimes \overline{\EC}$, where $Z(\EC)$ is the
monoidal center of $\EC$, and $\overline{\EC}$ is the same category as $\EC$
but with the braiding given by the anti-braiding in $\EC$. 
%Since a unitary modular tensor category $\EC$ is an example of $\BF_3$ category, such defined $\overline{\EC}$ is braided equivalent to the dual of $\EC$ as a BF category.
This also justifies our notation $\cZ(\cdot)$ in this case.  
\end{rema}

\subsection{Quasi-equivalence relation and a group structure}
\label{wgrp}

%The tensor product makes the BF categories in the same dimensions to form a monoid.  A monoid is not a group since the inverse may not exist.  
In an attempt to obtain a group structure, we introduce an equivalence relation
$\sim$ between BF categories,
hoping that the equivalence classes of such equivalence relation form a group.
\begin{defn} 
\label{equvCC}
Two $\BF_n$ categories, $\EC_n$ and $\EC'_n$, are called quasi-equivalent, denoted by $\EC_n
\sim \EC'_n$, if there exist $n$-dimensional closed BF categories $\ED_n$ and
$\ED_n'$ such that $ \EC_n\boxtimes \ED_n \boxtimes \overline{\ED}_n
=\EC_n'\boxtimes \ED_n' \boxtimes \overline{\ED}_n' $
\end{defn} \noindent
$\sim$ is indeed an equivalence relation since it satisfies
\begin{lemma} ~\\
(1) $\EC_n \sim \EC_n$.\\
(2) $\EC_n \sim \EC_n'$ implies that $\EC_n' \sim \EC_n$.\\
(3) $\EC_n \sim \EC_n'$ and $\EC_n' \sim \EC_n''$ implies that $\EC_n \sim
\EC_n''$.
\end{lemma} \noindent
\pf
We give the proof of the third condition. 
$\EC_n \sim \EC_n'$ and $\EC_n' \sim \EC_n''$ implies that $\EC_n \boxtimes
\ED_n \boxtimes \overline{\ED}_n= \EC_n'\boxtimes \ED_n' \boxtimes
\overline{\ED}_n'$ and $\EC_n' \boxtimes \EE_n \boxtimes \overline{\EE}_n=
\EC_n''\boxtimes \EE_n' \boxtimes \overline{\EE}_n'$.  So, we have  $\EC_n
\boxtimes \ED_n \boxtimes \overline{\ED}_n\boxtimes \EE_n \boxtimes
\overline{\EE}_n= \EC_n'\boxtimes \ED_n' \boxtimes \overline{\ED}_n'\boxtimes
\EE_n \boxtimes \overline{\EE}_n=\EC_n''\boxtimes \EE_n' \boxtimes
\overline{\EE}_n'\boxtimes \ED_n' \boxtimes \overline{\ED}_n'$. Thus $ \EC_n
\sim \EC_n''$.
\epf

We denote the set of equivalence classes of $n$-dimensional BF categories under
the quasi-equivalence relation $\sim$ by $\cA^n$, i.e. $\cA^n:= \cM^n/\sim$. 
We denote the equivalence class of $\EC_n$ in $\cA^n$ by $[\EC_n]$. 
We hope that the set $\cA^n$ can form a group under the tensor
product $\boxtimes$. 

From the definition of $\sim$ we can easily see that
\begin{lemma}~\\
(1) If $\EC_n \sim \EC'_n$ and $\ED_n \sim \ED'_n$, then $\EC_n \boxtimes \ED_n
\sim \EC'_n \boxtimes \ED'_n$.\\
(2) For a closed BF category $\EC_n$, $\EC_n\boxtimes \overline{\EC}_n \sim \one_n$.\\
\end{lemma}\noindent
Thus, the tensor product $\boxtimes$ is compatible with the equivalence
relation. In other words, $[\EC_n] \boxtimes [\ED_n] :=[\EC_n \boxtimes \ED_n]$ is a well-defined binary multiplication $\boxtimes: \cA^n \times \cA^n \to \cA^n$, and the quotient map $\cM^n \to \cA^n$ is a homomorphism between two commutative monoids. Moreover, for a closed BF category $\EC_n$, the inverse of $[\EC_n]$ is given by $[\overline{\EC}_n]$. 

\smallskip
Also, we note that taking the bulk $\cZ_n(\cdot)$ is compatible with the
equivalence relation that defines $\cA^n$:
\begin{lemma} 
If $\EC_n \sim \EC_n'$, then $\cZ_n(\EC_n) = \cZ_n(\EC_n')$.
\end{lemma}\noindent
\pf 
$\EC_n \sim \EC_n'$ implies that $\EC_n \boxtimes \ED_n \boxtimes
\overline{\ED}_n = \EC_n'\boxtimes \ED_n' \boxtimes \overline{\ED}_n'$. Thus
$\cZ_n(\EC_n \boxtimes \ED_n \boxtimes \overline{\ED}_n)= \cZ_n(\EC_n'\boxtimes
\ED_n' \boxtimes \overline{\ED}_n')$, which implies that $\cZ_n(\EC_n)
\boxtimes \cZ_n(\ED_n \boxtimes \overline{\ED}_n)= \cZ_n(\EC_n')\boxtimes
\cZ_n(\ED_n' \boxtimes \overline{\ED}_n')$.  Since $\ED_n$ and $\ED_n'$ are
closed, we have $\cZ_n(\ED_n \boxtimes \overline{\ED}_n)=\cZ_n(\ED_n'
\boxtimes \overline{\ED}_n')=\one_{n+1}$. Thus $\cZ_n(\EC_n) = \cZ_n(\ED_n)$.  
\epf

As a result, $\cZ_n$ canonically induces a monoid homomorphism, also denoted by $\cZ_n$, from $\cA^n$ to $\cA^{n+1}$, such that the following diagram:
$$
\xymatrix{
\cM^n \ar[r]^{\cZ_n} \ar[d]_{\sim} & \cM^{n+1} \ar[d]^{\sim} \\
\cA^n \ar[r]^{\cZ_n} & \cA^n 
}
$$
is commutative. 

\smallskip
From the definition of $\sim$, we can easily see that
if $\EC_n \sim \ED_n$ and $\EC_n$ is closed, then $\ED_n$ is also closed.
This allows us to show that
\begin{lemma}
Let $\cA^n_{\text{closed}}$ be a subset consisting of the quasi-equivalence
classes of all closed $\BF_n$ categories. Then, $\cA_\text{closed}^n$ is an
Abelian group under the stacking $\boxtimes$ operation. 
% of $\cA^n$. 
\end{lemma} \noindent
%Note that the Abelian group $\cA_\text{closed}^n$ is different from the Witt
%group.\cite{dmno,Wang10}

\begin{figure}[tb] 
\begin{center} 
\includegraphics[scale=0.5]{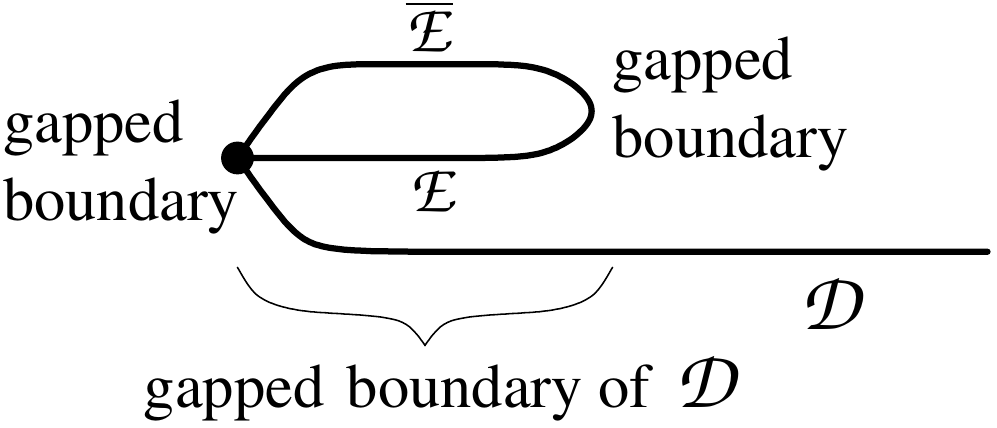} \end{center}
%Fig. 1
\caption{
$ \ED_n\boxtimes\EE_n \boxtimes
\overline{\EE}_n $ has a gapped boundary implies that
$ \ED_n$ has a gapped boundary if $\EE_n$ is closed.
} 
\label{DEbarE} 
\end{figure}

We can also show that
\begin{lemma} If $\EC_n \sim \ED_n$ and $\EC_n$ is exact, then $\ED_n$ is exact.
\end{lemma}\noindent
\pf 
We note that  $\EC_n \sim \ED_n$ implies that
there exist closed $\EE_n$ and $\EE_n'$ such that
that $\EC_n \boxtimes \EE_n' \boxtimes
\overline{\EE}_n'= \ED_n\boxtimes \EE_n \boxtimes \overline{\EE}_n$.
Since $\EC_n$ is closed, so is $\ED_n$.
Since $\EC_n \boxtimes \EE_n' \boxtimes
\overline{\EE}_n'$ is exact, so is $\ED_n\boxtimes\EE_n \boxtimes
\overline{\EE}_n $ (\ie has a gapped boundary).
Then from Fig. \ref{DEbarE}, we can show that $\ED_n$ is exact.
\epf \noindent
Since the center of the dual is the dual of the center: $\cZ(\overline{\EC}_n)=\overline{\cZ(\EC_n)}$, therefore, for exact BF categories, we also have:
\begin{lemma}
Let $\cA_\text{exact}^n$ be the quasi-equivalence classes of all exact BF
categories.  Then, $\cA_\text{exact}^n$ is an Abelian group under the stacking
$\boxtimes$ operation.
\end{lemma}\noindent

Recall an elementary result in mathematics. 
\begin{lemma}
Let  $f: X \to Y$ be a surjective homomorphism between two commutative monoids $X$ and $Y$. If both $\ker f$ and $Y$ are Abelian groups, so is $X$. 
\end{lemma}
\begin{proof} It is enough to show that any element $x$ in $X$ has a right
inverse $x'$, i.e.  $x x' = 1$ (by the commutativity,  $x' x = x x' = 1$).
Notice that if such $x'$ exists, it must be unique. Otherwise, let $x'$ and
$x''$ be so that     $x x' =1 = x x''$.  Then we have $x' = x' x x'' = x''$
(using the associativity and commutativity).
For any $x$ in $X$, let $y$ be the inverse of $f( x )$ and let $z$
be an element in $X$ such that $f( z ) = y$. Then we have $f(x z) =
f( x ) f( z ) = f( x ) y =1$.  Therefore, $x z\in \ker f$. Since $\ker f$ is an Abelian group, there is an element $d$ such that $xzd = 1$.  Hence $zd$ is the right inverse of $x$. 
\end{proof}

Therefore, we obtain an important result. 
\begin{prop}
The set $\cA^n$ of quasi-equivalence classes of $\BF_n$ categories form
an Abelian group under the stacking $\boxtimes$ operation. Moreover, $\cZ_n: \cA^n \to \cA^{n+1}$ is a group homomorphism. 
\end{prop}

\subsection{Witt equivalence relation}  \label{sec:witt-eq-relation}

The quasi-equivalence relation is not the only equivalence relation available. In this subsection, we will discuss more equivalence relations, among which Witt equivalence relation is the most important one. 

\begin{defn}  
\label{equvCC}
Two $n$-dimensional BF categories $\EC_n$ and $\EC'_n$ are called
$k$-equivalent for $k\geq n-1$ and denoted by $\EC_n\ksim \EC'_n$ if there
exist $k$-dimensional BF categories $\ED_k$ and $\ED_k'$ such that 
$$\EC_n\boxtimes \cZ_{n-1}[P_{n-1}(\ED_k)] \simeq \EC_n'\boxtimes
\cZ_{n-1}[P_{n-1}(\ED_k')].$$
\end{defn}

An immediate consequence of above definition is $\cZ_n(\EC_n) \simeq \cZ_n(\ED_n)$ if $\EC_n \ksim \ED_n$.  Namely, two $k$-equivalent BF categories $\EC_n$ and $\ED_n$ must share the same \bulk. 

\begin{rema}
When $n=3$, the $2$-equivalence $\overset{2}{\sim}$ for closed $\BF_3$-categories is the usual Witt equivalence relation\cite{dmno, fsv}. 
\end{rema}

The $k$-equivalence $\ksim$ is indeed a well-defined equivalence relation follows from the following Lemma. 
\begin{lemma} \label{lem:3-prop-er}~ For $k\geq n$, we have \\
(1) $\EC_n \ksim \EC_n$.\\
(2) $\EC_n \ksim \EC_n'$ implies that $\EC_n' \ksim \EC_n$.\\
(3) $\EC_n \ksim \EC_n'$ and $\EC_n' \ksim \EC_n''$ imply that $\EC_n \ksim \EC_n''$.
\end{lemma} \noindent

The $k$-equivalence relation is different for different $k$. A $k$-equivalence class is larger than a $k+1$-equivalence class. Namely, two $(k+1)$-equivalent $\BF_n$-categories is automatically $k$-equivalent. Namely, 
$$
\cdots \,\, \overset{k+1}{\sim} \,\, \Rightarrow \,\, \ksim \,\, \Rightarrow \,\, \cdots \,\, \Rightarrow \,\, \overset{n}{\sim} \,\, \Rightarrow \,\, \overset{n-1}{\sim}.
$$
Note that the $n$-equivalence relation is not the quasi-equivalence relation $\sim$. Instead, we have $\sim$ $\Rightarrow$ $\overset{n}{\sim}$. The $\overset{n-1}{\sim}$ is also called the Witt equivalence, also denoted by $\overset{\rw}{\sim}$. 

Using equation (\ref{eq:Z-boxtimes}) and Lemma\,\ref{lem:dual-Z}, we obtain immediately a few results from Definition\,\ref{equvCC}:
\begin{lemma}  \label{lem:sim}  ~ For $k\geq n-1$, \\
(1) If $\EC_n \ksim \EC'_n$ and $\ED_n \ksim \ED'_n$, then $\EC_n \boxtimes \ED_n
\ksim \EC'_n \boxtimes \ED'_n$.\\
(2) $\cZ_{n-1}(P_{n-1}(\EC_k)) \ksim \one_n$.\\
(3) If $\EC_n$ is closed and $\EC_n \ksim \ED_n$, then $\ED_n$ is closed.\\
%(4) If $\ED_n$ is closed and $\EC_n \boxtimes \EE_n \ksim \ED_n$, then $\cZ_n(\EC_n)\boxtimes \cZ_n(\EE_n)=\one_{n+1}$.\\
(4) If $\EC_n$ is closed, then $\EC_n\boxtimes \overline{\EC}_n \ksim \one_n$ for $k=n, n-1$.\\
(5) If $\EC_n$ is exact, i.e. $\EC_n\simeq \cZ_{n-1}(\ED_{n-1})$ for some $\ED_{n-1}$,  then $\EC_n \boxtimes \cZ_{n-1}(\overline{\ED_{n-1}}) \ksim \one_n$
for $k=n, n-1$.
\end{lemma}

%We denote the set of equivalence classes of $n$-dimensional BF categories under the equivalence relation $\nnsim$ as $\cA^n$, i.e. $\cA^n:= \cM^n/\nnsim$.  We denote the equivalence class of $\EC_n$ in $\cA^n$ by $[\EC_n]$. We hope that the set $\cA^n$ can form a group under the tensor product $\boxtimes$. Let $\cA^n_{\text{closed}/\text{exact}}$ be a subset consisting of the $n$-equivalence classes of all closed/exact $\BF_n$ categories. 

%The first result in Lemma\,\ref{lem:sim} says that the stacking operation (or the tensor product $\boxtimes$) is a well-defined binary multiplication on $\cA^n$; the second result in Lemma\,\ref{lem:sim} says that the subset $\cA^n_\text{closed}$ of $\cA^n$ is an Abelian group, with the inverse given by the dual BF category; the 6th result in Lemma\,\ref{lem:sim} says that the subset $A_\text{exact}^n$ of $\cA^n$ is an Abelian group with the inverse given by the dual BF category. 

%These results further imply the following theorem: 
%\begin{prop}
%The set $\cA^n$ of $n$-equivalence classes of $n$-dimensional BF categories, together with the multiplication given by the stacking $\boxtimes$ operation and the unit element given by $[\one_n]$, is an Abelian group. The inverse is given the dual of BF category. 
%\end{prop}
%
We give a sufficient condition for the equivalence relation $\overset{n}{\sim}$ among closed $BF$ categories below. 
\begin{prop} \label{prop:gw-w}
Let $\EC_n$ and $\ED_n$ be $n$-dimensional BF categories. If $\EC_n$ (or $\ED_n$) is closed and there exists an $n$-dimensional BF category $\EE_n$ such that $\EC_n \boxtimes \overline{\ED}_n \simeq \cZ_{n-1}(P_{n-1}(\EE_n))$, then $\ED_n$ (or $\EC_n$) is closed and $\EC_n \overset{n}{\sim} \ED_n$. 
\end{prop}
\begin{proof}
If $\EC_n$ is closed, then we have
$$
\cZ_n(\ED_n) \simeq \cZ_n(\EC_n \boxtimes \overline{\ED}_n) \simeq \cZ_n(\cZ_{n-1}(P_{n-1}(\EE_n))) \simeq \one_{n+1}.
$$
Namely, $\ED_n$ is closed. 
Multiplying $\ED_n$ on the both side of the condition $\EC_n \boxtimes \overline{\ED}_n \simeq \cZ_{n-1}(\EE_n)$ and applying the 3rd result in Lemma\,\ref{lem:dual-Z}, we obtain
$$
\EC_n \boxtimes \cZ_{n-1}(P_{n-1}(\ED_n)) 
\simeq \ED_n \boxtimes \cZ_{n-1}(P_{n-1}(\EE_n)),
$$ 
i.e. $\EC_n \overset{n}{\sim} \ED_n$. 
\end{proof}

For Witt equivalence relation, we denote the set of equivalence classes of $n$-dimensional BF categories under
the equivalence relation $\rwsim$ as $\cA_\rw^n$, i.e. $\cA_\rw^n:= \cM^n/\rwsim$. The set $\cA_\rw^n$ is a quotient of $B^n$. We denote the equivalence class of $\EC_n$ in $\cA_\rw^n$ by $[\EC_n]_\rw$. 
\begin{prop} \label{prop:Aw-group}
The set $\cA_\rw^n$, together with the multiplication given by the stacking $\boxtimes$ operation and the unit element $[\one_n]_\rw$, is an Abelian group. The inverse is given by the dual of a BF category. Moreover, $\cZ: \cA_\rw^n \to \cA_\rw^{n+1}$ is a group homomorphism.
\end{prop}

\medskip
The Witt equivalence between two closed $\BF_n$ categories can also be understood in a different way. We introduce another equivalence relation. 
\begin{defn}
Two $\BF_n$ categories $\EC_n$ and $\ED_n$ are called $\mathrm{gw}$-equivalent, denoted by $\EC_n \gwsim \ED_n$, if there exists an $(n-1)$-dimensional BF category $\EE_{n-1}$ such that 
\begin{equation} \label{eq:gw-eq}
\EC_n \boxtimes \overline{\ED}_n \simeq 
\cZ_{n-1}(\EE_{n-1}).
\end{equation}
\end{defn}

\begin{rema}  {\rm
When both $\EC_n$ and $\ED_n$ are closed, the physical meaning of $\mathrm{gw}$-equivalence is clear. It means that $\EC_n$ and $\ED_n$ are $\mathrm{gw}$-equivalent if and only if they can be connected by a gapped domain wall. 
}
\end{rema}

\begin{lemma}
Three defining properties of an equivalence relation (recall Lemma\,\ref{lem:3-prop-er}) hold. Namely, $\gwsim$ is a well-defined equivalence relation. 
\end{lemma}

Proposition\,\ref{prop:gw-w} implies the following result. 
\begin{lemma} \label{lemma:gw-implies-w}
If $\EC_n$ (or $\ED_n$) is closed, then $\EC_n \gwsim \ED_n$ implies that $\ED_n$ (or $\EC_n$) is closed and $\EC_n \rwsim \ED_n$. 
\end{lemma}

\begin{figure}[t] 
$$
 \begin{picture}(120, 100)
   \put(0, 0){\scalebox{1.5}{\includegraphics{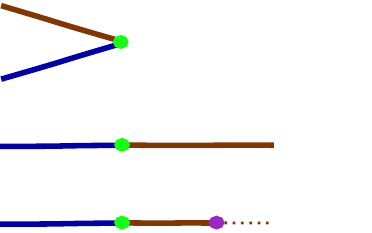}}}
   \put(0, 0){
     \setlength{\unitlength}{.75pt}\put(0,0){
     \put(25, 130)  {$\cZ_n(\overline{\ED}_n)$ } 
     \put(110, 57)     { $\cZ_n(\ED_n)$ } 
     \put(60, 60)      { $ \EE_n $}
     \put(75, 108)   { $ \EE_n $}
     \put(60, 15)     { $ \EE_n $}
     \put(116,15)    { $\ED_n$ }
     \put(20, 83)     { $\EC_{n+1}$ }
     \put(0,57)      { $\EC_{n+1}$ } 
     \put(0, 12)   { $\EC_{n+1}$ }     
     \put(157, 2)  { $\one_{n+1}$ }
     \put(70, -10)     { $\cZ_n(\ED_n)$ } 
     }\setlength{\unitlength}{1pt}}
  \end{picture}
$$  
\caption{{\small The three steps in the proof of Lemma\,\ref{lem:w-implies-exact}, i.e. $\EC_n \rwsim \one_n$ iff $\EC_n$ is exact.}}
  \label{fig:w=gw}
\end{figure}

Conversely, we will show that $\rwsim$ also implies $\gwsim$. We will start first prove an important lemma. 
\begin{lemma}  \label{lem:w-implies-exact}
$\EC_{n+1} \rwsim \one_{n+1}$ if and only if $\EC_{n+1}$ is exact. 
\end{lemma}
\begin{proof}
If $\EC_{n+1}$ is exact, it is obvious that $\EC_{n+1} \rwsim \one_{n+1}$. Conversely, by definition, $\EC_{n+1} \rwsim \one_{n+1}$ means that there are $\ED_n$ and $\EE_n$ such that 
$$
\EC_{n+1} \boxtimes \cZ_n(\overline{\ED}_n) \simeq \cZ_n(\EE_n). 
$$ 
It implies $\EC_{n+1}$ is closed. Moreover, its physical meaning is that the $(n+1)$-dimensional topological order $\cZ_n(\EE_n)$, as the \bulk of an $n$-dimensional boundary $\EE_n$, can be factorized as a double-layered system $\EC_{n+1} \boxtimes \cZ_n(\ED_n)$ (see Fig.\,\ref{fig:w=gw}). By unfolding this double-layered system along its $\EE_m$-boundary, we obtain two $(n+1)$-dimensional topological orders $\EC_{n+1}$ and $\cZ_n(\ED_n)$, which are connected by a gapped $n$-dimensional domain wall $\EE_n$. Since the topological order $\cZ_n(\ED_n)$ itself allows a gapped boundary, we are able to condense $\cZ_n(\ED_n)$ to the trivial phase $\one_{n+1}$. This condensation creates a gapped boundary given by $\ED_n$ and a narrow band bounded by $\EE_n$ and $\ED_n$. This narrow band connects the topological order $\EC_n$ to the trivial order $\one_{n+1}$ (see Fig.\,\ref{fig:w=gw}). Therefore, this narrow band should be viewed as an $n$-dimensional gapped boundary, which is of type $\EE_n \boxtimes_{\cZ_n(\ED_n)} \ED_n$, of an $(n+1)$-dimensional bulk phase $\EC_{n+1}$. Therefore, we must have
$\EC_{n+1} \simeq \cZ_n(\EE_n \boxtimes_{\cZ_n(\ED_n)} \ED_n)$. 
\end{proof}

\begin{prop}  \label{thm:w=gw}
For two closed BF categories $\EC_n$ and $\ED_n$, $\EC_n \rwsim \ED_n$ if and only if $\EC_n \gwsim \ED_n$, or equivalently, if and only if they are connected by an $(n-1)$-dimensional gapped domain wall. 
\end{prop}
\begin{proof}
By Lemma\,\ref{lemma:gw-implies-w}, it is enough to show that $\rwsim$ implies $\gwsim$. By Proposition\,\ref{prop:Aw-group}, $\EC_n \rwsim \ED_n$ implies that 
$[\EC_n]_\rw \boxtimes [\overline{\ED}_n]_\rw = [\one_n]_\rw$. By Lemma\,\ref{lem:w-implies-exact}, there exists an $(n-1)$-dimensional BF category $\EE_{n-1}$ such that
$\EC_n \boxtimes \overline{\ED}_n \simeq \cZ_{n-1}(\EE_{n-1})$, which means $\EC_n \gwsim \ED_n$.  
\end{proof}

\section{The cochain complex of the BF
categories}

\label{mcc}

%\subsection{The monoid-cochain complex}

%\subsection{The monoid-cochain complex of BF categories}

%However, not every  well-defined gapped state in $n+1$ space-time dimension
%has a gapped boundary of $n$ space-time dimension. In other words, not every
%closed BF category is the center of a BF category in one lower dimension.  We
%refer those closed BF categories, whose corresponding topological states have
%a gapped boundary, as exact BF categories.  In other words, the closed BF
%categories, that are  centers of BF categories in one lower dimension, are
%exact BF categories.  The  exact BF categories describe the well-defined
%topological states that have a gapped boundary.  
The BF categories in the same dimension form a commutative monoid, which is
denoted by $\cM^n$. The bulk of a boundary defines a homomorphism $\cZ_n: \cM^n \to \cM^{n+1}$ between commutative monoids.  These homomorphisms $\cZ_n$ for all non-negative integers $n$ satisfy the property that $\cZ_{n+1}(\cZ_n(\cM^n))\simeq \one_{n+2} \in \cM^{n+2}$. In other words, $\cZ_n$ is a differential operator in a cochain complex. Therefore, we obtain a commutative-monoid-valued cochain complex 
\begin{align} \label{diag:cochain}
\cdots \overset{\cZ_{n-1}}{\longrightarrow}
  \cM^n
\overset{\cZ_n}{\longrightarrow}
  \cM^{n+1}
\overset{\cZ_{n+1}}{\longrightarrow}
  \cM^{n+2}
\overset{\cZ_{n+2}}{\longrightarrow} \cdots .
\end{align}
Because a non-trivial BF category does not have an inverse in $\cM^n$, $\cM^n$
is not an Abelian group in general. Therefore, (\ref{diag:cochain}) is not a
usual cochain complex (which is valued in Abelian groups). 

\medskip
However, the sets $\cA^n$ of the equivalence classes of BF$_n$ categories in
different dimensions do form cochain complex because $\cA^n$ are Abelian groups:
\begin{thm}
The sets $\cA^n$ of the equivalence classes of $\BF$ categories for all $n$, together with the group homomorphisms $\cZ_n$, form a cochain complex:
\begin{align}
\cdots \overset{\cZ_{n-1}}{\longrightarrow}
  \cA^n
\overset{\cZ_n}{\longrightarrow}
  \cA^{n+1}
\overset{\cZ_{n+1}}{\longrightarrow}
  \cA^{n+2}
\overset{\cZ_{n+2}}{\longrightarrow} \cdots .
\end{align}
\end{thm}

\begin{rema} 
Since we don't yet have a precise mathematical definition of a BF category, above theorem should be understood as a physical theorem of the equivalence classes of topological orders. 
\end{rema}

We can define the $n$-th cohomology group as usual.
\begin{defn}
$\rH^n:=\text{ker}(\cZ_n)/\text{img}(\cZ_{n-1})$. 
\end{defn} 

\medskip
Similarly, we have the following result for $\cA_\rw^n$: 
\begin{thm}
The sets $\cA_\rw^n$ of the equivalence classes of $\BF$ categories for all $n$, together with the group homomorphisms $\cZ_n$, form a cochain complex:
\begin{align}
\cdots \overset{\cZ_{n-1}}{\longrightarrow}
  \cA_\rw^n
\overset{\cZ_n}{\longrightarrow}
  \cA_\rw^{n+1}
\overset{\cZ_{n+1}}{\longrightarrow}
  \cA_\rw^{n+2}
\overset{\cZ_{n+2}}{\longrightarrow} \cdots .
\end{align}
\end{thm}

We can define the $n$-th cohomology group as usual.
\begin{defn}
$\rH_\rw^n:=\text{ker}(\cZ_n)/\text{img}(\cZ_{n-1})$. 
\end{defn} 
In fact $\cZ_n(\cA_\rw^n)=\one_{n+1}$ and $\rH_\rw^n=\cA_\rw^n$.  By
Theorem\,\ref{thm:w=gw}, two closed BF$_n$ categories $\EC_n$ and $\ED_n$
belong to the same class in $\rH_\rw^n$ iff the boundary between $\EC_n$ and
$\ED_n$ can be gapped. Then it is clear that the $n$-th cohomology group
$\rH_\rw^n$ classify the types of $(n-1)$ space-time dimensional gapless
boundaries.

When $n=3$, closed $\BF_3$ categories are unitary modular tensor categories
(UMTC), and exact $\BF_3$ categories are those monoidal centers of unitary
fusion categories. So the cohomology group $\rH_\rw^3$ is nothing but the Witt
group\cite{dmno} for UMTCs. It classifies the types
of $2$ space-time dimensional gapless boundaries. It is not surprising that the
Witt group was originally introduced to classify $2$-dimensional rational
conformal field theories\cite{dmno}.

\section{General examples of low dimensional \hBF{}  categories with only particle-like excitations}  
\label{GBFexample}

In this section, we will discuss some general examples of \hBF{} categories in
low dimensions. Since only defects of codimension 2 are detectable by braiding
with other excitations, in the cases that spatial dimension is not more than 2,
we can only detect particle-like excitations via braiding. So we will restrict
ourselves to only those particle-like excitations in a \hBF{} category and
ignore higher dimensional defects. For simplicity, we will also ignore all
those $0$-dimensional defect nested on non-trivial higher dimension defects,
such as the one depicted as the blue point in Fig.\ref{toric}. Higher
dimensional excitations or defects will be studied in
Section\,\ref{sec:math-def}.

\subsection{0+1D topological phases}
\label{0dTO}
A 0+1D topological phase is just a quantum mechanics system. It is given by a
finite dimensional Hilbert space $V$ equipped with the local operator algebra
$A=\mathrm{End}(V)$. The data $V$ is redundant and can be recovered from $A$.
Therefore, a 0+1D topological phase can be described by a category with a
single object $\ast$ and $\hom(\ast, \ast)=A$. Equipping $A$ with the operator
normal, we can turn it into a $C^\ast$-algebra. 

\subsection{Particles in 1+1D topological phases}
\label{1dTO}

In a 1+1D phase,  there is no braiding between topological particles. So this phase is characterized by the following data: 
(for a detailed discussion see \Ref{LWstrnet,RSW0777,Wang10}):\\
(1) An integer $N$ that describes the number of nontrivial types of
particle-like topological excitations.\\
(2) An one-to-one map $i \to i^*$, $i,i^*=0,1,\cdots ,N$
that satisfy $0=0^*$ and $(i^*)^*=i$.\\
(3) A rank-3 tensor $N^{ij}_k$ that describes dimension of the fusion
spaces of the topological excitations. Moreover, $N_{0i}^j=\delta_{ij} = N_{i0}^j$ and $N^{ij}_k$ satisfies an associativity property.\\
(4) A rank-10 tensor $F^{ijk,m\al\bt}_{l,n\ga\la}$ satisfies the pentagon identities. It describes the linear
relations between the fusion spaces of the topological excitations.\\

Using categorical language, above data amounts to a unitary fusion category $\EC$
of topological excitations (simple types or or composite types), with only finite many 
simple types $i, j, k \in I$ where $|I|=N+1$. The hom space $\hom_\EC (X, Y)$ is a 
finite dimensional Hilbert space for any two excitations $X$ and $Y$. 
Moreover, fusion of two simple excitations $i$ and $j$ give the tensor product $i\otimes j$. 
This fusion product $i\otimes j$ is completely determined by 
\begin{align}
N_{ij}^k = \dim \hom_\EC(i\otimes j, k).
\end{align}
In particular, $N_{0i}^j=\delta_{ij} = N_{i0}^j$ implies that $0 \otimes i = i
= i\otimes 0$. $0$ is the tensor unit $\one$ of $\EC$.  In particular,
$N_{00}^0=1$ means that the vacuum degeneracy of the vacuum is trivial. In more
categorical language, we have
\begin{equation} \label{eq:vacuum-degeneracy}
\dim \hom_\EC(\one, \one) = 1.
\end{equation}
These structures $(\EC, \otimes, 0)$, together with the rank-10 tensor
$F^{ijk,m\al\bt}_{l,n\ga\la}$ satisfying the pentagon identities, equip the
unitary category $\EC$ with a structure of a monoidal or tensor category. The
existence of anti-particle $i^*$ for all $i$ further implies that $\EC$ is also
rigid. Combining all of these results, we have shown that $\EC$ has a structure
of a unitary fusion category (UFC).  Notice that the hom space should be viewed
as instantons in time direction.

\smallskip
Among all UFC's, the most trivial one is the category $\hilb$ of finite
dimensional Hilbert spaces. By \Ref{KK1251}, any UFC $\EC$ can be realized as
the boundary excitations of a Levin-Wen type of lattice model with bulk
excitations given by the category $\cZ(\EC)$, which is the monoidal center of
$\EC$.  Mathematically, it is well-known that $\cZ(\EC)\simeq \hilb$ if and
only if $\EC \simeq \hilb$. According to our general theory of bulk-boundary
relation in Section\,\ref{sec:boundary-bulk-relation}, the only anomalous free
(or closed) $\BF_2$ category is $\hilb$.  It is also clear that the trivial
$1+1$-dimensional phase is also the bulk of a trivial $0+1$-dimensional phase.
Therefore, $\hilb$ is the only closed and exact $\BF_2$ category (which can be
composite  $\BF_2$ categories. See Appendix \ref{SCBF}).

\void{
\subsection{Particles in 0+1D topological phases}

A BF category $\EM$ in 0+1D can be defined by the boundary theory of a
anomaly-free 1+1D topological phase $\EC$. As a particle in the 1D bulk move
closer to the boundary, it fuses into the boundary and becomes a particle on
the boundary. This implies that there is an action $\otimes: \EC \times \EM \to
\EM$, i.e. for $i\in \EC$ and $m\in \EM$, $i\otimes m \in \EM$. Moreover, this
action is associative, i.e. $i\otimes (j\otimes m) = (i\otimes j) \otimes m$,
and $\one \otimes m =  m = m \otimes \one$. In other words, $\EM$ is a module
category over $\EC$ or a $\EC$-module. We now assume that $\EM$ is
indecomposable as a $\EC$-module. Since the only anomaly-free
$\BF_{1+1}$-category is $\hilb$ (see Section\,\ref{GBFexample}), $\EM$, as an
indecomposable $\hilb$-module, is nothing but $\hilb$. Notice that the elements
in the hom spaces in $\hilb$ label the instantons in time direction. 
}

\subsection{Particles in 2+1D topological phases}
%\subsection{General examples of (potentially anomalous) 2+1D anyon systems}
\label{preM}

In a 2+1D topological phase, particles can fuse with each other and also braid
with each other.  We will now list the ingredients in these fusion and braiding
structures (for a more physical description of some of the following
properties, see Section \ref{uniprop}): 

1. There is a finite set $I$ of anyons. An anyon $i\in I$ corresponds to a
simple object in a category $\EC$. 
%A generic anyon is a the superposition of simple anyons, 
A generic object is a direct sum of simple objects, e.g. $i\oplus j \oplus k$ for $i,j,k\in I$, 
which corresponds to superposition of anyons. 
%which is formed by accidentally degenerate excitations $i$, $j$, 
%and $k$ (\ie the excitations $i$, $j$, and $k$ happen to have the same energy ??).

2. Between two generic objects $X$ and $Y$, there are fusion-splitting channels
which forms a vector space over $\Cb$: $\hom(X, Y)$. In particular, the vector
space $\hom(i, X)$ tells us how many ways the simple anyon $i$ can fuse into a
generic object $X$ and the vector space $\hom(X, i)$ tells us how many
splitting channels from $X$ to $i$. 

3. The unitarity of the anyon system is a physical requirement. It immediately
implies that the category $\EC$ has to be semisimple. 

4. The fusion of two objects $X$ and $Y$ gives arise to a tensor product $X
\otimes Y$ which must be associative and unital. The tensor unit $\one$ is given by the vacuum. 

5. The mutual statistics among anyons is given by the braiding $c_{X, Y}: X\otimes Y \to Y\otimes X$ for all $X, Y\in \EC$. This information of braiding is encoded in
the physically measurable linear map: $c_{X,Y}: \hom(i, X\otimes Y) \to \hom(i,
Y\otimes X)$. 

6. Anyons can be created or annihilated from the vacuum in pairs. In
particular, we need a dual object $X^\vee$ for each anyon $X \in \EC$, together
with morphisms $\ev_X: X^\vee \otimes X \to \one$, $\coev_X: \one \to X \otimes
X^\vee$ and their adjoints $\ev_X^\ast$ and $\coev_X^\ast$, satisfying the some
natural condition. This says that $\EC$ must be a rigid tensor category. 

7. Each anyon has spins. It amounts to an automormphism $\theta_X: X \to X$
satisfying some properties. This requires $\EC$ to be a ribbon category. 

%\begin{rema} {\rm
%The proof of above Theorem is very brief. It can certainly be expanded. 
%Kitaev's long paper might contain a complete argument. 
%}
%\end{rema}

In summary, above braiding-fusion structures of a system of anyons amount to a
{\it unitary premodular category}.  The unitary premodular categories are
BF$_3$ categories.  However, unitary premodular categories only represent a
subset of 3-dimensional BF$_3$ category.  The anomalous BF$_3$ category
$\EC_3^{sF\Zb_2}$ with only string-like topological excitations (see Example
\ref{C3sFZ2} and Section \ref{C3sFZ2long}) is not a unitary premodular
category. We will discuss those string-like excitations in
Section\,\ref{sec:math-def}.

%%%%%%%%%%%%%%%%%%%%%%%%

In general, a unitary premodular category is anomalous (\ie not realizable by
2+1D qubit models).  We have the following result characterizing the anomalous
free unitary premodular categories. 
\begin{thm}
If a unitary premodular category $\EC$ is not anomalous, \ie if it can be
realized by a $2+1$D local Hamiltonian qubit system, then $\EC$ must be
modular. 
\end{thm}\noindent
\pf
If a $2+1$ anyon system can be defined in $2+1$-dimension, all its particles
should be detectable by the braiding among themselves. As a consequence, if an
anyon $X$ is such that its mutual braidings with all other anyons are trivial,
i.e. $c_{X,Y} \circ c_{Y,X} =\id_{X\otimes Y}$ for all $Y\in \EC$, then $X$
must be uniquely fixed by this property. On the other hand, we know that the
mutual braidings between the vacuum $\one$ and any other object $Y$ is trivial.
Therefore, we must have $X=\one$. In other words, $\EC$ is modular. 
\epf

%\begin{rema} 
%In terms of the new language, a unitary modular tensor category (UMTC) is a
%closed \hBF{3} category. But a unitary premodular category is not closed in
%general.  
%Also premodular categories do not describe all 2+1D \hBF{3}  categories. 
%\end{rema}\noindent
%We note that UMTC's with chiral central charge $c \neq 0$ mod 8 are 0-1-2 extended TQFT's, UMTC's with chiral central charge $c = 0$ mod 8 are 0-1-2-3 extended TQFT's.  
\begin{rema} 
We note that the UMTC's are closed \hBF{3} categories (the path integral can be
defined on mapping tori -- fiber bundles over $S^1$). They are also closed \lBF{3} categories (the path integral can be defined for
any oriented space-time topologies). This fact is very subtle and will be discussed in Sections \ref{WWmdl} and \ref{WWmdl1}. 
\end{rema}\noindent

%This suggests that
%\begin{conj}
%\label{H1L0}
%the closed \hBF{} categories are extended TQFT's that extend to 1 while
%the closed \lBF{} categories are extended TQFT's that extend to 0.
%\end{conj}\noindent

\section{A mathematical definition of a BF$_{n+1}$-category}  \label{sec:math-def}

We have described the examples and the general structures of a BF category
without giving it a precise mathematical definition.  In this section, we will
try to outline a mathematical definition of a BF category.  Since our
understanding of topological orders in high dimensions at the current stage is
very limited, many assumptions and conjectures are imposed in order to proceed.
The mathematical definition we obtained is conjectural and far from being
complete. In particular, coherence morphisms in higher categories and their
properties are ignored. As we will show, many physically interesting and
important questions, which can be studied or formulated in this framework, are
quite irrelevant to the structures we have ignored. Indeed, by focusing on the
fusion and braiding structures alone, we are leading to a very rich theory. The
precise definition of a BF category is not important to us at this stage. But
many important questions and conjectures formulated in our framework will serve
as a blueprint for future studies.

%The layout of this section is as follows: we will first explain the replacement of excitations by general defects and the natural rise of $n$-category, then illustrate some basic topological properties of topological excitations in an explicit example: the toric code model, next we generalize these properties to fusion and braiding structures in general topological orders, then we will explain why the structures of an $n$-category can encode these structures, ???, we will try to introduce the notion of $\BF_n$-category based on the notion of $n$-category. 

\subsection{Defects of all dimensions and higher categories}  \label{sec:defect-ncat}

Topological excitations can also be viewed as defects in a topological phase. In this section, we will use the term ``excitation" and ``defect" interchangeably. We will use ``a domain wall" to refer to a 1-codimensional defect, or more generally, a defect of 1-lower dimension (1-higher codimension), and domain walls between domain walls for defects of 2-lower dimension. 

The main difficulty in describing an $(n+1)$-dimensional $\BF$-category precisely is the existence of topological excitations in different dimensions. Excitations in different dimensions carry different level of richness of structures. For example, a $p$-dimensional excitation can have particle-like excitations nested in it. Moreover, they can be fused and braided within the $p$-dimensional excitation. An higher dimensional excitation has much richer structures than a particle-like excitation. So it is clear that defects of different dimensions should belong to different layers in a multi-layered structure. It suggests us to arrange topological excitations according to their codimensions: at the 0-th level, there is a unique $n$-spatial dimensional bulk phase; the first level, there are domain walls or defects of codimension 1; at the second-level, there are walls between walls (or defects of codimension 2); at the $n$-th level, there are particle-like excitations (or $n$-codimensional defects); at the $(n+1)$-th level, there are instantons (or $(n+1)$-codimensional defects). This multi-layered structure coincides exactly with that of an $(n+1)$-category with one object. More precisely, the unique object (or 0-morphism) corresponds to the bulk-phase; 1-morphisms correspond to the domain walls; 2-morphisms between a pair of 1-morphisms correspond to walls between walls; ..., $(n+1)$-morphisms correspond to instantons.

\begin{rema}
The only reason that the notion of a phase (or order) of matter was invented is because there are phase transitions. A topological order $x$ should be uniquely determined by its relation or ``phase transition" to all topological orders, including $x$ itself. This relation between two $(n+1)$-dimensional phases $x$ and $y$ can be characterized by all possible $n$-dimensional domain walls. In this work, we only consider gapped domain walls because gapless phases are much richer than the finite category theory. In category theory, the relation between two objects is encoded by morphisms between them, and an object $x$ in a category $\EC$ can be determined uniquely (up to isomorphisms) by a family of sets of morphisms $\{ \hom_\EC(y,x) \}_{y\in \EC}$ and maps between them. This is called Yoneda Lemma in category theory. Therefore, we should consider the category of $(n+1)$-dimensional topological orders with morphisms given by domain walls. According to the philosophy of Yoneda lemma, a topological order should be characterized completely by the domain walls between itself and all topological orders. Moreover, notice that a domain wall is itself a topological phase. There are domain wall between domain walls and domain walls between domain walls between domain walls. Each time the dimension of the domain walls is reduced by one until we reach the instantons which is a localized defect in the time direction. As a consequence, the category of all $(n+1)$-dimensional topological orders, denoted by $\EB\EF_{n+1}$, must be an $(n+1)$-category with 1-morphisms given by domain walls, 2-morphisms given by domain walls between domain walls, ..., $(n+1)$-morphisms given by instantons. To determine a given topological order $x$, the information of all domain walls, although sufficient, are too large to work with. A small part of it is given by all the domain walls between a phase $x$ and itself, or the full subcategory of $\EB\EF_{n+1}$ supported on $x$, denoted by $\hat{x}$. This small part is nothing but a phase $x$ with gapped defects of all dimensions discussed in the previous paragraph. We believe that it is rich enough to characterize the topological phase $x$ uniquely. 
\end{rema}
 
At the current stage, an $n$-category is nothing but a name for a multi-layered structure. This mathematical notion contains many more structures. But whether these extra structures are relevant to topological order is not entirely clear. We will explain what additional structures are needed for a topological order by first looking at a simple example: the toric model\cite{K032} (a $\Zb_2$ spin liquid\cite{RS9173,W9164,MS0181}).

\subsection{The toric code model enriched by defects}

In this section, we will illustrate the additional structures that are needed, in particular, the fusion and braiding structures, in the toric code model. We will also explain its relation to other topological phases in $\EB\EF_{2+1}$ such as the trivial phase. For convenience, we will refer to a subcategory of $\EB\EF_{n+1}$ consists of a finite many objects as a {\it multi-$\BF_{n+1}$-category} or an $\MBF_{n+1}$-category. 

\medskip
The toric code model is a 2-dimensional lattice model depicted in Fig.\,\ref{toric}. In this subsection, we will review the results in \Ref{KK1251} in terms of a 3-category. We will also use the language used in \Ref{KK1251} freely. Let $\rep_{\Zb_2}$ be the category of representations of the $\Zb_2$ group. It is a unitary fusion category. 
In the language of Levin-Wen model, the bulk lattice is, by construction, determined by the unitary fusion category $\rep_{\Zb_2}$, thus will be referred to as an $\rep_{\Zb_2}$-bulk. If there is a domain wall, it was shown in \Ref{KK1251} that the lattice model near the domain wall can be constructed from an indecomposable $\rep_{\Zb_2}$-bimodule category $X$. All (bi)modules over a unitary fusion category are assumed to be semisimple.  We will refer to such lattice near the domain wall as an $X$-wall. Similarly, if there is a gapped boundary, or equivalently, a gapped domain wall between the toric code model and the trivial phase, the lattice model near the boundary can be constructed from a $\rep_{\Zb_2}$-module $Y$ and will be referred to as a $Y$-boundary. 

The trivial phase can be viewed as a Levin-Wen model based on the unitary fusion category $\hilb$ of finite dimensional Hilbert spaces. 

%{\bf Wen: changed BF$_3$ category to $3$-category}

The toric code model gives a $3$-category, denoted by $\TC_3$, in its full
complexity. Actually, it is quite convenient and perhaps more illustrative to
describe a slightly larger $3$-category, $\TC_3^b$, which contains two objects:
the toric code model and the trivial phase. $\TC_3$ can be obtained as
sub-category of $\TC_3^b$. There are four layers of structures in $\TC_3^b$. 

\newcommand\tc {\mathrm{tc}}

\smallskip
\noindent $\bullet$ {\it Objects or 0-morphisms}:  the trivial phase, denoted by $\one$ and toric code denoted by $\tc$.

%and given by the unitary modular tensor category $\hilb$; the other one is given by unitary modular tensor categories: $Z(\rep_{\Zb_2})$. The category $Z(\rep_{\Zb_2})$ consists of four types of anyons: $\one, e, m, \epsilon$ in the \bulk phase in Fig.\,\ref{toric}. 

\smallskip
\noindent $\bullet$ {\it $1$-morphism}: 
There are 4 types of 1-morphisms given by various types of defect lines or domain walls:
\begin{enumerate} 
\item 1-morphisms $\tc \to \tc$ are given by domain walls between two $\rep_{\Zb_2}$-bulks. They are classified by $\rep_{\Zb_2}$-bimodules. The trivial wall is the $\rep_{\Zb_2}$-wall, where $\rep_{\Zb_2}$ is viewed as an $\rep_{\Zb_2}$-bimodule. An example of non-trivial domain wall is given by the $\rep_{\Zb_2}$-bimodule: 
$$
{}_{\rep_{\Zb_2}} \hilb_{\rep_{\Zb_2}},
$$ 
which is the category of finite dimensional Hilbert spaces and is depicted as the dotted line in Fig.\,\ref{toric}. The trivial $\rep_{\Zb_2}$-wall can be any line other than the dotted line in the lattice, in particular it can be the vertical line connecting to the dotted line via the blue point in Fig.\,\ref{toric}. 

There are more simple $\rep_{\Zb_2}$-bimodules. For example, there is another $\rep_{\Zb_2}$-bimodule structure on $\hilb$.

\begin{figure}[t] 
  \centerline{\includegraphics[scale=1]{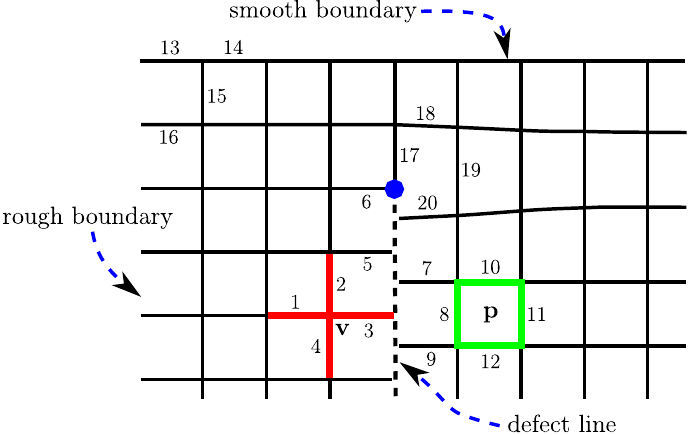}}
  \caption{Toric code with boundaries of two types, a transparent domain wall (the dotted line) and  a $0$-d defect depicted by the blue point at the corner of edge labeled by $6$ and $17$.}
  \label{toric}
\end{figure}

\item There are two simple 1-morphisms $\tc \to \one$. These two 1-morphisms correspond to two types of boundaries: the ${}_{\rep_{\Zb_2}}\hilb_{\hilb}$-boundary, which is called ``rough boundary" in Fig.\,\ref{toric}, and the ${}_{\rep_{\Zb_2}}(\rep_{\Zb_2})_{\hilb}$-boundary, which is called ``smooth boundary" in Fig.\,\ref{toric}. 

\item There are two simple 1-morphisms from $\one\to \tc$ given by the bimodule ${}_{\hilb}\hilb^\op_{\rep_{\Zb_2}}$ and ${}_{\hilb}(\rep_{\Zb_2}^\op)_{\rep_{\Zb_2}}$.

\item 1-morphisms $\one \to \one$ are domain walls in the trivial phase. The only simple one is the trivial wall given by the $\hilb$-bimodule ${}_{\hilb} \hilb_{\hilb}$. 

\end{enumerate}

\smallskip
\noindent {\it $\bullet$ 2-morphisms}: 2-morphisms are defects of codimension 2. These point-like defects (for example the bulk point in Fig.\,\ref{toric}) are completely classified in \Ref{KK1251} by bimodule functors. For example, 
\bnu
\item a 2-morphism from the 1-morphism ${}_{\rep_{\Zb_2}} (\rep_{\Zb_2})_{\rep_{\Zb_2}}$ to itself is given by a bimodule functor from $\rep_{\Zb_2}$ to $\rep_{\Zb_2}$. Since ${}_{\rep_{\Zb_2}} (\rep_{\Zb_2})_{\rep_{\Zb_2}}$ is a trivial defect line, a point-like defect on such defect line should be nothing but a bulk excitation. Therefore, the bulk excitations, or defects of codimension 2 on the trivial domain wall $\rep_{\Zb_2}$, are classified by the objects in the category $\fun_{\rep_{\Zb_2}|\rep_{\Zb_2}}(\rep_{\Zb_2}, \rep_{\Zb_2})$
of bimodule functors from $\rep_{\Zb_2}$ to $\rep_{\Zb_2}$. There are 4 such bimodule functors that are simple. They are denoted by $1,e,m,\epsilon$. The category $\fun_{\rep_{\Zb_2}|\rep_{\Zb_2}}(\rep_{\Zb_2}, \rep_{\Zb_2})$ is also called {\it monoidal center} of the monoidal category $\rep_{\Zb_2}$, often denoted by $Z(\rep_{\Zb_2})$. 

\item Similarly, 2-morphisms from 1-morphism ${}_{\rep_{\Zb_2}} \hilb_{\rep_{\Zb_2}}$ to itself is given by objects in $\fun_{\rep_{\Zb_2}|\rep_{\Zb_2}}(\hilb, \hilb)$, which is actually equivalent to $Z(\rep_{\Zb_2})$ as monoidal categories. Namely, it also contains four simple objects, corresponding to four simple wall excitations. 

\item 2-morphisms from 1-morphism $\rep_{\Zb_2}$ to $\hilb$ are given by objects in $\fun_{\rep_{\Zb_2}|\rep_{\Zb_2}}(\rep_{\Zb_2}, \hilb)$. An example of such 2-morphism is depicted as the lattice configuration around the blue point in Fig.\,\ref{toric}. There is a stabilizer operator 
\begin{equation} \label{eq:Q-op}
Q=\sigma_6^x \sigma_{17}^y\sigma_{18}^z \sigma_{19}^z \sigma_{20}^z,
\end{equation}
which commutes with other stabilizers (see eq.(8) in \Ref{KK1251}). Two eigenvalues of $Q$ correspond to two distinct simple 2-morphisms, or two simple objects $F_\pm$ in $\fun_{\rep_{\Zb_2}|\rep_{\Zb_2}}(\rep_{\Zb_2}, \hilb)$. 

\item 2-morphisms from 1-morphism $\hilb$ to $\rep_{\Zb_2}$ are given by objects in 
$\fun_{\rep_{\Zb_2}|\rep_{\Zb_2}}(\hilb, \rep_{\Zb_2})$. There are again two simple 2-morphisms given by $\overline{F}_\pm$ which is the two-sides adjoint of $F_\pm$. 

\item 2-morphisms from 1-morphisms ${}_{\rep_{\Zb_2}}(\ER)_{\hilb}$ to itself, if $\ER=\rep_{\Zb_2}$ or $\hilb$, are given by objects in $\fun_{\rep_{\Zb_2}}(\ER, \ER) \simeq \rep_{\Zb_2}$ for both cases. 

\enu

\medskip
\noindent {\it $\bullet$ 3-morphisms}: 3-morphisms are given by instantons which can be viewed as defects in the time direction. More precisely, in this case, they are natural transformation between bimodule functors. We recall a typical example of an instanton given in \Ref{KK1251}. Imagine a dotted vertical interval $\hilb$-wall in Fig.\,\ref{toric}, with a upper end $F_+$ and a lower end $\overline{F}_+$, shrinking in the time direction and finally disappeared. This is given by an instanton. More precisely, two defect junction $F_+$ and $\overline{F}_+$, when viewed from far away, fuse into a single defect junction on the trivial defect, i.e. a bulk excitation which was shown to be $1\oplus \epsilon$. Then the instanton describe above is the morphism $1 \oplus \epsilon \to 1$ in $Z(\rep_{\Zb_2})$. Moreover, it is clear that all morphisms in the categories: $Z(\rep_{\Zb_2})$, $\fun_{\rep_{\Zb_2}|\rep_{\Zb_2}}(\hilb, \hilb)$ and $\fun_{\rep_{\Zb_2}|\rep_{\Zb_2}}(\rep_{\Zb_2}, \hilb)$, etc. are instantons or 3-morphisms in $\TC_3^b$.

\bigskip
In addition to above 4 layers of structures: 0-,1-,2-,3-morphisms, there are much more structures naturally required by physics. We will illustrate them one by one.

\medskip
\noindent {\it $\bullet$ Composition of morphisms} 

\bnu

\item Physically, when two instantons move closer to each other in the time direction, they can be viewed as a single instanton. This says that 3-morphisms can be composed. Since they are given by morphisms in an ordinary 1-category, they can be composed just as usual. 

\item 2-morphisms, or the defect junctions, can be fused as particles. Indeed, mathematically, these particle-like excitations are given by module functors. So the fusion among these particles is exactly given by the composition of module functors. For example, the four bulk excitations, $1,e,m,\epsilon$ fuse exactly as the composition of functors in $Z(\rep_{\Zb_2})=\fun_{\rep_{\Zb_2}|\rep_{\Zb_2}}(\rep_{\Zb_2}, \rep_{\Zb_2})$. This gives arise to a monoidal structure on $Z(\rep_{\Zb_2})$. Notice that when two particles fuse, the two instantons living on the time line (about the same time), which ends at these two particles, also move close to each other. This process gives arise to a potentially new composition of 3-morphisms. 
$$
(e\xrightarrow{f} e, m\xrightarrow{g} m) \mapsto (e\otimes m \xrightarrow{f\otimes g} e\otimes m). 
$$
Mathematically, it is achieved by the fact that $\otimes$ is a functor which automatically fuse the instantons.

Similarly, the fusion of particles also give each of the following categories of 2-morphisms $\fun_{\rep_{\Zb_2}|\rep_{\Zb_2}}(\rep_{\Zb_2}, \rep_{\Zb_2})$, $\fun_{\rep_{\Zb_2}}(\rep_{\Zb_2}, \rep_{\Zb_2})$, $\fun_{\rep_{\Zb_2}}(\hilb, \hilb)$ a structure of a monoidal category. 

Moreover, a defect junction from $\rep_{\Zb_2}$ to $\hilb$ can fuse with a
defect junction from $\hilb$ to $\rep_{\Zb_2}$ to give a defect junction from
$\rep_{\Zb_2}$ to $\rep_{\Zb_2}$ (a bulk excitation), or from $\hilb$ to $\hilb$. More precisely,
by Eq.\,(35) in \Ref{KK1251}, we have
$$
F_+ \circ \overline{F}_+ = F_-\circ \overline{F}_- \simeq 1 \oplus \epsilon, 
$$
\vspace{-.7cm}
$$
F_+ \circ \overline{F}_- = F_-\circ \overline{F}_+ \simeq e \oplus m. 
$$

\item 1-morphisms can be composed. For example, consider two domain walls ${}_{\rep_{\Zb_2}} \EM_{\rep_{\Zb_2}}$ and ${}_{\rep_{\Zb_2}} \EN_{\rep_{\Zb_2}}$ sitting parallel and next to each other. Then, when viewed from far away, they simply fuse into a single domain wall, which is given by
$$
\EM \boxtimes_{\rep_{\Zb_2}} \EN
$$
where the tensor product $\boxtimes_{\rep_{\Zb_2}}$ is well-defined mathematically\cite{eno} and the resulting category is again an $\rep_{\Zb_2}$-bimodule. Notice that $\rep_{\Zb_2}$ is the trivial domain wall. It is trivial in the sense that if we replace $\EM$ by $\rep_{\Zb_2}$, then viewed from far away the fused wall $\rep_{\Zb_2}\boxtimes_{\rep_{\Zb_2}} \EN$ must be the same as a single $\EN$-wall. 
Yes, indeed, $\rep_{\Zb_2}\boxtimes_{\rep_{\Zb_2}} \EN \simeq \EN$ is guaranteed mathematically by the defining properties of the tensor product $\boxtimes_{\rep_{\Zb_2}}$. In other words, under the tensor product the trivial $\rep_{\Zb_2}$-wall acts like an identity 1-morphism. It is also clear that the composition of 1-morphisms are associative, i.e.
\begin{align}
(\EL \boxtimes_{\rep_{\Zb_2}} \EM) \boxtimes_{\rep_{\Zb_2}} \EN &  \nn
 & \hspace{-2cm} \simeq  
\EL \boxtimes_{\rep_{\Zb_2}} (\EM \boxtimes_{\rep_{\Zb_2}} \EN).
\end{align}
Also notice that fusion of domain wall also fuse excitations on different wall horizontally. This process provides a (potentially) new composition of 2-morphisms, 
and at the same time, it provides a (potentially) new composition of 3-morphisms.

\enu

\medskip
\noindent {\it $\bullet$ Bulk-to-wall maps}: This structure can be viewed as a substructure of the composition of morphisms. But due to its importance in our study later, it is beneficial to discuss them in detail now. 

Let $\EM$ be a $\rep_{\Zb_2}$-bimodule. The $1,e,m.\epsilon$-particles can fuse into the $\EM$-wall and becoming wall excitations. This process gives arise to two maps, called left/right bulk-to-wall maps, which are given by two monoidal functors $L$ and $R$: 
\begin{equation} \label{diag:cospan}
\hspace{0.6cm} Z(\rep_{\Zb_2}) \xrightarrow{L} Z(\EM) \xleftarrow{R} Z(\rep_{\Zb_2}),
\end{equation}
where $Z(\EM): =\fun_{\rep_{\Zb_2}|\rep_{\Zb_2}}(\EM, \EM)$ 
is a unitary fusion category of $\rep_{\Zb_2}$-bimodule functors from $\EM$ to $\EM$ and describes the excitations on the $\EM$-wall. More precisely, the functor $L$ and $R$ are defined as follows: 
\begin{align}
L: \quad 1/e/m/\epsilon  &\mapsto 1/e/m/\epsilon \otimes_{\rep_{\Zb_2}} \id_\EM. \nn
R: \quad 1/e/m/\epsilon  &\mapsto  \id_\EM \otimes_{\rep_{\Zb_2}} 1/e/m/\epsilon  \nonumber
\end{align}
where we have used the fact that the anyons $1,e,m,\epsilon$ can be viewed as bimodule endo-functors on $\rep_{\Zb_2}$ and that their images can be viewed as an endo-functors on $\rep_{\Zb_2} \otimes_{\rep_{\Zb_2}} \EM\simeq \EM$ for $L$ and
$\EM \otimes_{\rep_{\Zb_2}} \rep_{\Zb_2} \simeq \EM$ for $R$. The functors $L$ and $R$ can be combined into a two-side bulk-to-wall map: 
\begin{equation}  \label{eq:bulk-to-wall-0}
Z(\rep_{\Zb_2}) \boxtimes Z(\rep_{\Zb_2}) \xrightarrow{L \boxtimes R}
Z(\EM).
\end{equation}
%where $Z(\rep_{\Zb_2})^{\otimes^\op}$ is the same category as $Z(\rep_{\Zb_2})$ but equipped with the opposite tensor product $\otimes^\op$, i.e. $a\otimes^\op b := b\otimes a$.
The bulk-to-wall map (\ref{eq:bulk-to-wall-0}) is a {\it dominant functor}, which means that any object in $Z(\EM)$ appear as a subobject of an object in the image of $L \boxtimes R$. 
%In toric code case, or anyon condensation in general, the dominance of the bulk-to-wall map is sufficient to prove that the domain wall can be obtained by anyon condensation\cite{kong-anyon}. 
Moreover, the functors $L$ and $R$ are also central\cite{dmno,fsv,kong-anyon}.

\medskip
\noindent {\it $\bullet$ Braidings}: The bulk excitations can be braided. It is true for all defect junctions living in a trivial defect line. We will show later that the braiding structure is automatic for endo 2-morphisms of the identity 1-morphism in an $n$-category for $n\geq 2$. 

\medskip
\noindent {\it $\bullet$ Half Braidings}: A bulk excitation can be half-braided with a wall-excitation. The general braiding between bulk excitations $e/m/\epsilon$ and $F_\pm$ (or $\overline{F}_\pm$) is encoded in the following commutative diagrams up to isomorphisms $\phi_L$ and $\phi_R$:
\begin{equation} \label{diag:2-cospan}
\raisebox{4.5em}{
\xymatrix@R=2.5em@C=0.9em{
& Z(\rep_{\Zb_2})  \ar[d]^{F_\pm \circ -}  &  \\
Z(\rep_{\Zb_2}) \ar[ur]^{\id} \ar[dr]_{L_2^{[1]}} & 
\phi_L \Downarrow \,\,\,\,\,\,\ \EE^{[0]} \,\,\,\,\,\, \Downarrow \phi_R &
Z(\rep_{\Zb_2}) \ar[ul]_{\id} \ar[dl]^{R_2^{[1]}} \\
&  Z_{\rep_{\Zb_2}}(\hilb) \ar[u]^{-\circ F_\pm} & 
}}
\end{equation}
where $\EE^{[0]}=\fun_{\rep_{\Zb_2}|\rep_{\Zb_2}}(\rep_{\Zb_2},\hilb)$ and both functors $L_2^{[1]}$ and $R_2^{[1]}$ are invertible. These isomorphisms $\phi_L$ and $\phi_R$ give the so-called half-braiding between the bulk excitation and 
defects $F_\pm$ in $\EE^{[0]}$.  

\begin{rema}  \label{rema:Z2-crossed}
In above case, since both functors $L_2^{[1]}$ and $R_2^{[1]}$ are invertible,  the isomorphisms $\phi_L$ and $\phi_R$ actually gives the following full braiding: 
\begin{equation} \label{eq:double-braiding-toric}
F_\pm \otimes e \mapsto m \otimes F_\pm, \quad
F_\pm \otimes m \mapsto e \otimes F_\pm,
\end{equation}
which has an interesting $\Zb_2$-crossed braiding structure. The toric code enriched by the transparent domain wall as shown in Figure\,\ref{toric} actually gives an example of symmetry enriched topological order\cite{CGW1038}. The $\Zb_2$-crossed braiding in (\ref{eq:double-braiding-toric}) is a part of the structure of a $\Zb_2$-crossed braided fusion category, which is obtained by a $\Zb_2$-extension of $Z(\rep_{\Zb_2})$ and describes an symmetry enriched topological order. 
Both the modular tensor category $Z(\rep_{\Zb_2})$ and its $\Zb_2$-extension can be viewed as two different minimal descriptions of the toric code model from two different points of view. 
%We will return to this point in Remark\,\ref{rema:SET-toric}. 
%In this work, according to our convention, when we detect excitations via braiding, an excitation is not allowed to cross a non-trivial domain wall, even though the wall is invertible or transparent. 
\end{rema}

\medskip
In summary, the defects in the toric code model form a 4-layered structure: only one object or 0-morphism, which can be labeled by $\rep_{\Zb_2}$, 1-morphisms given by $\rep_{\Zb_2}$-bimodules, 2-morphisms given by bimodule functors and 3-morphisms given by natural transformations between bimodule functors. This 4-layered structure needs to be enriched in order to describe a physical topological order. 
In particular, the composition of $i$-morphisms should be introduced for $i>0$, and all compositions are associative and unital. Certain braiding structures, including the half-braiding, are needed to describe a physical topological order. We will show in later sections that the notion of $n$-category automatically encode these structures, thus can be used as a proper mathematical language to model the topological properties of excitations in a topological order. 

\begin{rema}  \label{rema:fus-3-cat}
For general Levin-Wen models, one simply replace the only 0-morphism $\rep_{\Zb_2}$ by a unitary fusion category $\EC$, everything else remains the same. We obtain a new 3-category with one object $\EC$. More generally, we have a 3-category $\EF\mathrm{us}$ with objects given by unitary fusion categories, 1-morphisms by bimodules, 2-morphisms by bimodule functors and 3-morphisms by natural transformations between bimodule functors. 
\end{rema}

\subsection{Defects in an $(n+1)$-dimensional topological order} \label{sec:defect-n-toporder}

In this subsection, we will generalize topological properties of defects in toric code model to an arbitrary $(n+1)$-dimensional topological order. These properties summarized in this subsection should serve as a guide to formulate a mathematical definition of a $\BF_{n+1}$-category. 

\medskip
\noindent 1. {\it Defects of all codimensions}: In each codimension $l$ ($1\leq l \leq n+1$), there are only finite types of simple topological excitations (or defects), labeled by $i^{[l]}, j^{[l]}, k^{[l]}$, etc.  Notice that the superscript of $i^{[l]}$ represent the codimension of the defect and will be omitted if it is clear from the context. The trivial pure $l$-codimensional defect is denoted by $1^{[l]}$. We assume that $1^{[l]}$ is simple. A general defect can be composite (not simple). These topological excitations are instantons for $l=n+1$; particle-like excitations for $l=n$; string-like excitations for $l=n-1$; surface-like excitations for $l=n-2$, etc. This is a very rough way of labeling these excitations. In general, a $(l+1)$-codimensional defect can be a domain wall between two (not necessarily different) $l$-codimensional defects, each of which again can be a domain wall between two $(l-1)$-codimensional defects, so on and so forth. 

All the gapped domain walls between two $l$-codimensional defect $x^{[l]}$ and $y^{[l]}$ and domain wall between domain walls, etc., form an interesting mutli-layered structure, which will be denoted by $\hom(x,y)$. It will become clear later that $\hom(x^{[l]},y^{[l]})$ is an $(n-l)$-category.

\void{
\begin{rema}
Since a defect of codimension 1 in an $n$-spatial dimensional topological phase can not be detected by other excitations from far away via braiding in $n$ dimension, a topological phase with non-trivial codimension 1 excitations are anomalous in general. In other word, a closed topological order does not contain any non-trivial codimension 1 excitations. For example, in a 2+1 dimensional phase, a small closed string viewed by an excitation from far away is nothing but a point-like particle. There is no way to braiding them along a path which has a non-trivial linking number with the closed string. But a non-trivial closed string can be potentially detected via braiding of particles from a 3+1 dimensional bulk. Therefore, a 2+1D phase with a non-trivial string-like excitation need have a non-trivial 3+1 dimensional bulk theory. 
\end{rema}
}
  
\medskip
\noindent 2. {\it Ground-state degeneracy}: 
We fixed the locations of all the topological defects, and assume that the topological excitations are well separated. In this limit, the topological degeneracy is robust against any local perturbations of the Hamiltonian.

The ground state is obtained by decorating the space $\Sigma$ with trivial defects $1^{[l]}$ for all $0\leq l\leq n$. Then we obtain the space of ground states, denoted by $\hom_{\Sigma}(1^{[n]}, 1^{[n]})$, which is also called {\it ground-state degeneracy}. It depends on the topology of $\Sigma$ in general. When $\Sigma=S^n$,  if the ground-state degeneracy is trivial, i.e.
\begin{equation} \label{eq:gs-deg}
\hom_{S^n}(1^{[n]}, 1^{[n]}) \simeq \Cb,
\end{equation}
then the theory is called {\it stable}. By the stable condition (\ref{eq:gs-deg}), this degeneracy is independent of how many trivial defect $1^{[n]}$ we introduce. Otherwise, it is clearly not true. This tells us why the stable condition (\ref{eq:gs-deg}) is natural. We will return to this point later. 

We can also define a general space of lowest energy states with the appearance of  nontrivial higher dimensional defects $i_1, \cdots, i_k$. The ground-state degeneracy in this case is given by the Hilbert space: $\otimes_{j=1}^k \hom_\Sigma(1^{[n]}_{1\cdots i_j \cdots}, 1^{[n]}_{1\cdots i_j \cdots})$, where $1^{[n]}_{1\cdots i_j \cdots}$ represents the trivial sub-defect on $i_j$. 

\smallskip
\noindent 3. {\it Fusion between defects of the same codimension}: Defects of the same codimension can be fused (see Fig.~\ref{fuse}). For example, any two adjacent domain walls, when viewed from far away, simply fuse to a single domain wall. This gives arise to a fusion product $i^{[l]} \otimes j^{[l]}$. The trivial domain walls
act like the unit under the fusion product. More precisely, if $f$ is domain wall between 
two defects $x^{[l]}$ and $y^{[l]}$, and the trivial domain walls $1_{x^{[l]}}, 1_{y^{[l]}}$ inside the defect $x^{[l]}$ and $y^{[l]}$, respectively, act like the units for the fusion product, i.e. $1_{x^{[l]}} \otimes f  = f$ and $f \otimes 1_{y^{[l]}} = f$. These structures match exactly with the composition of higher morphisms in a higher category. The trivial domain wall behave like the identity higher morphism in a higher category. 

The information of fusion product is encoded in the so-called ``hom space" or ``fusion rules", denoted by $\hom(i^{[l]} \otimes j^{[l]}, k^{[l]})$, which contains the information of all gapped domain walls between $i\otimes j$ and $k$ and domain walls between domain walls, etc. This generalized fusion rule $\hom(i^{[l]} \otimes j^{[l]}, k^{[l]})$ is actually an $(n-l)$-category. For example, for $l=n+1$ (instantons), the ``fusion rules" $\hom(i^{[n+1]} \otimes j^{[n+1]}, k^{[n+1]})$ is empty; for $l=n$ (particle-like defects), the fusion rules $\hom(i^{[n]} \otimes j^{[n]}, k^{[n]})$ is a finite dimensional Hilbert space; for $l=n-1$ (defect lines), the fusion rule $\hom(i^{[n-1]} \otimes j^{[n-1]}, k^{[n-1]})$ is a unitary $1$-category. 

For example, in Levin-Wen type of lattice models constructed in \Ref{KK1251}, if three domain walls, associated to three bimodule categories ${}_{\EC_1}(\EM_{13})_{\EC_3}, {}_{\EC_1}(\EM_{12})_{\EC_2}$ and ${}_{\EC_2} (\EM_{23})_{\EC_3}$, respectively, are connected by a defect junction as follows:
\begin{equation} \label{eq:3-way}
%\EF \in \fun_{\EC_1|\EC_3}\bigl(\EM_{12} \otimes_{\EC_2} \EM_{23}, \EM_{13}\bigr)\qquad
\figbox{1.0}{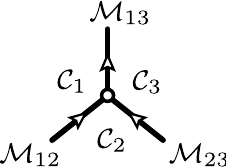}
\end{equation}
then the defect junctions are classified by the unitary 1-category $\fun_{\EC_1|\EC_3}(\EM_{12}\otimes_{\EC_2} \EM_{23}, \EM_{13})$ of $\EC_1$-$\EC_3$-bimodule functors from $\EM_{12}\otimes_{\EC_2} \EM_{23}$ to $\EM_{13}$.\cite{KK1251}

In general, the hom space $\hom(i^{[l]} \otimes j^{[l]}, k^{[l]})$ is an $(n-l)$-category. But once we select the particle-like excitations $a,b,c$ on defects $i^{[l]}, j^{[l]}$ and $k^{[l]}$, we should always obtain a finite dimensional Hilbert space
$$
\hom(a^{[n]}_{\cdots i^{[l]} \cdots} \otimes  b^{[n]}_{\cdots j^{[l]} \cdots}, 
c^{[n]}_{\cdots k^{[l]} \cdots}) \cong  \Cb^N
$$
for some finite $N\in \Zb_{\geq 0}$. 
These ``$(n-l)$-categorical fusion rules" define a fusion (or tensor) product $\otimes$ among $l$-codimensional defects. By natural physical requirements, these fusion products must be associative, i.e. existing an associator: $\alpha_{i,j,k}^{[l]}: i^{[l]} \otimes (j^{[l]} \otimes k^{[l]}) \xrightarrow{\simeq} (i^{[l]} \otimes j^{[l]}) \otimes k^{[l]}$. Moreover, the trivial type $1^{[l]}$ fuses as a tensor unit, i.e. $1^{[l]} \otimes i^{[l]} \simeq i^{[l]} \simeq i^{[l]} \otimes 1^{[l]}$. We require that these associators and unit isomorphisms are unitary. By that we mean, for arbitrary decoration of particle-like excitations on $l$-dimensional excitations $i^{[l]}, j^{[l]}$ and $k^{[l]}$, the data of the associators and unit isomorphisms boil down to finite number of linear maps between finite dimensional Hilbert spaces. We require these linear maps to be unitary. In categorical language, these fusion rules provide a monoidal structure among all $l$-codimensional defects. 

\medskip
\noindent 4. {\it General fusions}: Topological defects of different dimensions can also fuse. For example, for $l\geq l'$, a pure $l$-codimensional defect $x^{[l]}$ in the bulk can fuse into an $l'$-codimensional defect $y^{[l']}$ and becomes an $l$-codimensional defect nested in $y$. The information of this kind of fusion is automatically included in the fusion between $1_{1\cdots 1}^{[l']}$ and $y^{[l']}$ in the same dimension. 

\medskip
\noindent 5. {\it 1-dimensional bulk-to-wall maps}: A special case of the general fusion will be important to us later. It is called {\it bulk-to-wall maps}. Consider two gapped domain wall $f^{[l+1]}, g^{[l+1]}\in \hom(x^{[l]}, y^{[l]})$ between $x^{[l]}$ and $y^{[l]}$. An $(l+2)$- (or higher) codimensional defect nested in $x^{[l]}$ or $y^{[l]}$ can fuse into the wall $f$ and become a defect nested in $f$. Notice that non-trivial domain walls in $x^{[l]}$ or $y^{[l]}$ can not be included because they might change the type of the domain wall $f$. This process give left/right 1-dimensional bulk-to-wall maps: 
$$
\hom^{> 1}(x,x) \xrightarrow{L_f} \hom(f,f) \xleftarrow{R_f} \hom^{> 1}(y,y),
$$
or equivalently, the two-side bulk-to-wall map:
\begin{equation} \label{eq:bulk-to-wall}
L_f\boxtimes R_f: \hom^{> 1}(x,x)\boxtimes \hom^{> 1}(y,y) \to \hom(f,f). 
\end{equation}
Note that the notation $\hom^{> 1}(\cdot,\cdot)$ simply means that all non-trivial $(l+1)$-codimensional defects are excluded. %It will be formulated more precisely later in categorical language as the looping of the category $\hom(\cdot,\cdot)$. 

The set of $(l+2)$- (and higher) codimensional defects nested in $x^{[l]}$ or $y^{[l]}$ acts on the set of defects between $f^{[l+1]}$ and $g^{[l+1]}$ (including domain walls and wall between walls). More explicitly, for a sub-defect $a^{[k]}$ in $x$, $b^{[k]}$ in $y$ and a sub-defect $m^{[k]}$ in $\hom(f,g)$ and $k\geq l+2$, the action 
\begin{align}
 \hom^{>1}(x,x)\boxtimes \hom(f,g) \boxtimes \hom^{>1}(y,y) &\to \hom(f,g) \nn
(a, m, b)  &\mapsto a\otimes m \otimes b, \nonumber
\end{align}
The action is clearly associative and unital. Therefore, $\hom(f,g)$ as a topological phase is a $\hom^{>1}(x,x)$-$\hom^{>1}(y,y)$-bimodule in some sense. The two-side bulk-to-wall map in (\ref{eq:bulk-to-wall}) can be recovered from this action by taking $m$ to be the trivial sub-defect $1_{1\cdots_f}^{[k]}$ in $f$, i.e. $g=f$ and
\begin{equation}  \label{eq:bulk-to-wall-2}
L\boxtimes R:\, (a,b) \mapsto a\otimes 1_{1\cdots_f}^{[k]} \otimes b.
\end{equation}
In general, for fixed $m$, this action defines the left/right 1-dimensional bulk-to-wall maps: 
\begin{equation} \label{diag:L-m-R-m}
 \hom^{>1}(x,x) \xrightarrow{L_m} \hom^{>1}(f,g) \xleftarrow{R_m} \hom^{>1}(y,y), 
\end{equation}
and a two-side 1-dimensional bulk-to-wall map 
\begin{equation}  \label{diag:two-side-bulk-to-wall}
 \hom^{>1}(x,x) \boxtimes  \hom^{>1}(y,y) \xrightarrow{L_m\boxtimes R_m}
  \hom^{>1}(f,g).
\end{equation}

\medskip
\noindent 6. {\it $k$-dimensional bulk-to-wall maps}: 
In general, all $(l+k+1)$- and higher codimensional defects nested in a defect $x^{[l]}$ can be fused into a defect $z^{[l+k]}$ directly as long as $z$ is sitting adjacent to $x$ in the sense that either $z$ is nested in $x$ or $z$ lies in a boundary of $x$ connected to another defect $y^{[l]}$. We assume the later situation as it includes the former one  as a special case. %This fusion map is actually the composition of $k$ number of bulk-to-wall maps, each of which increase the codimension by 1 until the $k$-th map ends on $z^{[k]}$. 
For simplicity, we will refer to the fusion map
\begin{align} \label{eq:k-bulk-to-wall}
\hom^{> k}(x^{[l]},x^{[l]}) \boxtimes \hom^{> k}(y^{[l]},y^{[l]}) &  \nn
&  \hspace{-2cm} \to  \hom(z^{[l+k]}, z^{[l+k]})
\end{align}
as the {\it two-side $k$-dimensional bulk-to-wall map from $x^{[l]}\otimes y^{[l]}$ to $z^{[l+k]}$}. 
The two-side $1$-dimensional bulk-to-wall map in (\ref{diag:two-side-bulk-to-wall}) can also be generalized to higher dimensions by replacing the second $z^{[l+k]}$ in (\ref{eq:k-bulk-to-wall}) by another defect $u^{[l+k]}$.  

%In this context, the bulk-to-wall map in (\ref{eq:bulk-to-wall}) is $1$-dimensional. It is important to note that a $k$-dimensional bulk-to-wall map is not necessary a composition of $k$ number of $1$-dimensional bulk-to-wall maps. More precisely, if $z^{[l+k]}$ is the domain wall between $u^{[l+k-1]}$ and $v^{[l+k-1]}$, then the $k$-dimensional bulk-to-wall map from $x\otimes y$ to $z$ is not necessary the composition of the bulk-to-wall maps from $x\otimes y$ to $u, v$ and that from $u\otimes v$ to $z$. This situation is described in Remark\,\ref{rema:bulk-to-wall}.

\medskip
\noindent 7. {\it Anti-excitations}: For each $l$-codimensional defect $x$, there is an $l$-dimensional anti-excitation $\bar{x}$ such that a pair of such excitations can be created from (or annihilated to) the vacuum $1^{[l]}$. In categorical language, this amounts to a rigidity or duality structure on the fusion among all $l$-codimensional defects. For example, in the case of a $2+1$-topological phase, an anyon has a dual given by the anti-anyon. The category of anyons is, in particular, a rigid tensor category or tensor category with duals. 

\medskip
\noindent 6. {\it Braiding between defects}: Any defects of codimension 2 or higher nested in the trivial lower codimensional defects can be braided. Moreover, any $(l+2)$- (or higher) codimensional defects nested in an $l$-codimensional defect $x$ can be braided within $x$.

\medskip
\noindent 7. {\it Half braidings}: If $f^{[l+1]}$ is a gapped domain wall between $x^{[l]}$ and $y^{[l]}$. Then an $(l+2)$- (or higher) codimensional defect nested in $x^{[l]}$ or $y^{[l]}$ can be half-braided with defects in $f$. More explicitly, we consider two gapped domain wall $f^{[l+1]}, g^{[l+1]}\in \hom(x^{[l]}, y^{[l]})$ between $x$ and $y$. For a gapped domain wall $\psi^{[l+2]}$ in $\hom(f,g)$, defects in $\hom(f,f)$ and $\hom(g,g)$ can fuse onto $\psi$. This gives arise to two maps denoted by $\psi_\ast$ and $\psi^\ast$. Thus we obtain the following diagram (recall the diagram (\ref{diag:2-cospan})): 
\begin{equation} \label{diag:id-f-g-phi}
\raisebox{0em}{
\xymatrix@R=1.8em@C=0.3em{ 
&  & \hom(f,f) \ar[d]^{\psi_\ast} & & \\
\hom^{>1}(x,x)  \ar[urr]^{L_f} \ar[drr]_{L_g}  & \Downarrow \beta_L & \hom(f,g) & \Downarrow \beta_R & \hom^{>1}(y,y) \ar[dll]^{R_g} \ar[ull]_{R_f} \\
& & \hom(g,g) \ar[u]^{\psi^\ast} &  &
}}
\end{equation}
where $\beta_L$ represents the physical process of deforming the fusion path from $\psi_\ast \circ L_f$ to 
$\psi^\ast \circ L_g$ and $\beta_R$ is similar. These processes are nothing but the (left/right) half-braidings. 

\smallskip
Notice that each domain wall $\psi$ between $f$ and $g$ defines a 1-dimensional bulk-to-wall map 
\begin{equation}  \label{diag:L-phi-R-phi}
\hom^{>1}(x,x) \xrightarrow{L_\psi} \hom(f,g) \xleftarrow{R_\psi} \hom^{>1}(y,y),
\end{equation}
where $L_\psi:= \psi_\ast \circ L_f \simeq \psi^\ast \circ L_g$ and 
$R_\psi:= \psi_\ast \circ R_g \simeq \psi^\ast \circ R_f$. The defects nested in $x^{[l]}$ or $y^{[l]}$ are not allowed to cross the domain wall, unless the wall $f$ is invertible (or transparent), e.g. $x=y$ and $f$ is the trivial domain wall $1_x^{[l+1]}$.

\medskip
\noindent 8. {\it Half braiding between wall excitations}: The most general half-braiding occurs between two wall excitations. It is illustrated in Figure.\,\ref{fig:half-braiding-2-wall-excitations} (see \Ref{dkr} for a discussion of this braiding in 1+1D conformal field theory). The half braiding discussed before is a special case of this general half braiding. It is actually a special case of equ. (\ref{eq:h-v-compatibility}). Half braidings generate all full braidings.

\begin{figure}[bt]
\centerline{\begin{tabular}{@{}c@{\quad}c@{\quad}c}
\raisebox{-50pt}{
  \begin{picture}(105,120)
   \put(0,8){\scalebox{0.6}{\includegraphics{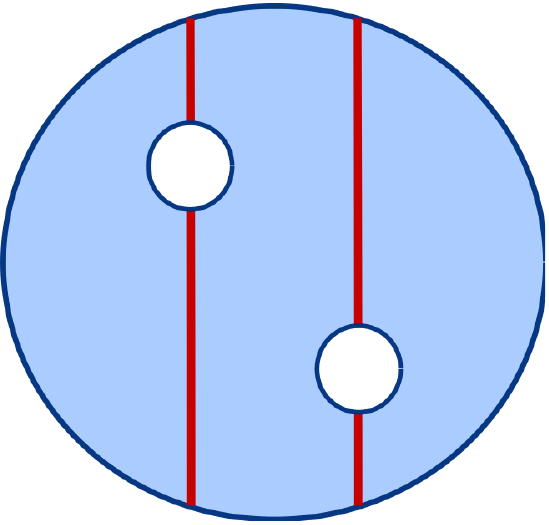}}}
   \put(0,8){
     \setlength{\unitlength}{.75pt}\put(-18,-19){
     \put(50, 32)      {\scriptsize $ g $}
     \put(47, 120)       {\scriptsize $ g' $}
     \put(50, 75)       {\scriptsize $ g $} 
     \put(88, 120)       {\scriptsize $ f' $}
     \put( 88, 75)       {\scriptsize $ f' $}
     \put( 90, 30)       {\scriptsize $ f $}
     \put(37, 55)        {\scriptsize $ x $}
     \put(71, 55)       {\scriptsize $ y $}
     \put(118, 55)    {\scriptsize $ z $} 
     \put(59, 100)     {\scriptsize $\phi$}
     \put(97, 53)     {\scriptsize $\psi$}
     }\setlength{\unitlength}{1pt}}
  \end{picture}}
& $\rightsquigarrow$ &
\raisebox{-50pt}{
  \begin{picture}(105,120)
   \put(0,8){\scalebox{0.6}{\includegraphics{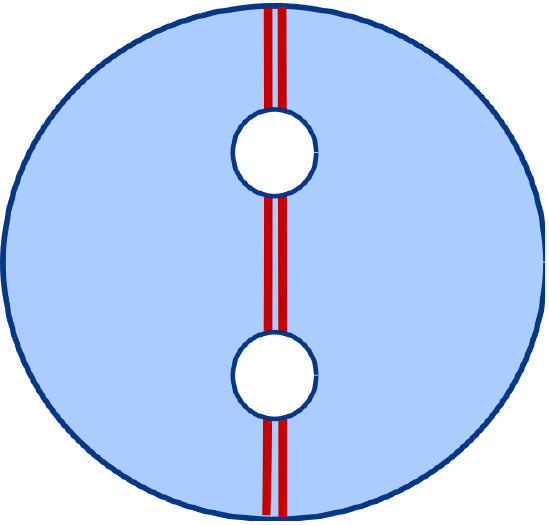}}}
   \put(0,8){
     \setlength{\unitlength}{.75pt}\put(-18,-19){
     \put(53, 120)       {\scriptsize $ g'f' $}
     \put(58, 75)       {\scriptsize $ gf' $} 
     \put(60, 32)      {\scriptsize $ gf $}
      \put(37, 55)        {\scriptsize $ x $}
     \put(118, 55)    {\scriptsize $ z $} 
     \put(77, 103)     {\scriptsize $\phi$}
     \put(77, 52)     {\scriptsize $\psi$}
     }\setlength{\unitlength}{1pt}}
  \end{picture}}
\\[5pt]
\raisebox{-50pt}{
  \begin{picture}(105,120)
   \put(0,8){\scalebox{0.6}{\includegraphics{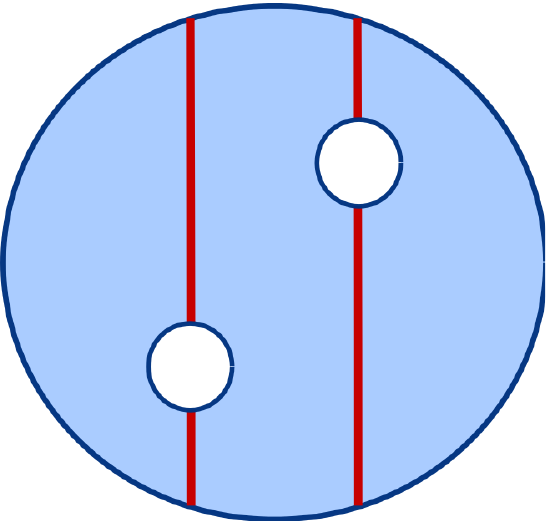}}}
   \put(0,8){
     \setlength{\unitlength}{.75pt}\put(-18,-19){
     \put(50, 32)      {\scriptsize $ g $}
     \put(47, 120)       {\scriptsize $ g' $}
     \put(47, 75)       {\scriptsize $ g' $} 
     \put(88, 120)       {\scriptsize $ f' $}
     \put( 90, 75)       {\scriptsize $ f $}
     \put( 90, 30)       {\scriptsize $ f $}
      \put(37, 55)        {\scriptsize $ x $}
     \put(78, 55)       {\scriptsize $ y $}
     \put(118, 55)    {\scriptsize $ z $} 
     \put(97, 100)     {\scriptsize $\psi$}
     \put(59, 53)     {\scriptsize $\phi$}
     }\setlength{\unitlength}{1pt}}
  \end{picture}}
& $\rightsquigarrow$ &
\raisebox{-50pt}{
  \begin{picture}(105,120)
   \put(0,8){\scalebox{0.6}{\includegraphics{comm-def-fields-a2}}}
   \put(0,8){
     \setlength{\unitlength}{.75pt}\put(-18,-19){
     \put(53, 120)       {\scriptsize $ g'f' $}
     \put(57, 75)       {\scriptsize $ g'f $} 
     \put(60, 32)      {\scriptsize $ gf $}
      \put(37, 55)        {\scriptsize $ x $}
     \put(118, 55)    {\scriptsize $ z $} 
     \put(77, 103)     {\scriptsize $\psi$}
     \put(77, 52)     {\scriptsize $\phi$}
     }\setlength{\unitlength}{1pt}}
  \end{picture}}
\end{tabular}}
\caption{Above four picture illustrate the half braiding between two walls between walls. More precisely, let $x^{[l]}$, $y^{[l]}$, $z^{[l]}$ be three $l$-codimensional defects, $g, g'$  domain walls between $x$ and $y$, $f,f'$ are walls between $y$ and $z$, and $\phi^{[l+2]}$ is a wall between $g$ and $g'$ and $\psi^{[l+2]}$ a wall between $f$ and $f'$. The two squigarrows $\rightsquigarrow$ represents fusion of $(l+1)$-codimensional walls. The compatibility between the first column implies the compatibility of the second column, which is a half braiding. 
}
\label{fig:half-braiding-2-wall-excitations}
\end{figure}

\medskip
\noindent 9. {\it Compatibility of fusions in different directions}: 
Actually, the half braiding depicted in Figure\,\ref{fig:half-braiding-2-wall-excitations} is a consequence of the compatibility of between the horizontal and vertical fusions, which is illustrated in Figure\,\ref{fig:h-v-compatibility}. If we denote the vertical fusion by $\bullet$ and the horizontal fusion by $\circ$, we must have 
\begin{equation}  \label{eq:h-v-compatibility}
(\phi' \circ \psi') \bullet (\phi \circ \psi) \simeq (\phi' \bullet \phi) \circ (\psi' \bullet \psi). 
\end{equation}

\begin{figure}[bt]
\centerline{\begin{tabular}{@{}c}
\raisebox{-50pt}{
  \begin{picture}(105,120)
   \put(0,8){\scalebox{0.6}{\includegraphics{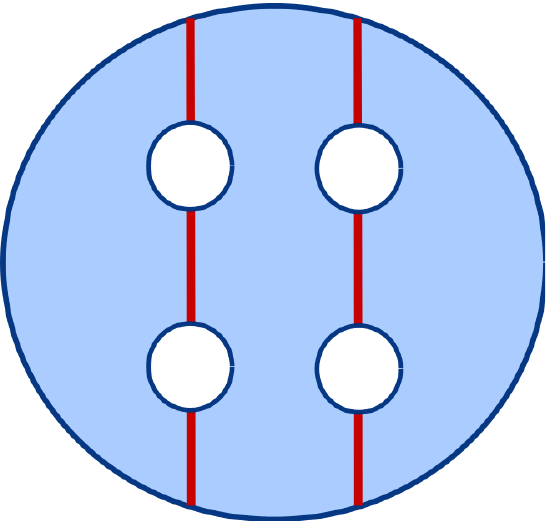}}}
   \put(0,8){
     \setlength{\unitlength}{.75pt}\put(-18,-19){
     \put(50, 32)      {\scriptsize $ g $}
     \put(47, 120)       {\scriptsize $ g'' $}
     \put(50, 75)       {\scriptsize $ g' $} 
     \put(88, 120)       {\scriptsize $ f' $}
     \put( 88, 75)       {\scriptsize $ f' $}
     \put( 90, 30)       {\scriptsize $ f $}
     \put(37, 55)        {\scriptsize $ x $}
     \put(81, 55)       {\scriptsize $ y $}
     \put(118, 55)    {\scriptsize $ z $} 
     \put(57, 99)     {\scriptsize $\phi'$}
     \put(97, 53)     {\scriptsize $\psi$}
     \put(59, 53)     {\scriptsize  $\phi$}
     \put(95, 99)   {\scriptsize $\psi'$}
     }\setlength{\unitlength}{1pt}}
  \end{picture}}
\end{tabular}}
\caption{Let $x^{[l]}$, $y^{[l]}$, $z^{[l]}$ be three $l$-codimensional defects, $g, g',g''$  domain walls between $x$ and $y$, $f,f',f''$ are walls between $y$ and $z$, and $\phi,\phi'$ and $\psi, \psi'$ are $(l+2)$-codimensional defects. First fusing horizontally then vertically must be compatible with first fusing vertically then horizontally. This compatibility leads to the equation (\ref{eq:h-v-compatibility}). 
}
\label{fig:h-v-compatibility}
\end{figure}

\medskip
In summary, an $(n+1)$-dimensional topological order with defects of all codimensions has an $(n+2)$-layered structure similar to that of an $(n+1)$-category, together with additional fusion structures, which also correspond to the composition of $i$-morphisms in an higher category, and certain (half-)braiding structures. We will show in the next subsection that these structures are automatically encoded in an $(n+1)$-category.

\subsection{Fusion and braiding structures in $n$-categories} \label{sec:fb-in-ncat}

In this subsection, we will show how the notion of $n$-categories automatically encodes all the fusion and braiding structures. This result is well-known in mathematics. We will follows Baez-Dolan in \Ref{bd} and Baez's note \Ref{baez}. A physics oriented reader should keep in mind the following dictionary: an $(n+1)$-category with only one object roughly corresponds to an $(n+1)$-dimensional topological order; 1-morphisms correspond to defects of codimension 1; and 2-morphisms correspond to defects of codimension 2, so on and so forth, $(n+1)$-morphisms are defects of codimension $(n+1)$ lying in the time direction and are also called instantons. 

For convenience, we introduce some notations. Let $x^{[l]}$ (or $x$ for simplicity) be an $l$-morphism in an $n$-category ($n>l$). We denote the full sub-$(n-l)$-category consists of one object $x$ and $\hom(x,x)$ as $\hat{x}$. The notation $\id_x^k$ means $\id_x^1:=\id_x$, $\id_x^2:=\id_{\id_x}$, so on and so forth. For $l>0$, an $l$-morphism $x^{[l]}: \id_\ast^{l-1} \to \id_\ast^{l-1}$ is called {\it a pure $l$-morphism}. A $1$-morphism is automatically pure. 

\medskip
A $0$-category is just a set. We define a $\Cb$-linear $0$-category to be a set with a $\EC$-linear structure, i.e. a vector space over $\Cb$. 

\medskip
In a $1$-category $\EC$ (see also Section\,\ref{CatT}), there are the set of objects $\ob(\EC)$ and the set of morphisms $\hom_\EC(x,x)$ for $x\in \ob(\EC)$. The composition of 1-morphisms are required to satisfy the associativity and unit properties, i.e. $(h\circ g) \circ f = h\circ (g\circ f)$, $f\circ \id_x = f$, $\id_y \circ f = f$ for $f:x\to y$, $g: y\to z$ and $h: z \to u$. Therefore, the set $\hom_\EC(x,x)$ is a monoid with the unit given by the identity 1-morphism $\id_x$ and the multiplication given by the composition of 1-morphisms. 1-categories that are relevant to physics are often $\Cb$-linear Abelian 1-categories, which requires, in particular, each hom space $\hom_\EC(x,y)$ to be a finite dimensional vector space over $\Cb$, and the direct sum (or the coproduct) $x\oplus y$ is well-defined. The direct sum is characterized by the fact that $\hom_\EC(z, x\oplus y) \simeq \hom_\EC(z, x) \oplus \hom_\EC(z, y)$ as vector spaces. In a $\Cb$-linear 1-category with one object $\ast$, $\hom(\ast,\ast)$ is a 0-category with a composition, i.e. an algebra over $\Cb$. 

Examples of $\Cb$-linear Abelian 1-categories are abundant. We list a few that are familiar to physicists. 
\begin{enumerate}
\item The category of vector spaces over $\Cb$, denoted by $\vect$ with 1-morphisms given by linear map, i.e. $\hom_{\vect}(x,x) = \End_\Cb(x)$. 
\item The category of representations of a group $G$, denoted by $\rep_G$, with 1-morphisms given by linear maps that intertwine the $G$-action. 
\end{enumerate}

In a $2$-category $\EC$ (see also Section\,\ref{CatT}), there are the set of object $\ob(\EC)$ and a set of 1-morphisms $\{ f: x \to y\}$ for $x,y\in \ob(\EC)$ and a set of 2-morphisms $\{ \phi: f\Rightarrow g\}$ for two 1-morphisms $f,g: x\to y$ (see the diagram below). 
$$
\raisebox{1em}{
\xymatrix@R=1em@C=0.2em{ 
x \ar@/^1.2pc/[rr]^f \ar@/_1.2pc/[rr]_g & \Downarrow \phi &  y 
}}
$$
The 2-morphisms can be composed. Namely, for $\phi: f \Rightarrow g$ and $\psi: g\Rightarrow h$, their composition is also called {\it vertical composition} and will be denoted by $\psi \bullet \phi$, i.e.
$$
\xy 
(-8,0)*+{x}="1"; 
(8,0)*+{y}="2"; 
{\ar "1";"2"}; 
{\ar@/^1.75pc/ "1";"2"}; 
{\ar@/_1.75pc/ "1";"2"}; 
{\ar@{=>}^{\phi} (-1,5)*{};(-1,2)*{}} ; 
{\ar@{=>}^{\psi} (-1,-2)*{};(-1,-5)*{}} ; 
\endxy \quad \xrightarrow{\bullet} \quad 
\xy 
(-8,0)*+{x}="1"; 
(8,0)*+{y}="2"; 
{\ar@/^1.25pc/ "1";"2"}; 
{\ar@/_1.25pc/ "1";"2"}; 
{\ar@{=>}^{\psi\bullet \phi} (-2,3)*{};(-2,-3)*{}} ; 
\endxy
\,\, .
$$
The associativity and unit properties hold on the nose. Namely, we have
$$
\psi \bullet (\phi \bullet \chi) = (\psi \bullet \phi) \bullet \chi, \quad \psi \bullet \id_f = \psi, \quad \id_g \bullet \phi = \phi. 
$$
Therefore, these 1-morphisms and 2-morphisms form a 1-category $\hom_\EC(x,y)$ with objects being 1-morphisms in $\EC$ and 1-morphisms being 2-morphisms in $\EC$. 

The 1-morphisms in $\EC$ can also be composed. This composition is, by definition, a functor
$$
\hom_\EC(x,y) \times \hom_\EC(y,z) \xrightarrow{\circ} \hom_\EC(x,z)
$$
which is illustrated in the following diagrams: 
\begin{equation} \label{diag:h-composition}
\raisebox{0.3em}{
\xymatrix@R=1em@C=0.2em{ 
x \ar@/^1.2pc/[rr]^f \ar@/_1.2pc/[rr]_{f'} & \Downarrow \phi &  y 
\ar@/^1.2pc/[rr]^g \ar@/_1.2pc/[rr]_{g'} & \Downarrow \psi & z
}}
\quad \mapsto \quad
\raisebox{0.3em}{
\xymatrix@R=1em@C=0.2em{ 
x \ar@/^1.2pc/[rr]^{g\circ f} \ar@/_1.2pc/[rr]_{g' \circ f'} & \Downarrow \psi \circ \phi &  z
}}
\end{equation}
which induces a new composition, called {\it horizontal composition}, among 2-morphisms, and is denoted by $\psi \circ \phi$. 

The composition $\circ$ of 1-morphisms also satisfies similar associativity and unit properties. The difference is that they do not hold on the nose, but hold up to invertible 2-morphisms. For example, for $x,y,z\in \ob(\EC)$ and $x\xrightarrow{f} y \xrightarrow{g} z \xrightarrow{h} u$, there are an associator 2-isomorphism 
$$
\alpha_{h,g,f}: h \circ (g\circ f) \Rightarrow (h\circ g) \circ f
$$ 
and unit 2-isomorphisms 
$$
l_f: \id_y \circ f \Rightarrow f, \quad 
r_f: f \circ \id_x \Rightarrow f.
$$ 
The associator and unit 2-isomorphisms satisfy the usual pentagon relations and triangle relations. Then it is easy to see that the 1-category $\hom_\EC(x,x)$ is a monoidal category with the tensor unit given by $\id_x$ and the tensor product $\otimes$ given by the composition of 1-morphisms. The unit properties of horizontal composition $\circ$ does not play an explicit role in physics. For convenience, we assume that the unit properties of horizontal composition hold on the nose, i.e.
\begin{equation} \label{eq:unit}
\id_y \circ f = f = f \circ \id_x 
\end{equation}
%Therefore, we obtain a well-known result. \begin{lemma}A $2$-category with a single object is equivalent to a monoidal 1-category. \end{lemma}

We give a few examples of 2-categories below: 
\begin{enumerate}
\item Any 1-category can always be lifted to a 2-category by adding identity 2-morphisms. 
\item Any monoidal 1-category can be viewed as a 2-category with a single object. 
\item The category of 1-categories is a 2-category with 1-morphisms given by functors and 2-morphisms by natural transformations. 
\item The category of algebras in $\vect$ with 1-morphisms given by bimodules and 2-morphisms by bimodule maps. More explicitly, for algebras $a, b, c$ in $\vect$, the bimodule ${}_am_b, {}_bn_c$ are two 1-morphisms $a \to b$ and $b\to c$, respectively. The composition of $m$ and $n$ is defined by the tensor product $m\otimes_b n$.  
\end{enumerate}

In a 2-category, the set of 2-morphisms satisfying an interesting commutativity property, which follows from an argument of Eckmann-Hilton for the commutativity of the higher homotopy groups and generalizes to higher categories\cite{bd}. Consider four 2-morphisms in the following diagram:
\begin{equation} \label{diag:braiding-eq-1}
\xy 
(-16,0)*+{x}="4"; 
(0,0)*+{x}="6"; 
{\ar "4";"6"}; 
{\ar@/^1.75pc/^{\id_x} "4";"6"}; 
{\ar@/_1.75pc/_{\id_x} "4";"6"}; 
{\ar@{=>}^<<<{\scriptstyle \phi} (-9,5)*{};(-9,2)*{}} ; 
{\ar@{=>}^<<{\scriptstyle \id_{\id_x}} (-10,-2)*{};(-10,-5)*{}} ; 
(0,0)*+{x}="4"; 
(16,0)*+{x}="6"; 
{\ar "4";"6"}; 
{\ar@/^1.75pc/^{\id_x} "4";"6"}; 
{\ar@/_1.75pc/_{\id_x} "4";"6"}; 
{\ar@{=>}^<<<{\scriptstyle \id_{\id_x}} (6,5)*{};(6,2)*{}} ; 
{\ar@{=>}^<<{\scriptstyle \psi} (8,-2)*{};(8,-5)*{}} ; 
\endxy 
\rightsquigarrow \raisebox{0.2em}{\xymatrix@R=1em@C=0.4em{
x \ar@/^1.2pc/[rr]^{\id_x} \ar@/_1.2pc/[rr]_{\id_x} & \Downarrow \psi \circ \phi & x
}}
\end{equation}
where the arrow $\rightsquigarrow$ represents two vertical composition followed by a horizontal composition. By the unit property, we can switch the order of the 2-morphisms in each row in above diagram, we obtain
\begin{equation} \label{diag:braiding-eq-2}
\xy 
(-16,0)*+{x}="4"; 
(0,0)*+{x}="6"; 
{\ar "4";"6"}; 
{\ar@/^1.75pc/^{\id_x} "4";"6"}; 
{\ar@/_1.75pc/_{\id_x} "4";"6"}; 
{\ar@{=>}^<<<{\scriptstyle \id_{\id_x}} (-10,5)*{};(-10,2)*{}} ; 
{\ar@{=>}^<<{\scriptstyle \psi} (-9,-2)*{};(-9,-5)*{}} ; 
(0,0)*+{x}="4"; 
(16,0)*+{x}="6"; 
{\ar "4";"6"}; 
{\ar@/^1.75pc/^{\id_x} "4";"6"}; 
{\ar@/_1.75pc/_{\id_x} "4";"6"}; 
{\ar@{=>}^<<<{\scriptstyle \phi} (8,5)*{};(8,2)*{}} ; 
{\ar@{=>}^<<{\scriptstyle \id_{\id_x}} (6,-2)*{};(6,-5)*{}} ; 
\endxy 
\rightsquigarrow \raisebox{0.2em}{\xymatrix@R=1em@C=0.4em{
x \ar@/^1.2pc/[rr]^{\id_x} \ar@/_1.2pc/[rr]_{\id_x} & \Downarrow \phi \circ \psi & x
}}
\end{equation}
where the arrow $\rightsquigarrow$ represents again two vertical composition followed by a horizontal composition. The axiom of 2-category requires certain compatibility between the vertical and horizontal composition. For $x\xleftarrow{g,g',g''} y \xleftarrow{f,f',f''} z$ and $g\xrightarrow{\phi} g'\xrightarrow{\phi'} g''$, it says 
\begin{equation} \label{eq:h-v-compatibility-2}
(\psi \circ \phi) \bullet (\psi' \circ \phi') = (\psi \bullet \psi' ) \circ (\phi \bullet \phi'),
\end{equation}
which is nothing but the equation (\ref{eq:h-v-compatibility}) and further leads to the following identities: 
\begin{align}
\psi \circ \phi &= (\psi \bullet 1)\circ (1 \bullet \phi) = (\psi \circ 1) \bullet (1\circ \phi)  \nn
&= \psi \bullet \phi  \nn
&= (1\circ \psi) \bullet (\phi \circ 1) = (1\bullet \phi) \circ (\psi \bullet 1)   \nn
&= \phi \circ \psi.   
\end{align}
As a consequence, the set $\hom_\EC(\id_x, \id_x)$ must be a commutative monoid. We will refer to this commutativity as 2-dimensional commutativity.

\medskip
In a $3$-category $\EC$, there are $0,1,2,3$-morphisms (see the diagram below),
%\begin{equation} \label{diag:lw-data}
$$
\xy 0;/r.22pc/:
(0,15)*{};
(0,-15)*{};
(0,8)*{}="A";
(0,-8)*{}="B";
%{\ar@{=>} "A" ; "B"};
{\ar@{=>}@/_1pc/ "A"+(-4,1) ; "B"+(-3,0)_{\phi}};
{\ar@{=}@/_1pc/ "A"+(-4,1) ; "B"+(-4,1)};
{\ar@{=>}@/^1pc/ "A"+(4,1) ; "B"+(3,0)^{\psi}};
{\ar@{=}@/^1pc/ "A"+(4,1) ; "B"+(4,1)};
%{\ar@3{->} (-6,0)*{} ; (-2,0)*+{}};
{\ar@3{->} (-4,0)*{} ; (6,0)*+{}^{i}};
(-15,0)*+{x}="1";
(15,0)*+{y}="2";
{\ar@/^2.75pc/ "1";"2"^{f}};
{\ar@/_2.75pc/ "1";"2"_{g}};
\endxy 
$$
where $0$-morphisms are objects. Note that, for $x,y\in \ob(\EC)$, $\hom_\EC(x,y)$ is a 2-category; and, for $f,g: x\rightarrow y$, $\hom_\EC(f,g)$ is a 1-category; 
and, for $\phi, \psi: f\Rightarrow g$, $\hom_\EC(\phi, \psi)$ is a 0-category.

The 3-morphisms can be composed with the associativity and unit properties that
hold on the nose. The 2-morphisms can be composed so that the associativity and
unit properties hold up to 3-isomorphisms that satisfy the pentagon and
triangle relations. The 1-morphisms can be composed and satisfy a new kind of
associativity and unit properties such that the associator and unit
isomorphisms satisfying the pentagon relations and the triangle relations only
up to 3-isomorphisms which satisfy further coherence properties\cite{KK1251, kong-icmp12, LW1384}. 

The composition of 1-morphisms (or 2-morphisms) defines a tensor product of any pair of 1-morphisms (or 2-morphisms). Therefore, in a 3-category $\EC$, the 2-category $\hom_\EC(x,x)$ is a monoidal 2-category for $x\in \ob(\EC)$, and the 1-category $\hom_\EC(f,f)$, for $f: x\rightarrow y$, is a monoidal 1-category. We denote $\psi \circ \phi$, for $\phi, \psi: f \Rightarrow f$, by $\phi\otimes \psi$. Similar to 2-category case, 2-morphisms in 3-category also satisfy some commutativity when $f=\id_x$. The difference is that, due to the existence of 3-morphisms, the 2-dimensional commutativity does not hold on the nose but only up to a 3-isomorphism. In other words, there is an 3-isomorphism 
$$
c_{\phi, \psi}:   \phi \otimes \psi \xrightarrow{\simeq} \psi \otimes \phi
$$
for each $\phi,\psi: \id_x \Rightarrow \id_x$. These 3-isomorphisms are required to satisfy certain coherence properties, which includes the usual hexagon relations for a braiding tensor category. As a consequence, the 1-category $\hom_\EC(\id_x, \id_x)$ is a braided tensor category. 

We give two examples of 3-category below: 
\begin{enumerate}
\item the category of 2-categories is a 3-category with objects given by 2-categories, 1-morphisms given by 2-functors,  2-morphisms given by 2-natural transformations and 3-morphisms given by modifications.
\item the category $\EF\mathrm{us}$ of fusion categories is a 3-category with objects given by fusion categories, 1-morphisms by bimodule categories, 2-morphisms by bimodule functors and 3-morphisms by natural transformations between bimodule functors. 

In the 3-category $\EF\mathrm{us}$, let $\EC$ be an object, or equivalently, a fusion category. Then the identity 1-morphism $\id_\EC$ is simply the trivial $\EC$-$\EC$-bimodule ${}_\EC\EC_\EC$. Then the 1-category $\hom_{\EF\mathrm{us}}(\id_\EC, \id_\EC)$ is nothing but the category $\fun_{\EC|\EC}(\EC, \EC)$ of $\EC$-$\EC$-bimodule functors from ${}_\EC\EC_\EC$ to ${}_\EC\EC_\EC$. This category, denoted by $Z(\EC) := \fun_{\EC|\EC}(\EC, \EC)$, is a braided monoidal category, which is also called the monoidal center of $\EC$. 
\end{enumerate}

For an object $x$ in a 3-category $\EC$, the 3-morphisms in $\hom_\EC(\id_{\id_x}, \id_{\id_x})$ satisfy a new kind of commutativity. More precisely, two 3-morphisms can be composed in 3 different ways, in which $\circ_3$ arise from the usual composition of 3-morphisms, $\circ_2$ from the composition of 2-morphisms and $\circ_1$ from the composition of 1-morphisms. If we illustrate these three compositions by a diagram similar to the first diagram in (\ref{diag:braiding-eq-1}), it will give a 3-dimensional diagram in which $\circ_1$ is horizontal ($x$-direction), $\circ_2$ is vertical ($y$-direction) and $\circ_3$ is in the third direction ($z$-direction). This leads to a new kind of commutativity which contains three independent 2-dimensional commutativity's in $xy$-plane, $xz$-plane and $yz$-plane. We will refer to this kind of commutativity as 3-dimensional commutativity. 

\medskip
An $n$-category $\EE$ can be viewed as a category enriched by $(n-1)$-categories. Namely, the hom spaces of an $n$-category are $(n-1)$-categories. All $i$-morphisms for $i>0$ can be composed. Therefore, $\hom_\EE(x,x)$ is a monoidal $(n-1)$-category for $x\in \ob(\EE)$; and the $(n-2)$-category $\hom_\EE(\id_x, \id_x)$ is monoidal and equipped with a 2-dimensional braiding structure which arises from the compositions of 1,2-morphisms; and the $(n-3)$-category $\hom_\EE(\id_{\id_x}, \id_{\id_x})$ is monoidal and equipped with a 3-dimensional braiding structure which arises from the compositions $1,2,3$-morphisms. 
For example, 
\begin{enumerate}
\item when $n=4$, $\hom_\EE(\id_x, \id_x)$ is a braided monoidal 2-category\cite{kapranov-voevodsky, baez-neuchl}, 
\item when $n=4$, the 3-morphisms in $\hom_\EE(\id_{\id_x}, \id_{\id_x})$ form a 1-category equipped with a 3-dimensional braiding structure. This 1-category is nothing but symmetric monoidal 1-category. 
\item when $n=5$, $\hom_\EE(\id_{\id_x}, \id_{\id_x})$ is a 2-category equipped with a 3-dimensional braiding structure. This 2-category is called a weakly involuntary monoidal 2-category\cite{bd}. 
\end{enumerate}
More examples of this type was shown in \Ref{bd} and was called $k$-tuply monoidal $n$-categories. 

\medskip
For us, we will use $(n+1)$-category to describe a $(n+1)$-dimensional topological order. The 1-morphisms correspond to defects of codimension 1; and the 2-morphisms correspond to defects of codimension 2, so on and so forth. Notice that only defects of codimension $2$ or higher can be braided. This coincides with the commutativity of 2- or higher morphisms in $\hom_\EE(\id_x, \id_x)$ in an $n$-category $\EE$. 

Note that braiding between defects of different dimensions are automatically encoded in the braiding between defects of the same dimension. For example, one can braid a particle-like excitation with a string-like excitation. In order to see the braiding, the particle should be viewed as a $0$-dimensional defect nested in a trivial string; and the string-like excitation should be viewed as a string equipped with the trivial particle. Then this braiding can be encoded by the braiding between the string and the trivial string. 

We will illustrate this structure in the case when the spatial dimension $n=3$. In this case, string-like excitations, or 2-codimensional defects nested in the trivial 1-codimensional defect, form a braided monoidal 2-category $\ED$. The braiding in $\ED$ is a pseudo-natural isomorphism $R: \otimes \to \otimes^\op$ between two tensor product functors, where $\otimes^\op$ is defined by $a \otimes^\op b:= b\otimes a$ for $a,b\in \ob(\ED)$. The defining data of such a pseudo-natural isomorphism includes\cite{kapranov-voevodsky,baez-neuchl}:
\begin{enumerate}
\item equivalences $R_{a,b}: a\otimes b \xrightarrow{\simeq} b\otimes a$ for $a,b\in \ob(\ED)$ 
\item invertible 2-morphisms $R_{f,g}$: 
\begin{equation} \label{diag:R-ab}
\raisebox{1em}{
\xymatrix@R=1.2em@C=1.2em{ 
&  a' \otimes b' \ar[dr]^{R_{a',b'}} & \\
a \otimes b \ar[ur]^{f\otimes g} \ar[dr]_{R_{a,b}}  & \Downarrow R_{f,g} & b' \otimes a' \\
& b \otimes a \ar[ur]_{g\otimes f} &  
}}
\end{equation}
for each pair of 1-morphisms $(a\xrightarrow{f} a', b\xrightarrow{g} b')$. 
\end{enumerate}
$R_{f,g}$ is the braiding between the particle-like excitation $f$, which connects the string $a$ and $a'$, and
the particle-like excitation $g$, which connects the string $b$ and $b'$. The braiding discussed in the previous paragraph is a special case of the pair $(R_{a,b}, R_{f,g})$ when $a=1_\ED=a'$, $b=b'$ and $g=\id_b$. 

Similar to $R_{f,g}$, in an $(n+1)$-topological order, the braiding of an $i$-codimensional defect $x$ and a $j$-dimensional defect $y$ with arbitrary decoration of higher codimensional defects for $i\geq j$ is encoded in the braiding between the trivial $j$-codimensional defect and $y$. If no decoration of higher codimensional defect is specified, it automatically means that the default decorations by the trivial higher dimensional defects up to codimension $n$ (excluding the instantons) are chosen. Therefore, all the information of braiding is encoded in the braiding among defects of the same codimension.

\begin{rema} \label{rema:n-braiding}
Notice that the data in the braiding $R_{x,y}: x\otimes y \to y\otimes x$ boils down to a family of $k$-isomorphisms for $j<k\leq n+1$. Since, for $j<k<i$, the $k$-morphisms are trivial due to our convention (\ref{eq:unit}), the braiding data further boils down to a family of $k$-isomorphisms for $i\leq k\leq n+1$. If we assume that Yoneda lemma holds for any properly defined $m$-categories for all $m\in \Zb_+$, then we can further reduce the data of a $k$-isomorphisms to a  set of data of $(n+1)$-isomorphisms, which is an uncontrollable large set in general. But if we assume certain finiteness to our $(n+1)$-category as we will do in the next subsection, the braiding data can be reduced to a finite number of $(n+1)$-isomorphisms, which are directly measurable in physics. More precisely, these finite number of $(n+1)$-isomorphisms are invertible linear maps between two spaces of lowest energy states associated to two configurations of excitations including at least the $i$-codimensional defect $x$ and $j$-codimensional defect $y$, and these two configurations can be deformed from one to the other by moving $x$ around $y$ without crossing other defects. 
\end{rema}

\medskip
Recall the half-braiding discussed in Section\,\ref{sec:defect-n-toporder}. This structure is also automatically encoded in the structure of higher category. To see this, consider the following diagrams. The composition is obtained by first composed horizontally then vertically, 
\begin{equation} \label{diag:half-braiding-1}
\xy 
(-16,0)*+{x}="4"; 
(0,0)*+{x}="6"; 
{\ar "4";"6"}; 
{\ar@/^1.75pc/^{\id_x} "4";"6"}; 
{\ar@/_1.75pc/_{\id_x} "4";"6"}; 
{\ar@{=>}^<<<{\scriptstyle \phi} (-9,5)*{};(-9,2)*{}} ; 
{\ar@{=>}^<<{\scriptstyle \id_{\id_x}} (-10,-2)*{};(-10,-5)*{}} ; 
(0,0)*+{x}="4"; 
(16,0)*+{y}="6"; 
{\ar "4";"6"}; 
{\ar@/^1.75pc/^{f} "4";"6"}; 
{\ar@/_1.75pc/_{g} "4";"6"}; 
{\ar@{=>}^<<<{\scriptstyle \id_f} (6,5)*{};(6,2)*{}} ; 
{\ar@{=>}^<<{\scriptstyle \psi} (8,-2)*{};(8,-5)*{}} ; 
\endxy 
\rightsquigarrow \raisebox{0.2em}{\xymatrix@R=1em@C=0.4em{
x \ar@/^1.2pc/[rr]^{\id_x} \ar@/_1.2pc/[rr]_{\id_x} & \Downarrow \psi \bullet \phi & x
}}
\end{equation}
where $\phi$ in the second diagram is actually a abbreviation for $\id_f \circ \phi$. 
Due to the unit property of the identity $(l+2)$-morphisms composed vertically, we can exchange the $(l+2)$-morphisms in the first row in the first diagram in (\ref{diag:half-braiding-1}) with the second row. We have he following diagrams and composition. 
\begin{equation} \label{diag:half-braiding-2}
\xy 
(-16,0)*+{x}="4"; 
(0,0)*+{x}="6"; 
{\ar "4";"6"}; 
{\ar@/^1.75pc/^{\id_x} "4";"6"}; 
{\ar@/_1.75pc/_{\id_x} "4";"6"}; 
{\ar@{=>}^<<<{\scriptstyle \id_{\id_x}} (-10,5)*{};(-10,2)*{}} ; 
{\ar@{=>}^<<{\scriptstyle \phi} (-9,-2)*{};(-9,-5)*{}} ; 
(0,0)*+{x}="4"; 
(16,0)*+{y}="6"; 
{\ar "4";"6"}; 
{\ar@/^1.75pc/^{f} "4";"6"}; 
{\ar@/_1.75pc/_{g} "4";"6"}; 
{\ar@{=>}^<<<{\scriptstyle \psi} (8,5)*{};(8,2)*{}} ; 
{\ar@{=>}^<<{\scriptstyle \id_g} (6,-2)*{};(6,-5)*{}} ; 
\endxy 
\rightsquigarrow \raisebox{0.2em}{\xymatrix@R=1em@C=0.4em{
x \ar@/^1.2pc/[rr]^{\id_x} \ar@/_1.2pc/[rr]_{\id_x} & \Downarrow \phi \bullet \psi & x
}}
\end{equation}
where $\phi$ in the second diagram is an abbreviation for $\id_g \circ \phi$. 
The axioms of the $(n+1)$-category requires that there is a half-braiding 
\begin{equation}  \label{eq:half-braiding}
c_{\phi, \psi}: \phi \bullet \psi \xrightarrow{\simeq} \psi \bullet \phi.
\end{equation}
More precisely, $c_{\phi, \psi}$ is a natural transformation between the two functors 
$\hom(\id_x, \id_x) \boxtimes \hom(f,g) \to \hom(f,g)$ described by (\ref{diag:braiding-eq-1}) and (\ref{diag:braiding-eq-2}). 

Another way to visualize the half braiding is to look at the following diagram: 
\begin{equation} \label{diag:id-f-g-phi-2}
\raisebox{3.7em}{
\xymatrix@R=1.8em@C=0.3em{ 
&  & \hom(f,f) \ar[d]^{\psi_\ast} & & \\
\hom(\id_x,\id_x)  \ar[urr]^{L_f} \ar[drr]_{L_g}  & \Downarrow \beta_L & \hom(f,g) & \Downarrow \beta_R & 
\hom(\id_y,\id_y) \ar[dll]^{R_g} \ar[ull]_{R_f} \\
& & \hom(g,g) \ar[u]^{\psi^\ast} & &  
}}
\end{equation}
where $L_{f/g}$ and $R_{f/g}$ are functors defined by the composition (also called left/right bulk-to-wall maps) and
$\psi_\ast = \psi \circ -$ and $\psi^\ast = - \circ \psi$. Above diagram is commutative up to two isomorphisms
\begin{align}  
\beta_L: \psi_\ast \circ L_f  & \xrightarrow{\simeq} \psi^\ast \circ L_g  \label{eq:left-right-half-braiding-ncat-1} \\
\beta_R: \psi_\ast \circ R_g & \xrightarrow{\simeq} \psi^\ast \circ R_f  \label{eq:left-right-half-braiding-ncat-2} 
\end{align}
which are nothing but the half-braidings. 

Notice that $c_{\psi, \phi}$ can not be defined in general. That is why it is called the half-braiding. Another important difference with a full braiding, comparing it to (\ref{diag:braiding-eq-1}) and (\ref{diag:braiding-eq-2}), is that only the vertical composition $\bullet$ (not the horizontal one!) is related to the half-braiding. This fact corresponds exactly to the half-braiding between defects nested in an $l'$-dimensional defect and those in a $(l-1)$-dimensional wall which is only well-defined the $l'$-th direction that is normal to the wall. 

When $f$ and $g$ are invertible (corresponding to transparent domain walls), the bulk-to-wall maps $L_f, L_g, R_f, R_g$ are all invertible, then we can fully braid morphisms in $\hom(\id_{x/y}, \id_{x/y})$ with those in $\hom(f,g)$. More explicitly, we obtain the following two double braidings: 
\begin{align}
\psi_\ast \circ L_f \circ L_g^{-1} \circ R_g  & \xrightarrow{ \beta_L^{-1}\beta_R} \psi_\ast  \circ R_f, 
\label{eq:double-braiding-ncat-1} \\
\psi^\ast \circ  L_g \circ L_f^{-1} \circ R_f   & \xrightarrow{ \beta_L\beta_R^{-1}} \psi^\ast  \circ R_g. 
\label{eq:double-braiding-ncat-2}
\end{align}
An example of such double braiding and transparent domain wall is given in (\ref{eq:double-braiding-toric}) in the toric code model. 

\medskip
More general half braidings can be obtained by replacing the ``=" in 
equation (\ref{eq:h-v-compatibility-2}) by an higher isomorphism $\simeq$. This is nothing but the categorical description of the compatibility of two different orders of fusions in different direction (recall equation (\ref{eq:h-v-compatibility}) and Figure\,\ref{fig:h-v-compatibility}).

\medskip
In the remaining part of this subsection, we will discuss $k$-dimensional bulk-to-wall maps in an $(n+1)$-category. We assume that certain additive structure to an $n$-category and a symmetric tensor product $\boxtimes$ and the direct sum $\oplus$ for $n$-category is well-defined. When $n=1$, the symmetric tensor product $\boxtimes$ is the Deligne tensor product. 

\smallskip
Recall the 1-dimensional bulk-to-wall map (\ref{eq:bulk-to-wall})(\ref{eq:bulk-to-wall-2}) introduced in Section\,\ref{sec:defect-n-toporder}. Since the defect $x^{[l]}$ and $y^{[l]}$ correspond to two $l$-morphisms and $f^{[l+1]}$ is an $(l+1)$-morphism $f: x \to y$. The two-side 1-dimensional bulk-to-wall map in an $(n+1)$-category is the following $(n-l-1)$-functor (i.e. a functor between two $(n-l-1)$-categories): 
\begin{equation} \label{eq:LR-ncat}
\hom(\id_x,\id_x) \boxtimes \hom(\id_y,\id_y) \xrightarrow{L_f\boxtimes R_f} \hom(f,f)
\end{equation}
defined by the horizontal composition induced from the composition of $(l+1)$-morphisms as shown in the following diagrams:
\begin{equation} \label{diag:bulk-to-wall-ncat}
\raisebox{0.3em}{
\xymatrix@R=1em@C=0.2em{ 
x \ar@/^1.2pc/[rr]^{\id_x} \ar@/_1.2pc/[rr]_{\id_x} & \Downarrow \phi &  x }}
 \boxtimes  
\raisebox{0.3em}{
\xymatrix@R=1em@C=0.2em{ 
y \ar@/^1.2pc/[rr]^{\id_y} \ar@/_1.2pc/[rr]_{\id_y}  & \Downarrow \psi & y
}}
 \mapsto 
\raisebox{0.3em}{
\xymatrix@R=1em@C=0.2em{ 
x \ar@/^1.25pc/[rr]^{f} \ar@/_1.25pc/[rr]_{f} 
& \Downarrow \mbox{{\scriptsize $\psi \circ \id_f \circ \phi$}} &  y\, .
}}
\end{equation}
And the functor $L_f$/$R_f$ is called left/right 1-dimensional bulk-to-wall map. 
More generally, we will also refer to the following map 
$$
\hom(\id_x,\id_x) \boxtimes \hom(\id_y,\id_y) \xrightarrow{L_\psi \boxtimes R_\psi} \hom(f,g),
$$
where $L_\psi=\psi_\ast \circ L_f$ ($R_\psi=\psi_\ast \circ R_f$) is the left (right) 1-dimensional bulk-to-wall map.

The $k$-dimensional bulk-to-wall map introduced in (\ref{eq:k-bulk-to-wall}) is also automatically encoded in the horizontal composition induced from the composition of $(l+1)$-morphisms in an $(n+1)$-category. For example, let $f, g: x\to y$ be two $(l+1)$-morphisms and $z: f \Rightarrow g$ an $(l+2)$-morphism, a two-side $2$-dimensional bulk-to-wall map from two $x\boxtimes y$ to $z$ is an $(n-l-2)$-functor
\begin{equation} \label{eq:2-bulk-to-map-ncat}
\hom(\id_{\id_x}, \id_{\id_x}) \boxtimes \hom(\id_{\id_y}, \id_{\id_y})
\xrightarrow{} \hom(z, z),
\end{equation}
the definition of which is naturally included in the definition of composition of $(l+1)$-morphisms. More generally, if $z$ is a $(l+k)$-morphism, the left/right $k$-dimensional bulk-to-wall maps (\ref{eq:2-bulk-to-map-ncat}) are two $(n-l-k)$-functors:
\begin{equation} \label{eq:k-bulk-to-map-ncat-0}
\hom(\id_x^k, \id_x^k) \xrightarrow{L} \hom(z, z) \xleftarrow{R} \hom(\id_y^k, \id_y^k)
\end{equation}
or equivalently, a two-side bulk-to-wall map:
\begin{equation} \label{eq:k-bulk-to-map-ncat}
\hom(\id_x^k, \id_x^k) \boxtimes \hom(\id_y^k, \id_y^k)
\xrightarrow{L\boxtimes R} \hom(z, z),
\end{equation}
where the notation $\id_x^k$ means $\id_x^1:=\id_x$, $\id_x^2:=\id_{\id_x}$, so on and so forth. More generally, if $z'$ is another $(l+k)$-morphism, for a fixed morphism in $\hom(z,z')$, there is a two-side $k$-dimensional bulk-to-wall map from the right hand side of (\ref{eq:k-bulk-to-map-ncat}) to $\hom(z,z')$.

\medskip
In summary, the rich inner structures of higher category exactly catch the information of the fusion and braiding of defects in a topological order.

\subsection{Unitary $(n+1)$-categories and $\BF_{n+1}^{pre}$-categories}

In this section, we introduce the notion of unitary $(n+1)$-category and that of a $\BF_{n+1}^{pre}$-category. 

\medskip
A unitary $0$-category is a finite dimensional Hilbert space. 

\medskip
Now we introduce the definition of a unitary $1$-category. We recommend M\"{u}ger's review \Ref{mueger4}.
\begin{defn} \label{def:ast-cat}
A $\ast$-$1$-category $\EC$ is a $\Cb$-linear category equipped with a functor $\ast: \EC \to \EC^\op$ 
which acts trivially on the objects and is antilinear and involutive on morphisms, i.e. 
$\ast:\hom_\EC(x,y) \to \hom_\EC(y,x)$ is defined so that 
\begin{equation} \label{eq:dagger}
(g \circ f)^\ast = f^\ast \circ g^\ast, \quad\quad (\lambda f)^\ast = \bar{\lambda} f^\ast,\quad\quad f^{\ast\ast} = f. 
\end{equation}
for $f: u \to v$, $g: v \to w$, $h: x \to y$, $\lambda \in \Cb^\times$.
\end{defn}

\begin{defn} \label{def:finite-cat}
A $1$-category $\EC$ is called {\it finite} if $\EC$ is Abelian $\Cb$-linear and semisimple so that
there are only finite number of simple objects and hom spaces are all finite dimensional vector spaces over $\Cb$. 
\end{defn}

\begin{defn}
A $\ast$-$1$-category is called {\it unitary} if it is finite and $\ast$ satisfies the positivity condition: $f\circ f^\ast =0$ implies $f=0$. 
\end{defn}

\begin{defn}
A finite 1-category is a $\Cb$-linear Abelian semisimple category with only finitely many inequivalent simple objects and finite dimensional hom spaces. A fusion category is a finite monoidal category such that $\hom(\one,\one)\simeq \Cb$, where $\one$ is the tensor unit. 
\end{defn}

A unitary fusion 1-category has  a lot of nice properties. In particular, it is automatically spherical,  and automatically a $C^\ast$-category\cite{mueger3}, and the left dual of an object automatically coincides with the right dual. 

\medskip
In a $2$-category $\EC$, the direct sum $x\oplus y$ (or the coproduct) of two objects $x$ and $y$ is characterized by the property that $\hom(z, x\oplus y) \simeq \hom(z, x) \oplus \hom(z, y)$ 
as $1$-categories for all $z$. 
\begin{defn} \label{def:finite-cat}
A $2$-category $\EC$ is called {\it finite} if there are only finite number of simple objects; every object is a direct sum of simple objects; and all hom spaces are finite $1$-categories. 
\end{defn} 

\begin{rema}
Although we need the additive structure in a finite $2$-category, we don't require it to be Abelian. In particular, we expect to have non-trivial morphisms between two simple objects. 
\end{rema}

\begin{defn}
A $2$-category $\EC$ said to have adjoints for 1-morphisms if for every 1-morphism $f: X \to Y$, there exists another 1-morphism $g: Y \to X$, the unit 2-morphisms $\eta: \id_X \Rightarrow g\circ f$ and $\tilde{\eta}: \id_Y \Rightarrow f\circ g$ and the co-unit 2-morphisms $\epsilon: f\circ g \Rightarrow \id_Y$ and 
$\tilde{\epsilon}: g\circ f \Rightarrow \id_X$ such that they form an ambidextrous adjunction, i.e. left and right adjoints exist and coincide. 
\end{defn}

\begin{rema} \label{rema:adjoint}
The above notion is stronger than the usual notion of a $2$-category $\EC$ having adjoints for 1-morphisms, which only requires the existence of the right adjoint and the left adjoint \cite{lurie}. The 2-morphisms $\eta, \tilde{\eta}, \epsilon, \tilde{\epsilon}$ are called duality maps. 
\end{rema}

\begin{defn}  \label{def:2cat-unitary}
A unitary $2$-category is a finite $2$-category having adjoints for 1-morphisms such that all hom spaces are unitary $1$-categories and all coherence isomorphisms are unitary, i.e. 
$$
\alpha_{f,g,h}^\ast = \alpha_{f,g,h}^{-1}, \quad l_f^\ast = l_f^{-1}, \quad r_f^\ast = r_f^{-1}
$$ 
$$
\tilde{\epsilon} = \eta^\ast, \quad \tilde{\eta}= \epsilon^\ast, 
$$
for $1$-morphisms $f,g,h$, where $\alpha$, $l$, $r$ and $\eta, \tilde{\eta}, \epsilon, \tilde{\epsilon}$ are the associator, the left unit isomorphism, the right unit isomorphism and duality maps, respectively. 
\end{defn}

It is well-known that a $2$-category with a single object is a monoidal category (see Section\,\ref{sec:fb-in-ncat}). It is very easy to obtain the following result. 
\begin{lemma} \label{lemma:2-cat+ufc}
For a unitary $2$-category with a single object, the 1-category of 1-morphisms is a unitary fusion 1-category, in which the tensor product is given by the composition of 1-morphisms and the tensor unit is the identity $1$-morphism. 
\end{lemma}

%%%%%%%%%%%%%%%%%%%%%
We would like to define a unitary $n$-category recursively. We assume that a good definition of $n$-category is chosen with certain additive structure such that the direct sums (or coproducts) are well-defined. In an $n$-category, the direct sum $x\oplus y$ (or the coproduct) of two objects $x$ and $y$ (if defined) is characterized by the property that 
$\hom(z, x\oplus y) \simeq \hom(z, x) \oplus \hom(z, y)$ 
as $(n-1)$-categories for all $z$. An $i$-morphism that is not a direct sum of two non-zero $i$-morphisms is called {\it simple}. In the $n$-categories that we are interested in, finite coproducts (or direct sum) of $i$-morphisms always exist. 

\begin{defn}
An $n$-category $\EC$ is called {\it finite} if 
\begin{enumerate}
\item the homotopy 1-category $\mathrm{h}\EC$ obtained from $\EC$ by taking the same set of objects and defining 1-morphisms as the equivalence classes of $1$-morphisms in $\EC$ is (not Abelian in general) semisimple with only finite many simple objects, and the identity morphism
$\id_f$ is simple for all simple morphisms $f$; 
\item for any pair of object $x,y\in \EC$, the $(n-1)$-category $\hom_\EC(x,y)$ is finite. 
\end{enumerate}
\end{defn}

\begin{rema}
In a finite $(n+1)$-category, the fusion and braiding data can all be encoded efficiently by considering only the simple $l$-morphisms for $1\leq l \leq n+1$, which can be further reduced to a finite set of $(n+1)$-isomorphisms (recall Remark\,\ref{rema:n-braiding}), which should be directly measurable in a physical realization of this category. 
\end{rema}

We will follow the footsteps of Lurie in \Ref{lurie}. Let $\EC$ be an $n$-category for $n\geq 2$. Its homotopy $2$-category $\mathrm{h}_2\EC$ is a 2-category with the same objects and 1-morphisms as those in $\EC$ and 2-morphisms being the isomorphism classes of those in $\EC$. 
$\EC$ has adjoints for $1$-morphisms if $\mathrm{h}_2\EC$ has
adjoints for $1$-morphisms. For $1<l<n$, $\EC$ has adjoints for $k$-morphisms if, for any pair of objects $x, y\in \EC$, the $(n-1)$-category $\hom_\EC(x,y)$ has adjoints for all $(l-1)$-morphisms. $\EC$ is said to have adjoints if $\EC$ has adjoints for $l$-morphisms for all $0<l<n$. 

\begin{rema}
If the $n$-category $\EC$ has a monoidal structure, $\EC$ has duals for the objects if the homotopy 1-category $\mathrm{h}\EC$ is a rigid monoidal category such that two side duals coincide. Then $\EC$ is said to have duals if $\EC$ has duals for objects and adjoints for $l$-morphisms for all $0<l<n$. 
\end{rema}

We propose an incomplete definition of a unitary $n$-category recursively as follows.
\begin{defn} \label{def:unitary-n-cat}
A {\it unitary $n$-category} $\EC$ is a finite $n$-category such that it has adjoints and, for any pair $x,y\in \EC$, $\hom_\EC(x,y)$ is a unitary $(n-1)$-category and all coherence isomorphisms are unitary. If a unitary $n$-category $\EC$ has a monoidal structure with simple tensor unit and duals, it is called a {\it unitary fusion $n$-category}. 
\end{defn}

\begin{rema}
It is unclear to us how to generalize the definition the unitarity for all the coherence isomorphisms in Definition\,\ref{def:2cat-unitary} to higher categories.  This drawback does not play any role in our later discussion. 
\end{rema}

\begin{rema}
We have assumed that the tensor unit is simple in the unitary fusion $n$-category. If this condition is not satisfied, we will call such $n$-category as unitary multi-fusion $n$-category. This notion is also very useful. Multi-fusion $n$-categories can be naturally produced in the process of dimensional reduction or by the general tensor product $\boxtimes_{\EE_{n+1}}$ discussed in Remark\,\ref{rema:general-boxtimes}. For example, consider a toric code model defined on a strip of lattice such that boundaries on both left and right side of the strip are smooth (see Fig.\,\ref{toric}). When viewed from far away, this strip narrows down to a line, a closed $1+1$D topological order. Excitations on this line are given by the multi-fusion 1-category $\fun_{\hilb}(\rep_{\Zb_2},\rep_{\Zb_2})$, the center of which is again trivial\cite{eno}. More generally, for any $l$-morphism $x$ in a unitary $n$-category, the category $\hom(x\oplus x, x\oplus x)$ is a multi-fusion $(n-l-1)$-category. The topological phases associated to multi-fusion $n$-categories are not stable. For example, the stable condition (\ref{eq:gs-deg}) is violated for such multi-fusion $n$-categories. For this reason, we choose not to discuss them further in this work. We will do that in \Ref{kong-wen-zheng}. Also notice that fusion $n$-categories are closed under the stacking operation: $\boxtimes$. 
\end{rema}

\begin{defn}
A unitary $n$-functor $f$ between two unitary $n$-categories is an $n$-functor preserving the adjoints and the additive structures (e.g. direct sums). 
\end{defn}

%\begin{defn}A {\it pre-$\BF_{n+1}$-category} is a unitary $(n+1)$-category with a single object. \end{defn}

\smallskip
Similar to Lemma\,\ref{lemma:2-cat+ufc}, we have the following result. 
\begin{lemma}
The fully subcategory of a unitary $(n+1)$-category consisting of a single object is automatically a unitary fusion $n$-category. These two notions are equivalent. %If a unitary fusion $n$-category is denoted by $\EC$, then we denote its corresponding unitary $(n+1)$-category with a single object by $B\EC$. 
\end{lemma}

%\begin{rema} \label{rema:pre-bf-cat}
A unitary $(n+1)$-category with one object can describe a $(n+1)$-dimensional topological order with a given (not necessarily complete) set of defects. The 1-morphisms are 1-codimensional defects, 2-morphisms are 2-codimensional defects that are connecting two (not necessarily distinct) 1-codimensional defects, so  on and so forth. $(n+1)$-isomorphisms are instantons. For this reason, we will introduce a new terminology.

\begin{defn}
A {\it pre-$\BF_{n+1}$-category} or a {\it $\prebf_{n+1}$-category} is an $(n+1)$-category with a single object $\ast$ such that $\hom(\ast, \ast)$ is a unitary fusion $n$-category. 
\end{defn}

\begin{rema}
A $\prebf_{n+1}$-category is not a unitary $(n+1)$-category because we do not assume an additive structure on objects. An alternative definition is to define a $\prebf_{n+1}$-category as a unitary $(n+1)$-category with one simple object. So all finite coproducts (direct sums) are included. There is no essential difference between these two approaches. But our approach makes some later discussion easier. 
\end{rema}

\begin{expl}
For each $n>0$, there is a trivial and the smallest $\prebf_{n+1}$-category $\one_{n+1}$, which is defined as the smallest $\prebf_{n+1}$-category that contains $\{ \ast, \id_\ast, \id_{\id_\ast}, \cdots, \id_\ast^n \}$ and $\hom(\id_\ast^n, \id_\ast^n)=\Cb$. More explicitly, we have 
\bnu
\item for $n=0$, $\one_{0+1}$ is the 1-category with a single object $\ast$ and $\hom(\ast, \ast) = \Cb$; 
\item for $n=1$, $\one_{1+1}$ is the 2-category with a single object and $\hom(\ast, \ast) \simeq \hilb$; 
\item for $n=2$, $\one_{1+1}$ is a 3-category with a single object $\ast$ and a unique simple 1-morphism $\id_\ast$ and $\hom(\id_\ast, \id_\ast) \simeq \hilb$. An non-simple 1-morphism is a finite direct sum of $\id_\ast$. 
\enu
\end{expl}

\begin{expl} 
Some non-trivial examples of $\prebf_{n+1}$-categories in low dimensions: 
\bnu
\item A $\prebf_{0+1}$-category is just a simple (because we require $\id_\ast$ to be ``simple") $C^\ast$-algebra over $\Cb$, i.e. a matrix algebra over $\Cb$. 
%satisfying positive condition, i.e. $aa^\ast = 0$ if and only if $a=0$. 
\item A $\prebf_{1+1}$-category is a 2-category with one object $\ast$ such that $\hom(\ast, \ast)$ is a unitary fusion 1-category. 
\item A unitary braided fusion 1-category $\ED$ determines a $\prebf_{2+1}$-category, which consists of one object $\ast$ and only one simple 1-morphism $\id_\ast$ such that $\hom(\id_\ast, \id_\ast) =\ED$. 
\item A unitary fusion category $\EC$ determines a $\prebf_{2+1}$-category, which consists of one object $\EC$, 1-morphisms given by $\EC$-bimodules, 2-morphisms given by $\EC$-$\EC$-bimodule functors and 3-morphisms by natural transformations between bimodule functors, i.e. the full subcategory of $\EF\mathrm{us}$ consisting of a single object $\EC$ (recall Remark\,\ref{rema:fus-3-cat}). 
\enu
\end{expl}

Due to the connection to topological order, an $l$-morphism $f$ in a $\prebf_{n+1}$-category will also be called an $l$-codimensional defect. Higher morphisms in $\hom(f,f)$ are higher codimensional defects nested on $f$. The only object will be denoted by $\ast$ unless we specify it otherwise. 

\medskip
In order to study situations that involve multi topological phases, such as a topological phase with a gapped boundary, we would like to introduce the notion of an $\MBF_{n+1}^{pre}$-category, where the letter ``M'' stands for ``multi topological phases". 
\begin{defn}
An $\MBF_{n+1}^{pre}$-category is an $(n+1)$-category with only finitely many objects such that the hom space $\hom(x,y)$ for a pair of objects $(x,y)$ is a unitary $n$-category  and a unitary fusion $n$-category if $x=y$. 
\end{defn}

\begin{rema}
Notice that in an $\MBF_{n+1}^{pre}$-category, a 1-morphism in $\hom(x,y)$ might not have an adjoint if $x\neq y$ (not required by definition). But all other 1-morphisms and higher morphisms have adjoints. Also note that there is no additive structure on objects. It is a good thing. It allows us to describe real physical situations, in which there are finite many topological phases connected by gapped domain walls. But the additive structure of gapped domain walls (as objects in $\hom(x,y)$) is unavoidable because it can be generated automatically from the additive structure of $\hom(x,x)$ by the composition of morphisms. 
\end{rema}

\begin{expl} 
Examples of $\prembf_{2+1}$-category are full subcategories of the 3-category $\EF\mathrm{us}$ (see Remark\,\ref{rema:fus-3-cat}) consisting of only finite number of objects. 
\end{expl}

\begin{rema}  \label{rema:infty-cat}
We have ignored the coherence properties. Subtleness may arise in the study of the universal perturbative gravitational responses in one dimension higher (recall Section\,\ref{sec:univ-grav-response}). For example, it is possible that there is no non-trivial topological excitation in an $n$ space-time dimensional topological order (e.g. the $E_8$ quantum Hall system), corresponding to a trivial $\prebf_n$-category, but it has non-trivial gravitational responses, which can be captured by $(n+1)$-dimensional mathematical structures (see discussion in Section\,\ref{sec:univ-grav-response} and \ref{invTop}). In other words, an $n$-category is not adequate to describe an $n$-dimensional topological order with non-trivial gravitational responses. We believe that the universal perturbative gravitational responses can have non-trivial effects on the coherence property. Namely, the coherence properties of $n$-morphisms might not hold on the nose due to the gravitational anomalies. Then we have to introduce higher invertible morphisms. For this reason, perhaps, a more complete framework to describe topological order is to use $(\infty, n)$-category\cite{lurie} instead of $n$-category. In an $(\infty, n)$-category, $l$-morphisms exist for all $l\in \Zb_+$, and all $l$-morphisms are invertible if $l>n$. This is beyond the scope of this paper. We leave it to the future. 
\end{rema}

\subsection{Conceptual gaps between a $\BF^{pre}$-category and a topological order}

The reason that a pre-$\BF$-category is not yet a $\BF$-category is because a $\prebf_{n+1}$-category describes an $(n+1)$-dimensional topological order equipped with a chosen (not necessarily complete) set of defects. We can choose the set to contain only the trivial defects or choose one containing the large amount of defects, which are obtained from condensation and closed by fusion. Then it creates an ambiguity because different choices of the set of defects will give different $\prebf_{n+1}$-categories, which are all associated to the same topological order. We give an explicit example below. 

\begin{expl}  \label{expl:toric-code}
We can associate the following two different $\prebf_{2+1}$-categories to the same toric code model: 
\begin{enumerate}
\item the modular tensor category $Z(\rep_{\Zb_2})$, which is viewed as a unitary $3$-category with only one object $\rep_{\Zb_2}$, only one 1-morphism $\rep_{\Zb_2}$ (viewed as the trivial $\rep_{\Zb_2}$-bimodule), 2-morphisms given by $\rep_{\Zb_2}$-bimodule functors and 3-morphisms given by natural transformations between bimodule functors. Notice that the 1-category of 2-morphisms is nothing but the modular tensor category
$Z(\rep_{\Zb_2})=\fun_{\rep_{\Zb_2}|\rep_{\Zb_2}}(\rep_{\Zb_2}, \rep_{\Zb_2})$.

\item the full subcategory of $\EF\mathrm{us}$ with the only object $\rep_{\Zb_2}$, denoted by $\EF\mathrm{us}|_{\rep_{\Zb_2}}$. Its 1-morphisms includes not only the trivial $\rep_{\Zb_2}$-bimodule $\rep_{\Zb_2}$ but also all other bimodules such as ${}_{\rep_{\Zb_2}}\hilb_{\rep_{\Zb_2}}$. This is the toric code model enriched by gapped domain walls and defects of higher codimensions\cite{KK1251}.
\end{enumerate}
Notice that $Z(\rep_{\Zb_2})$ viewed as a $3$-category is a proper subcategory of the second $3$-category $\EF\mathrm{us}|_{\rep_{\Zb_2}}$. Moreover, when $Z(\rep_{\Zb_2})$ is viewed as a $2$-category with one object, it is the looping of the second $3$-category $\EF\mathrm{us}|_{\rep_{\Zb_2}}$, often denoted by $\Omega(\EF\mathrm{us}|_{\rep_{\Zb_2}})$. 
\end{expl}

This is a general phenomenon in Levin-Wen type of lattice models. More precisely, we can replace $\rep_{\Zb_2}$ by any unitary fusion 1-category $\EC$. The category $\EC$ determines a lattice model and all its extensions by defects. One can associate to this Levin-Wen model either by its bulk excitations which are given by the unitary modular tensor category $Z(\EC)$, or by the full-subcategory of $\EF\mathrm{us}$ consisting of only one object $\EC$. These two associated categories are different as $\prebf_3$-categories.

\medskip
The point of above examples is that one can associate different $\prebf_{n+1}$-categories to the same topological order. More generally, one can also paste a finite many quantum Hall systems to a $(2+1)$-dimensional defect in a higher dimensional topological phase without changing the phase. This process creates further ambiguities. We want to find a way to characterize topological orders in a unique way that is also minimal (or efficient) and complete. 

Ideally, we would like to find a minimal generating set of all excitations. In order to do that, we would like find precise mathematical descriptions of those {\it condensed} excitations that can be obtained from lower dimensional excitations by condensations, and those called {\it elementary excitations} and the mixture of these two types called {\it quasi-elementary excitations}. Our hope is to use elementary excitations only to characterize the topological phase. %It is reasonable to believe that such characterization is minimal and complete. 

\subsection{Condensed/elementary excitations and the definition of a $\BF$-category}

In this subsection, we will discuss how to characterize condensed/elementary topological excitations in physical language. We will also propose a definition of a $\BF_{n+1}$-category. But we can not say that we have achieved our goal. We will discuss a few problems naturally associated to this definition. 

\medskip
%\begin{rema}  \label{rema:closed-factor}
In general, a simple $l$-codimensional excitation $x^{[l]}$ in a topological phase can be very complicated. For example, it can contain a quantum Hall system for $l=2$, or a closed $(l+1)$-dimensional topological order in general. More precisely, $x$ can be factorized as follows:
$$
x^{[l]} = x_0^{[l]} \boxtimes x_1^{[l]} \boxtimes \cdots \boxtimes x_k^{[l]}, 
$$
where $x_1, \cdots, x_k$ are all simple closed topological orders and $\boxtimes$ is the stacking operation. We will refer to $x_1, \cdots, x_k$ as closed factors of $x$. The fusion product between two excitations both with non-trivial closed factors is shown in the following equation: 
$$
(x_0^{[l]} \boxtimes x_1^{[l]}) \otimes (y_0^{[l]} \boxtimes y_1^{[l]}) = (x_0^{[l]} \otimes y_0^{[l]})\boxtimes x_1^{[l]} \boxtimes y_1^{[l]}. 
$$
If $x_0^{[l]}, x_1^{[l]}$ and $x_2^{[l]}$ are simple, so is $x_0^{[l]} \boxtimes x_1^{[l]} \boxtimes x_2^{[l]}$. Since the closed factors $x_1^{[l]}$ and $y_1^{[l]}$ are long-range entangled, they can not cancel each other via the stacking operation. For this reason, this closed factor are of infinite type. Namely, by applying fusion product repeatedly, we obtain infinite number of such factors. If we assume that there are only finitely many simple $l$-codimensional defects, then an $l$-codimensional defect can not contain any closed factor. We will assume this finiteness from now on so that we can ignore the closed factors completely for $\BF_{n+1}$-category.

\medskip
If an excitation $x$ can be obtain from a condensation of other excitations, then we would like to say that $x$ is condensed. But there are a lot of ambiguities in this terminology. For example, it is possible that a set $S_1$ of excitations can be obtained from another set $S_2$ of excitations via condensations. At the same time, $S_2$ can be obtained from $S_1$ via condensations. Let's consider a more explicit situation. If a defect $x$ can be obtained from a condensation involving an $p$-dimensional defect $y$, it seems reasonable that $x$ is at least $(p+1)$-dimensional. Our intuition is that this condensation involves a large number of $y$, one can fine tune the system before the condensation so that all of these $y$-defects are arranged sitting next to (but separated from) each other. They fill a subspace of dimension at least $(p+1)$. When you turn on the local interaction (or condensation), this subspace simply turns into a new excitation of dimension at least $(p+1)$. But what makes the situation much more complicated is that this new excitation can be a trivial one. This means that a pure excitation might be obtained from a condensation of defects that are domain walls between higher dimensional defects. On the other hand, at least some of the later defects can also be generated from pure excitations via condensation and fusion. So it does not make any sense to say that an excitation is condensed in general. We must specify where it is condensed from. In other words, the generating sets for all topological excitations are not unique in general. We have to make some choices. A natural choice is to include pure excitations in the generating set and hope that they can generate other defects. 

\begin{defn} \label{defn:quasi-elem-physics}
An excitation (or defect) $x$ is called {\it condensed} if $x$ can be obtained from pure excitations of lower dimensions via nontrivial condensations. An excitation that is not condensed is called {\it quasi-elementary}. 
\end{defn}

\begin{rema}
In above definition, a condensed excitation can be a pure excitation or a domain wall between two other defects, which can be domain walls between higher dimensional defects. 
\end{rema}

A pure particle-like excitation is automatically quasi-elementary. The trivial defect $1^{[l]}$ is quasi-elementary by definition because it can be reproduced only by the trivial condensation. In the 3+1D $\Zb_2$ gauge theory (recall Example\,\ref{C4Z2}), a particle and a vortex line are both quasi-elementary. 

The direct sum of two condensed/quasi-elementary excitations is still condensed/quasi-elementary. The fusion product of two condensed excitations is also condensed. In general, the fusion product of a condensed excitation and a non-trivial quasi-elementary excitation is quasi-elementary. The fusion product of two quasi-elementary excitations can be either condensed or quasi-elementary.  

\begin{rema}
A more complicated examples of quasi-elementary excitations can be obtained by considering the defect junctions. More precisely, in (\ref{eq:3-way}), if one of the three defect lines is condensed and one is quasi-elementary, then the defect junction is quasi-elementary. These generic defects can be obtained from the set of all elementary excitations via condensations followed by  fusions, thus can be ignored in a minimal description of a topological order. 
\end{rema}

Our main goal is to find a characterization of a subset of quasi-elementary excitations called elementary excitations. In general, a simple $l$-codimensional quasi-elementary excitation $x^{[l]}$ can be factorized as follows: 
\begin{equation}  \label{eq:elem-factors}
x^{[l]} = x_0^{[l]}\otimes c^{[l]},
\end{equation}
where $\otimes$ is the fusion product, $c$ is a simple condensed excitation, which will be called a {\it condensed factor} of $x$. It is tempting to say that $x_0$ in (\ref{eq:elem-factors}) does not contain any condensed factor and in some sense ``elementary". Unfortunately, such factorization is not unique in general. It is possible that $e_0 =e_0' \otimes c_0'$ and $c_0' \otimes c_0=1$, where $1^{[l]}$ is the trivial $l$-codimensional excitation. Therefore, it makes no sense to say that an excitation is free of condensed factors. We must find a better way to characterize an elementary excitation. 

\medskip
According to the discussion in Section\,\ref{unbele}, a condensed and finite $l$-codimensional pure excitation should have gapped domain walls with the trivial $l$-codimensional excitation. If the condensation produces a gapless domain wall that can not be gapped, similar to quantum Hall systems, we believe that the condensed excitation, which can be viewed as an anomalous topological order, is long range entangled and can not be canceled by fusion products. Therefore, it is reasonable to believe that such a condensed excitation can not be finite and should be ignored according to the finiteness of a $\prebf$-category. We will also refer to such condensation as of infinite type. 

A condensation is of finite type if the domain wall between the condensed and uncondensed phases can be gapped. The connection by gapped domain wall defines an equivalence relation, called Witt equivalence, on the set of simple $l$-codimensional excitations. For an $l$-codimensional defect $x$, we denote its equivalence class by $[x]$, which is called {\it Witt class} of $x$. Different equivalence classes are disconnected. It is clear that different elementary excitations must be disconnected.

\smallskip
This picture of Witt class is very important. Although it does not tell us how to select the elementary excitation from each Witt class. It immediately implies that in a real $\BF_{n+1}$ category, containing only elementary excitations, simple objects are all disconnected. Therefore, it suggests us to give a definition of $\BF$-category simply as follows: 
\begin{defn}  \label{def:bf-category}
A $\BF_{n+1}$-category is a $\prebf_{n+1}$-category such that $\hom(x,y)=0$ for any pair of  non-equivalent ($x\ncong y$) simple $l$-morphisms $x^{[l]}$ and $y^{[l]}$. 
\end{defn}

\begin{rema}
A $\BF_{n+1}$-category is still not Abelian because $\hom(x,x)$ is nontrivial for a simple $l$-morphism $x$ in general. Each simple $l$-morphism in a $\BF_{n+1}$-category corresponds to an elementary excitation. 
\end{rema}

\subsection{A $\BF_{n+1}$-category as the core of a $\prebf_{n+1}$-category}

There are a lot of natural questions associated to Definition\,\ref{def:bf-category}. For example, given a $\prebf_{n+1}$-category, can we determine which $\BF_{n+1}$-category is associated to it? Is it unique? It is not hard to see that the uniqueness will be a serious problem. For example, the dotted domain wall in Figure\,\ref{toric} can have gapped domain wall with the trivial domain wall as shown as the blue dot in Figure\,\ref{toric}. But we can also chose not to include the blue dot in our lattice model and create a $\prebf_3$-category, in which the morphism associated to the dotted wall in Figure\,\ref{toric} is not connected (by gapped walls) to that associated to the trivial domain wall. This choice is very arbitrary. So it seems that it is impossible to associated a $\BF$-category uniquely to a $\prebf$-category unless the $\prebf$-category satisfies certain maximal condition, i.e. all possible gapped domain walls and walls between walls are included in the $\prebf$-category. Even if the $\prebf$-category we start with satisfies the maximal condition, there is still a problem of selecting which one in a Witt class to be the elementary excitation.

\medskip
Given an $l$-codimensional excitation $y$, condensations of higher codimensional subdefects in $y$ will produce many different $l$-codimensional excitations, which can be further condensed. 
%Normally, you can not go backwards. Namely, if $z$ is obtained from $y$ via a non-trivial condensation, one cannot obtain $y$ from $z$ via condensation. 
The condensation provides the excitations in a Witt class of $y$ a partial order $\leq$. As shown in Section\,6.2 in \Ref{kong-anyon}, in 2+1D, two anyon systems, which are given by two modular tensor categories and connected by a gapped domain wall, can be obtained from another anyon system via two condensations. We believe that this phenomenon generalizes to higher dimensions. 
More precisely, we believe that if two $l$-codimensional excitations can be connected by gapped domain wall, these two excitations together with the gapped domain wall can be obtained from condensations of the higher codimensional sub-defects in a possibly third excitations. This picture suggests the following conjecture. 
\begin{conj}  \label{conj:root}
The excitations in a Witt class with the particle order $\leq$ form a lattice, which is a mathematical notion and means a partially ordered set such that every two elements have a least upper bound and a greatest lower bound. Moreover, 
in each equivalence class, there is a unique (up to isomorphisms) minimum of the whole class, called the {\it root} of the class. In other words, all other defects in the class and the associated gapped domain walls can be obtained from the root via condensations. The trivial $l$-codimensional defect $1^{[l]}$ is the root of the class $[1^{[l]}]$. 
\end{conj}

The existence of least upper bound is a little hard to see. Its evidence comes from the theory of anyon condensation in 2+1D, in which such lattice structures have already appeared in the Witt classes of non-degenerate braided fusion categories (or modular tensor categories)\cite{dmno,kong-anyon}. We will call such lattice a rooted lattice. Notice that a rooted lattice is not a rooted tree because the former contains loops in general. 

\begin{defn}
A simple $l$-codimensional excitation $x$ called {\it elementary} if $x$ is the root of its class $[x]$. 
\end{defn}

\begin{rema}
Using above picture, condensed $l$-codimensional excitations are those excitations in the Witt class $[1^{[l]}]$ except the root $1^{[l]}$. All the rest are quasi-elementary. 
\end{rema}

A pure particle-like excitation is automatically elementary. The trivial defect $1^{[l]}$ is elementary by definition. In the 3+1D $\Zb_2$ gauge theory (recall Example\,\ref{C4Z2}), a particle and a vortex line are both elementary. If $x$ is (quasi-)elementary, then its anti-particle $\bar{x}$ is also (quasi-)elementary. The direct sum of elementary excitations are also called elementary. 

\begin{conj}
The fusion product of two elementary excitations is also elementary. 
\end{conj}

As a consequence, the set of all elementary excitations is closed under fusion product.

\begin{rema}  \label{rema:leaves}
By our assumption of the finiteness, the condensation can not go forever. This gives another interesting set of excitations in a given Witt class. They are sitting at the other ends of the rooted lattice, thus will be called {\it leaves} of the class. They can be obtained from the root excitation via a maximal condensation such that no further condensation is possible (see Section\,6.1 in \Ref{kong-anyon} for closed 2+1D cases). 
\end{rema}

The next important question is how to determine the root of each Witt class. We believe that if the complete set of observable data is given for a topological order, in principle, one should be able to determine the root. In order to achieve it, we need a yet-to-be-developed theory of condensation for higher dimensional topological order. 

\begin{conj}
%In a closed topological order, it is possible to talk about the maximal $\prebf$-category associated to it such that it includes all possible defects. 
A full-fledged condensation theory allow us to identify the elementary excitation as the root of each lattice associated to each Witt class. 
\end{conj}

A $\prebf$-category might not be maximal, but it already contains all the elementary defects. 
We propose the following conjecture. 
\begin{conj}
For a given $\prebf$-category $\EC$, a full-fledged condensation theory allows us to find all new defects obtained from condensations (of finite type, i.e. with gapped domain walls) of excitations in $\EC$. Then we can extend $\EC$ by adding all these new defects so that we obtain a maximal $\prebf$-category $\EC^{max}$ which contains $\EC$ as a sub-$\prebf$-category. 
\end{conj}

\begin{rema}
For practical purpose, it is not necessary to work out the maximal $\prebf$-category before we identify the root of a Witt class. But it is important that the entire Witt class can be recovered in principle.
\end{rema}

Assuming this, it provides us a new way to define a $\BF_{n+1}$-category from a $\prebf_{n+1}$-category. 
\begin{defn}
For a given $\prebf_{n+1}$ category $\EC$, the core of a $\EC$ is the smallest sub-$\prebf_{n+1}$-category containing all elementary morphisms, each of which is defined by the root of a Witt class of morphisms in $\EC^{max}$. We denote the core of $\EC$ by $\core(\EC)$. 
\end{defn}

\begin{defn}
A $\BF_{n+1}$-category is a $\prebf_{n+1}$-category that is equivalent to the core of another $\prebf_{n+1}$-category.  
\end{defn}

\begin{rema}  \label{rema:closure}
Let $S$ be a set of morphisms in an $(n+1)$-category $\ED$. We will call the smallest sub-$(n+1)$-category that contains the set $S$ as the closure of $S$. In this language, $\core(\EC)$ is the closure of the set of elementary morphisms in $\EC$. The notion of closure only make sense in a given $\EC$. But when the set $S$ itself is already a special kind $(n+1)$-category, it is sometimes clear what it means by a smallest unitary sub-$(n+1)$-category containing it without referring to a large category in which it lives. For example, in the third example in Example\,\ref{expl:closed-prebf}, the extension only involves adding the additive structure to $S$. From now on, we will refer to such extension as a $\prebf$-closure of $S$. 
\end{rema}

\subsection{Closed/anomalous $\BF_{n+1}$-categories} \label{sec:closed-bf-def}

In this subsection, we want to give a mathematical definition of a closed/anomalous $\BF_{n+1}$-category.

%Conversely, in a $x^{[0]}$ be an arbitrary topological order. If a domain wall $z^{[1]}$ between $x^{[0]}$ and $x^{[0]}$ can be factorized as $z=\text{Im}(f) \boxtimes z_1$, where $f$ is the bulk-to-wall map, by the philosophy discussed in Section\,\ref{sec:detect-excitation} (need modified??), the excitations in $z_1$ is not detectable by excitations in $x^{[0]}\boxtimes x^{[0]}$. Since the excitations in $z_1$ are mutually symmetric with those in $\text{Im}(f)$, excitations in $z_1$ must be able to detect themselves via braiding. In order words, $z_1$ is a closed topological phase.  In seeking a minimal description of topological order, the condensed part $\text{Im}(f)$ can be ignored, but not $z_1$. Therefore, it is important to give a mathematical characterization of a closed $\BF_{n+1}^{pre}$-category. 

%\begin{rema}The proof of the identity $z=\text{Im}(f) \boxtimes z_1$ should follow a higher categorical generalization of the proof of Proposition\,7.7 in \Ref{mueger2}.\end{rema}

\medskip
We have shown in Section \ref{sec:fb-in-ncat} that if we view a $\prebf_{n+1}$-category as a topological phase with defects of all codimensions, $i$-morphisms corresponding to $i$-codimensional defect, the information of the braiding between any two defects $R_{x,y}: x\otimes y \to y\otimes x$ is automatically included in the defining structure of an $(n+1)$-category. Without lose of generality, we can assume $x$ and $y$ are of the same dimension. 
What motivates our definition of a closed $\prebf_{n+1}$-category is the conjecture that all excitations in a closed topological order can be detected by braiding with other defects. Therefore, if an excitation double braids with any other excitations does not give any physically detectable difference, this excitation must be the trivial excitation. 

\begin{defn}
In a $\prebf_{n+1}$-category with the object $\ast$, two pure morphisms $x^{[l]}$ and $y^{[l]}$  are said to be {\it mutually symmetric} if one of the following conditions is satisfied:
\begin{enumerate}
\item at least one of the braidings $x\otimes y \xrightarrow{\cong} y\otimes x$ and $y\otimes x \xrightarrow{\cong} x\otimes y$, or simply the double braiding, is not defined as a sub-structure of a $\prebf_{n+1}$-category, 
\item both of the braidings $R_{x,y}: x\otimes y \to y\otimes x$ and $R_{y,x}: y\otimes x \to x\otimes y$ are well-defined and double braiding $R_{y,x} \circ R_{x,y}$ is trivial in all higher homotopies. By ``trivial in all higher homotopies", we mean that $R_{y,x} \circ R_{x,y}=\id_{x\otimes y}$ if $i=n$, or for $1<i< n$, there are $r_{i+2}^{[i+2]}$ such that 
\begin{equation} \label{eq:xy-symmetric}
R_{y,x} \circ R_{x,y} \overset{{\scriptsize r_{i+2}^{[i+2]}}}{\simeq} \id_{x\otimes y},  
\end{equation} 
and $r_{i+2}'$ such that $r_{i+2} \circ r_{i+2}' \overset{r_{i+3}}{\simeq} \id_{\id_{x\otimes y}}$, so on and so forth, until $r_{n+1}\circ r_{n+1}'=\id_{x\otimes y}^{n-i+1}$.

Moreover, for all $x\xrightarrow{f} x$ and $y\xrightarrow{g} y$, we require that $R_{g,f} \circ R_{f,g}$, as shown in the following diagram, 
$$
\raisebox{1em}{
\xymatrix@R=1.2em@C=1.2em{ 
&  x \otimes y \ar[dr]^{R_{x,y}} & & \\
x \otimes y \ar[ur]^{f\otimes g} \ar[dr]_{R_{x,y}}  & \Downarrow R_{f,g} & y \otimes x \ar[rd]^{R_{y,x}} & \\
& y \otimes x \ar[ur]^{g\otimes f} \ar[dr]_{R_{y,x}} &  \Downarrow R_{g,f} & x\otimes y   \\
& & x\otimes y \ar[ur]_{f\otimes g}  &
}}
$$
is trivial in all higher homotopies. And the same is true for all higher morphisms in $\hom(x,x)$ and $\hom(y,y)$. 
\end{enumerate}
The two $(l+1)$-morphisms $f$ and $g$ above are said to be mutually symmetric if both $R_{y,x}\circ R_{x,y}$ and $R_{g,f} \circ R_{f,g}$ are trivial to higher homotopies and the same is true for all higher morphisms in $\hom(f,f)$ and $\hom(g,g)$. 
\end{defn}

\begin{rema}
The intuition for being trivial in all higher homotopies is that there should not be any detectable difference between $R_{y,x} \circ R_{x,y}$ and $\id_{x\otimes y}$ even with decoration by higher codimensional defects . %We will refer to such an $(i+2)$-codimensional defect $r_{i+2}$ as an {\it entirely trivial $(i+2)$-codimensional defect}. 
\end{rema}

More generally, in a $\prebf_{n+1}$-category with the object $\ast$, the braiding between two (not necessarily pure) morphisms $f^{[l+1]}: x_1^{[l]} \to x_2^{[l]}$ and $g^{[l+1]}: y_1^{[l]} \to y_2^{[l]}$ can also be automatically defined by the axioms of an $(n+1)$-category. We want to define when they can be said to be mutually symmetric. For simplicity, we assume that $x_1,x_2$ and $y_1,y_2$ are pure. More general cases are similar. We omit the details. 
\begin{defn}
$f$ and $g$ are said to be {\it mutually symmetric} if one of the following conditions is satisfied:
\begin{enumerate}
\item the double braiding is not defined. 
\item All three double braidings: $R_{y_1,x_1} \circ R_{x_1,y_1}$, $R_{y_2,x_2} \circ R_{x_2,y_2}$ and $R_{g,f} \circ R_{f,g}$ in the following diagram: 
$$
\raisebox{1em}{
\xymatrix@R=1.2em@C=1.2em{ 
&  x_2 \otimes y_2 \ar[dr]^{R_{x_2,y_2}} & & \\
x_1 \otimes y_1 \ar[ur]^{f\otimes g} \ar[dr]_{R_{x_1,y_1}}  & \Downarrow R_{f,g} & y_2 \otimes x_2 \ar[rd]^{R_{y_2,x_2}} & \\
& y_1 \otimes x_1 \ar[ur]^{g\otimes f} \ar[dr]_{R_{y_1,x_1}} &  \Downarrow R_{g,f} & x_2\otimes y_2   \\
& & x_1\otimes y_1 \ar[ur]_{f\otimes g}  &
}}
$$
are trivial in all higher homotopies. The same is true for all higher morphisms in $\hom(f,f)$ and $\hom(g,g)$. 
\end{enumerate}
\end{defn}

In a $\prebf_{n+1}$-category or a $\BF_{n+1}$-category, the pure identity $n$-morphism $\id_\ast^n$ is automatically mutually symmetric to all other morphisms. The set of all excitations that are mutually symmetric to all excitations including themselves forms a sub-$\prebf$-category which is called the braiding-center of $\EC$, denoted by $Z(\EC)$. Notice that for any $\prebf$-category, its braiding-center must contain the trivial $\prebf$-category $\one_{n+1}$, which is the smallest $\BF_{n+1}$-category that contains $\{ \ast, \id_\ast, \id_{\id_\ast}, \cdots, \id_\ast^n \}$.

\begin{rema}
The notation for the braiding-center $Z(\EC)$ is slightly different from the bulk (or center) $\cZ(\EC)$ of $\EC$ because the bulk $\cZ(\EC)$ is a $\prebf_{n+2}$-category. But this two notions are related. We will study them in \Ref{kong-wen-zheng}. 
\end{rema}

\begin{defn}  \label{def:closed-prebf}
A closed $\BF_{n+1}$-category $\EC_{n+1}$ is a $\BF_{n+1}$-category such that $Z(\EC_{n+1}) \simeq \one_{n+1}$. A $\BF_{n+1}$-category that is not closed is called {\it anomalous}. 
 \end{defn}

The simplest closed $\BF_{n+1}$-category is just the trivial one $\one_{n+1}$.
The following result follows immediately from the definition. 
\begin{lemma} \label{lemma:1-morphism-condensed}
In a closed $\BF_{n+1}$-category, the only simple 1-morphism is $\id_\ast$.
\end{lemma}

%\begin{rema}  \label{rema:redefine-closed-BF}
%Since a closed $\BF_{n+1}$-category is completely determined by $\hom(\id_\ast, \id_\ast)$, which is a unitary braided fusion $(n-1)$-category. For this reason, we can remove the additive structure of 1-morphisms by taking only one 1-morphism $\id_\ast$ and removing all finite direct sum of $\id_\ast$ for a closed $\BF_{n+1}$-category. This removal  causes very little confusion. So we will assume that from now on. \end{rema}

\begin{expl} \label{expl:closed-prebf}
We give a few families of examples of closed $\BF_{n+1}$-categories. 
\bnu

\item All $\BF_{0+1}$-categories are closed. %, i.e. a 1-category with one object $\ast$ and $\hom(\ast, \ast)=\Cb$. We will denoted this $1$-category by $\one_{1}$. 

\item We know that any unitary $2$-category with one object is equivalent to a unitary fusion 1-category. Since simple 1-morphisms in a closed $\BF_{2}$ mutually symmetric to all 1-morphisms because there is no braiding, the only closed $\BF_{2}$-category must be the trivial one $\one_2$, or equivalently, the unitary 1-category $\hilb$. Any unitary $\BF_2$-category with at least one non-trivial simple 1-morphisms is anomalous.

\item An non-degenerate braided fusion category (including all modular tensor categories) $\EC$ can be viewed as $3$-category with a single object $\ast$ and a single 1-morphism $\id_\ast$ and with $\hom(\id_\ast, \id_\ast)=\EC$. 
Its $\prebf$-closure determines a closed $\prebf_3$-category by adding additive structures to the 1-morphisms and associated 2,3-morphisms. 
This family of examples includes all quantum Hall systems. A braided fusion 1-category that is not non-degenerate gives an anomalous $\prebf_{3}$-category. 

%\item The $\BF_{3+1}$-category given by the full subcategory of $\EF\mathrm{us}$ with only one object is closed. Indeed, since braiding that need the particle to crossing a non-trivial defect line is not allowed, the only possible braidings are those between the particle-like defects nested in the trivial defect of codimension 1. These particle-like defects form a modular tensor category. Then the closeness follows immediately from the non-degeneracy of a modular tensor category. 

\item The $3+1$-dimensional $\Zb_2$ gauge theory gives a closed $\prebf_4$-category. In this case, all the particle-like excitations are mutually symmetric to each other, but they can be distinguished by braiding them with the vortex line. 

\enu
\end{expl}

\begin{lemma}
If $\EC$ and $\ED$ are closed $\prebf_{n+1}$-categories, then $\EC\boxtimes \ED$ is also closed. If one of them is anomalous, then $\EC\boxtimes \ED$ is anomalous.
\end{lemma}

The following conjecture is one of the main goal of this work. We believe that it is true up to additional anomalies such as the spins and universal perturbative gravitational responses. 
\begin{conj}
There is a one-to-one correspondence between the equivalence classes of closed/anomalous $\BF_{n+1}$-categories and the set of closed/anomalous topological orders. 
\end{conj}

\subsection{Closed $\prebf_{n+1}$-categories}

Sometimes, it is also convenient to have a notion of a closed $\prebf_{n+1}$-category because they can describe many $\prebf_{n+1}$-categories constructed from lattice models. 

\medskip
Since a $\prebf_{n+1}$-category is not a minimal way of describing topological order, we expect to have non-trivial domain walls and walls between walls in a closed $\prebf_{n+1}$-category. But these walls and walls between walls must be condensed. By Conjecture\,\ref{conj:root}, all simple 1-morphisms in a closed $\prebf_{n+1}$-category should be connected to the trivial domain wall $\id_\ast$ by gapped domain walls. 
But in a generic $\prebf_{n+1}$-category, gapped walls between walls are randomly included. In other words, a condensed domain wall might be superficially disconnected to $\id_\ast$. Therefore, we would like to find an alternative description of a condensed domain wall.

A necessary condition for a domain wall to be condensed $x$ is that the two-side 1-dimensional bulk-to-wall maps 
$$
\hom(\id_\ast, \id_\ast) \boxtimes \hom(\id_\ast, \id_\ast) \xrightarrow{L \boxtimes R} \hom(x,x) 
$$
are dominant. For a generic $l$-codimensional defect $x$, the dominance of the two-side $l$-dimensional bulk-to-wall map to $x$ is not a sufficient condition for $x$ to be condensed. For example, the vortex line in 2+1D $\Zb_2$ gauge theory satisfies this dominance condition, but it is an elementary excitation
and is detectable by double braiding with particles. 

Situation is quite different for a sub-defect in a 1-codimensional domain wall. It can only be half-braided (but not double-braid) with pure excitations in the bulk. This half-braiding is not enough to detect such a sub-defect by pure excitations. However, motivated by Levin's work \Ref{L1309}, we propose the following conjecture. 
\begin{conj}  \label{conj:dominance}
In a domain wall, if an sub-excitation lies in the image of the two-side bulk-to-wall map, it is detectable by bulk pure excitations, otherwise, it is not detectable by bulk pure excitations. 
\end{conj}

As a consequence, the dominance of bulk-to-wall maps to a domain wall $x$ is sufficient for $x$ to be a condensed domain wall. Otherwise, the topological order can not be closed. Similar results should also hold for walls between walls. Therefore, we propose the following definition of a closed $\prebf_{n+1}$-category. 
\begin{defn}
A $\prebf_{n+1}$-category $\EC$ is called closed if the only simple morphism in $\hat{\id_\ast} \cap Z(\EC)$ is $\id_\ast^n$, and all domain walls and walls between walls satisfy the dominance condition. 
\end{defn}

\begin{expl}
Any $\prebf_3$-category arises from Levin-Wen type of lattice models enriched by defects of all codimensions\cite{KK1251} %(see Remark\,\ref{rema:SET-toric}) 
is closed. More precisely, such a $\prebf_3$-category is 
the full subcategory of $\EF\mathrm{us}$ consisting of a single unitary fusion category $\EC$. The 1-morphisms are $\EC$-bimodules. They are condensed from bulk excitations\cite{KK1251,kong-anyon}. This follows from the mathematical fact that the monoidal functor $L\boxtimes R: Z(\EC)^{\boxtimes 2} \to \fun_{\EC|\EC}(\EM, \EM)$ is dominant for any indecomposable semisimple $\EC$-bimodule $\EM$.\cite{dmno} The gapped domain walls (2-morphisms) between the 1-morphisms $\EM$ and  $\EC$ is given by the category $\fun_{\EC|\EC}(\EC, \EM)$ (see the blue dot in Fig.\,\ref{toric} in the toric code model for an example of wall of this kind).
\end{expl}

According to Lemma\,\ref{lemma:1-morphism-condensed}, the only simple 1-morphism in a core of a closed $\prebf_{n+1}$-category is $\id_\ast$. For $n=2$, the core of a closed $\prebf_3$-category is the $\prebf$-closure of $\ast, \id_\ast$ and $\hom(\id_x, \id_x)$, which is a non-degenerate unitary braided fusion 1-category. In other words, the core of an closed $\prebf_3$-category is equivalent to an non-degenerate unitary braided fusion 1-category. 

%%%%%%%%%%%%%%%%%%%%%%%%%%%%%%%%%%%%%%

\begin{defn}
Two closed $\prebf_{n+1}$-categories $\EC$ and $\ED$ are called {\it core equivalent} if
$\core(\EC) \simeq \core(\ED)$ as $\prebf_{n+1}$-categories. 
\end{defn}

\begin{conj}
The set of core equivalence classes of closed $\prebf_{n+1}$-categories are one-to-one corresponding to the set of closed $\BF_{n+1}$-categories. % thus also to the set of closed topological orders. 
\end{conj}

\subsection{Closed $\prembf_{n+1}$-category}

Using the same idea in the last subsection, we can define a closed $\prembf_{n+1}$ as follows: 
\begin{defn}  \label{def:closed-prembf}
An $\prembf_{n+1}$-category $\EC$ is closed if 
\bnu
\item $\hat{x}$ is a closed $\prebf_{n+1}$-category for all $x\in \ob(\EC)$, 
\item all two-side bulk-to-wall maps to walls and walls between walls lying outside of $\hat{x}$ for all $x\in \ob(\EC)$ are dominant.
\enu
\end{defn}

\begin{rema}
The intuition is that a closed $\prembf_{n+1}$-category should contain only condensed domain walls. According to Conjecture\,\ref{conj:dominance}, we believe that the dominance of the two-side bulk-to-wall maps guarantees that all domain walls and walls between walls in a closed $\prembf_{n+1}$-category are condensed. If the left/right bulk-to-wall maps to a domain wall $f^{[1]}\in \hom(x,y)$ are dominant and also fully faithful, then two maps are equivalences, and $f$ is a transparent domain wall between $x$ and $y$. All of these results are not necessarily true if the $\prembf$-category $\EC$ is not closed. 
\end{rema}

\begin{defn} \label{defn:eq-closed-prembf}
Two closed $\prembf_{n+1}$-categories $\EC$ and $\ED$ are equivalent if there is an $(n+1)$-functor which gives bijection on objects and unitary $n$-equivalences on $\hom(x,y) \to \hom(F(x), F(y))$ for all $x,y\in \ob(\EC)$.   
\end{defn}

Note that a closed $\prembf_{n+1}$-category is not necessarily connected. For example, it is known that two quantum Hall systems can not be connected by gapped domain walls if these two systems belong to different Witt classes\cite{dmno, fsv, kong-anyon}. More generally, it is not possible to connect a closed $(n+1)$-dimensional topological phase to an anomalous $(n+1)$-dimensional phase by a gapped domain wall because they must share the same $(n+2)$-dimensional bulk.

\begin{expl} We give a few families of closed $\prembf_{n+1}$-categories below. 
\bnu
\item  Any $\prembf_3$-category arising from Levin-Wen models enriched by domain walls and defects are closed. More precisely, any full subcategories of the 3-category $\EF\mathrm{us}$ (see Remark\,\ref{rema:fus-3-cat}) consisting of only finite number of objects are closed $\prembf_3$-categories.
In particular, the 3-category $\TC_3^b$ from the toric code model with gapped boundaries is an example of closed $\prembf_{3}$-category.

\item A few quantum Hall systems with a gapped domain walls form a closed $\MBF_{n+1}^{pre}$-category. If all quantum Hall systems selected are Witt equivalent\cite{dmno,fsv,kong-anyon}, then such a closed $\prembf_{n+1}$-category is also connected. 

\enu
\end{expl}

\begin{defn}  \label{def:gapped-boundary}
For a closed $\prebf_{n+1}$-category $\EC$, 
a {\it gapped boundary of $\EC$} is a closed
$\prembf_{n+1}$-category $\EC_f^b$ consisting of only two objects $x$ and $e$, only one 1-morphism $f: x\to e$ and no 1-morphism from $e$ to $x$ except the zero morphism, such that
\bnu
\item $\hat{x}=\EC$ and $\hat{e}=\one_{n+1}$ as $\prebf_{n+1}$-categories,
\item $\hat{f}$ is a $\prebf_n$-category and $\hat{f} \simeq \mathrm{Im}(L) \boxtimes \ED$, where $\ED$ is a closed $\prebf_n$-category and $L$ is the bulk-to-wall map $L: \hom(\id_\ast, \id_\ast) \to \hom(f,f)$, which is a dominant central monoidal $(n-1)$-functor. ``Central" means that there are half braidings 
(recall \eqref{eq:half-braiding})
$$c_{L(a),b}: L(a) \bullet b \xrightarrow{\simeq} b\bullet L(a)$$ 
for $a\in \hom(\id_\ast, \id_\ast)$ and $b\in \hom(f,f)$.  
\enu
\end{defn}

\begin{rema}
Above definition is motivated from the anyon condensation theory in 2+1D\cite{kong-anyon}. 
The dominance condition in Definition\,\ref{def:gapped-boundary} is automatic by the definition of an $\prembf_{n+1}$-category. It gives a good control of higher morphisms lying in the domain wall $f$. Note that the category $\hom(\id_\ast, \id_\ast)$ is a braided monoidal $(n-1)$-category. That $L$ is a central monoidal $(n-1)$-functor implies that the functor $L$ factors through the monoidal center of the monoidal $(n-1)$-category $\hom(f,f)$. 
$$
\xymatrix{
\hom(\id_\ast, \id_\ast) \ar@{.>}[r]^{\exists ! \, \tilde{L}}  \ar[dr]_L  & \cZ(\hom(f,f)) \ar[d]^{forget} \\
& \hom(f,f)\, .
}
$$
%The centralness is also automatic by the half-braiding structure in an higher category. 
By the dominance and the centralness, we conjecture that the monoidal $(n-1)$-category $\hom(f,f)$ can be recovered from $\hom(\id_\ast, \id_\ast)$
as the category of modules over the commutative algebra $L^\vee(\id_f)$ in $\hom(\id_\ast, \id_\ast)$, where $L^\vee$ is the two-side adjoint functor of $L$. One can see that this is nothing but a higher categorical analog of the condensation theory of 2+1D developed in \Ref{kong-anyon}. We will give more details in the future. 
\end{rema}

\begin{defn}
If a closed $\BF_{n+1}$-category $\EC$ has a gapped boundary, it is called {\it exact}. 
\end{defn}

\section{Tensor network (TN) approach to topological phases in any dimensions}
\label{TNappr}

We have discussed \hBF{}  category and topological order in any dimensions via
their universal properties (\ie their topological invariants).  
%Since only simple BF categories correspond to stable topological orders, in
%the rest of this paper we will only discuss simple BF categories.  For
%simplicity, we will drop the word ``simple'' and use BF category to refer to
%simple  BF category in the rest of this paper.
In this section, we are going to discuss how to realize those  \hBF{}
categories in physical systems defined by path integrals.  

However, path integrals described by tensor network naturally
describe a lbL system, and \lBF{}  categories.  This is because the path
integrals constructed using tensors can be defined on space-time of any
topologies. 
%Hence those path integrals describe \lBF{}  categories.  
The  path integrals that describe \hBF{}  categories are only required to be
well defined on space-time which is a mapping torus.  So it is more
natural to realize  \lBF{} categories through tensor network and path
integrals.  

Therefore, we will first concentrate on the  tensor network and path integral
construction of \lBF{}  categories.  We will try to write down the most general
form of path integrals using tensors, hoping our construction can produce all
the  \lBF{} categories.  This way, we can obtain an alternative definition and
a classification of \lBF{}  categories in any dimensions.  We have seen that is
quite difficult to define \lBF{}  category via their topological invariants
(such as the fusion and the braiding properties).  The tensor network and
path-integral way to define and to classify \lBF{} categories may be more
practical.  
%Here, we are going to use the tenser network representation of path integral.
Since  \lBF{}  categories form a subset of \hBF{}  category, we can at least
understand a subset of \hBF{}  categories this way.  

\subsection{TN realization of exact \lBF{}  category}
\label{TNeBF}

A \lBF{}  category is a description of the topological properties of
$p$-dimensional topological excitations, such as their fusion and braiding
properties.  It is quite difficult to formulate such a theory.  On the other
hand, we may use TN and path integrals\cite{LN0701,GW0931} to realize exact
\lBF{}  categories in a concrete physical way.  Thus it is possible to study the
exact \lBF{}  category via its TN realization or its path integral realization,
without directly involving the fusion and braiding of $p$-dimensional
topological excitations.  In this section, we will describe such an approach.

All lbL systems, include L-type topologically ordered
states, are described by path-integrals.  A path-integral can be described by a
TN with finite dimensional tensors defined on a space-time lattice (or a
space-time complex) (see Appendix \ref{path}).  Even though L-type
topologically ordered states are all gapped, only some of them can be described
by fixed-point path-integrals which are \emph{topological path integrals}:
\begin{defn} \textbf{Topological path integral} \\
(1) A topological path integral has an action 
amplitude that can be described by 
a TN with \emph{finite dimensional} tensors.\\
(2)  It is a sum of the action amplitudes for all the paths.
(The summation corresponds to the tensor contraction.) \\ 
(3)  Such a sum (called the partition function $Z(M)$) on a closed space-time
$M$ only depend on the topology of the space-time.  The partition function is
invariant under the local deformations and reconnections of the TN.\\
(4)  To describe a lbL system, we 
require the partition function to be well defined
on space-time with any topology.\\
(5)  To describe a local Hamiltonian qubit system, we only
require the partition function to be well defined
on space-time which is a mapping torus.
\end{defn}\noindent
In the next a few sections, we will give concrete examples of the topological
path integrals. The topological path integrals are closely related to exact
\lBF{} categories.\cite{TV9265,LWstrnet,kirillov,balsam-kirillov} We like to conjecture that 
\begin{conj} 
All exact \lBF{}  categories (\ie topological states with gapped boundary)
are described by topological path integrals.
\end{conj} \noindent
We make such a conjecture because we believe that the tensor network
representation that we are going to discuss is the most general one.  It can
capture all possible fixed-point tensors under renormalization flow, and those
fixed-point tensors give rise to  topological path integrals.  Note that the
fixed point tensors with a finite dimensions always produce entanglement
spectrums\cite{LH0804} that have a finite gap.  As a result, the topological
path integrals (described by tensors of finite dimensions) can only produce the
exact  \lBF{}  categories.  This is because exact \lBF{}  categories have gapped
boundary and gapped entanglement spectrum, while closed  \lBF{}  categories that
are not exact have gapless boundary and gapless entanglement spectrum.

We also like to remark that we cannot say that all topological path integrals
describe  exact \lBF{}  categories, since some  topological path integrals are
stable while others are unstable (see Definition \ref{stTopP}).  We will see
that only the stable topological path integrals describe  exact \lBF{}
categories, and unstable topological path integrals do not describe exact
\lBF{}  categories.  
The definition of  stable topological path integrals is given in Definition
\ref{stTopP}.
Here we like to point out that
\begin{conj} 
\label{Zstable}
A topological path integral in $(n+1)$-dimensional space-time constructed with
finite dimensional tensors is stable iff $|Z(S^1\times S^n)|=1$.
\end{conj}\noindent
Note that $Z(S^1\times S^n)$ is the ground state degeneracy on $n$-dimensional
space $S^n$.  If a system has a gap and the ground degeneracy is 1, a small
perturbation cannot do much to destabilize the state.  So $Z(S^1\times S^n)=1$
is the sufficient condition for a stable  topological path integral.  This
argument implies that if the ground degeneracy is 1 on $S^n$, then the system
has no locally distinguishable ground state, and the ground state degeneracy on
space with other topologies are all robust against any small perturbations.  In
\Ref{KW13}, it is shown that if the entanglement entropy of the region $M$ has
the following area-law structure $S_M=c A+\ga A^{0}+o(1/A)$, then
the ground state degeneracy is indeed 1 on $S^n$.  Here $A$ is the
``area'' of the boundary of the region $M$.

Since the topological path integrals are independent of re-triangulation of the
space-time, the partition function on a closed space-time only depends on the
topology of the space-time.
\begin{conj} 
\label{LBF categoryeq}
Two exact \lBF{}  categories are the same iff their topological path integrals
produce the topology-dependent partition functions that belong to the same
connected component in the space of topological partition functions
$V^L_{Z(M)}$.
%, up to a factor
%$W^\chi(M)$, on closed space-time with any orientable topology.
\end{conj}\noindent
Here $V^L_{Z(M)}$ is defined as
\begin{defn} 
A stable topological path integral produces a topology-dependent partition
function $Z(M)$ for  closed space-time $M$ with any orientable topologies.
$V^L_{Z(M)}$ is the space of all such topological partition functions.
\end{defn} \noindent
In fact, we know that each connected component in $V^L_{Z(M)}$ contains
many topology-dependent partition functions (see Section \ref{TNtrivial}).
Two topological path integrals, $Z(M)$ and $Z'(M)$, can belong to the same
connected component in the space of topological partition functions if the two
topological path integrals differ by 
\begin{align}
\label{eupon}
 Z'(M)/Z(M)=W^{\chi(M)}  
\ee^{\ii \sum_{\{n_i\}} \phi_{n_1n_2\cdots} \int_M P_{n_1n_2\cdots}} ,
\end{align}
where $\chi(M)$ is the Euler number of $M$ and $P_{n_1n_2\cdots}$ are
combinations of Pontryagin classes: $P_{n_1n_2\cdots}=p_{n_1}\wedge
p_{n_2}\wedge \cdots$ on $M$.  $Z(M)$ and $Z'(M)$ are connected since complex
numbers $W$ and $\phi_{n_1n_2\cdots}$ are not quantized.  

Since the path integrals are local, thus the ratio $Z'(M)/Z(M)$ is also local.
This means that if we triangulate the space-time manifold $M$ into a complex $C$,
then $Z'(M)/Z(M)$ can be expressed as a product of the amplitudes from each
simplex in $C$.  Eqn. (\ref{eupon}) may be the only topological
invariant local path integral that is not quantized (\ie $W$ and
$\phi_{n_1n_2\cdots}$ can be any complex numbers). Thus
\begin{conj} 
\label{ZZpEq}
It is possible that
 $Z(M)$ and $Z'(M)$ are connected iff they are related by \eqn{eupon}.
\end{conj}\noindent
In other words, if two  topological path integrals produce two
topology-dependent partition functions that differ by a factor $W^{\chi(M)}
\ee^{\ii \sum_{\{n_i\}} \phi_{n_1n_2\cdots} \int_M P_{n_1n_2\cdots}}$, then the
two topological path integrals describe the same exact \lBF{} category.
%\begin{thm} Let $Z(M)$ be a topology-dependent partition function: $Z(M) \in
%V^L_{Z(M)}$.  Then $Z'(M)$ that satisfy $Z'(M) = W^{\chi(M)}Z(M)$ is also a
%topology-dependent partition function: $Z'(M) \in V^L_{Z(M)}$. $Z(M)$ and
%$Z'(M)$ belong to same connected component in $V^L_{Z(M)}$.  \end{thm}
%\noindent Here $\chi(M)$ is the Euler characteristics of the space-time
%complex $M$.  
%Because all exact \lBF{}  categories  are described by stable topological path
%integrals and all stable topological path integrals describe exact \lBF{}
%categories, we may view the stable topological path integrals as a concrete
%definition of  the exact \lBF{}  categories by using the above two conjectures.  

Summarizing the above discussions:\\
(1) All exact \lBF{}  categories (\ie all L-type topological states with gapped
boundary) are described by \emph{stable} topological path integral constructed
with finite dimensional tensors.\\
(2) All stable topological path integrals describe  exact \lBF{}  categories.
\\
So, we may view the stable topological path integrals as a concrete definition
of the exact \lBF{}  categories.  The stable topological path integrals also
classify the exact \lBF{}  categories.  
%Since exact \lBF{}  categories form a subset of exact \hBF{}  categories.
%This allows us to understand a large class of exact \hBF{} categories.
Since exact \lBF{}  categories are also exact \hBF{} categories, the above
topological path integrals and TN also describe a subset of exact \hBF{}
categories.

\subsection{Examples of TN realization of \lBF{n}  category}

\subsubsection{TN realization of 0+1D exact \lBF{1}  category}

\begin{figure}[tb]
\begin{center}
\includegraphics[scale=0.5]{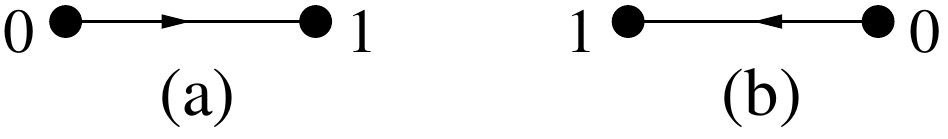}
%Fig. 1
\end{center}
\caption{The tensor $\As01{\pm}$ is associated with a segment, with a branching
structure.  The branching structure gives the vertices a local order: the
$i^{th}$ vertex has $i$ incoming edges.  The segment in (a) has an orientation
$s_{01}=+$  and the segment in (b) has an orientation $s_{01}=-$.  
}
\label{seg}
\end{figure}

The topological path integral that describes a 0+1D topologically ordered state
can be constructed from two complex tensors $\As01{\pm}$.  The tensor
$\As01{\pm}$ can be associated with a segment, which has a branching structure.
(For details about the branching structure, see Appendix \ref{stcomp}.)  A
branching structure is a choice of orientation of each edge (see Fig.
\ref{seg}).  Here the $v_0$ index is associated with the vertex-0 (See Fig.
\ref{seg}).  It represents the degrees of freedom on the vertices.

Using the tensors, we can define the topological path integral on any 1-complex
that has no boundary:
\begin{align}
 Z&=\sum_{ v_0,\cdots}
\prod_\text{edge} \As01{s_{01}}
\end{align}
where $\sum_{v_0,\cdots}$ sums over all the vertex indices,
$s_{01}=+$ or
$-$ depending on the orientation of edge(see Fig.  \ref{seg}), and
\begin{align}
 \As01+ &= \A01 ,
\nonumber\\
 \As01- &= \Amm01 .
\end{align}
We want to choose the tensors $\As01{\pm}$ such that the path integral is
re-triangulation invariant.  Such a topological path integral describes a
topologically ordered state in 1-space-time dimension and also define an exact
\hBF{1}  category in 1 dimension.

\begin{figure}[tb]
\begin{center}
\includegraphics[scale=0.5]{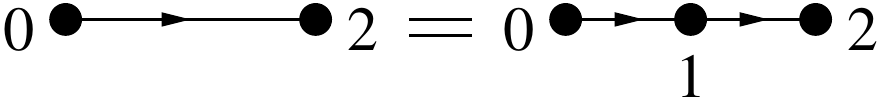}
%Fig. 1
\end{center}
\caption{
A triangulation of a 1D complex.
}
\label{1to2}
\end{figure}

The invariance of $Z$ under the re-triangulation in Fig. \ref{1to2} requires
that
\begin{align}
\label{AA12}
\A02
&=
\sum_{v_1} 
\A01
\A12 .
\end{align}
We would like to mention that there are other similar conditions for different
choices of the branching structures.  We obtain a total of three conditions
\begin{align}
\label{AA}
\A02 &= \sum_{v_1} \A01 \A12 ,
\nonumber\\
\A01 &= \sum_{v_2} \A02 \Amm12 ,
\nonumber\\
\A12 &= \sum_{v_0} \Amm01 \A02 ,
\end{align}
If we  view 
$\A01$ as a matrix $A$, 
and $\Amm01$ as a matrix $\t A$, 
the above can be rewritten as
\begin{align}
 A= A A= A \tilde A = \tilde A A.
\end{align}
We see that the fixed point $A$ is a matrix with a form
\begin{align}
\label{AUI}
 A=\t A &=U^{-1} \begin{pmatrix}
 I_{n\times n} & 0_{n\times m}\\
 0_{m\times n} & 0_{m\times m} \\
\end{pmatrix} U
\end{align}

Here, we would like to introduce a notion of stable fixed point.  If $A$ is not
a fixed point, when we combine two segments into one segment, $A$ transforms as
\begin{align}
\label{AA}
 & \sum_{v_1} \A01 \A12 \to \Ga \A02,
\nonumber\\
 & \sum_{v_2} \A02 \Amm12 \to \Ga \A01,
\nonumber\\
 & \sum_{v_0} \Amm01 \A02 \to \Ga \A12,
\end{align}
where $\Ga$ is a rescaling factor.  Such a transformation defines a
renormalization flow. Then, $A$ is a stable fixed point if $A+\del A$ will
alway flow back to fixed-point tensor $A'$ that have the same
topology-dependent partition function for any $\del A$ and a proper  rescaling
factor $\Ga$.  This leads to the following definition of stable topological
path integral:
\begin{defn} 
\label{stTopP}
A topological path integral described by a finite dimensional tensor $T$
produces a partition function $Z(M)$ that only depend on the topology of the
closed space-time.  If we perturb the tensor $T\to T+\del T$, the perturbed
tensor may flow to another fixed-point tensor $T'$ under a renormalization flow
which give rise to another topological path integral.  The new topological path
integral produce a topological partition function $Z'(M)$.  If $Z'(M)$ and
$Z(M)$ belong to the same connected component of $V^H_{Z(M)}$ for any
perturbation $\del T$, then the topological path integral described by $T$ is a
stable topological path integral, and $T$ is a stable fixed-point tensor.
\end{defn}\noindent
The above definition describes the physical meaning of stable  topological path
integral. However, it is very hard to use such a definition to determine if a
topological path integral is stable or not. However, if the Conjecture
\ref{Zstable} is true, it will allow us to determine which topological path
integrals are stable and which are not.

Most of the fixed points in \eq{AUI} are not stable.  The only stable fixed
point $A$ has  a form
\begin{align}
 A=\tilde A= U^{-1}\begin{pmatrix}
 1 & 0_{1\times m}\\
 0_{m\times 1} & 0_{m\times m} \\
\end{pmatrix}
U
\end{align}
$A$ for different $U$ give rise to the same topological partition function, and
describe the same exact \lBF1 category.  Thus there is only a trivial
topological order in 0+1D.
\begin{cor} 
All 1-dimensional exact \lBF{1}  categories and exact \hBF{1}  categories are trivial.
\end{cor}

\subsubsection{TN realization of 1+1D exact \lBF{2}  category}
\label{TN2D}

\begin{figure}[tb]
\begin{center}
\includegraphics[scale=0.5]{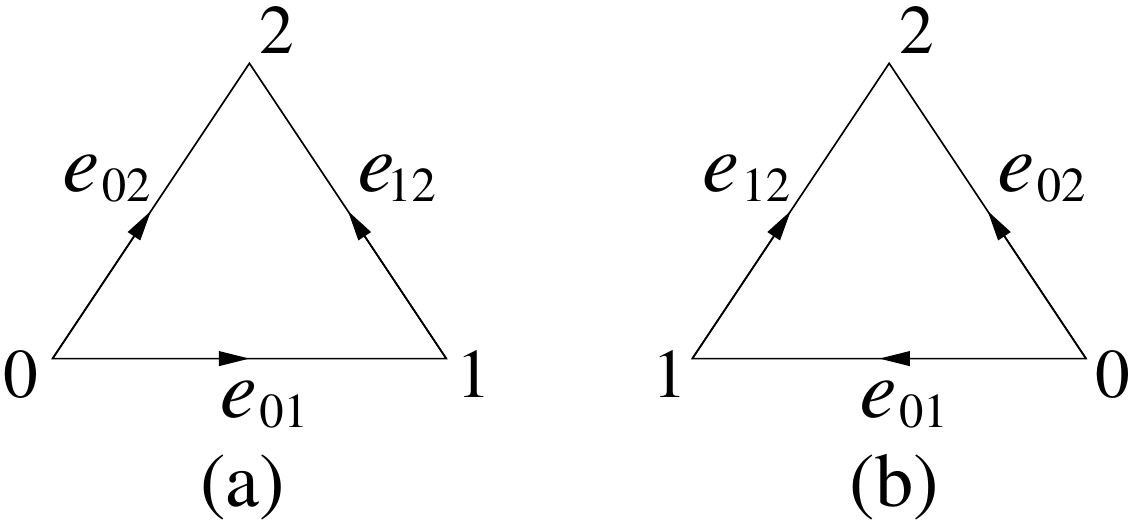}
%Fig. 1
\end{center}
\caption{The tensor $\Bs012{\pm}$ is associated with a triangle, with a
branching structure.  The branching structure gives the vertices a local order:
the $i^{th}$ vertex has $i$ incoming edges.  The triangle in (a) has an
orientation $s_{012}=+$  and the triangle in (b) has an orientation
$s_{012}=-$.  
}
\label{tri}
\end{figure}

The topological path integral that describes a 1+1D topologically ordered state
can be constructed from a tensor set $T_2$ of two complex tensors
$T_2=(w_{v_0}, \Bs012{\pm})$.  The tensor $\B012$ is complex and can be
associated with a triangle, which has a branching structure (see Fig.
\ref{tri}).  $w_{v_0}$ is real and can be associated with a vertex.  A
branching structure is a choice of orientation of each edge in the complex so
that there is no oriented loop on any triangle (see Fig.  \ref{tri} and Fig.
\ref{tetr}).  Here the $v_0$ index is associated with the vertex-0, the
$e_{01}$ index is associated with the edge-$01$ (See Fig. \ref{tri}).  They
represent the degrees of freedom on the vertices and the edges.

Using the tensors, we can define the topological path integral on any 3-complex
that has no boundary:
\begin{align}
\label{Z2d}
 Z&=\sum_{ v_0,\cdots; e_{01},\cdots}
\prod_\text{vertex} w_{v_{0}} 
\prod_\text{face} \Bs012{s_{012}}
\end{align}
where $\sum_{v_0,\cdots; e_{01},\cdots}$ sums over all the vertex indices and
the edge indices,  $s_{012}=+$ or $-$ depending on the orientation of
triangle (see Fig.  \ref{tri}), and
\begin{align}
\Bs012+ &= \B012 ,
\nonumber\\
\Bs012- &= \Bmm012 .
\nonumber\\ 
\end{align}
We want to choose the tensors $(w_{v_0}$, $\Bs012{\pm})$ such that the
path integral is re-triangulation invariant.  Such a topological path integral,
if stable, describes a topologically ordered state in 2-space-time dimensions
and also define an exact \lBF{2}  category.

\begin{figure}[tb]
\begin{center}
\includegraphics[scale=0.5]{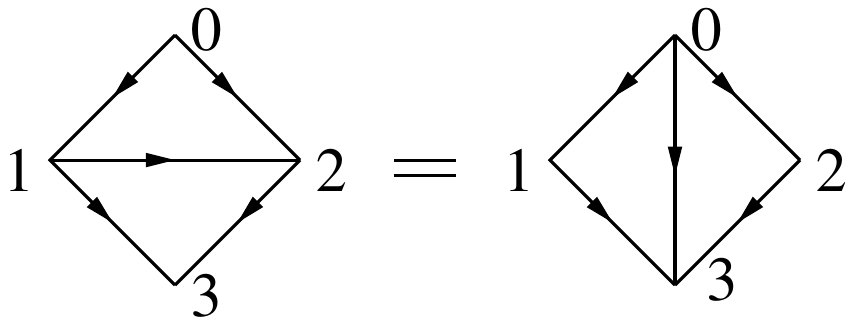}
%Fig. 1
\end{center}
\caption{
A triangulation of a 2D complex.
}
\label{2to2}
\end{figure}

\begin{figure}[tb]
\begin{center}
\includegraphics[scale=0.5]{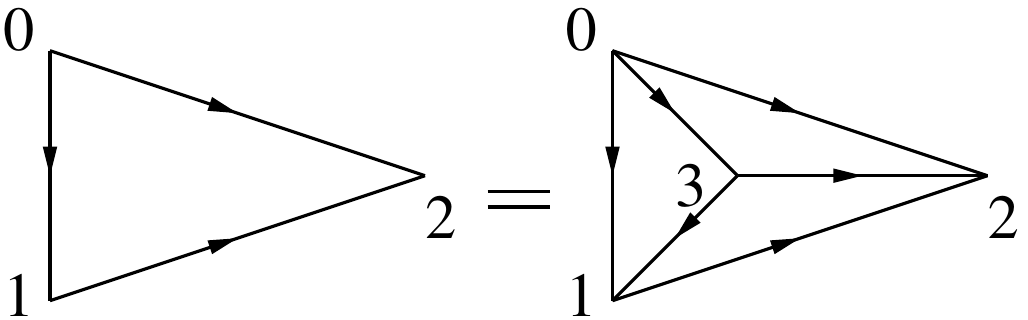}
%Fig. 1
\end{center}
\caption{
A triangulation of another 2D complex.
}
\label{1to3}
\end{figure}

The invariance of $Z$ under the re-triangulation in Fig. \ref{2to2}
requires that
\begin{align}
\label{BB22}
\sum_{e_{12}}
\B012
\Bmm123
&=
\sum_{e_{03}} 
\B013
\Bmm023 .
\end{align}
We would like to mention that there are other
similar conditions for different choices of the branching structures.  The branching structure of a tetrahedron affects the labeling of the vertices.

%\begin{Question}
%Is the last sentence correct?
%\end{Question}

The invariance of $Z$ under the triangulation in Fig. \ref{1to3}
requires that
\begin{align}
\label{BB13}
& \B012 =
\\
&
\sum_{e_{03}e_{13}e_{23},v_3} w_{v_3}
\Bmm031
\B312
\B032.
\nonumber 
\end{align}
Again there are other similar conditions for different choices of the branching
structures.

The above two types of the conditions are sufficient for producing a
topologically invariant partition function $Z$.

Here we would like to point out that two different solutions
are regarded as equivalent (\ie describe the same exact \lBF{} category) if\\
(1) they can be connected by a one parameter family of the solitons continuously.\\
(2) they can be mapped into each other by the relabeling of the indices
$i \to \t i=f(i)$.\\
There may be additional equivalence relations.  In general, two fixed-point
tensor sets $T_2$ and $T'_2$ are regarded as equivalent if their corresponding
topological partition functions for any closed orientable space-time mapping
tori are the same: $Z(M)=Z'(M)$.

It turns out that in 1+1D, all the stable solutions have a trivial
topology-dependent partition function $Z(M)=1$, since there is no nontrivial
topological order in 1+1D.\cite{VCL0501,CGW1107}  Thus 
\begin{cor} 
all 2-dimensional exact \lBF{2}  categories and exact \hBF{2}  categories are trivial.
\end{cor}\noindent
Since the
boundary of 1+1D gapped state can always be gapped,
\begin{cor} 
all 2-dimensional closed \lBF{2}  categories and closed \hBF{2}  categories are exact.
\end{cor}\noindent
As a result,
\begin{cor} 
all 2-dimensional closed \lBF{2}  categories and closed \hBF{2}  categories are trivial.
\end{cor}

\subsubsection{TN realization of 2+1D exact \hBF{3}  category}

\begin{figure}[tb]
\begin{center}
\includegraphics[scale=0.6]{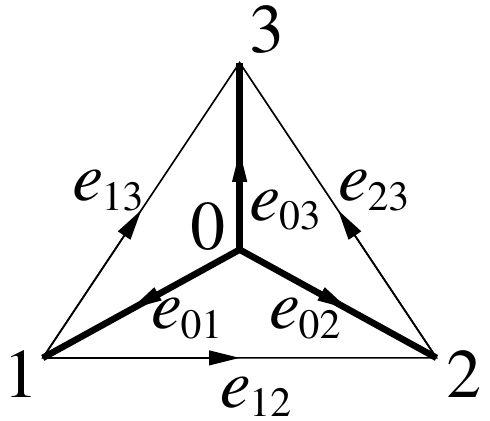}
%Fig. 1
\end{center}
\caption{
The tensor $\Cs0123{\pm}$ is associated with a tetrahedron, which has a branching
structure.  If the vertex-0 is above the triangle-123, then the tetrahedron
will have an orientation $s_{0123}=-$.  If the vertex-0 is below the
triangle-123, the tetrahedron will have an orientation $s_{0123}=+$. The
branching structure gives the vertices a local order: the $i^{th}$ vertex has
$i$ incoming edges.  
%The $\al$ index lives on the triangle facing vertex-0; the
%$\bt$ index facing vertex-1; the $\ga$ index facing vertex-2; the $\del$ index
%facing vertex-3. 
}
\label{tetr}
\end{figure}

The topological path integral that describes a 2+1D topologically ordered state
with a gapped boundary can be constructed from a tensor set $T_3$ of two real
and one complex tensors $T_3=(w_{v_0}, \Aw01,\Cs0123{\pm})$.  The complex
tensor $\Cs0123{\pm}$ can be associated with a tetrahedron, which has a
branching structure (see Fig.  \ref{tetr}).  A branching structure is a choice
of orientation of each edge in the complex so that there is no oriented loop on
any triangle (see Fig.  \ref{tetr}).  Here the $v_0$ index is associated with
the vertex-0, the $e_{01}$ index is associated with the edge-$01$, and the
$\phi_{012}$ index is associated with the triangle-$012$.  They represents the
degrees of freedom on the vertices, edges, and the triangles.

Using the tensors, we can define the topological path integral on any 3-complex
that has no boundary:
\begin{align}
\label{Z3d}
 Z=\sum_{ v_0,\cdots; e_{01},\cdots; \phi_{012},\cdots}
&\prod_\text{vertex} w_{v_{0}} 
\prod_\text{edge} \Aw01\times
\\
&
\prod_\text{tetra} \Cs0123{s_{0123}}
\nonumber 
\end{align}
where $\sum_{v_0,\cdots; e_{01},\cdots; \phi_{012},\cdots}$ sums over all the
vertex indices, the edge indices, and face indices, 
$s_{0123}=+$ or $-$
depending on the orientation of tetrahedron (see Fig.  \ref{tetr}), and
\begin{align}
 \Cs0123+ &= \C0123 
\nonumber\\
 \Cs0123- &= \Cmm0123
\end{align}
We want to choose the tensors $(w_{v_0}$, $\Aw01$, $\Cs0123{\pm})$ such that
the path integral is re-triangulation invariant.  Such a topological path
integral describes a L-type topologically ordered state in 3-space-time
dimensions and also define an exact \lBF{3}  category.

\begin{widetext}
The invariance of $Z$ under the re-triangulation in Fig. \ref{2to3}
requires that
\begin{align}
\label{CC23}
&\ \ \
\sum_{\phi_{123}} \C0123 \C1234
\nonumber\\
&=
\sum_{e_{04}} \Aw04
\sum_{ \phi_{014} \phi_{024} \phi_{034} }
\C0124 
\Cmm0134
\C0234 .
\end{align}
We would like to mention that there are other similar conditions for different
choices of the branching structures.  The branching structure of a tetrahedron
affects the labeling of the vertices.

%\begin{Question}
%Is the last sentence correct?
%\end{Question}

The invariance of $Z$ under the re-triangulation in Fig. \ref{1to4}
requires that
\begin{align}
\label{CC14}
&
\C0234
=
\sum_{e_{01}e_{12}e_{13}e_{14},v_1} w_{v_1} 
\Aw01 \Aw12 \Aw13 \Aw14 
\sum_{ 
\phi_{012} 
\phi_{013} 
\phi_{014} 
\phi_{123} 
\phi_{124} 
\phi_{134} 
}
\\
&\ \ \ \ \ \ \ \ \ \ \ \ 
\C0123
\Cmm0124 
\C0134
\C1234
\nonumber 
\end{align}
Again there are other similar conditions for different choices of the branching
structures.
\end{widetext}

The above two types of the conditions are sufficient for producing a
topologically invariant partition function $Z$, which is nothing but the
topological invariant for three manifolds introduced by Turaev and
Viro.\cite{TV9265} Again, two different solutions are regarded as equivalent if
they produces the same topology-dependent partition function 
for any closed space-time.
% (see Conjecture \ref{BF categoryeqDD}).

\begin{figure}[b]
\begin{center}
\includegraphics[scale=0.5]{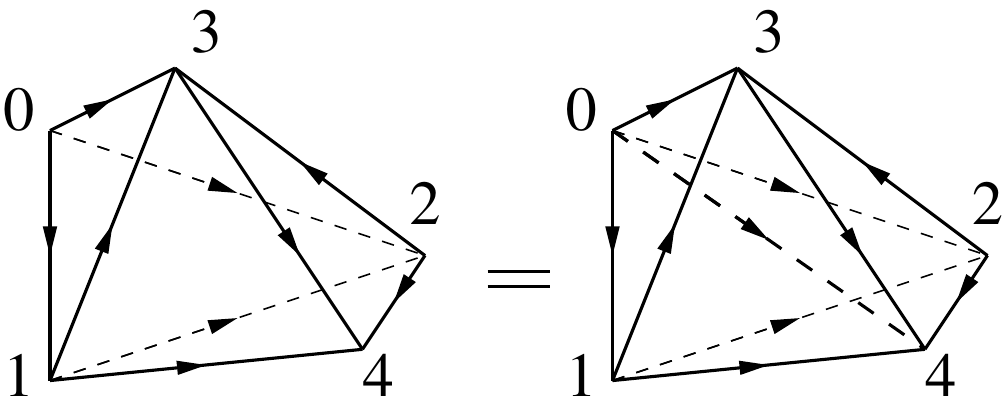}
%Fig. 1
\end{center}
\caption{
A triangulation of a 3D complex.
}
\label{2to3}
\end{figure}
\begin{figure}[b]
\begin{center}
\includegraphics[scale=0.5]{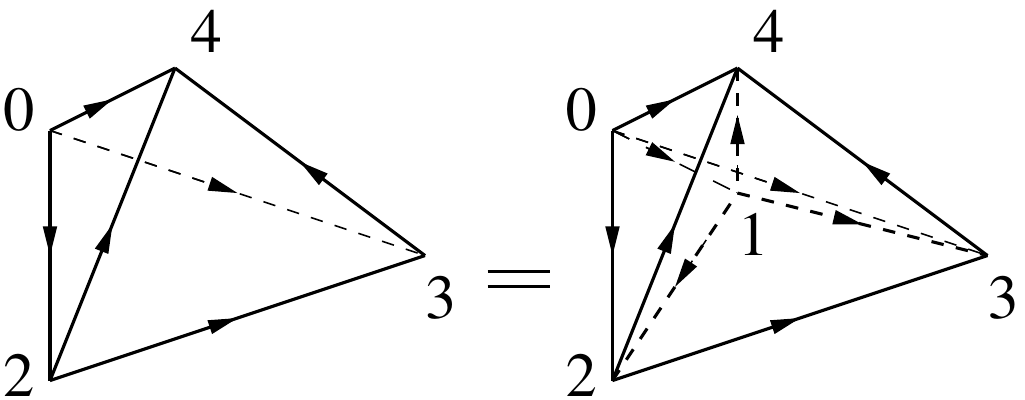}
%Fig. 1
\end{center}
\caption{
A triangulation of another 3D complex.
}
\label{1to4}
\end{figure}

\subsection{The unitarity condition}

Let we choose the space-time to have a form $M^d\times I$, where $M^d$ is the
space and the 1D segment $I$ is the time.  The path integral on $M^d\times I$
gives rise to a transfer matrix $U=\ee^{-\tau H}$. The indices on one boundary
of $M^d\times I$ correspond of the row index of $U$ and the indices on the
other boundary of $M^d\times I$ correspond of the column index of $U$.  Since
$H$ is hermitian, $U$ must has non-negative eigenvalues.  This is the unitarity
condition for the path integral.  One can show that if the two real and one
complex tensors $(w_{v_0}, \Aw01,\Cs0123{\pm})$ satisfy (see Appendix
\ref{pathham})
\begin{align}
& w_{v_0}>0, \ \ \ \  \Aw01>0,
\\
& \Cs0123+= (\Cs0123-)^*,
\nonumber 
\end{align}
then the path integral described by $(w_{v_0}, \Aw01,\Cs0123{\pm})$ is unitary.
The above is for 2+1D path integral. We have similar unitarity condition for
TN path integral in other dimensions.

Now it is clear that the above construction of topological path integrals can
be easily generalized to any dimensions.  Such a construction can be viewed as
concrete definition of exact \lBF{} categories (or L-type topological orders
with gappable boundaries).

\subsection{TN realization of generic \lBF{}  category}
\label{TNBF}

After using TNs and topological path integral to define/construct exact \lBF{}
categories, in this section, we would like to use TNs and topological path
integrals to construct generic \lBF{}  categories.
The idea is every simple, for a TN
defined by a stable topological path integral, its boundary can give us an
explicit construction of a generic \lBF{}  category.

Let us try to construct a path integral that gives rise to a generic
topological theory described by a generic \lBF{}  category. Note that such a
topological theory may be anomalous.  So the topological theory may not be
described by a path integral in the same dimension.  The trick is try to
describe the  topological theory using a path integral in higher dimension.

To define  a generic and potentially anomalous topological theory (\ie a
generic \lBF{n}  category $\EC_n$) in $n$-dimensional space-time, let us assume
that space-time cell-complex $M^n$ has a topology of $S^1\times S^{n-1}$.  We
then view $S^{n-1}$ as a boundary of $n$-dimensional solid ball $D^n$ and
extend $M^n$ to $M^{n+1}=S^1\times D^n$.  (For a more general discussion, see
Appendix \ref{extF}.) Now we can give  $M^n$ a triangulation, and then extend
that triangulation to $M^{n+1}$.  

Next, we put a topological path integral described by a  tensor set $T_{n+1}$
of finite-dimensional tensors on cell-complex $M^{n+1}$. Such a topological
path integral defines a $(n+1)$-dimensional exact \lBF{n+1}  category
$\EC^\text{exact}_{n+1}$.  More generally, we can also modify the tensor set on
the boundary from $T_{n+1}$ to $T^b_{n+1}$, such that the path integral on
$M^{n+1}$ is still a topological path integral (\ie re-triangularization
independent).  

We note that the  topological path integral on $M^{n+1}$, defined by the pair
of tensor sets $(T_{n+1},T_{n+1}^b)$, only depends on the fields (the indices
of the tensors) on the boundary $M^n$, since the topological path integral does
not change under the re-triangularization on $M^{n+1}$, as long as we fix the
triangularization and the fields on the boundary $M^n$.  So the  topological
path integral on $M^{n+1}$ also defines a path integral (a quantum theory) on
$M^n$ (an $n$-dimensional space-time).  Since the topological path integral is
invariant under the retrianglations both in the bulk and on the boundary, we
like to make the following conjecture:
\begin{conj} 
Assume that the tensor set $T_{n+1}$ describes a stable topological path
integral in $(n+1)$-dimensional space-time, and the pair $(T_{n+1},T_{n+1}^b)$
describes  a topological path integral in $(n+1)$-dimensional space-time
$M^{n+1}$ with a boundary $M^n$.  Such a theory is gapped in the bulk and on
the boundary.  
\end{conj} 

Since the boundary is gapped. The topological path integral on $M^{n+1}$ also
defines the gapped topological excitations on the boundary $M^n$ which is
described by an $n$-dimensional \lBF{n}  category $\EC_n$.  We have
$\EC^\text{exact}_{n+1}=\cZ_n(\EC_n)$.  Now, we can see that
\begin{conj} 
\label{SdDd1}
The theory on $M^n$, defined by a pair of stable tensor sets
$(T_{n+1},T_{n+1}^b)$ as outlined above, describes a generic \lBF{n}  category
$\EC_n$ in $n$ dimension.  The center of $\EC_n$ is given by
$\EC^\text{exact}_{n+1}=\cZ_n(\EC_n)$, where $\EC^\text{exact}_{n+1}$ is the
\hBF{n+1} category realized by the $(n+1)$-dimensional topological path
integral defined by the tensor set $T_{n+1}$.  
\end{conj}\noindent 
Note that the $n$-dimensional theory is only defined on $M^n=S^1\times S^{n-1}$
with a particular extension $M^{n+1}=S^1\times D^n$. Since
$\EC^\text{exact}_{n+1}$ is nontrivial, the different extensions of $M^n$ to
different $(n+1)$-dimensional manifolds may leads to different values of the
path integral (see discussions in the next section). This is a sign of
gravitational anomaly.

We also note that the path integral defined above on $M^n$ is enough to give
rise to $p$-dimensional topological excitations on the space $S^{n-1}$ and
determine their fusion and braiding properties.  This is why we believe the
Conjecture \ref{SdDd1}.  We like to further conjecture that
\begin{conj} 
All  generic \lBF{n}  categories $\EC_n$ in $n$ dimension can be realized this
way by pairs of tensor sets $(T_{n+1},T_{n+1}^b)$ of finite dimensional
tensors.
\end{conj}

\subsection{TN realization of closed \lBF{}  category}
\label{TNcBF}

With the tensor  formulation of generic \lBF{}  categories discussed in last
section, it is quite natural to reduce it into a tensor formulation of closed
\lBF{}  category.  We just need to figure out which (generic) \lBF{}  categories
constructed above are closed \lBF{}  categories.

First, let us start with the  tensor  formulation of a generic $n$-dimensional
\lBF{n}  category $\EC_n$ defined by a pair of tensor sets
$(T_{n+1},T_{n+1}^b)$.  
%We ask, which pair of tensor sets,
%$(T_{n+1},T_{n+1}^b)$, describes a closed \hBF{n+1} category?  
%To answer such a
%question, we note that the tensor formulation of exact \hBF{n+1}  category in
%$(n+1)$-dimension, $\EC^\text{exact}_{n+1}$, gives rise to a topological path
%integral whose partition function $Z(M)$ on a closed space-time $M$ is
%independent of the triangulation of $M$.  But $Z(M)$ may depend on the topology
%of $M$.  However, if $Z(M)=1$ on any closed $M$, then the path integral of
%$(T_{n+1},T_{n+1}^b)$ on $M$ will only depend on the boundary of $M$, which
%define a theory on $\prt M$. Thus
%To answer such a question, 
We conjecture that
\begin{conj} 
$(T_{n+1},T_{n+1}^b)$ describes a closed \lBF{n} category if $T_{n+1}$ gives
rise to a topological partition function $Z(M^{n+1})$ that describes a trivial
topological phase. Every closed \lBF{n} category can be obtained this way.
%(that only depend on the topology of space-time $M^{n+1}$).
\end{conj}\noindent

To understand the conjecture, we note that the topological partition function
for a trivial phase has a form
\begin{align}
Z(M^{n+1})=
W^{\chi(M^{n+1})}
\ee^{\ii \sum_{\{n_i\}} \phi_{n_1n_2\cdots} \int_{M^{n+1}} P_{n_1n_2\cdots}}
\end{align}
We can introduce a trivial topological state described by the
following partition function
$$
Z^\text{tri}(M^{n+1})=
W^{-\chi(M^{n+1})}
\ee^{-\ii \sum_{\{n_i\}} \phi_{n_1n_2\cdots} \int_{M^{n+1}} P_{n_1n_2\cdots}}
$$
(The boundary of such a trivial topological state is the gCS anomalous
topological state discussed in Section \ref{tprod}.) 
We then stack the two systems. The partition function for the
combined system satisfies
\begin{align}
 Z^\text{combined}(M^{n+1})=1.
\end{align} 

%We like to conjecture that (this is a
%special case of Conjecture \ref{BF categoryeq\cdots})
%\begin{conj} 
%\label{trvlZ}
%When $n<7$, if $Z(M^{n+1})=1$ for any cell-complex $M^{n+1}$ which is a
%triangulation of a mapping torus, then the corresponding \hBF{n+1}  category
%$\EC^\text{exact}_{n+1}$ is trivial.  The corresponding ground state is in the
%same phase with a product state.
%\end{conj}
%
%Now, we are ready to have a tensor formulation of closed \hBF{}  category
%\begin{conj} 
%The theory on a closed $n$-dimensional manifold $M^n$, described by a
%topological path integral on a $(n+1)$-dimensional manifold  $N^{n+1}$ with
%$M^n$ as its boundary and defined by a pair of tensor sets $(T_{n+1},T_{n+1}^b)$ of
%finite dimensional tensors, describes a closed \hBF{}  category $\EC_n$ in $n$
%dimension, provided that the $(n+1)$-dimensional topological path integral
%defined by the tensor $T_{n+1}$ describes a trivial BFTC in $n+1$
%dimensions.  In other words, the topological excitations on the boundary $M^n$
%can be described by a pure $n$-dimensional path integral (which may not be
%topological).
%\end{conj}
%We note that when 

We see that for the combined system, the
$(n+1)$-dimensional path integral on any two $(n+1)$-dimensional cell complex,
$N^{n+1}$ and  $N^{\prime n+1}$, will be the same, as long as \\
(1) the boundaries of $N^{n+1}$ and  $N^{\prime
n+1}$ are the same:  $\prt N^{n+1}=\prt N^{\prime n+1} =M^n$,\\
(2) the triangulations on $N^{n+1}$ and  $N^{\prime n+1}$ reduce to the same
triangulations on $M^n$,\\
Thus, the combined $(n+1)$-dimensional  path integral on $N^{n+1}$ only
depend on the fields on the $n$-dimensional boundary $M^n$.  The  path integral
does not depend on how we extend $M^n$ into $N^{n+1}$.  Therefore the
$(n+1)$-dimensional  path integral on $N^{n+1}$ define an
$n$-dimensional path integral on $M^n$, which in turn describes a well defined
topological theory in $n$-dimensional space-time, whose excitations are
described by a closed \lBF{n}  category in $n$ dimensions.

\begin{rema}
We note that although the  $(n+1)$-dimensional path integral on $N^{n+1}$
described by the tenser set $(T_{n+1},T_{n+1}^b)$ is topological (\ie
independent of the retriangulizations both in the bulk and on the boundary),
Only the path integral of the combined system leads to a well defined
$n$-dimensional path integral on $M^n=\prt N^{n+1}$ which describes a gapped
topological state, a closed \lBF{n} category.  Such an $n$-dimensional path
integral on $M^n$ for the combined system may not be re-triangularization
invariant.  Thus the $n$-dimensional path integral on $M^n$ may not describe an
exact \lBF{n} category.
\end{rema}

\section{Examples of TN realization of topologically ordered states}

\subsection{TN realization of exact \lBF{}  categories}

\subsubsection{TN realization of the trivial \lBF{}  category in any dimensions}
\label{TNtrivial}

One way to obtain topological path integral in any dimensions is to assume that
all tensors in the tensor set $T_n$ are 1-dimensional (\ie all the indices in
the tensor have a range 1).  Such kind of tensor set describes a trivial
topological state. Here we like to use such a trivial example to perform some
nontrivial check for some of our conjectures.

When all the tensor are 1-dimensional, we assign a weight
$W_k$ to each $k$-cell in the $n$-dimensional space-time complex $M^n$,
where $W_k$ is real for $k=0,\cdots,n-1$ and is complex for $k=n$.
The partition function has a form 
\begin{align}
 Z(M^n)= W_n^{N^+_n} (W_n^*)^{N^-_n} \prod_{k=0}^{n-1} W_k^{N_k} .
\end{align}
where $N_k$ is the number of $k$-cells
and $N^\pm_n$ is the number of $n$-cells with the $\pm$ orientation.

The re-triangulation invariance requires that
\begin{align}
 W_k=W^{(-)^k}, \ \  W\in \R.
\end{align}
In this case
\begin{align}
 Z(M^n)=W^{\chi(M^n)}
\end{align}
where $\chi(M)$ is the Euler characteristics of cell complex $M$.  We see that
a trivial topological theory can give rise to a nontrivial partition function
with a nontrivial dependence on the topology of the space-time. Such a
seemingly non-trivial ``topological'' partition function actually describes a
trivial topological order since $W$ is not quantized.  $W\neq 1$ and $W= 1$
correspond to the same phase.

%However, when we consider a topologically ordered state in a
%local Hamiltonian qubit system, we only need the partition function on a
%space-time complex $M^n$ which is a mapping torus.  The Euler characteristics
%$\chi(M^n)=0$ and $Z(M^n)=1$ for such space-time topology.  $Z(M^n)=1$ for any
%mapping torus $M^n$ is the expected result for a trivial topological state.
% (see Conjecture \ref{trvlZ}).

\subsubsection{TN realization of a 1+1D unstable topological path integral}

One way to construct 1+1D topological path integral is to use the elements of a
finite group $G$ to label the edge degrees of freedom and assume there is no
degrees of freedom on the vertices (\ie the range of the vertex index is 1:
$v_i=1$).
We choose
$\B012$, $\Bmm012$, and $w_{v_0}$ in \eq{Z2d} as
\begin{align}
 B_{111;e_{02}}^{ e_{01} e_{12}}&=
 B^{111;e_{02}}_{ e_{01} e_{12}}=
1 \text{ if } e_{01}e_{12}=e_{02},
\nonumber\\
 B_{111;e_{02}}^{ e_{01} e_{12}}&=
 B^{111;e_{02}}_{ e_{01} e_{12}}=
0 \text{ otherwise}, \ \
w_1 =|G|^{-1},
\end{align}
where $e_{ij}\in G$ and $|G|$ is the number of the elements in $G$.  The
resulting path integral is a topological path integral (\ie re-triangulation
invariant).

We find that the partition function on $S^1\times S^1$ is
\begin{align}
 Z(S^1\times S^1)=|G|.
\end{align}
According to the Conjecture \ref{Zstable}, the topological path integral is
unstable.  This is a correct result. The  topological path integral actually
describes a 1+1D gauge theory in zero coupling limit, where $G$ is the gauge
group.  In 1+1D, a gauge theory alway confine for any finite coupling even for
discrete gauge group.  Thus 1+1D gauge theory is unstable and does not describe
a topological phase.

\subsubsection{TN realization of the 3+1D trivial BF$_4$  category}
\label{WWmdl}

In \Ref{CY9362,W06,WW1132}, a 3+1D topological path integral is constructed
using the data of a UMTC.  The partition
function of the 3+1D topological theory on a closed 3+1D space-time is given by
\begin{align}
\label{Z4dUMTC}
Z(M^4)= \ee^{\frac{2 \pi\ii}{8} (c_R-c_L) \si(M^4) }W^{\chi(M^4)},
\end{align}
where $c_R-c_L$ is the chiral central charge of the UMTC, $\si(M^4)$ the
signature of $M^4$, and $W$ an arbitrary real number that can be
continuously deformed to $1$.

Now we would like to show that the construction using UMTC data give rise to a
trivial BF$_4$ category (\ie a trival 3+1D topological order).
First, we can have 3+1D bosonic lattice model
that breaks the time-reversal symmetry
and produces an effective action
\begin{align}
 S = \ii \ka_{gCS} \int_{M^4} p_1,
\end{align}
where $p_1$ is the first Pontryagin class.  The lattice model has no
topological excitations and is a trivial topologically ordered state, since it
is continuously connected to the product state as $\ka_{gCS}\to 0$.

If we stack the lattice model with the  3+1D topological path integral
constructed from UMTC, we can make the combined theory to have a trivial
partition function $Z(M^4)=1$, if we choose $\ka_{gCS}=-\frac{2 \pi\ii}{24}
(c_R-c_L) $, since the signature  $\si(M^4)$ can be expressed as
\begin{align}
 \si(M^4) = \int_{M^4} p_1/3.
\end{align}
According to Conjecture \ref{LBF categoryeq}, the combined theory realize a
trivial topological order (\ie a trivial BF category).  The original theory
must be trivial since its stacking with a trivial phase is trivial.

%we can continuously deform the above partition function to identity by adding a
%local topological invariant term $\ee^{\ii \phi \int_{M^4} p_1/3}$ and by
%deforming $W$ to 1. Therefore, the above 3+1D topological path integral
%describes a trivial \lBF{4} category which is also a trivial \hBF{4} category.

%the non-trivial
%dependence of $Z(M^4)$ on  $\si(M^4)$ indicates that the 3+1D topological path
%integral describe a non-trivial \lBF{4}  category in 3+1D. But such a
%topological path integral describe a trivial  \hBF{4}  category in 3+1D.
%
%This is because when $M^4$ is a mapping torus, its signature and its
%Euler number vanishes $\si(M^4)=\chi(M^4)=0$ (see Appendix \ref{extF}).  The
%fact that $Z(M^4)=1$ suggests that the topological path integral constructed
%from UMTC describes a trivial topologically ordered state (or a trivial \hBF{4}
%category), if we view the path integral as an effective theory of a local
%Hamiltonian qubit system.  (However, such a topological path integral describes
%a nontrivial topological theory if we view the path integral as an effective
%theory of a lbL system.) The fact that $Z(S^1\times S^3)=1$ suggests that the
%topological path integral is stable and describes a trivial gapped phase.

By choosing different UMTC's, we can construct many different topological path
integrals that describe the same trivial  category.  Later in Section
\ref{WWmdl1}, we will take advantage of this many-to-one representation of
trivial BF  category and use them to construct closed BF  categories in one
lower dimension.

\subsubsection{TN realization of a 2+1D exact \lBF{3}  category}

One way to construct a 2+1D topological path integral is to use the elements of
a finite group $G$ to label the edge degrees of freedom and assume there is no
degrees of freedom on the vertices and the faces (\ie the range of the index is
1: $v_i=1$ and $\phi_i=1$).  We choose $\C0123$, $\Cmm0123$, $w_{v_0}$, and
$\Aw01$ in \eq{Z3d} as
\begin{align}
&\ \ \ \C0123=\Cmm0123
\nonumber\\
&=1, \text{ if every } e_{ij}e_{jk}=e_{ik},
\nonumber\\
&\ \ \ \C0123=\Cmm0123
\nonumber\\
&=0, \text{ otherwise}, 
\nonumber\\
& 
w_{v_0} =|G|^{-1},\ \ \
\Aw01=1.
\end{align}
where $e_{ij}\in G$.  The resulting path integral is a topological path
integral (\ie re-triangulation invariant).

We find that the partition function on $S^1\times S^2$ is
\begin{align}
 Z(S^1\times S^2)=1.
\end{align}
According to the Conjecture \ref{Zstable}, the topological path integral is
stable.  This is a correct result. The  topological path integral actually
describes a 2+1D gauge theory in zero coupling limit, where $G$ is the gauge
group.  In 2+1D, a discrete gauge theory is alway in the deconfined phase for
small enough coupling.  Thus 2+1D gauge theory is stable and describe a
topological phase.

We can construct a more general 2+1D exact \lBF{3}  category by twisting the
above topological path integral by the cocycle $\om(g_0,g_1,g_2)$ in the group
cohomology class $H^3(G,\R/\Z)$:\cite{DW9093,HW1267,HWW1295}
\begin{align}
\label{T3d}
&\C0123 =\om(e_{01},e_{12},e_{23}),
\nonumber\\
&\ \ \ \ \ \text{ if every } e_{ij}e_{jk}=e_{ik},
\nonumber\\
&\Cmm0123 =\om^*(e_{01},e_{12},e_{23}),
\nonumber\\
&\ \ \ \ \ \text{ if every } e_{ij}e_{jk}=e_{ik},
\nonumber\\
&\ \ \ \C0123=\Cmm0123
\nonumber\\
&=0, \text{ otherwise}, 
\nonumber\\
& 
w_{v_0} =|G|^{-1},\ \ \
\Aw01=1.
\end{align}
Certainly, such a construction can also be easily generalized to any higher
dimensions.

\subsection{TN realization of generic \lBF{}  categories}

Eq. \ref{T3d} gives rise to an exact \lBF{}  category in 2+1D.  Its boundary
will give rise an generic \lBF{}  category in 1+1D.  So using the TN
realization of exact \lBF{}  categories, we can obtain the TN realization of
generic \lBF{}  categories in one lower dimensions.

\subsection{TN realization of closed \lBF{}  categories}

%\subsubsection{TN realization of closed 2+1D \hBF{}  categories}
\label{WWmdl1}

In \Ref{CY9362,W06,WW1132}, a 3+1D topological path integral is constructed
using the data of a UMTC, which describe a trivial topologically ordered state
in 3+1D as discussed before (see Section \ref{WWmdl}).  However, the natural
boundary of the 3+1D topological path integral is very interesting. It is
a 2+1D topologically ordered state whose particle-like topological
excitations are described by the same UMTC that was used to construct the 3+1D
topological path integral.\cite{WW1132,KBS1307,BCFV1372,CFV1350} So all the
closed 2+1D \lBF{3} categories that correspond to the UMTC can be described by
TN and its topological path integral in one higher dimensions.  Since the 3+1D
topological path integral describes a trivial topologically ordered state in
3+1D, we believe that we can also use a path integral in 2+1D to describe the
closed 2+1D \lBF{3}  categories that correspond to UMTC's. 

Suppose we have a 2+1D path integral that describe a UMTC, what is the nature
of its fixed-point partition function (\ie the volume-independent partition
function $Z_0(M^3)$). Since the UMTC describes a gapped topological state,
should we expect the fixed-point partition function $Z_0(M^3)$ to be
topological, \ie independent of the deformation of the shape of space-time
$M^3$.  The answer is no if the chiral central charge $c_R-c_L$ of the UMTC is
not multiple of 8.  The volume-independent partition function $Z_0(M^3)$
must contain a gravitational Chern-Simons term
\begin{align}
\label{ZM3UMTC}
 Z_0(M^{3}) = 
=\ee^{\ii \frac{2\pi (c_R-c_L)}{24} \int_{M^{d+1}} \om_{3} }
\end{align}
which make $Z_0(M^{3})$ not topological.

On the other hand, if we introduce a framing to $M^3$ and allow the fixed-point
partition function to depend on the framing, or we extend $M^3$ to a $M^4$ with
$\prt M^4=M^3$, then we can obtain a fixed-point partition function that is
topological (see \eqn{Z4dUMTC}).  But such a topological partition function for
the 2+1D UMTC is ``anomalous'' since it either depends on the framing, or
depends the signature of its 4-dimensional extension.  However, such an anomaly
can be canceled by stacking with an invertible gCS anomalous topological order
described in Section \ref{tprod}. The price we pay is that the combined anomaly free partition
function (see \eqn{ZM3UMTC}) cannot be topological unless $c_R-c_L=0$ mod 8.

%But such a 2+1D path
%integral may not be topological.  It is possible that all the closed 2+1D
%\lBF{3}  categories correspond to UMTC's (plus some $E_8$ quantum Hall states).
%In this case, all the closed 2+1D \lBF{3}  categories can be described by  3+1D
%topological path integrals.

\section{Probing and measuring BF  categories 
(\ie topological orders and gravitational anomalies)}
\label{PMtopo}

In this paper, we pointed out a direct connection between gravitational
anomalies and topological orders in one higher dimension.  Using such a
connection, we have developed a systematic theory of topological order and
gravitational anomaly in any dimensions.  In this section, we will discuss
another important issue: How to probe and measure different topological orders
and gravitational anomalies. Or in other words, how to probe and measure
different BF  categories.  Here ``probe and measure'' means the methods in
experiments and/or numerical calculations that allow us to distinguish
different topological orders and gravitational anomalies.  
%We have already
%discussed some topological invariants that allow us to  distinguish  different
%topological orders and gravitational anomalies. We will also summarize them
%here for completeness.

\subsection{How to probe and measure the closed \lBF{}  categories described by
non-fixed-point path integrals}

If the path integral described by a TN has no long range correlations, it will
describe a closed \lBF{} category.  But how to determine which closed \lBF{}
category the path integral can produce?  How to determine if two path
integrals give rise to the same closed \lBF{} category or not?

Let us consider, for simplicity, a 2+1D path integral defined by a TN on a 2+1D
space-time complex $M$.  We assume that the path integral has no long range
correlations.  We consider the limit where the space-time complex is formed by
many 3-cells (the thermal dynamical limit).  In this limit, the partition
function have a form
\begin{align}
 Z_\text{path}(M^3)&=\ee^{c_0N_0+c_1N_1+c_2 N_2+c_3N_3} 
Z_0(M)
\\
&\ \ \ \
\ee^{ O(1/N_0)+ O(1/N_1)+ O(1/N_2)+ O(1/N_3)}
\nonumber 
,
\end{align}
where $N_d$ is the number of the $d$-cells in the space-time complex.  Note that
that term $c_iN_i$ is proportional to the volume of space-time.  So $c_i$'s
correspond to the density of ground state energy and are not universal.  On the
other hand, $Z_0(M^3)$ is the volume-independent part of the partition function
which contains universal structures.

To understand the universal structures in $Z_0(M^{d+1})$, 
let us use $\cM_{M^{d+1}}$ to denote the
moduli space the closed space-time $M^{d+1}$ with different metrics but the
same topology.  Then  the volume-independent part of the partition function
$Z_0(\cdot)$ can be viewed as a map from $\cM_{M^{d+1}}$ to $\Cb$.
\begin{conj} 
If $Z_0(M^{d+1}_0) \neq 0$ for a point $M^{d+1}_0$ in  $\cM_{M^{d+1}}$, then
$Z_0(M^{d+1}) \neq 0$ for every point $M^{d+1}$ in  $\cM_{M^{d+1}}$.
\end{conj}\noindent
We note that $\cM_{M^{d+1}}$ is connected. So $Z_0(M^{d+1}_0)$ and $Z_0(M^{d+1}_1)$
for two points $M^{d+1}_0, M^{d+1}_1 \in \cM_{M^{d+1}}$ are partition functions of
topological states that belong to the same gapped phase. As a result
$Z_0(M^{d+1}_0)/Z_0(M^{d+1}_1)=W^{\chi(M)}  \ee^{\ii \sum_{\{n_i\}}
\phi_{n_1n_2\cdots} \int_M P_{n_1n_2\cdots}}$ (see Conjecture \ref{ZZpEq}).  So
the partition function $Z_0(\cdot)$ is actually a map $Z_0: \cM_{M^{d+1}} \to \Cb
-\{0\}\sim U(1)$.  If $\pi_1(\cM_{M^{d+1}})\neq 0$, such map may have a non-trivial
winding number. 

To understand the winding number, let us use $G_\text{homeo}(M^{d+1})$ to
denote the homeomorphism group of the space-time $M^{d+1}$.  Note that
$G_\text{homeo}(M^{d+1})$ only depends on the topology of $M^{d+1}$  and
is the same for every point $M^{d+1} \in \cM_{M^{d+1}}$.
Let us use $G^0_\text{homeo}(M^{d+1})$ to denote the subgroup of
$G_\text{homeo}(M^{d+1})$ which is the connected component of
$G_\text{homeo}(M^{d+1})$ that contain identity.  The mapping class group is
formed by the discrete components of the homeomorphism group:
\begin{defn} \textbf{mapping class group}\\
$\text{MCG}(M^{d+1})\equiv  G_\text{homeo}(M^{d+1})/G^0_\text{homeo}(M^{d+1})=\pi_0[G_\text{homeo}(M^{d+1})]$.
\end{defn} \noindent
We note that every homeomorphism $f: M^{d+1}\to M^{d+1}$ in
$\text{MCG}(M^{d+1})$ defines a mapping torus $M^{d+1} \rtimes_f S^1$ that
describes how $M^{d+1}$ deform around a loop $S^1$, and correspond to an
element in $\pi_1(\cM_{M^{d+1}})$.

Since $\pi_1(\cM_{M^{d+1}})=\text{MCG}(M^{d+1})$, the winding number is a group
homomorphism $\text{MCG}(M^{d+1}) \to \Zb$.  So the  winding numbers (\ie the
group homomorphisms) always form integer classes $\Z$. This leads us to believe
that the  winding numbers (or the group homomorphism $\text{MCG}(M^{d+1}) \to
\Zb$) are always realized by the partition function $Z_0(M^{d+1})$ that
contains the gravitational Chern-Simons term $\om_{d+1}$ 
\begin{align}
\label{Zom1}
 Z_0(M^{d+1})\sim \ee^{\ii \ka_{gCS} \int_{M^{d+1}} \om_{d+1} }
\end{align}
where $\dd \om_{d+1}=P_{n_1n_2\cdots}$ is a combination of Pontryagin classes
which are the only integer characteristic classes of oriented manifolds.  If the
values of $P_{n_1n_2\cdots}$ on mapping tori are not always non-zero, then
$\ka_{gCS}$ is quantized since we require
\begin{align}
 \ee^{\ii \ka_{gCS} \int_{M^{d+1}\rtimes_f S^1} P_{n_1n_2\cdots} } =1
\end{align}
for any mapping torus $M^{d+1}\rtimes_f S^1$.  In this case,
$\frac{\ka_{gCS}}{2\pi} \int_{M^{d+1}\rtimes_f S^1} P_{n_1n_2\cdots}$ is the
winding number for the loop in  $\cM_{M^{d+1}}$ described by the mapping torus
$M^{d+1}\rtimes_f S^1$.

Such type of winding numbers and the  partition function exist only when
$d+1=4k+3$.  We also note that there is always one and only one combination of
Pontryagin classes for each $d+1=4k+3$ whose value on mapping torus is always
zero.  (They correspond to the signature $\si$ of the manifold.) For such
Pontryagin class, the corresponding gravitational Chern-Simons term
$\om_{d+1}^\si$ can have a unquantized coefficient.

Clearly, two bosonic systems that give rise to partition functions with
different winding numbers must belong to two different phases.  So the  winding
numbers of partition functions are a type of topological invariants that can be
used to probe and measure the closed \lBF{d+1}  categories.

To gain a better understanding of what part of the \lBF{d+1}  categories 
that the winding numbers characterize, we note
that invertible topological order are described by partition functions that are
pure $U(1)$ phase. In particular the $\Zb$-class of invertible topological
order (see Section \ref{invTop}), such as the $E_8$ quantum Hall state in
$d+1=3$, are described by
\begin{align}
 Z_0(M^{d+1}) = \ee^{\ii \ka_{gCS} \int_{M^{d+1}} \om_{d+1} }
\end{align}
and
$\ka_{gCS}$ is quantized since we require
\begin{align}
 \ee^{\ii \ka_{gCS} \int_{M^{d+2}} P_{n_1n_2\cdots} } =1
\end{align}
for any closed $M^{d+2}$.
% that satisfy $\prt M^{d+2}=M^{d+1}$.  
We note that even $\om_{d+1}^\si$ is required to have a quantize coefficient in
order to be diffeomorphic invariant.  For example, in 2+1D
\begin{align}
\label{Zom2}
 Z_0(M^{3}) = 
\ee^{\ii \ka_{gCS} \int_{M^{d+1}} \om_{3} }
=\ee^{\ii \frac{2\pi c}{24} \int_{M^{d+1}} \om_{3} }
\end{align}
where $c\equiv  12\ka_{gCS}/\pi$ must be quantized as 0 mod 8.  In fact $c$ is
the chiral central charge of the edge states and the above partition function
describes the stacking of $c/8$  $E_8$ quantum Hall states.

We note that the group homomorphism $\text{MCG}(M^{d+1}) \to \Zb$ is additive
under the stacking $\boxtimes$ operation. By comparing \eqn{Zom1} and
\eqn{Zom2}, we find that
\begin{thm} 
For any \lBF{d+1} category $\EC_{d+1}$, there always exists
an invertible \lBF{d+1} category $\EC^\text{invertible}_{d+1}$, such that
the partition function $Z_0(\cdot)$ for the combined
\lBF{d+1} category $\EC_{d+1}\boxtimes \EC^\text{invertible}_{d+1}$
has vanishing winding numbers.
\end{thm} \noindent
We like to stress that having vanishing winding numbers does not imply the
partition function must be constant locally.  In fact, the $E_8$ quantum Hall
state is an example that the  partition function has zero winding numbers
(since the Pontryagin number for $p_1$ is always zero for mapping torus), but
the partition function is not a constant due to the non-zero thermal Hall
effect.

\subsection{How to probe and measure the exact \lBF{}  categories described by
non-fixed-point path integrals}  \label{sec:detect-excitation}

We know that an exact \lBF{} can be described by a topological path integral
that is independent of retriangulation of space-time and independent of local
change of space-time metrics (\ie $Z_0(M^{d+1})$ is constant on $\cM_{M^{d+1}}$
locally).  Such a topological path integral is a fixed-point of the
renormalization group transformation.

For a non-fixed-point path integral that describes an exact \lBF{}  category,
it will flow to a fixed-point path integral under renormalization group
transformation.\cite{LN0701,GW0931} Since the renormalization group
transformation change the volume of the space-time, the fixed-point path
integral has no volume dependent part and correspond to the volume-independent
partition function $Z_0(M^{d+1})$.  The fixed-point path integral should be
closely related to the topological path integral:
\begin{conj} 
The topological path integral that describes
an exact \lBF{} category coincide with the
volume-independent part $Z_0(M^{d+1})$ of the partition function that realizes
the \lBF{} category.  Thus we can use the volume-independent part of the
partition functions, $Z_0(M^{d+1})$, to probe the topological orders described
by exact \lBF{} categories.
\end{conj}\noindent
This conjecture has lead to some related researches and is confirmed for simple
exact \lBF{} categories.\cite{HW1339,MW1418,HMW1457} Since the topological
path integral is re-triangulation invariant, we see that $Z_0(M^{d+1})$ is not
only independent of volume, it is also independent of shape. It only depends on
the topology of $M^{d+1}$.  Therefore, the topological partition function
$Z_0(M^{d+1})$ is a topological invariant for $d+1$-manifold $M^{d+1}$, and
different \lBF{d+1} categories give different topological invariants for
$M^{d+1}$.  In 2+1D,  the topological invariants from exact \lBF{3} categories
are the Turaev-Viro invariants for 3-manifolds.\cite{TV9265}

We have introduced several related concepts, non-fixed point partition
functions $Z(M^{d+1})$ of local bosonic path integral, volume-independent
partition functions $Z_0(M^{d+1})$, topological partition functions
$Z_\text{top}(M^{d+1})$ (assumed to be stable here), and closed/exact \lBF{d+1}
categories.  We will summarize their relations here.  We first note that all of
them are monoids under the stacking operation $\boxtimes$. Thus we can describe
their relations using surjective monoid homomorphisms
\begin{align}
& \text{non-fixed point partition functions
$Z(M^{d+1})$} 
\twoheadrightarrow
\nonumber\\
& \text{volume-independent partition
functions $Z_0(M^{d+1})$}
\twoheadrightarrow
\nonumber\\
&\text{closed \lBF{d+1} categories} 
\end{align}
The reduction from non-fixed point partition functions $Z(M^{d+1})$ to
volume-independent partition functions $Z_0(M^{d+1})$ is the renormalization
group flow.  Volume-independent partition functions $Z_0(M^{d+1})$ may have
non-zero winding numbers that force them to have a non-trivial dependence on the
metrics of space-time (via the gravitational Chern-Simons terms).
The relation between volume-independent partition functions $Z_0(M^{d+1})$
and the closed \lBF{d+1} categories is many-to-one.

We also have a short exact sequence
\begin{align}
\label{ZElBF}
1\to & \{
\text{
$ W^{\chi(M)}
\ee^{\ii \sum_{\{n_i\}} \phi_{n_1n_2\cdots} \int_M P_{n_1n_2\cdots}}$
} \}
 \to
\nonumber\\
&\text{topological partition functions
$Z_\text{top}(M^{d+1})$} \to
\nonumber\\
&\text{exact \lBF{d+1} categories} \to 1
\end{align}
The relation between volume-independent partition functions $Z_0(M^{d+1})$ and
exact \lBF{d+1} categories is one-to-one only if we mod out the factor like
$W^{\chi(M)} \ee^{\ii \sum_{\{n_i\}} \phi_{n_1n_2\cdots} \int_M
P_{n_1n_2\cdots}}$.
Last, we have
\begin{align}
&1\to \text{topological partition functions
$Z_\text{top}(M^{d+1})$} \to
\nonumber\\
& \text{volume-independent partition
functions $Z_0(M^{d+1})$.}  \nonumber
\end{align}

\subsection{How to probe and measure the closed \hBF{}  categories }

A closed \lBF{d+1}  category is described by a local bosonic path integral that
is required to be well defined for arbitrary space-time $M^{d+1}$, while a
closed \hBF{d+1}  category is described by a local bosonic Hamiltonian that is
required to be well defined for arbitrary space $\Si^d$.  
%As a result, we only
%require the \hBF{d+1}  category to have a well defined  path integral for
%space-time which are mapping tori $\Si^d\rtimes S^1$.
Since closed \hBF{d+1}  categories are gapped, we require the Hamiltonian on
a closed space $\Si^d$ to be gapped, whose degenerate ground states form a
finite dimensional vector space $V$ which is a subspace of the total Hilbert space $H_{\Si^d}$ of the boson system.  Let $\cM_{\Si^d}$ be the moduli space for
closed space $\Si^d$ with different metrics and $\cM$ the disjoint union of these moduli spaces.  We see that we have a ground-state
vector space $V$ for every point in $\cM_{\Si^d}$. Therefore, for each $\Si^d$, a closed
\lBF{d+1} category gives us a complex vector bundle on $\cM_{\Si^d}$, which is a sub-bundle of the trivial bundle $\cM_{\Si^d} \times H_{\Si^d}$.
\begin{conj} 
The complex vector bundle of degenerate ground states on $\cM$ 
may fully characterize the closed \hBF{d+1} category.
\end{conj}\noindent
We note that, $\pi_1(\cM_{\Si^d})=\text{MCG}(\Si^{d})$.  Along a loop $g$ in
$\pi_1(\cM_{\Si^{d}})$, the fiber bundle gives us a monodromy $U(g)$ which is a
unitary matrix acting on the ground state vector space $V$.  We may view $g$ as
an element in the group $\text{MCG}(\Si^{d})$.  So $U(g)$ gives an projective
representation of $\text{MCG}(\Si^{d})$.

To understand why we only get a projective representation, we note that the
topological robustness of the ground state degeneracy implies that the unitary
matrix $U_0$ for contractible loop must be a pure over-all phase (which can be
path dependent),  so that $U_0$ cannot distinguish (or split) the degenerate
ground states.  Similarly, $U(g)$ may also have a path-dependent over-all
phase, which leads to the projective representation of $\text{MCG}(\Si^{d})$.
We also like to mention that the trace of $U(g)$ is the partition function on
the corresponding mapping torus:
$$
 \Tr \, U(g) =Z_0(\Si^{d} \rtimes_g S^1) .
$$
As a result, we obtain
$$
 |Z_0(\Si^{d} \times S^1)| =\text{ground state degeneracy on }\Si^d .
$$

For space with different topologies, we will get different projective
representations.  Those finite dimensional projective representations are the
non-Abelian geometric phases of the degenerate ground states introduced in
\Ref{Wrig,KW9327}.  Certainly, the non-Abelian geometric phases contain more
information than the projective representations.  They contain all the
information about the vector bundle $\cE_{\Si^d}$ on $\cM_{\Si^d}$, and thus fully characterize the
closed \hBF{d+1} category.

If the vector bundle $\cE_{\Si^d}$ is not flat, the partition function on mapping torus
$Z_0(\Si^{d} \rtimes_g S^1)$ cannot be topological.  It will depend on the
metrics of the space-time $\Si^d \rtimes_g S^1$.  It is very strange since the
bosonic system has short range correlation and a finite energy gap.  In the
thermal dynamical limit, the space-time becomes flat, and bosonic system should
not be able to sense the geometry of the space-time.  The fact that the
partition function does depend on the metrics of the space-time means that the
entanglement in the ground state can still sense the geometry of the space in
the flat limit.  We like to link such a geometry sensitivity to the gapless
nature of boundary excitations and entanglement spectrum:
\begin{conj} 
\label{flatVB}
The boundary of a closed \hBF{d+1} category is gappable iff the ground state
vector bundle $\cE_{\Si^d}$ over $\cM_{\Si^d}$ is flat.
\end{conj}\noindent

What is the obstruction that prevent the vector bundle to be flat?  First, for a
contractible loop $g$, $U(g)$ is a pure $U(1)$ phase. So the
non-flat part is only contained in the $U(1)$ phase of the complex vector
bundle.  We can examine it by considering the determinant bundle $\cE^\text{det}_{\Si^d}$ of the vector
bundle $\cE_{\Si^d}$, which is a complex line bundle over $\cM_{\Si^d}$.
Let us consider closed submanifold $B \subset \cM_{\Si^d}$.
Then, Chern number of the line bundle $\cE^\text{det}_{\Si^d}$ on $B$ should be given by a
certain Pontryagin number on $\Si^d\rtimes B$:
\begin{align}
\label{ChPon}
 \int_B C = \int_{\Si^d\rtimes B} P_{n_1n_2\cdots}
\end{align}
due to some localness consideration.  Here $\Si^d\rtimes B$ is a fiber bundle
with the space $\Si^d$ as the fiber and $B$ as the base manifold.  We see that
the Pontryagin classes in all dimensions could be the obstructions to have a
flat vector bundle $\cE_{\Si^d}$.

Let us consider an example of 2+1D theory whose gravitational response contain
the gravitational Chern-Simons term:
\begin{align}
 Z_0(\Si^2\rtimes S^1) = 
\ee^{\ii \frac{2\pi c}{24} \int_{\Si^2\rtimes S^1} \om_{3} }
\end{align}
where $c$ is the chiral central charge of the edge states.
For such a theory, the Chern number in \eqn{ChPon} is given by
\begin{align}
 \int_{B^2} C = \frac{c}{24} D_g \int_{\Si^2\rtimes B^2} p_1 = \text{integer},
\end{align}
for any surface bundle $\Si^2\rtimes B^2$, where $D_g$ is the ground state
degeneracy on $\Si^2$,  and $g$ is the genus of $\Si^2$.  

Since $\int_{\Si^2\rtimes B^2} p_1 \neq 0$ for some surface bundle, $\int_{B^2}
C \neq 0$ for some $B$ and the vector bundle $\cE_{\Si^d}$ is not flat if
$c\neq 0$.  So the appearance of the gravitational Chern-Simons term implies
the gapless edge excitations.  

It was shown that $\int_{\Si^2\rtimes B^2} p_1 =0$ mod 12 for any orientable
surface bundles.\cite{CFT1275,GMT0759} 
If the genus of the fiber $\Si^2$ is less than 2, then  $\int_{\Si^2\rtimes
B^2} p_1 =0$.\cite{CFT1275,E0612}
If the genus of the fiber $\Si^2$ is greater than 2, then
we can always find a base manifold $B^2$ with a genus equal or less than 111,
such  that there is a surface bundle $\Si^2\rtimes B^2$ with $\int_{\Si^2\rtimes
B^2} p_1 =\pm 12$.\cite{E9815}  Thus 
\begin{thm} 
For a 2+1D gapped quantum liquid (\ie a closed \hBF{3} category), 
the chiral central charge of the edge state is quantized 
as $c D_g/2=$ an integer, for each $g>2$.  
\end{thm}\noindent
\begin{app} 
For a bosonic quantum Hall state with one branch of edge mode (\ie
$c=1$), the ground state degeneracy $D_g$ must be even for $g>2$.  
\end{app} 
\begin{app} 
For closed
\hBF{3} categories with fusion rule $i\otimes j = \oplus_k N^k_{ij} k$, the
ground state degeneracy $D_g$ is given by\cite{BW0932}
\begin{align}
 D_g=\sum_i (N_i N_{\bar i})^{g-1}
\end{align}
where $\bar i$ is the antiparticle of $i$ and the matrix $N_i$ is given by
$(N_i)^k_j=N_{ij}^k$.
For $\nu=1$ bosonic Pfaffian quantum Hall state, we have
\begin{align}
 N_1=\bpm
1 & 0 & 0\\
0 & 1 & 0\\
0 & 0 & 1\\
\epm,\ \
 N_\psi=\bpm
0 & 1 & 0\\
1 & 0 & 0\\
0 & 0 & 1\\
\epm, \ \
 N_\si=\bpm
0 & 0 & 1\\
0 & 0 & 1\\
1 & 1 & 0\\
\epm.
\end{align}
We find that $D_1=3$, $D_2=10$, $D_3=36$, $D_4=136$, $D_5= 528$, \etc.
Therefore the chiral central charge must be quantized as $c=0$ mod $1/2$, which
agrees with $c=3/2$.  We also see that $c D_g/2=$ integer is not valid for
$g=2$. This allows us to prove that 
\begin{cor} for a 4-dimensional orientable
surface bundle $E$ with fiber of genus 2,
$\int_{E} p_1 =0$ mod 24 (or the signature is 0 mod 8).
\end{cor} 
\end{app} 
\begin{app} 
The chiral central charge of invertible \hBF{3} category is quantized as $c=0$
mod 2, since $D_g=1$.  
However, at the moment, we do not know if the minimal chiral central
charge $c=2$ can be realized by an invertible \hBF{3} category.  In
contrast, the chiral central charge of invertible \lBF{3} category is quantized
as $c=0$ mod 8, where the minimal chiral central charge $c=8$ is realized by
the $E_8$ quantum Hall state.  
\end{app} 

If we have a fermionic system in 2+1D, both $\Si^d$ and $\Si^2\rtimes B^2$
should be chosen to be spin manifolds.  In this case $\int_{\Si^2\rtimes B^2}
p_1 =0$ mod 48 for any spin surface bundles.\cite{CFT1275,E0612} We find that 
\begin{thm} 
For fermionic  invertible topological orders,
the chiral central charge is quantized as $c=0$ mod 1/2.  
\end{thm}\noindent
The minimal chiral central charge $c=1/2$ can be realized by $p+\ii p$
superconductor, which contain no non-trivial topological excitations.

Next let us consider bosonic 1+1D topological orders (\ie closed \hBF{2}
categories).  Since $\text{MCG}(S^1)$ is trivial, $\cM_{\Si^1}$ is simply
connected.  Since the Pontryagin classes for circle bundle $S^1\rtimes B$ all
vanishes, the determinant bundle of the vector bundle $\cE^\text{det}_{\Si^1}$ over $\cM_{\Si^1}$ is flat.
Thus the vector bundle $\cE_{\Si^1}$ is flat, and the  vector bundle is trivial since
$\cM_{\Si^1}$ is simply connected.  Therefore, all bosonic closed \hBF{2}
categories are trivial (if we assume that all non-trivial closed \hBF{2}
categories have non-trivial vector bundle $\cE_{\Si^d}$).  
%Similarly, under the same
%assumption, all fermionic closed \hBF{2} categories are trivial.

It appears that the vector bundle $\cE_{\Si^d}$ on $\cM_{\Si^d}$ is a high
resolution characterization of the closed \hBF{d+1} category.  The non-trivial
closed \hBF{d+1} category should lead to a non-trivial vector bundle
$\cE_{\Si^d}$.  On the other hand, since the structure of the vector bundle can
be so rich, it is very likely that not every allowed  vector bundle
$\cE_{\Si^d}$ on $\cM_{\Si^d}$ can be realized by  closed \hBF{d+1} categories.

\subsection{How to probe and measure the exact \hBF{}  categories }

For an exact \hBF{} category, the ground state vector bundle is always flat and
the partition function on mapping torus always topological.  The Conjecture
\ref{flatVB} implies the reverse: a flat vector bundle always correspond to an
exact \hBF{} category.  Also, for a flat vector bundle, the unitary matrices
$U(g)$ form a representation of the mapping class group $\text{MCG}(\Si^{d})$
which fully characterize the flat bundle. Thus
\begin{conj} 
An exact \hBF{d+1} category is fully characterized by a collection of
representations of the mapping class groups $\text{MCG}(\Si^{d})$ for various
spatial topologies.
\end{conj}\noindent
In particular, the representations of $\text{MCG}(\Si^{d})$ can be computed via
the universal wave function overlap\cite{HW1339,MW1418,HMW1457} or tensor
network calculations.\cite{ZGT1251,TZQ1251,ZMP1233,CV1308}

\subsection{How to probe the gravitational anomaly through quasiparticle
statistics}

\begin{figure}[tb]
  \centering
  \includegraphics[scale=0.6]{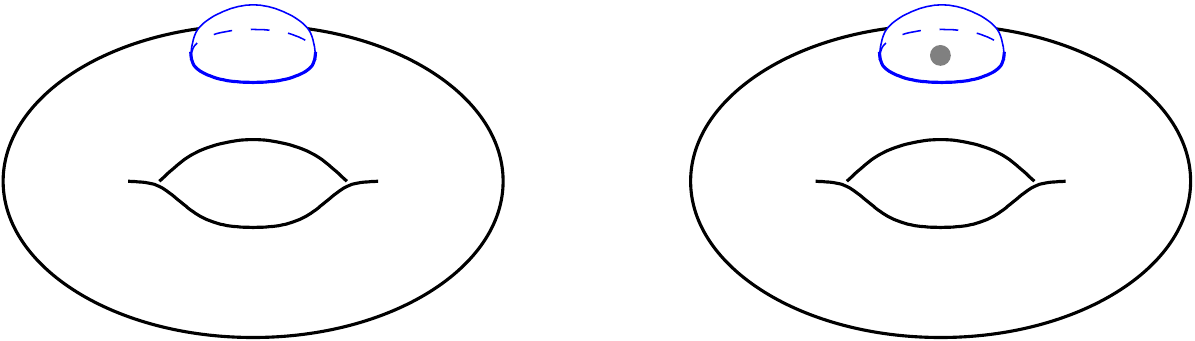}
  \caption{
The thick blue line is a string-like topological excitation in $d$-dimensional
space created at the boundary of a membrane operator.  The 1-dimensional
topological excitations condense on a membrane-like topological excitation
which form a torus. A particle-like topological excitation on the membrane-like
topological excitation can be probed by calculating the average of the membrane
operator with its boundary on the torus which may or may not enclose the
 particle-like topological excitation.
}
  \label{hbrd}
\end{figure}

We also have the following two useful conjectures. The first one is
\begin{conj} 
\label{msta}
$\EC_n$ is  a closed \hBF{n}  category iff\\
(1) any nontrivial pure $p$-dimensional topological excitations in $\EC_n$ can
be detected by their nontrivial mutual braiding properties with some other
topological excitations.\\
(2) any nontrivial pure $p'$-dimensional topological excitations on a
$p$-dimensional topological excitation $M^p$ can be detected by their
nontrivial mutual braiding properties with some other topological excitations
on $M^D$ or by their  different ``mutual half-braiding'' properties with
some other topological excitations in $\EC_n$ which condense on $M^p$.
\end{conj}\noindent
We like to point out that the above conjecture is not fully formulated.  We
state it here just to illustrate an idea.

Here the mutual braiding mean that we fix one topological excitation and move
other  topological excitations around the first one.  A nontrivial mutual
braiding property means that the mutual braiding generate a nontrivial
(non-Abelian) geometric phase. See Section\,\ref{sec:closed-bf-def} 
for a mathematical description. 

Let us explain what is the  ``mutual half-braiding'' property (see Fig.
\ref{hbrd}).\cite{L1309} We know that a $d$-dimensional excitation in the bulk
$\EC_n$ can be created at the boundary of a $(d+1)$-brane operator $\hat
O_{d+1}$.  If a  $d$-dimensional excitation condense on the subspace $M^p$,
then we have $\<\Psi_{M^p,0}|\hat O_{d+1}|\Psi_{M^p,0}\>\neq 0$ if the boundary
of the $(d+1)$-brane operator $\hat O_{d+1}$ lie within $M^p$.  Here
$|\Psi_{M^p,0}\>$ is the wave function of the system where the describe a pure
$p$-dimensional topological excitation on $M^p$.  Let $|\Psi_{M^p,i}\>$ be
the wave function of the system where $M^p$ contains some other topological
excitations.  Then a topological excitation on $M^D$ can be distinguished by
their  different ``mutual half-braiding'' properties  with some other
topological excitations in $\EC_n$ if 
\begin{align}
\frac{\<\Psi_{M^p,i}|\hat
O_{D+1}|\Psi_{M^p,i}\>}{\<\Psi_{M^p,0}|\hat
O_{D+1}|\Psi_{M^p,0}\>}
\neq 1
\end{align}
when the topological excitation $i$ on $M^p$ is enclosed by the boundary of
$\hat O_{p+1}$.

The second one is a generalization of a result by Levin\cite{L1309}:
\begin{conj} An $\EC_n$ is the bulk (or center) of a \hBF{n}  category
$\EC_{n-1}$ iff they satisfy the following condition: all the topological
excitations in $\EC_{n-1}$ can be distinguished by their different mutual
braiding properties with some other topological excitations in $\EC_{n-1}$, or
by their  different ``mutual half-braiding'' properties with some other
topological excitations in $\EC_n$ which condense on the boundary.
\end{conj}
%The above result will be usful if we know how to compute the center of a \hBF{n}  category.

\section{Topological orders that have no
non-trivial topological excitations}
\label{invTop}

As an application of the above conjectures, in this section, we are going to
try to classify a very simple class of topological orders
% in $n$ space-time dimensions 
that has no non-trivial topological excitations in the bulk.  One may wonder,
without any non-trivial topological excitations, such class of topological
orders may only contain the trivial one. In fact, even  without any non-trivial
topological excitations in the bulk, the topological order can still be
non-trivial since the boundary may be non-trivial. The $E_8$ bosonic quantum
Hall state in 2+1D is an example of such kind of topological order, whose
boundary must be gapless.

We know that in 2+1D, the number of point-like topological excitations is equal
to the ground state degeneracy on $T^2$.  In higher dimensions, the ground
state degeneracy on $S^1\times S^n$ and on other spatial topologies are
directly related to the number of point-like and other topological excitations.
Thus we have the following result. 
\begin{thm} 
A H-type topological order (a closed \hBF{n} category) has no non-trivial
elementary topological excitations iff it has no ground state degeneracy on any
closed spaces.
\end{thm} 

For such a topological order $\EC_n$ (a closed $\BF_n$-category), due to the absence of ground state degeneracy, topological partition function $Z(M^n)$ on the space-time $M^n$ (with/without boundaries) is an non-zero $\Cb$-number. Due to unitarity, it must be a pure $U(1)$ phase on any closed
space-time $M^n$ which is a mapping torus for $H$-type theory.  
We have a parallel result for $L$-type. 
\begin{thm} 
A L-type topological order (a closed \lBF{n} category) has no non-trivial
elementary topological excitations iff its topological partition function
$Z(M^n)$ is alway a pure $U(1)$ phase on any closed orientable space-time, up
to a factor $W^{\chi(M^n)}  \ee^{\ii \sum_{\{n_i\}} \phi_{n_1n_2\cdots} \int_{M^n}
P_{n_1n_2\cdots}}$.
\end{thm} 

Given such a topological phase $\EC_n$ without any ground state degeneracy, if
we stack the time-reversed system $\bar{\EC}_n$ on the top of $\EC$, all the
phases are canceled, and we must obtain the trivial topological order, in which
all topological partition functions are $1\in \Cb$. Namely, we must have $\EC_n
\boxtimes \bar{\EC}_n = \one_n$. So such kind of topological orders are {\it
invertible}.  For a generic closed \lBF{n} category $\EC_n$, the inverse of its
partition function, $1/Z_{\EC_n}({M^n})$, may not be the partition function of any
topological order. But when $Z_{\EC_n}({M^n})$ is a pure $U(1)$ phase,
$1/Z_{\EC_n}({M^n})$ will be a partition function of a topological order. In fact
$1/Z_{\EC_n}({M^n})=Z_{\overline{\EC}_n}({M^n})$.  So when the partition function
$Z_{\EC_n}({M^n})$ is a pure $U(1)$ phase, the corresponding topological order
$\EC_n$ is invertible.
An non-zero quantum field theory (L-type theory) with
1-dimensional state spaces is also called {\it invertible} by Freed and
Teleman\cite{FT1292,F14,freed2014}. 

\smallskip
Let us use $C_n$ to denote the $n$-complex obtained by triangulating the
space-time.  Due to locality, we require that, at least for some simple
space-time topologies,  the $U(1)$ phase $Z(M^n)$ comes from the product of
local $U(1)$ phases for each $n$-simplex:\cite{W1313}
%\begin{align}
$$
 Z(M^n)=
\<C_{n},\om_{n}\>
=\prod_{i\in C_n} \<S_n^{(i)},\om_n\>,\ \ \ \
\<S_n^{(i)},\om_n\> \in U(1),
$$
%\end{align}
where $S_n^{(i)}$ is the $i^\text{th}$ $n$-simplex in the complex $C_n$, and
$\om_n$ is a $U(1)$-valued $n$-cochain.  Such a partition function will be
called \emph{local}. In general $\om_n$ may  depend on some local geometric
structures (such as connections and vielbein\cite{Z1253} on the $n$-complex)
that can still affect the gapped ground state.

To find the $n$-cochains $\om_n$ that can describe invertible closed L-type
\lBF{n} categories, let us consider a partition function constructed via the
Pontryagin classes $P_{n_1n_2\cdots}=p_{n_1}\wedge p_{n_2}\wedge \cdots$:
\begin{align}
  Z(M^n)=
\<C_{n},\om_{n}\>
=\ee^{\ii \sum_{n_1n_2\cdots} \phi_{n_1n_2\cdots} \int_{M^n} P_{n_1n_2\cdots} }
\end{align}
Such a partition function is local and is a pure $U(1)$ phase that does not
depend on the volume of the space-time.  So it describes an invertible \lBF{n}
category.  But such an invertible \lBF{n} category is trivial since the
partition function can be continuously deformed to 1.  We see that, although
Pontryagin classes can give rise to local topological partition functions,
since the coefficients $\phi_{n_1n_2\cdots}$ of the Pontryagin classes are not
quantized, they do not give rise to non-trivial invertible \lBF{n} categories.
So a key to obtain non-trivial invertible \lBF{n} categories is to find
topological terms with quantized coefficients.

We note that the cobordism group of 5-dimensional closed oriented manifolds is
$\Om^{SO}_5=\Zb_2$ (see Appendix \ref{cob}).  It was proposed recently in
\Ref{K1459}, that there is a corresponding quantized topological term given by
a Stiefel-Whitney class $w_2\wedge w_3$:
\begin{align}
  Z(M^5)=
\<C_{5},\om_{5}\>
=\ee^{\pi \ii \int_{M^5} w_2\wedge w_3}.
\end{align}
Let us assume that there exists a 4+1D gapped local bosonic theory, integrating
out the matter field will produce the above partition function.  Such a model
realizes a non-trivial exact \lBF{5} category $\EC_5^{L,w_2w_3}$ since the
value of the partition function is a non-trivial $-1$ on $M^5=SU(3)/SO(3)$
(see Appendix \ref{cob}) and the partition function is a topological
invariant.  Such a model also realizes a non-trivial exact \hBF{5} category
$\EC_5^{H,w_2w_3}$ since the value of the partition function is non-trivial
on a 5-dimensional mapping torus $\Cb P^2 \rtimes_* S^1$ generated by the
complex conjugation $*: \Cb P^2\to \Cb P^2$ (see Appendix \ref{cob}).  The
above local topological partition function, being a pure $U(1)$ phase,
describes an invertible  \lBF{5} category $\EC_5^{L,w_2w_3}$ (also an
invertible  \hBF{5} category $\EC_5^{H,w_2w_3}$), which is its own inverse,
i.e. $\EC_5^{L,w_2w_3}\boxtimes \EC_5^{L,w_2w_3} =\one_5$.  We believe that
$\ee^{\pi \ii \int_M w_2\wedge w_3}$ is the only quantized topological term in
5-dimensional space-time. Thus in 4+1D, the invertible \lBF{5} categories form
a group $\Zb_2$. 

The boundary of the exact  \lBF{5} category $\EC_5^{L,w_2w_3}$ gives rise to an
anomalous \lBF{4} category $\EC_4^{L,w_2w_3}$.  The partition function for
$\EC_4^{L,w_2w_3}$ is not gauge invariant on $\Cb P^2$.  Under the complex
conjugation $*: \Cb P^2\to \Cb P^2$, it changes sign $Z_0(\Cb P^2)\to -Z_0(\Cb
P^2)$, since the phase change of the partition function  is given by $\ee^{\ii
\pi \int_{\Cb P^2 \rtimes_* S^1} w_2w_3}=-1$.  This represents a new type of
global gravitational anomaly in a 3+1D bosonic theory.

The above example describes one class of  quantized topological terms, which
leads to one class of invertible \lBF{n} categories.  The partition functions
for this class of topological orders is a topological invariant. Therefore,
this class of topological orders is exact and has gapped boundaries which
contain non-trivial topological excitations.

There is another class of quantized topological terms.  Let us consider a 2+1D
example.  Let $\om_3^{p_1}$ be the three form whose derivative is the first
Pontryagin class: $\dd \om_3^{p_1} =p_1$.  $\om_3^{p_1}$ is a gravitational
Chern-Simons term.\cite{Z1253} We can use $\om_3^{p_1}$ to construct a local
topological partition function integral in 2+1D:
\begin{align}
  Z(M^3)= \<C_{3},\om_{3}\>
=\ee^{\ii \phi \int_{M^3} \om_3^{p_1}}.
\end{align}
However, for some 3-manifold $M^3$, gravitational Chern-Simons term
$\om_3^{p_1}$ is only well defined on patches of $M^3$, with discontinuity
between the patches.  In this case $\int_{M^3} \om_3^{p_1}$ is not well
defined.  Since the cobordism group of 3-dimensional closed oriented manifolds
is $\Om^{SO}_3=0$ (see Appendix \ref{cob}), we can view $M^3$ as a boundary of
$M^4$: $\prt M^4=M^3$, and rewrite the 2+1D topological partition function as
\begin{align}
 Z(M^3)= \<C_{3},\om_{3}\>
=\ee^{\ii \phi \int_{M^4} p_1} .
\end{align}
The above is well defined only if it does not depend on how we extend $M^3$ to
$M^4$.  This requires $\phi$ to be quantized as $\phi=0$ mod $2\pi/3$ (see
Appendix \ref{cob}). So gravitational Chern-Simons term gives rise to a
quantized topological term:
\begin{align}
  Z(M^3)= \<C_{3},\om_{3}\>
=\ee^{2\pi \ii k \int_{M^3} \om_3^{p_1}/3},\ \ \ k\in \Zb .
\end{align}
We see that, in 2+1D, the invertible \lBF{3} categories form a group $\Zb$ .
Such invertible \lBF{3} categories are generated by the $E_8$ quantum Hall
state\cite{PMN1372} (see Example \ref{E8}).

Similarly, using the properties of Pontryagin classes
$p_1$ and $p_2$ in 8-dimensions (see Appendix \ref{cob}):
\begin{align}
 \int_{M^8} \frac{p_1^2-2p_2}{5} \in \Z, \ \ \
 \int_{M^8} \frac{-2p_1^2+5p_2}{9} \in \Z .
\end{align}
we can construct the following topological
partition function in 6+1D:
\begin{align}
Z(M^7) &= \exp\Big(2\pi\ii k_1 \int_{M^7} \frac{\om_7^{p_1^2}-2\om_7^{p_2}}{5} \Big) \times
\nonumber\\
 &\ \ \ \
 \exp\Big(2\pi\ii k_2 \int_{M^7} \frac{-2\om_7^{p_1^2}+5\om_7^{p_2}}{9} \Big) ,
\nonumber\\
\dd \om_7^{p_1^2}&=p_1^2, \ \ \ 
\dd \om_7^{p_2}=p_2,\ \ \ k_1,k_2 \in \Zb.
\end{align}
We see that the invertible \lBF{7} categories
are classified by two integers $(k_1,k_2)$ and form a group $\Z\oplus \Z$.

The partition functions for this second class of topological orders is a
topological invariant up to $U(1)$ phases.  Thus the second class of
topological orders is closed and not exact.  Their boundary must be gapless.

\section{Acknowledgement}

We would like to thank Dan Freed, Zheng-Cheng Gu, Anton Kapustin, Jacob Lurie,
Dmitri Nikshych, Ulrike Tillmann, Kevin Walker, Bei Zeng, and Hao Zheng for
many very useful discussions.  LK is supported by the Basic Research Young
Scholars Program and the Initiative Scientific Research Program of Tsinghua
University, and NSFC under Grant No.  11071134. X-G. W is supported by NSF
Grant No.  DMR-1005541 and NSFC 11274192.  He is also supported by the BMO
Financial Group and the John Templeton Foundation.  Research at Perimeter
Institute is supported by the Government of Canada through Industry Canada and
by the Province of Ontario through the Ministry of Research.

\appendix

\section{Lattice model defined by a path integral} 
\label{path}

\begin{figure}[tb]
\begin{center}
\includegraphics[scale=0.6]{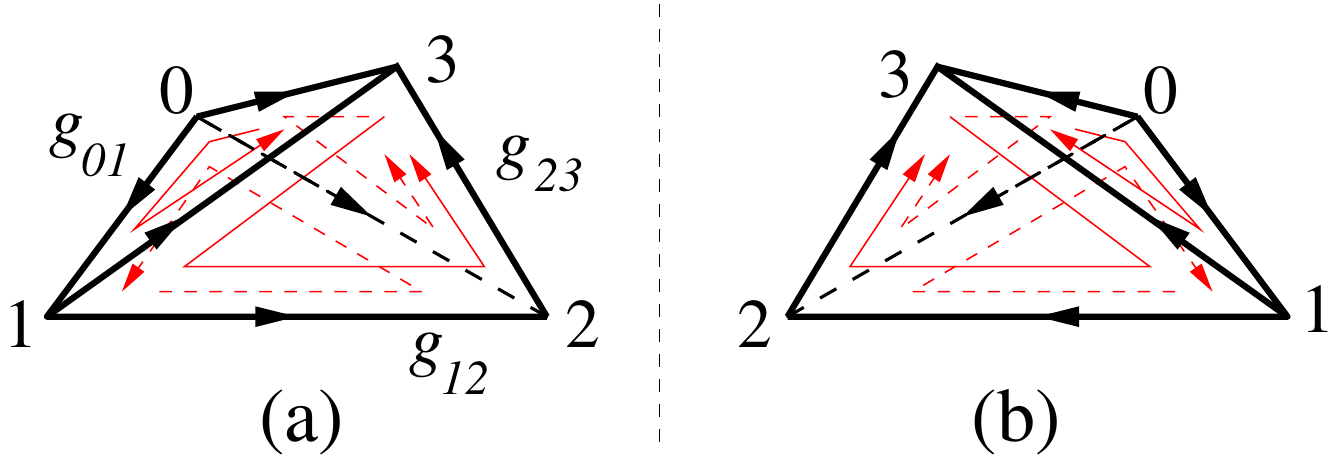} \end{center}
%Fig. 6
\caption{ (Color online) Two branched simplices with opposite orientations.
(a) A branched simplex with positive orientation and (b) a branched simplex
with negative orientation.  }
\label{mir}
\end{figure}

\subsection{Space-time complex}
\label{stcomp}

To define a lattice model through a space-time path integral, we first
triangulate of the $n$-dimensional space-time to obtain a space-time complex
$M_\text{tri}$.  We will call a cell in the space-time complex as a simplex.
In order to define a generic lattice theory on the space-time complex
$M_\text{tri}$, it is important to give the  vertices of each simplex a local
order.  A nice local scheme to order  the  vertices is given by a branching
structure.\cite{C0527,CGL1172,CGL1204} A branching structure is a choice of
orientation of each edge in the $n$-dimensional complex so that there is no
oriented loop on any triangle (see Fig. \ref{mir}).

The branching structure induces a \emph{local order} of the vertices on each
simplex.  The first vertex of a simplex is the vertex with no incoming edges,
and the second vertex is the vertex with only one incoming edge, \etc.  So the
simplex in  Fig. \ref{mir}a has the following vertex ordering: $0,1,2,3$.

The branching structure also gives the simplex (and its sub simplexes) an
orientation denoted by $s_{ij \cdots k}=\pm$.  Fig. \ref{mir} illustrates two
$3$-simplices with opposite orientations $s_{0123}=+$ and $s_{0123}=-$.  The
red arrows indicate the orientations of the $2$-simplices which are the
subsimplices of the $3$-simplices.  The black arrows on the edges indicate the
orientations of the $1$-simplices.

\subsection{Path integral on a space-time complex}
\label{stpath}

The degrees of freedom of our lattice model live on the vertices  (denoted by
$g_i$ where $i$ labels the vertices), on the edges (denoted by $h_{ij}$ where
$ij$ labels the edges), and on other high dimensional cells of the space-time
complex.  The action amplitude $\ee^{-S_\text{cell}}$ for an $n$-cell $(ij
\cdots k)$ is complex function of $g_i$, $h_{ij},\cdots$: $V_{ij \cdots
k}(\{g_i\},\{h_{ij}\},\cdots)$.  The total action amplitude $\ee^{-S}$ for
a configuration $\{g_i\},\{h_{ij}\},\cdots$ (or a path) is given by
\begin{align}
\label{eS}
\ee^{-S}=
\prod_{(ij \cdots k)} [V_{ij \cdots k}(\{g_i\},\{h_{ij}\},\cdots)]^{s_{ij \cdots k}}
\end{align}
where $\prod_{(ij \cdots k)}$ is the product over all the $n$-cells $(ij \cdots
k)$.  Note that the contribution from an $n$-cell $(ij \cdots k)$ is
$V_{ij \cdots k}(\{g_i\},\{h_{ij}\},\cdots)$ or $V^*_{ij \cdots k}(\{g_i\},\{h_{ij}\},\cdots)$
depending on the orientation $s_{ij \cdots k}$ of the cell.
Our lattice theory is defined
by the following imaginary-time path integral (or partition function)
\begin{align}
\label{Zpath}
 Z=\sum_{ \{g_i\},\{h_{ij}\},\cdots }
\prod_{(ij \cdots k)} [V_{ij \cdots k}(\{g_i\},\{h_{ij}\},\cdots)]^{s_{ij \cdots k}}
\end{align}

We would like to point out that, in general, the path integral may also depend
on some additional weighting factors $w_{g_i}$, $A_{g_i,g_j}^{h_{ij}}$, \etc
(see \eq{Z2d} and \eq{Z3d}).  In this section, for simplicity, we will assume
those  weighting factors are all equal to $1$.

Here, we like to introduce an important concept:
\begin{defn}
\textbf{Uniform path integral}\\
In the above path integral \eq{Zpath}, we have assigned the same action
amplitude $V_{ij \cdots k}(\{g_i\},\{h_{ij}\},\cdots)$ to each simplex $(ij
\cdots k)$.  Such a  path integral is called a uniform path integral.
\end{defn} \noindent
In this paper, we only study systems described by uniform path integral.  In
physics, they correspond to ``locally'' translation invariants system in space
and time directions, where the breaking of the exact translation symmetries
only come from the global topology of space-time.  (We also like to point out
that the doubling of unit cell does not break the ``local'' translation
invariance, since after doubling the unit cell, there still some  ``local''
translation invariance left.)

\begin{figure}[tb]
\begin{center}
\includegraphics[scale=0.5]{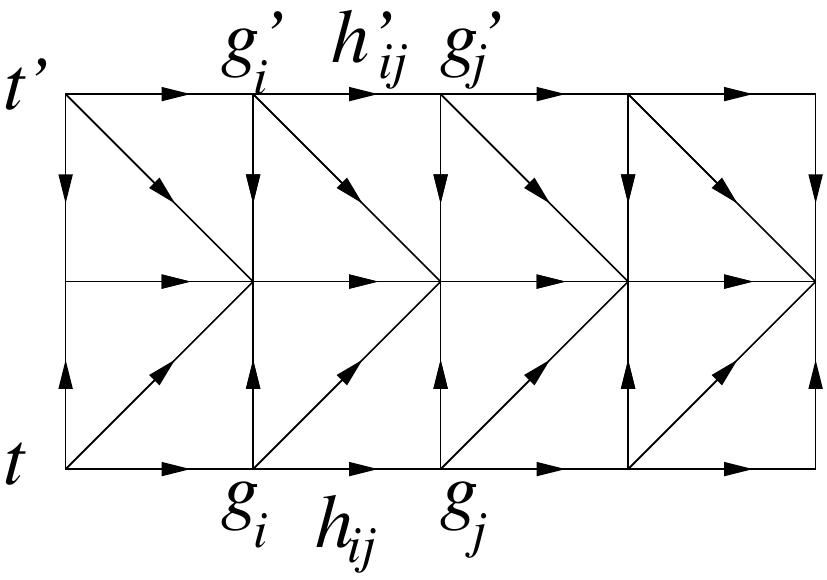} \end{center}
%Fig. 7
\caption{
Each time-step of evolution is given by the path integral on a particular form
of branched graph.  Here is an example in 1+1D. 
}
\label{tStep}
\end{figure}

\begin{figure}[tb]
\begin{center}
\includegraphics[scale=0.45]{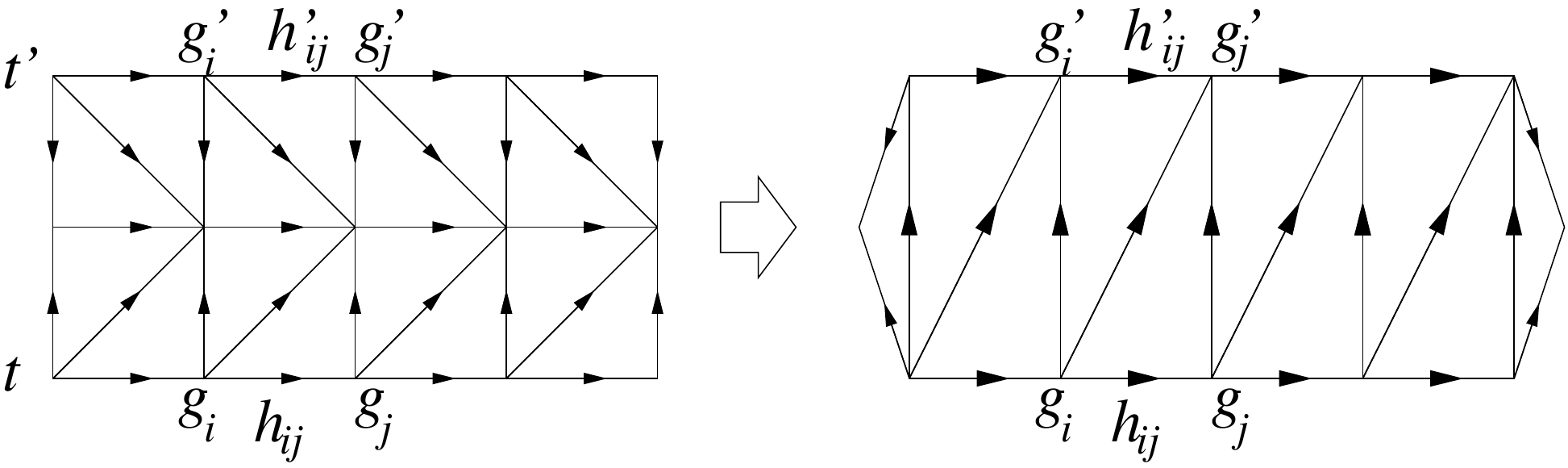} \end{center}
%Fig. 8
\caption{
The reduction of double-layer time-step to single-layer time-step on space with
boundary for an 1+1D topological path integral.
}
\label{stStep}
\end{figure}

\subsection{Path integral and Hamiltonian}
\label{pathham}

Consider a space-time complex of topology $M_\text{space}\times I$ where
$I=[t,t']$ represents the time dimension and $M_\text{space}$ is a closed space
complex (see Fig. \ref{tStep}).  The space-time complex $M_\text{space}\times
I$ has two boundaries: one at time $t$ and another at time $t'$.  A path
integral on the space-time complex $M_\text{space}\times I$ give us an
amplitude $Z[\{ g_i', h_{ij}',\cdots  \}, \{ g_i, h_{ij},\cdots \}]$ from a
configuration $\{ g_i, h_{ij},\cdots \}$ at $t$ to another configuration $\{
g_i', h_{ij}',\cdots  \}$ at $t'$.  Here, $\{ g_i, h_{ij},\cdots \}$ and $\{
g_i', h_{ij}',\cdots  \}$ are the degrees of freedom on the boundaries (see
Fig. \ref{tStep}).  We like to interpret $Z[\{ g_i', h_{ij}',\cdots  \}, \{
g_i, h_{ij},\cdots  \}]$ as the amplitude of an evolution in imaginary time by
a Hamiltonian:
\begin{align}
& \ \ \ \
 Z[\{ g_i', h_{ij}',\cdots  \}, \{ g_i, h_{ij},\cdots \}] 
\nonumber\\
&=\< g_i', h_{ij}',\cdots  | \ee^{-(t'-t)H} |
g_i, h_{ij},\cdots  \> .
\end{align}
However, such an interpretation may not be valid since $ Z[\{ g_i',
h_{ij}',\cdots  \}, \{ g_i, h_{ij},\cdots \}]$ may not give raise to a
Hermitian matrix.  It is a worrisome realization that path integral and
Hamiltonian evolution may not be directly related.

Here we would like to use the fact that the path integral that we are
considering are defined on the branched graphs with a ``reflection'' property
(see \eq{eS}). We like to show that such path integral are better related
Hamiltonian evolution.  The key is to require that each time-step of evolution
is given by  branched graphs of the form in Fig. \ref{tStep}.  One can show
that $Z[\{ g_i', h_{ij}',\cdots  \}, \{ g_i, h_{ij},\cdots \}]$ obtained by summing over all in
the internal indices in the  branched graphs Fig. \ref{tStep}
has a form
\begin{align}
&\ \ \ \
 Z[\{ g_i', h_{ij}',\cdots  \}, \{ g_i, h_{ij},\cdots \}]
\\
&=\sum_{\{ g_i'', h_{ij}'',\cdots  \}}
 U^*[\{ g_i'', h_{ij}'',\cdots  \}, \{ g_i', h_{ij}',\cdots \}]
\nonumber\\
&\ \ \ \ \ \ \ \ \ \ \ \ \ \ \ \ \
 U[\{ g_i'', h_{ij}'',\cdots  \}, \{ g_i, h_{ij},\cdots \}]
\nonumber
\end{align}
and represents a positive-definite Hermitian matrix.  Thus the path integral of
the form \eq{eS} always correspond to a Hamiltonian evolution in imaginary
time.  In fact, the above $Z[\{ g_i', h_{ij}',\cdots  \}, \{ g_i, h_{ij},\cdots
\}]$ can be viewed as an imaginary-time evolution $T=\ee^{-\Del \tau H}$ for a
single time step.

\subsection{Time-reversal transformation}

Consider a lattice model $\La$ described by a space-time path integral defined
by the action amplitude $V^\La_{ij \cdots k}(\{g_i\},\{h_{ij}\},\cdots)$. If we
fold the time direction as Fig. \ref{TConj}, we will get a time-reversal
transformed lattice model $\bar\La$.  The lattice model $\bar\La$ is described
by a different space-time path integral defined by the action amplitude
$V^{\bar\La}_{ij \cdots k}(\{g_i\},\{h_{ij}\},\cdots)$, which is given by 
$$
 V^{\bar\La}_{ij \cdots k}(\{g_i\},\{h_{ij}\},\cdots)
=[V^\La_{ij \cdots k}(\{g_i\},\{h_{ij}\},\cdots)]^*.
$$
The above defines the time-reversal transformation.

\section{Simple and composite \hBF{}  categories}
\label{SCBF}

Consider a $p$-dimensional topological excitation on a $p$-dimensional
subspace $M^p$ of the space $M^d$. We note that the  $p$-dimensional space
$M^p$ can support excitations whose dimensions are less than $p$.  So we can
view the $p$-dimensional topological excitation as a \hBF{p+1} category
(note that $p$ is the space dimension and $p+1$ is the space-time dimension).
This suggests that a $p$-dimensional topological excitation corresponds to a
\hBF{p+1} category.

Since a $p$-dimensional topological excitations can be simple or  
composite, the  \hBF{p+1} categories can also be simple or composite.
Following the definition of the simple and composite topological excitations
in Section \ref{simpcomp},
we can have the following definition of simple and composite \hBF{n}
categories:
\begin{defn} \textbf{Simple/composite  \hBF{n} category}:\\
A \hBF{n} category is simple if its ground state degeneracy on any
closed space is robust against any small perturbations.
Otherwise, the \hBF{n} category is composite.
\end{defn} \noindent

The fractional quantum Hall states and the 2+1D $Z_2$ spin liquid are example
of simple \hBF{3} categories.  To give an example of composite \hBF{3}
category, let us consider a family of Hamiltonian $H(g)$ parametrized by $g$.
The ground state of $H(0)$ is a product state with trivial topological order
and the ground state of $H(1)$ is the  2+1D $Z_2$ spin liquid.  At $g=g_c$
there is a \emph{first order} phase transition between the  product state and
the $Z_2$ spin liquid state.  Then the gapped ground state of $H(g_c)$ (at the
transition point) is an example of composite \hBF{3} category, which can be
expressed as a sum ($\oplus$) of a trivial \hBF{3} category and a 2+1D $Z_2$
topological order.

We see that the composite \hBF{} categories are unstable. 
%We also introduced the notion of simple/composite BF categories.  Simple BF
%categories describe generic gapped ground states whose ground state degeneracy
%is stable against all local perturbations, while composite BF categories
%describe gapped ground states with accidental degeneracy (see Section
%\ref{SCBF}). 
For simplicity, in this paper, we will use ``BF category'' and ``topological
order'' to only refer simple BF category. We will use ``potentially composite
BF category'' to refer the generic BF category that can be simple or
composite.

\section{Examples of BF categories}
\label{app:examples}

\subsection{Examples of exact BF categories (\ie gapped 
phases of qubit models with gapped boundaries)}

\subsubsection{2+1D $\Zb_2$ topological order} 

The 3-dimensional \lBF{3}  category $\EC_3^{\Zb_2}$ in Example \ref{C3Z2} is an
exact \lBF{3}  category.  (Note that an exact \lBF{3}  category is also an
exact \hBF{3}  category.) It has three and only three particle-like topological
excitations labeled by $e$, $v$, and $\eps$.  Those  topological excitations
are their own anti-particles (\ie satisfy a $\Zb_2$ fusion rule).  $e$ and $v$
are bosons, while $\eps$ is a fermion.  Such a  3-dimensional \lBF{3}  category
$\EC_3^{\Zb_2}$ can be realized by a toric code model\cite{K032} in 2+1 dimensions.
As a topological phase, it coincides with the $\Zb_2$-spin-liquid\cite{RS9173,W9164,MS0181}.  
Since \lBF{3}  category corresponds to effective theory in
physics, we write a \lBF{3}  category as a gapped effective theory.  In fact we
have 
\begin{align}
\EC_3^{\Zb_2} &= \text{2+1D $\Zb_2$ gauge theory}, 
\end{align}
where $e$ is the $\Zb_2$ charge, $v$ the $\Zb_2$ vortex, and $\eps$ the bond
state of $e$ and $v$.  Note that $\EC_3^{\Zb_2}$ also correspond to a
$U(1)\times U(1)$ Chern-Simons theory in \eqn{csK} (see
\Ref{BW9045,R9002,FK9169,WZ9290,BM0535,KS1193})
\begin{align}
\EC_3^{\Zb_2} 
&= U(1)\times U(1) \text{ Chern-Simons theory }
\nonumber\\
&\ \ \ \  \text{ with } K\text{-matrix }
K=\begin{pmatrix}
 0&2\\
 2&0\
\end{pmatrix}
\end{align}
where $e$ is the unit-charge of the first $U(1)$ and $v$ the unit-charge of the
second $U(1)$.

As an exact \lBF{3} category, $\EC_3^{\Zb_2}$ must a center of some \lBF{2}
category. In fact $\EC_3^{\Zb_2}$ can be a center of the $\EC_2^{F\Zb_2}$
category discussed in Example \ref{C2FZ2}.

\subsubsection{Double semion model} 

The 3-dimensional \lBF{3}  category $\EC_3^{\Zb_2ds}$ in Example \ref{C3Z2ds}
is another exact \lBF{3}  category.  It has three and only three particle-like
topological excitations labeled by $e$, $v$, and $\eps$.  Those  topological
excitations are their own anti-particles.  $e$ and $v$ are independent semions
with statistics $\pm \pi/2$, while $\eps$ is the bound state of $e$ and $v$ and
is a boson.  Such a  3-dimensional \lBF{3}  category $\EC_3^{\Zb_2ds}$ can be
realized by the so called double-semion string-net model\cite{LWstrnet} or a
double-layer $(2,-2,0)$ fractional quantum Hall state\cite{H8375} in 2+1
dimensions. In fact $\EC_3^{\Zb_2ds}$ is a $U(1)\times U(1)$ Chern-Simons
theory described in \eqn{csK} (see
\Ref{BW9045,R9002,FK9169,WZ9290,BM0535,KS1193})  
\begin{align}
\EC_3^{\Zb_2ds} &= U(1)\times U(1) \text{ Chern-Simons theory }
\nonumber\\
&\ \ \ \  \text{ with } K\text{-matrix }
K=\begin{pmatrix}
 2&0\\
 0&-2\\
\end{pmatrix}
, 
\end{align}
where $e$ is the unit-charge of the first $U(1)$ and $v$ the unit-charge of the
second $U(1)$.
% $\EC_3^{\Zb_2ds}$ can be a center of the $\EC_2^{F\Zb_2}$
%category discussed in Example \ref{C2FZ2}.

\subsubsection{3+1D $\Zb_2$ topological order} 

The 4-dimensional \lBF{4}  category $\EC_4^{\Zb_2}$ in Example \ref{C4Z2} is
also an exact \lBF{4}  category.  It has one particle-like topological
excitation denoted by $e$ and one string-like topological excitation denoted by
$s$, and no other topological excitations.  Such a 4-dimensional \lBF{4}
category $\EC_4^{\Zb_2}$ can be realized by a $\Zb_2$-spin-liquid\cite{HZW0507}
in 3+1 dimensions. In fact  \begin{align} \EC_4^{\Zb_2} = \text{3+1D $\Zb_2$
gauge theory}, \end{align} where $e$ is the $\Zb_2$ charge and $s$ the $\Zb_2$
vortex-line.

All the above topological states can have a gapped boundary. Thus they are
exact \lBF{n}  categories. They are also exact \hBF{n}  categories, since every
exact \lBF{n}  category is an exact \hBF{n}  category.

\subsection{Examples of closed BF categories (\ie gapped phases of qubit
models)}

\subsubsection{$\nu=1/2$ bosonic Laughlin state} 
\label{nuhalf}

The 3-dimensional \hBF{3}  category $\EC_3^{F\Zb_2s}$ in Example \ref{C3FZ2s} is a
closed \hBF{3}  category.  It has only one particle-like topological excitation
labeled by $e$, which is its own anti-particles and has a semion statistics.
Such a 3-dimensional \hBF{3}  category $\EC_3^{F\Zb_2s}$ can be realized by a
filling-fraction $\nu=1/2$ fractional quantum Hall state, the Laughlin state,
in 2+1 dimensions. In fact
\begin{align}
\EC_3^{F\Zb_2s} &= U(1) \text{ Chern-Simons theory }
\nonumber\\
&\ \ \ \  \text{ with } K\text{-matrix }
K=\begin{pmatrix}
 2\\
\end{pmatrix}
, 
\end{align}
where $e$ is the unit-charge of the $U(1)$.

The closed \hBF{3}  category $\EC_3^{F\Zb_2s}$ illustrates the Conjecture
\ref{msta}.  Every topological excitation in  $\EC_3^{F\Zb_2s}$ (which is the
semion $e$) has a nontrivial mutual statistics with at least one other
topological excitation (which is also $e$).  According to the Conjecture
\ref{msta},  the \hBF{3}  category $\EC_3^{F\Zb_2s}$ should be closed.

We like to point out that 2+1D topological theory with the semion as the only
type of topological excitation is a closed \hBF{3}  category, but it is not
closed \lBF{3}  category.  It is an anomalous \lBF{3}  category.  This is
because the theory has a L-type gravitational anomaly. It cannot be defined as a
lbL system in 2+1D, because the lbL system is required to be well defined on
space-time with any topology that is orientable.  The theory can only be
defined as a boundary of a lbL system in 3+1D.  So the theory corresponds to a
non-closed (\ie anomalous) \lBF{3}  category.  

In contrast, the theory has no H-type gravitational anomaly.  It can be realized
by a qubit model on a 2D lattice.  Hence, the corresponding \hBF{3}  category
$\EC_3^{F\Zb_2s}$ is a closed \hBF{3}  category.  Note that to be a closed
\hBF{3} category, we only require the path integral representation of the
theory to be well defined on space-time which is a mapping torus.  Thus the same theory can be free of H-type gravitational
anomaly but not free of L-type gravitational anomaly.  For more details, see
Sections \ref{WWmdl} and \ref{WWmdl1}.

\subsubsection{A three-fermion $\Zb_2$ topological state} 

The 3-dimensional \hBF{3}  category $\EC_3^{\Zb_2f^3}$ in Example \ref{C3Z2f3}
is the second closed \hBF{3}  category.  It has three and only three
particle-like topological excitations labeled by $e$, $v$ and $\eps$, which are
their own anti-particles.  All the three topological excitations are fermions
with mutual $\pi$ statistics.  Such a 3-dimensional \hBF{3}  category
$\EC_3^{\Zb_2f^3}$ can be realized by a four-layer fractional quantum Hall
state in 2+1 dimensions. In fact (see \Ref{VS1306})  
\begin{align}
\EC_3^{\Zb_2f^3} &= U^4(1) \text{ Chern-Simons theory }
\nonumber\\
&\ \ \ \  \text{ with } K\text{-matrix }
K=\begin{pmatrix}
 2 & 1 &1 &1 \\
 1 & 2 &0 &0 \\
 1 & 0 &2 &0 \\
 1 & 0 &0 &2 \\
\end{pmatrix}
.
\end{align}
Again, the above topological theory has a L-type gravitational anomaly,
although it has no H-type gravitational anomaly.  Thus it is a closed \hBF{3} 
category but not a closed \lBF{3}  category.

\subsubsection{Gapless edge state and chiral central charge} 

The above two examples are not exact \hBF{3}  categories, since they have gapless
edge excitations that are robust against any local interactions on the edge.
The \hBF{3}  category $\EC_3^{F\Zb_2s}$ has an edge with a chiral central charge
$c_R-c_L=1$, and $\EC_3^{\Zb_2f^3}$ has an edge with a chiral central charge
$c_R-c_L=4$.  The non-zero chiral central charge implies gapless edge states.
In fact, there is a quite direct relation between the chiral central charge and
the statistics of the topological excitations:\cite{K062,Wang10}
\begin{align}
 \frac{1}{\sqrt{ \sum_\al d_\al^2 } }\sum_\al
d_\al^2 \ee^{\ii \th_\al} =\ee^{\ii 2\pi (c_R-c_L)/8}
\end{align}
where $\al$ labels all the particle-like topological excitations (including the
trivial one).  Here $ \th_\al$ is the statistical angle and $d_\al$ the quantum
dimension of the topological excitations. Such a relation can help us to
determine which \hBF{3}  category cannot be exact.  Also, when $c_R-c_L \neq 0$ mod
8, the topological theory will have a L-type gravitational anomaly.  For more
details, see Sections \ref{WWmdl} and \ref{WWmdl1}.

\subsection{Examples of anomalous BF categories (\ie gapped anomalous theories)}

\subsubsection{An anomalous \lBF{2} category as an edge of 2+1D $\Zb_2$ topological state} 
\label{FZ2}

The 2-dimensional \lBF{2}  category $\EC_2^{F\Zb_2}$ in Example \ref{C2FZ2} is an
anomalous \lBF{2}  category.  It has only one particle-like topological excitation
labeled by $e$, which is its own anti-particles and has
a Bose statistics.  Such a 2-dimensional \lBF{2} 
category $\EC_2^{F\Zb_2}$ can be realized by a \emph{boundary} of a
$\Zb_2$-spin-liquid\cite{RS9173,W9164,MS0181,K032} (described by $\EC_3^{\Zb_2}$)
in 2+1 dimensions.  In other words,
$$
\cZ_2(\EC_2^{F\Zb_2}) = \text{ 2+1D } \Zb_2 \text{ gauge theory }
=\EC_3^{\Zb_2},
$$
where $e$ is the $\Zb_2$ charge.  The 2+1D  $\Zb_2$-spin-liquid has two
particle-like topological excitations: the $\Zb_2$ charge $e$ and the $\Zb_2$
vortex $v$.  The 2+1D $\Zb_2$-spin-liquid can have many different kinds of
boundaries.  The boundary created by the condensation of the $\Zb_2$ vortices $v$
realizes the \lBF{2}  category $\EC_2^{F\Zb_2}$, which is the simplest example
of UFC.

\subsubsection{An anomalous \lBF{3} category as a boundary of 3+1D $\Zb_2$ topological state} 

Similarly, the 3-dimensional \lBF{3}  category $\EC_3^{F\Zb_2b}$ in Example
\ref{C3FZ2b} is another anomalous \lBF{3}  category.  It has only one particle-like
topological excitation labeled by $e$, which is its own anti-particles and has
a Bose statistics.  Such a  3-dimensional \lBF{3}  category $\EC_3^{F\Zb_2b}$ can be
realized by the \emph{boundary} of a $\Zb_2$-spin-liquid\cite{HZW0507} (described
by $\EC_4^{\Zb_2}$) in 3+1 dimensions.  In other words,
$$
\cZ_3(\EC_3^{F\Zb_2b}) = \text{ 3+1D } \Zb_2 \text{ gauge theory }
=\EC_4^{\Zb_2},
$$
where $e$ is the $\Zb_2$ charge.  The 3+1D  $\Zb_2$-spin-liquid has a particle-like
and a string-like topological excitations $e$ and $s$.  The boundary created by
the condensation of the string-like topological excitations $s$ realizes the \lBF{3} 
category $\EC_3^{F\Zb_2b}$.

\subsubsection{Another anomalous \lBF{3} category as another boundary of 3+1D $\Zb_2$ topological state} 
\label{C3sFZ2long}

The 3-dimensional \lBF{3}  category $\EC_3^{sF\Zb_2}$ in Example \ref{C3sFZ2} is also
an anomalous \lBF{3}  category.  It has only one string-like topological excitation
labeled by $s$, which satisfies a $\Zb_2$ fusion rule.  Such a  3-dimensional \lBF{3} 
category $\EC_3^{sF\Zb_2}$ can be realized by the \emph{boundary} of a
$\Zb_2$-spin-liquid\cite{HZW0507} (described by $\EC_4^{\Zb_2}$) in 3+1 dimensions.
In other words,
$$
\cZ_3(\EC_3^{sF\Zb_2}) = \text{ 3+1D } \Zb_2 \text{ gauge theory }
=\EC_4^{\Zb_2},
$$
where the above string-like topological excitation $s$ correspond to the $\Zb_2$
vortex line in the $\Zb_2$-spin-liquid.  The boundary created by the condensation
of the particle-like topological excitations $e$ in the $\Zb_2$-spin-liquid
realizes the \lBF{3}  category $\EC_3^{sF\Zb_2}$.

\subsubsection{An anomalous \lBF{3} category as a boundary of 3+1D twisted $\Zb_2$ topological state
with emergent fermions} 

The 3-dimensional \lBF{4}  category $\EC_3^{F\Zb_2f}$ in Example \ref{C3FZ2f} is yet
another anomalous \lBF{4}  category.  It has only one particle-like topological
excitation labeled by $e$, which is its own anti-particles and has a Fermi
statistics.  Such a 3-dimensional \lBF{4}  category $\EC_3^{F\Zb_2b}$ can be realized
by the \emph{boundary} of a twisted $\Zb_2$ string-net
state\cite{LW0316,LWstrnet} (described by the \lBF{4}  category $\EC_4^{t\Zb_2}$) in
3+1 dimensions.  In other words,
$$
\cZ_3(\EC_3^{F\Zb_2f}) = \text{3+1D twisted } \Zb_2 \text{ gauge theory }
=\EC_4^{t\Zb_2},
$$
where $e$ is the $\Zb_2$ charge.  The 3+1D twisted $\Zb_2$ gauge theory is a $\Zb_2$
gauge theory where the $\Zb_2$ charge carries a Fermi statistics.  Such a twisted
$\Zb_2$ gauge theory can emerge from a lattice qubit model.\cite{LW0316,LWstrnet}
The 3+1D  twisted $\Zb_2$-spin-liquid also has a particle-like and a string-like
topological excitations $e$ and $s$.  The boundary created by the condensation
of the string-like topological excitations $s$ realizes the \lBF{3}  category
$\EC_3^{F\Zb_2f}$.

It is very strange to see that a simple theory with one fermion is anomalous.
We like to point out that it is easy to realize a gapped effective theory with
a fermion excitation using a \emph{fermionic} local Hamiltonian system in the
same dimension.  However, we cannot realize a gapped state, whose \emph{only}
type of topological excitations is fermionic, using a bosonic local Hamiltonian
system (\ie a lattice qubit model) in the same dimension.  This implies that,
by definition, a gapped theory with only one type of fermion excitation is
anomalous. Such a theory has to be a boundary of a gapped qubit state in
one-higher dimension.

The anomalous \lBF{3}  category $\EC_3^{F\Zb_2f}$ also illustrate the Conjecture
\ref{msta}.  The fermion $e$ in $\EC_3^{F\Zb_2f}$ has a trivial mutual statistics
with all other topological excitations (which is also $e$).  According to the
Conjecture \ref{msta},  the \lBF{3}  category $\EC_3^{F\Zb_2f}$ should be anomalous.

\subsubsection{An anomalous \lBF{4} category as a boundary of 4+1D $\Zb_2$ topological state} 

The 4-dimensional \lBF{4}  category $\EC_4^{mF\Zb_2}$ in Example \ref{C4mFZ2} is an
anomalous \lBF{4}  category.  It has only one membrane-like topological excitation
labeled by $m$, which satisfies a $\Zb_2$ fusion rule.  Such a  4-dimensional \lBF{4} 
category $\EC_4^{mF\Zb_2}$ can be realized by the \emph{boundary} of a
$\Zb_2$-spin-liquid (described by $\EC_5^{\Zb_2}$) in 4+1 dimensions.  In other
words,
$$
\cZ_4(\EC_4^{mF\Zb_2}) = \text{ 4+1D } \Zb_2 \text{ gauge theory }
=\EC_5^{\Zb_2}.
$$
The 4+1D  $\Zb_2$-spin-liquid $\EC_5^{\Zb_2}$ has a particle-like and a
membrane-like topological excitations, $\Zb_2$ charge particle and $\Zb_2$ vortex
membrane.  The above membrane-like topological excitation $m$ corresponds to
the $\Zb_2$ vortex membrane in $\EC_5^{\Zb_2}$.  The 4+1D $\Zb_2$-spin-liquid can
have many different kinds of boundaries.  The boundary created by the
condensation of the $\Zb_2$ charge particles realizes the \lBF{4}  category
$\EC_4^{mF\Zb_2}$.

\subsubsection{An anomalous \lBF{4} category as a boundary of a 4+1D membrane condensed state} 

The 4-dimensional \lBF{4}  category $\EC_4^{sF\Zb_2}$ in Example \ref{C4sFZ2} is our
last example of anomalous \lBF{4}  category.  It has only one string-like topological
excitation labeled by $s$, which satisfies a $\Zb_2$ fusion rule.  Such a
4-dimensional \lBF{4}  category $\EC_4^{sF\Zb_2}$ can be realized by the
\emph{boundary} of a non-oriented membrane condensed state (described by
$\EC_5^{\Zb_2m}$) in 4+1 dimensions.  In other words, 
$$
\cZ_4(\EC_4^{mF\Zb_2}) = \text{ 4+1D membrane condensed state} =\EC_5^{\Zb_2m} .  
$$ 
Such a 4+1D \lBF{4+1}  category is described by the following
effective Lagrangian
\begin{align}
\label{cs5}
 {\cal L}= \frac{K_{IJ}}{4\pi} 
b_{I\mu\nu} \prt_\la b_{J\rho\si}\eps^{\mu\nu\la\rho\si} ,
\end{align}
with
\begin{align}
K=\begin{pmatrix}
 0&2\\
 2&0\
\end{pmatrix} .
\end{align}
The 4+1D  membrane condensed state $\EC_5^{\Zb_2m}$ has two kinds of string-like
topological excitations labeled by $s_1$ and $s_2$.  The above string-like
topological excitations $s$ corresponds $s_1$ in  $\EC_5^{\Zb_2m}$.  The 4+1D
membrane condensed state can have many different kinds of boundaries.  The
boundary created by the condensation of $s_2$ realizes the \lBF{4}  category
$\EC_4^{mF\Zb_2}$.

The gapped effective theories discussed in this section all have   L-type (as
well as the H-type) gravitational anomalies. They represent global
gravitational anomalies.

\section{A brief introduction to category theory}
\label{CatT}

\newcommand\Ob       {\text{Ob}}
\newcommand\Mor       {\hom}

In this section, we will give a brief introduction to category theory, which is
basically an abstract theory about relations (maps) and the composition of the
relations (maps). The language of the category theory is used in the main text
to define \hBF{n} category as an unitary $n$-category.

\begin{defn}
\label{definition-category}
A 1-{\it category} $\EC$ consists of the following data:
\begin{enumerate}
\item A set of objects $\Ob(\EC)$.
\item For each pair $x, y \in \Ob(\EC)$ a set of morphisms
$\Mor_\EC(x, y)$.  \item For each triple $x, y, z\in \Ob(\EC)$
a composition map $\Mor_\EC(y, z) \times \Mor_\EC(x, y) \to
\Mor_\EC(x, z) $, denoted as 
$(f, g) \mapsto f \circ g$ for $f
\in  \Mor_\EC(x, y),\ g \in  \Mor_\EC(y, z)$.
\end{enumerate}
These data are to satisfy the following rules:
\begin{enumerate}
\item For every element $x\in \Ob(\EC)$ there exists a
morphism $\text{id}_x\in \Mor_\EC(x, x)$ such that
$\text{id}_x \circ \phi = \phi$ and $\psi \circ \text{id}_x = \psi $ whenever
these compositions make sense.
\item Composition is associative, i.e., $(f \circ g) \circ h =
f \circ ( g \circ h)$ whenever these compositions make sense.
\end{enumerate}
\end{defn}

We list a few examples of 1-category below: 
\begin{enumerate}
\item The category $\Set$ of sets consists of sets as objects and maps as 1-morphisms. The composition of 1-morphisms is just the usual composition of maps. The identity morphism is just the identity map. 

\item The category $\vect$ of vector spaces over $\Cb$ consists of vector spaces as objects and linear maps as morphisms, i.e. $\hom_{\vect}(x,x)=\End_\Cb(x)$. 

\item The category $\rep_G$ of representations of a group $G$ consists of representations of the group $G$ as objects and linear maps that intertwine the $G$-action as 1-morphisms. 

\end{enumerate}

\begin{rema}
A morphism $f \in  \Mor_\EC(x, y)$ is also referred as an arrow from
$x$ to $y$, denoted by $x\xrightarrow{f} y$. 
A morphism $f : x \to y$ is called an {\it isomorphism} of the category
$\EC$ if there exists a morphism $g : y \to x$
such that $f \circ g = \text{id}_y$ and
$g \circ f = \text{id}_x$.
%The composition can be written as $(x \to y) \circ (y \to z) = (x\to y \to z) = (x \to z)$.
\end{rema}

\begin{defn}
\label{definition-subcategory}
A {\it subcategory} of a category $\EB$ is
a category $\EA$ whose objects and arrows
form subsets of the objects and arrows
of $\EA$ and such that source, target
and composition in $\EA$ agree with those
of $\EB$. We say $\EA$ is a
{\it full subcategory} of $\EB$ if $\Mor_\EA(x, y)
= \Mor_\EB(x, y)$ for all $x, y \in \Ob(\EA)$.
We say $\EA$ is a {\it strictly full} subcategory of $\EB$
if it is a full subcategory and given $x \in \Ob(\EA)$ any
object of $\EB$ which is isomorphic to $x$ is also in $\EA$.
\end{defn}

\begin{defn}
\label{definition-opposite-category}
For any category $\EC$, the \emph{opposite category} $\EC^\op$ is defined so that $\ob(\EC^\op)=\ob(\EC)$ and $\hom_{\EC^\op}(x,y)=\hom_\EC(y,x)$ for all $x,y\in \ob(\EC)$. 
\end{defn}

\begin{rema}
\label{remark-unique-identity}
It follows directly from the definition that any two identity morphisms
of an object $x$ of $\EA$ are the same. Thus we may and will
speak of {\it the} identity morphism $\text{id}_x$ of $x$.
\end{rema}

\begin{defn}
\label{definition-functor}
A {\it functor} $F : \EA \to \EB$
between two categories $\EA, \EB$ is given by the
following data:
\begin{enumerate}
\item A map $F : \Ob(\EA) \to \Ob(\EB)$.
\item For every $x, y \in \Ob(\EA)$ a map
$F : \Mor_\EA(x, y) \to \Mor_\EB(F(x), F(y))$ such that $f \mapsto F(f)$.
\end{enumerate}
These data should be compatible with composition and identity morphisms
in the following manner: $F(f \circ g) =
F(f) \circ F(g)$ for a composable pair $(f, g)$ of
morphisms of $\EA$ and $F(\text{id}_x) = \text{id}_{F(x)}$.
\end{defn}

\noindent
Note that there is an {\it identity} functor $\text{id}_\EA$ for every category $\EA$. 
In addition, given a functor $G : \EB \to \EC$
and a functor $F : \EA \to \EB$ there is
a {\it composition} functor $G \circ F : \EA \to \EC$
defined in an obvious manner.

\begin{defn}
\label{definition-faithful}
Let $F : \EA \to \EB$ be a functor.
\begin{enumerate}
\item We say $F$ is {\it faithful} if
for any objects $x, y$ of $\Ob(\EA)$ the map
$$
F : \Mor_\EA(x, y) \to \Mor_\EB(F(x), F(y))
$$
is injective.
\item If these maps are all bijective then $F$ is called
{\it fully faithful}.
\item
The functor $F$ is called {\it essentially surjective} if for any
object $y \in \Ob(\EB)$ there exists an object
$x \in \Ob(\EA)$ such that $F(x)$ is isomorphic to $y$ in
$\EB$.
\end{enumerate}
\end{defn}

\begin{defn}
A natural transformation $\beta: F\to G$ between two functors $F,G: \EA \to \EB$ is a set of maps $\{ \beta_x: F(x) \to G(x) \}_{x\in \ob(\EA)}$ such that the following diagram: 
$$
\xymatrix{
F(x) \ar[r]^{\beta_x} \ar[d]_{F(f)} & G(x) \ar[d]^{G(f)} \\
F(y) \ar[r]^{\beta_y} & G(y)
}
$$
is commutative. 
\end{defn}

Let $\one$ be the 1-category with a single object $\bullet$ and a single morphism $\id_\bullet$.

\begin{defn} \label{def:bicat} 
A 2-category $\EE$ consists of a set of objects and a 1-category of morphisms $\hom(A,B)$ for each pair of objects $a$ and $b$ together with 
\begin{enumerate}
\item {\it identity morphism}: there is a functor $\one_a:  \one \to \Mor(a,a)$ for all $a\in \ob(\EE)$.
\item {\it composition functor}: 
\begin{eqnarray}
\circ_{a,b,c}: \hom(a, c) \times \hom(a, b) &\to& \hom(a, c) \nn
(f, g) &\mapsto& f\circ g   \nonumber 
\end{eqnarray}
\item {\it associativity isomorphisms}: for $a, b, c, d\in \ob(\EE)$, there is a natural isomorphism:
$$
\alpha:  \circ_{a,c,d} \circ (\circ_{a,b,c} \times \id) \rightarrow 
\circ_{a,b,d} \circ (\id \times \circ_{b,c,d}). 
$$
\item {\it left and right unit isomorphisms}:
$l: \circ_{a,a,b} \circ (\one_a \times \id ) \to \id$ and 
$r: \circ_{a,b,b} \circ (\id \times \one_b)\to  \id$. 
\end{enumerate}
satisfying the following coherence conditions:
\begin{enumerate}

\item {\it associativity coherence}: 
\begin{equation}  \label{diag:asso-bicat}
\xymatrix{
((e \circ f)\circ g) \circ h \ar[r]^{\alpha 1}  
\ar[d]_{\alpha}  & (e\circ (f \circ g)) \circ h \ar[d]^{\alpha} \\
(e\circ f) \circ (g\circ h) \ar[d]_{\alpha} &   e\circ ((f\circ g) \circ h) \ar[ld]^{1\alpha} \\
e\circ (f \circ (g\circ h)) &  
} 
\end{equation}

\item {\it identity coherence}: 
\begin{equation}  \label{diag:triangle-bicat}
\xymatrix{
(f \circ \one_b) \circ g  \ar[rd]_{r1} \ar[rr]^{\alpha} & & f\circ (\one_b \circ g) \ar[ld]^{1l}  \\
& f\circ g & 
}
\end{equation}
\end{enumerate}
\end{defn}

\begin{rema}
It immediately follows from the axioms of the 2-category that the 1-category $\hom(x,x)$ is a monoidal 1-category with the tensor product given by the composition and the tensor unit by the identity 1-morphism $\id_x:=\one_x(\bullet)$. 
\end{rema}

\section{An extension problem}
\label{extF}

Let $M^n$ be a fiber bundle with $S^1$ as the base and $(n-1)$-dimensional
manifold $F^{n-1}$ as the fiber (\ie $M^n$ is a mapping torus). Assume $M^n$ is
oriented.  Is $M^n$ always a boundary of a manifold $M^{n+1}$, where $M^{n+1}$
is a fiber bundle with $S^1$ as the base and $n$-dimensional manifold $F^n$ as
the fiber.  The answer is no. This is because $F^n$ needs to be the boundary of
$F^{n+1}$ and some manifolds cannot be realized as boundaries of other
manifolds (such as 4-manifolds with non-zero signature).

Now let us consider a slightly different extension problem.  Let $M^n$ be a
fiber bundle with $S^1$ as the base and $(n-1)$-dimensional cell complex
$F^{n-1}$ as the fiber. Assume $M^n$ is oriented.  Is $M^n$ always a boundary
of a cell-complex $M^{n+1}$, where $M^{n+1}$ is a fiber bundle with $S^1$ as
the base and $n$-dimensional cell-complex $F^n$ as the fiber.  The answer is
yes, since we can take $F^{n+1}$ to be the cone on $F^n$.\cite{W1315} 

We also has the following result: the signature is multiplicative in oriented
fiber bundles of odd fiber dimension.\cite{E0919} Combine the above two
results, we conclude that the path integral of Chern-Simons theory can always
be defined in 2+1 dimension where the space-time is a mapping torus. 

\section{The oriented cobordism groups}
\label{cob}

Two oriented smooth $n$-dimensional manifolds $M$ and $N$ are said to be
equivalent if $M\cup (-N)$ is a boundary of another manifold, where $-N$ is the
$N$ manifold with a reversed orientation.  With the multiplication given by the disjoint union, the corresponding equivalence classes has a structure of an Abelian group $\Om^{SO}_n$, which is called the cobordism group
of closed oriented smooth manifolds.  For low dimensions, we have\cite{OCob}\\
$\Om^{SO}_{0}=\Z$, generated by a point.\\
$\Om^{SO}_{1}=0$, since circles bound disks.\\
$\Om^{SO}_{2}=0$, since all oriented surfaces bound handlebodies.\\
$\Om^{SO}_{3}=0$.\\
$\Om^{SO}_{4}=\Z$, generated by $\Cb P^2$.\\
$\Om^{SO}_{5}=\Z_2$, generated by the Wu manifold $SU(3)/SO(3)$,\\
{\white{.}} ~~~~~~ detected by the deRham invariant or\\
{\white{.}} ~~~~~~ Stiefel-Whitney number $\int_M w_2\wedge w_3$.\\
$\Om^{SO}_{6}=0$.\\
$\Om^{SO}_{7}=0$.\\
$\Om^{SO}_{8}=\Z\oplus \Z$ generated by $\Cb P^4$ and $\Cb P^2\times \Cb P^2$.
\\
The free part of the cobordism groups $\Om^{SO}_n$ can be fully detected by the
Pontryagin numbers $P_{n_1n_2\cdots}(M)$.  In particular, we have\cite{S1421}
% (see http://arxiv.org/pdf/1010.3274.pdf)
\begin{align}
P_1(\Cb P^2) &= \int_{\Cb P^2} p_1 =3;
\nonumber\\
P_{1,1}(\Cb P^2\times \Cb P^2) &= \int_{\Cb P^2\times \Cb P^2} p_1^2 =18,  
\nonumber\\
P_{2}(\Cb P^2\times \Cb P^2) &= \int_{\Cb P^2\times \Cb P^2} p_2 =9,  
\nonumber\\
P_{1,1}(\Cb P^4) &=\int_{\Cb P^4} p_1^2 =25,  
\nonumber\\ 
 P_{2}(\Cb P^4) &=\int_{\Cb P^4} p_2 =10 .
\end{align}
We also have the following fundamental theorem:
\begin{thm}
Two closed oriented $n$-manifolds $M_0$ and $M_1$ are cobordism equivalent iff
they have the same Stiefel-Whitney and Pontryagin numbers.
\end{thm} \noindent

%(See http://www.map.mpim-bonn.mpg.de/Oriented\_bordism)

In \Ref{K7659,M7973,B8337}, some cobordism groups $\Omega_n^{MT}$ of
$n$-dimensional mapping tori are obtained
\begin{align}
\Omega_{2k}^{MT} &= 
\Omega_{2k-1}^{SO}\oplus 
\hat \Omega_{2k}^{SO}, \ \ \text{ for } k>2
\nonumber\\
\Omega_{4}^{MT} &= 0,
\end{align}
where $\hat \Omega_{2k}^{SO}$ is the subgroup of  $\Omega_{2k}^{SO}$ with
vanishing signature.
The structure of $\Omega_{2k}^{MT}$ suggests that we can use the following two types of
Stiefel-Whitney and Pontryagin numbers to detect/distinguish the elements
of $\Omega_{2k}^{MT}$:
\begin{align}
 &\int_M P_{n_1n_2\cdots}, & & \int_M \dd \th \wedge P_{n_1n_2\cdots},
\nonumber\\
 &\int_M W_{n_1n_2\cdots}, & & \int_M \dd \th \wedge W_{n_1n_2\cdots},
\end{align}
where $\dd\th$ is the one form on the base manifold $S^1$ which is parametrized
by $\th \in [0,2\pi)$, $P_{n_1n_2\cdots}$ are combinations of Pontryagin
classes: $P_{n_1n_2\cdots}=p_{n_1}\wedge p_{n_2}\wedge \cdots$, and
$W_{n_1n_2\cdots}$ are combinations of Stiefel-Whitney classes:
$W_{n_1n_2\cdots}=w_{n_1}\wedge w_{n_2}\wedge \cdots$ on $M$.

For mapping tori with odd dimensions, the result is more complicated:
for $k>1$, the homomorphism
\begin{align}
\Omega_{2k+1}^{MT} & \to 
\Omega_{2k+1}^{SO}\oplus \hat \Omega_{2k}^{SO}
\oplus W_{(-)^k}(\Zb,\Zb)
\end{align}
is an isomorphism (for $k$ even) or is injective with cokernel $Z_2$ (for $k$
odd).  Here $W_\pm(\Zb,\Zb)\simeq \Zb^\infty \oplus \Zb_2^\infty\oplus \Zb_4^\infty$ is the Witt group of isometries of free
finite-dimensional $\Zb$-modules with a symmetric (antisymmetric) unimodular
bilinear form. We also have\cite{B8337}
\begin{align}
\Omega_{3}^{MT} &= \Zb^\infty \oplus \Zb_2^\infty,
\end{align}
We note that  $\Omega_{5}^{MT}$ contains $\Omega_{5}^{SO}$.  In fact
$\Omega_{5}^{SO}$ is also generated by the mapping torus of the
complex-conjugation-map $\Cb P^2\to \Cb P^2$.\cite{K1414} This implies that the
invertible \hBF{5} categories contain a $\Zb_2$ class, which is also the
$\Zb_2$ class for invertible \lBF{5} categories.

%M. Kreck, (Bordism of diffeomorphisms Bull. Amer. Math. Soc. Volume 82, Number
%5 (1976), 641-789, with details in Cobordism of odd-dimensional
%diffeomorphisms. Topology 15 (1976). 353-361) 

\vfill\eject

\bibliography{./wencross,./all,./publst,./local} 

\end{document}